%% file: main.tex
\PassOptionsToPackage{hyphens}{url}
\documentclass[a4paper, amsfonts, amssymb, amsmath, reprint, showkeys, nofootinbib, twoside, superscriptaddress]{revtex4-1}
\usepackage[english]{babel}
\usepackage[utf8]{inputenc}
\input{preamble}
\usepackage[pdftex, pdftitle={Article}, pdfauthor={Author}]{hyperref} % For hyperlinks in the PDF

\makeatletter
\def\@fnsymbol#1{\ensuremath{\ifcase#1\or \dagger\or *\or \ddagger\or
\mathsection\or \mathparagraph\or \|\or **\or \dagger\dagger
\or \ddagger\ddagger \else\@ctrerr\fi}}
\makeatother

\begin{document}
\title{Requirements for a processing-node quantum repeater on a real-world fiber grid}

\author{Guus Avis}
    \thanks{These authors contributed equally.}
    \affiliation{QuTech, Delft University of Technology, Lorentzweg 1, 2628 CJ Delft, The Netherlands}
    \affiliation{Kavli Institute of Nanoscience, Delft University of Technology, Lorentzweg 1, 2628 CJ Delft, The Netherlands}
\author{Francisco Ferreira da Silva}
    \thanks{These authors contributed equally.}
    \affiliation{QuTech, Delft University of Technology, Lorentzweg 1, 2628 CJ Delft, The Netherlands}
    \affiliation{Kavli Institute of Nanoscience, Delft University of Technology, Lorentzweg 1, 2628 CJ Delft, The Netherlands}
\author{Tim Coopmans}
    \affiliation{QuTech, Delft University of Technology, Lorentzweg 1, 2628 CJ Delft, The Netherlands}
    \affiliation{Kavli Institute of Nanoscience, Delft University of Technology, Lorentzweg 1, 2628 CJ Delft, The Netherlands}
\author{Axel Dahlberg}
    \affiliation{QuTech, Delft University of Technology, Lorentzweg 1, 2628 CJ Delft, The Netherlands}
    \affiliation{Kavli Institute of Nanoscience, Delft University of Technology, Lorentzweg 1, 2628 CJ Delft, The Netherlands}
\author{Hana Jirovsk\'{a}}
    \affiliation{QuTech, Delft University of Technology, Lorentzweg 1, 2628 CJ Delft, The Netherlands}
    \affiliation{Kavli Institute of Nanoscience, Delft University of Technology, Lorentzweg 1, 2628 CJ Delft, The Netherlands}
\author{David Maier}
    \affiliation{QuTech, Delft University of Technology, Lorentzweg 1, 2628 CJ Delft, The Netherlands}
    \affiliation{Kavli Institute of Nanoscience, Delft University of Technology, Lorentzweg 1, 2628 CJ Delft, The Netherlands}
\author{Julian Rabbie}
    \affiliation{QuTech, Delft University of Technology, Lorentzweg 1, 2628 CJ Delft, The Netherlands}
    \affiliation{Kavli Institute of Nanoscience, Delft University of Technology, Lorentzweg 1, 2628 CJ Delft, The Netherlands}
\author{Ariana Torres-Knoop}
    \affiliation{SURF Utrecht,  Postbus 19035, 3501 DA Utrecht, The Netherlands}
\author{Stephanie Wehner}
    \thanks{Corresponding author: \href{mailto:s.d.c.wehner@tudelft.nl}{s.d.c.wehner@tudelft.nl}}
    \affiliation{QuTech, Delft University of Technology, Lorentzweg 1, 2628 CJ Delft, The Netherlands}
    \affiliation{Kavli Institute of Nanoscience, Delft University of Technology, Lorentzweg 1, 2628 CJ Delft, The Netherlands}
\date{\today} % Leave empty to omit a date

\begin{abstract}
We numerically study the distribution of entanglement between the Dutch cities of Delft and Eindhoven realized with a processing-node quantum repeater
and determine minimal hardware requirements for verifiable blind quantum computation using color centers and trapped ions.
Our results are obtained considering restrictions imposed by a real-world fiber grid and using detailed hardware-specific models.
By comparing our results to those we would obtain in idealized settings we show that
simplifications lead to a distorted picture of hardware demands, particularly on memory coherence and photon collection.
We develop general machinery suitable for studying arbitrary processing-node repeater chains using NetSquid, a discrete-event simulator for quantum networks.
This enables us to include time-dependent noise models and simulate repeater protocols with cut-offs, including the required classical control communication.
We find minimal hardware requirements by solving an optimization problem using genetic algorithms on a high-performance-computing cluster.
Our work provides guidance for further experimental progress, and showcases limitations of studying quantum-repeater requirements in idealized situations.
\end{abstract}

%\keywords{first keyword, second keyword, third keyword}

\maketitle

\section{Introduction} \label{sec:introduction}
Quantum communication unlocks network applications that are provably impossible to realize using only classical communication.
One striking example is secure communication using quantum key distribution~\cite{ekert1991quantum, bennett2020quantum}, but many other applications are already known.
Examples of these are secret sharing~\cite{hillery1999quantum} and clock synchronization~\cite{komar2014quantum}.
Several stages of quantum network development have been identified~\cite{wehner2018quantum}, where a higher stage of network development offers the potential to execute ever more advanced quantum network applications, at the expense of making higher demands on the end nodes running applications, as well as on the network that connects the end nodes.

Efficiently distributing quantum states over long distances is an outstanding technological challenge. 
Direct photon transmission that is used to carry quantum information over optical fibers is subject to a loss that is exponential in the length of the fiber.
Quantum repeaters~\cite{briegel1998quantum, dur1999quantum} promise to enable quantum communication over global distances, mitigating the loss in the fiber by the introduction of intermediary nodes.
A variety of different repeater platforms have been proposed (see, e.g.,~\cite{munro2015inside,muralidharan2016optimal}) including repeaters featuring quantum memories such as atomic ensembles~\cite{sangouard2011quantum, duanLongdistanceQuantumCommunication2001} %, atoms % not sure what we're supposed to cite for atoms~\cite{recentWeinfurterExperiment},
or processing nodes~\cite{rozpedek2019near, duan2010colloquium, uphoff2016integrated} that are capable not only of storing quantum information but also of performing quantum gates.
Examples of processing nodes include trapped ions~\cite{monroe2013, reiserer2015a}, neutral atoms~\cite{reiserer2015a, langenfeld2021a}, or color centers such as nitrogen-vacancy (NV), silicon-vacancy (SiV) or tin-vacancy (SnV) centers in diamond~\cite{ruf2021quantum}.
Despite proof-of-principle demonstrations of repeater nodes~\cite{bhaskar2020experimental, langenfeld2021a}, as well as entanglement swapping via an intermediary processing node~\cite{pompili2021realization}, at present no quantum repeater has been realized that bridges long distances.

Part of the challenge in building quantum repeaters is that their hardware requirements remain largely unknown.
Extensive studies have been conducted to estimate such requirements both analytically (see, e.g.,~\cite{amirloo2010quantum, asadi2018quantum, bernardes2011rate, borregaard2015heralded, bruschi2014repeat, chen2007fault, collins2007multiplexed, guha2015rate, hartmann2007role, jiang2009quantum, nemoto2016photonic, razavi2009quantum, razavi2006long, simon2007quantum, vinay2017practical, wu2020near, sangouard2007long, sangouard2008robust, borregaard2020one, luong2016overcoming, rozpedek2018parameter, rozpedek2019near, vanloock2020, kamin2022}), as well as using numerical simulations (see, e.g.,~\cite{abruzzo2013quantum, brask2008memory, muralidharan2014ultrafast, pant2017rate, ladd2006hybrid, van2006hybrid, zwerger2017quantum, jiang2007fast, wu2021b}).
While greatly informative in helping us understand minimal hardware requirements needed to bridge long distances, they have mostly been conducted in idealized settings where all repeaters are equally spaced, and one assumes a uniform loss of typically $0.2$ dB km$^{-1}$ on each fiber segment (exceptions are~\cite{luong2016overcoming, rozpedek2018parameter, rozpedek2019near}).
Furthermore, with few exceptions~\cite{guha2015rate, zwerger2017quantum, rozpedek2018parameter, rozpedek2019near, kamin2022}, such studies only provide rough approximations of time-dependent noise, and do not take into account platform-specific physical effects such as noise on the memory qubits during entanglement generation on NV centers~\cite{kalb2018dephasing} or collective Gaussian dephasing in ion traps (see Figure~\ref{fig:delft_eindhoven_setup}).

\section{Results} \label{sec:results} \label{sec:discussion}
\subsection{Summary of Results}
\begin{figure*}[!ht]
\centering
\includegraphics[width=1.5\columnwidth]{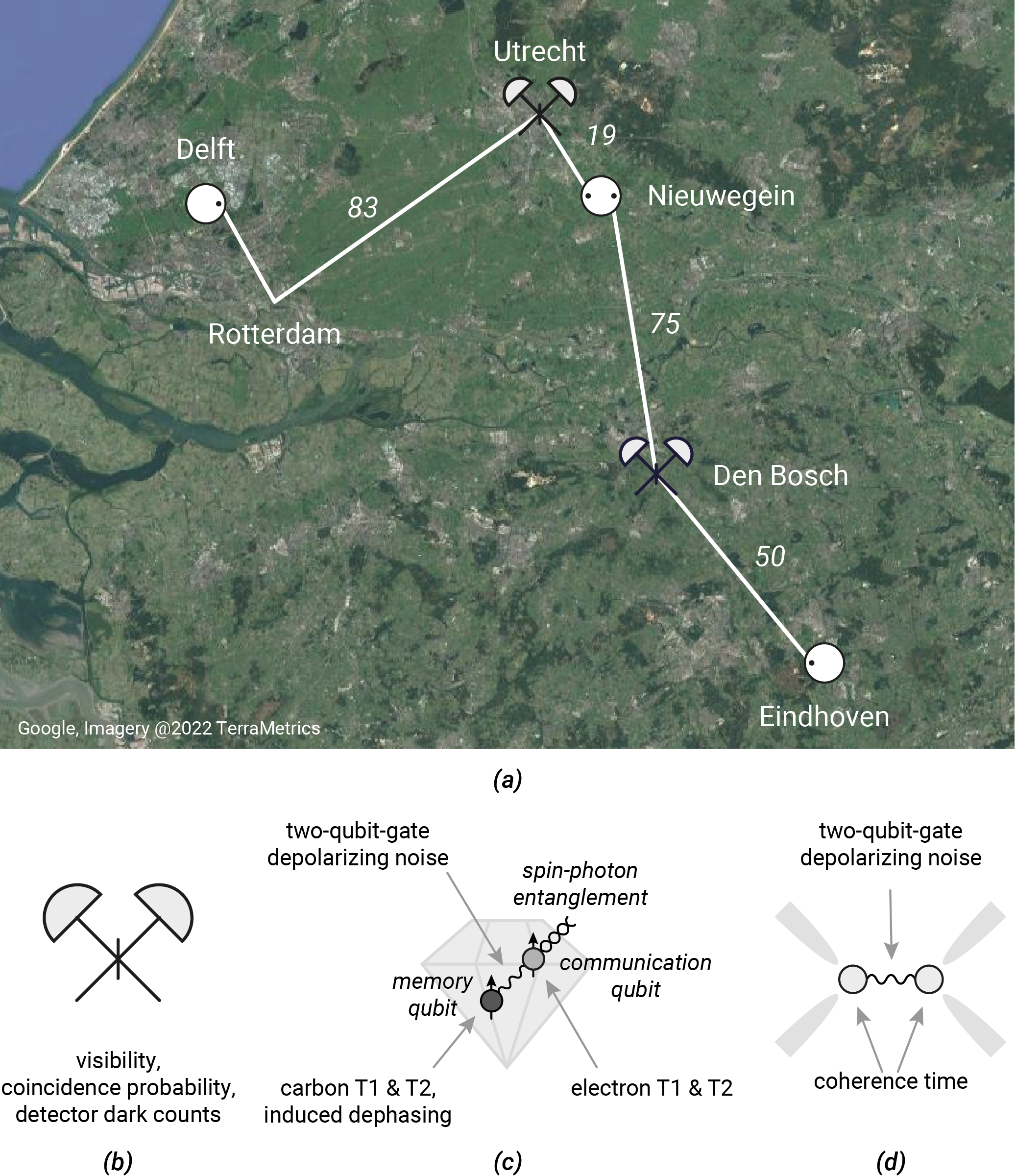}
\caption{
Investigated setup.
\textbf{(a)} Satellite photo of the Netherlands overlaid with a depiction of the hypothetical one-repeater connection between the Dutch cities of Delft and Eindhoven that we investigate.
The white circles represent processing nodes, connected to each other and to heralding stations through fiber drawn in white.
The black dots within the processing nodes represent qubits (the distinction between communication and memory qubits is not represented here).
The placement of nodes and heralding stations is constrained by the fiber network, and their position on the figure roughly approximates their actual geographic location.
All distances are given in kilometers, with a total fiber distance between Delft and Eindhoven of 226.5 km.
\textbf{(b)} Heralding station.
Photons emitted by a processing node travel through the optical fiber and are interfered at a beam splitter.
Photon detection heralds entanglement between processing nodes.
This process is affected by the overall probability that emitted photons are detected, the coincidence probability, i.e., the probability that photons arrive in the same time window, the imperfect indistinguishability of the photons as measured by the visibility and dark counts in the detector.
\textbf{(c)} Color center in diamond, one of the processing nodes we investigate.
We consider an optically-active electronic spin used as a communication qubit, and a carbon spin used as a memory qubit.
Decoherence in both qubits is modeled through amplitude damping and phase damping channels with characteristic times $T_1$ and $T_2$, respectively.
These are different for the two qubits.
The existence of an always-on interaction between the qubits allows for the execution of two-qubit gates, but also means that entangling attempts with the communication qubit induce noise on the memory qubit.
\textbf{(d)} Ion trap, the other processing node we investigate.
We consider two optically active ions trapped in an electromagnetic field generated by electrodes, whose energy levels are used as qubits.
The ions interact through their collective motional modes, which enables the implementation of two-qubit gates.
They are subject to collective Gaussian dephasing noise characterized by a coherence time.
}
\label{fig:delft_eindhoven_setup}
\end{figure*}

Here, we present the first study that takes into account time-dependent noise, platform-specific noise sources and classical control communication, as well as constraints imposed by a real-world fiber network, and optimizes over parameters of the repeater protocols used to generate entanglement.
Our investigation is conducted using fiber data from SURF, an organization that provides connectivity to educational institutions in the Netherlands.
Specifically, we will consider a network path connecting the Dutch cities of Delft and Eindhoven, separated by 226.5 km of optical fiber (see Figure~\ref{fig:delft_eindhoven_setup} (a)). 
In placing equipment, we restrict ourselves to SURF locations, which leads to the repeater being located closer to Delft than to Eindhoven.
Intermediary stations used for heralded entanglement generation (see Figure~\ref{fig:delft_eindhoven_setup} (b)) cannot be placed equidistantly from both nodes either, as is generally assumed in idealized studies.
We emphasize that we restrict ourselves to existing infrastructure, and therefore do not investigate the possibility of altering the fiber links.
We direct the interested reader to related work which focuses on determining hardware requirements while taking into account how many repeaters to use and their placement~\cite{da2023requirements}.

We consider the case where the network path is used to support an advanced quantum application, namely Verifiable Blind Quantum Computation (VBQC)~\cite{leichtle2021verifying}, with a client located in Eindhoven and a powerful quantum-computing server located in Delft.
We chose VBQC because since their introduction blind-quantum-computing protocols have attracted a lot of interest, being widely cited as one of the principal future applications of quantum networks (see, e.g.,~\cite{fitzsimons2017unconditionally, morimae2012blind, huang2017experimental, gheorghiu2015robustness, dunjko2012blind, broadbent2009universal, barz2012demonstration, leichtle2021verifying}).
While it is true that VBQC is somewhat unique in that it is highly asymmetrical in terms of the resources it requires from client and server, it is representative for many other quantum-networking applications in that it requires multiple live qubits.
Additionally, the noise resilience of the specific VBQC protocol we consider~\cite{leichtle2021verifying} makes it particularly suitable to study the performance of such applications in the presence of hardware imperfections.
Specifically, we consider the smallest instance of VBQC, where two entangled pairs are generated between the client and the server.
Such entanglement is used to send qubits from the client to the server.
We show in Appendix \ref{app:target_metric} that this can be done through remote state preparation~\cite{bennett2001remote}.
To set the requirements of our quantum-network path, we impose that its hardware must be good enough to execute VBQC with the largest acceptable error rate~\cite{leichtle2021verifying}.
This demand can be translated to requirements on the fidelity and rate at which entanglement is produced.
Both depend on the lifetime of the server's memory, as the server needs to be able to wait until both qubit states have been generated before it can begin processing.
Additionally, the requirements on the fidelity and rate can also be understood as the fidelity and rate at which we can deterministically teleport unknown data qubits between the client and the server.
Therefore, while our investigation focuses on VBQC, our results can also be interpreted from the perspective of quantum teleportation.

In our study, we obtain the following results, described in more detail below:
First, we investigate the \emph{minimal hardware requirements} that are
needed to realize target fidelities and rates that allow executing VBQC using our network path. 
These correspond to the minimal improvements over state-of-the-art hardware parameters that enable meeting the targets.
Specifically, we consider parameters measured for networked color centers (specifically, for NV centers in diamond)~\cite{bernien2013heralded, hensen2015loophole, kalb2017entanglement, humphreys2018deterministic, pompili2021realization, hermans2022qubit, abobeih2018one, bradley2019ten} and ion traps~\cite{tirepeater, krutyanskiy2023, krutyanskiy2019light, schupp_interface_2021, krutyanskiy2017, myerson2008, roos2006, tiprivate}.
We find that considerable improvements are needed even to bridge relatively modest distances, with our study also shining light on which parameters require significantly more improvement than others.
To obtain this result, we have built an extensive simulation framework on top of the discrete event simulator NetSquid~\cite{coopmans2021netsquid}, which includes models of color centers (specifically adapted from NV centers in diamond), ion traps, a general abstract model applicable to all processing nodes, as well as different schemes of entanglement generation. 
Our framework can be readily re-configured to study other network paths of this form, including the ability to configure other types of processing-node hardware, or entanglement-generation schemes.
Being able to simulate the Delft-Eindhoven path, we then perform parameter optimization based on genetic algorithms to search for parameter improvements that minimize a cost function (see Section~\ref{sec:methods} for details) on SURF's high-performance-computing cluster Snellius.

Second, we examine the \emph{absolute minimal requirements} for all parameters in our models (for color centers and ion traps), if all other parameters are set to their perfect value (except for photon loss in fiber).
We observe that the minimal hardware requirements impose higher demands on each individual parameter than the absolute minimal requirements.
This highlights potential dangers in trying to maximize individual parameters without taking into account global requirement trade-offs.
However, somewhat surprisingly, we find that the absolute minimal requirements are typically of the same order of magnitude as the minimal requirements, and can therefore still be valuable as a first-order approximation.
Our results are obtained using the same NetSquid simulation, by incrementally increasing the value of a parameter until the target requirements are met. 

Finally, we investigate whether the idealized network paths usually employed in the repeater literature would lead to significantly different minimal hardware 
improvements. Specifically, in such idealized setups all repeaters and heralding stations are equally spaced, all fibers are taken to have 0.2 dB km$^{-1}$ attenuation, and the models employed for the processing-node hardware are largely platform-agnostic.
We find that considering real-world network topologies such as the SURF grid imposes significantly more stringent demands. 

Let us now be more precise about the setup of our network path, as well as the requirements imposed by VBQC:
\\

\subsection{Quantum-Network Path}
The network path we consider consists of three processing nodes that are assumed to all have the same hardware.
That is, the stated hardware requirements are sufficient for all nodes and we do not differentiate between the three nodes.
On an abstract level, all processing nodes have at least one so-called communication qubit, which can be used to generate entanglement with a photon.
The repeater node in the middle (Nieuwegein, Figure~\ref{fig:delft_eindhoven_setup} (a)) has two qubits available (at least one of which is a communication qubit) that it can use to simultaneously hold entanglement with the node in Delft, as well as the one in Eindhoven.
Once entanglement has been generated with both Delft and Eindhoven, the repeater node may perform an entanglement swap~\cite{briegel1998quantum} in order to create end-to-end entanglement between Delft and Eindhoven (see Figure~\ref{fig:protocol_figure}).
On processing nodes, such a swap can be realized deterministically, i.e., with success probability $1$, since it can be implemented using quantum gates and measurements on the processor.
We note that even when the gates and measurements are noisy the swap remains deterministic, although it will induce noise on the resulting entangled state.

\begin{figure}[!htpb]
\centering
\includegraphics[width=\columnwidth]{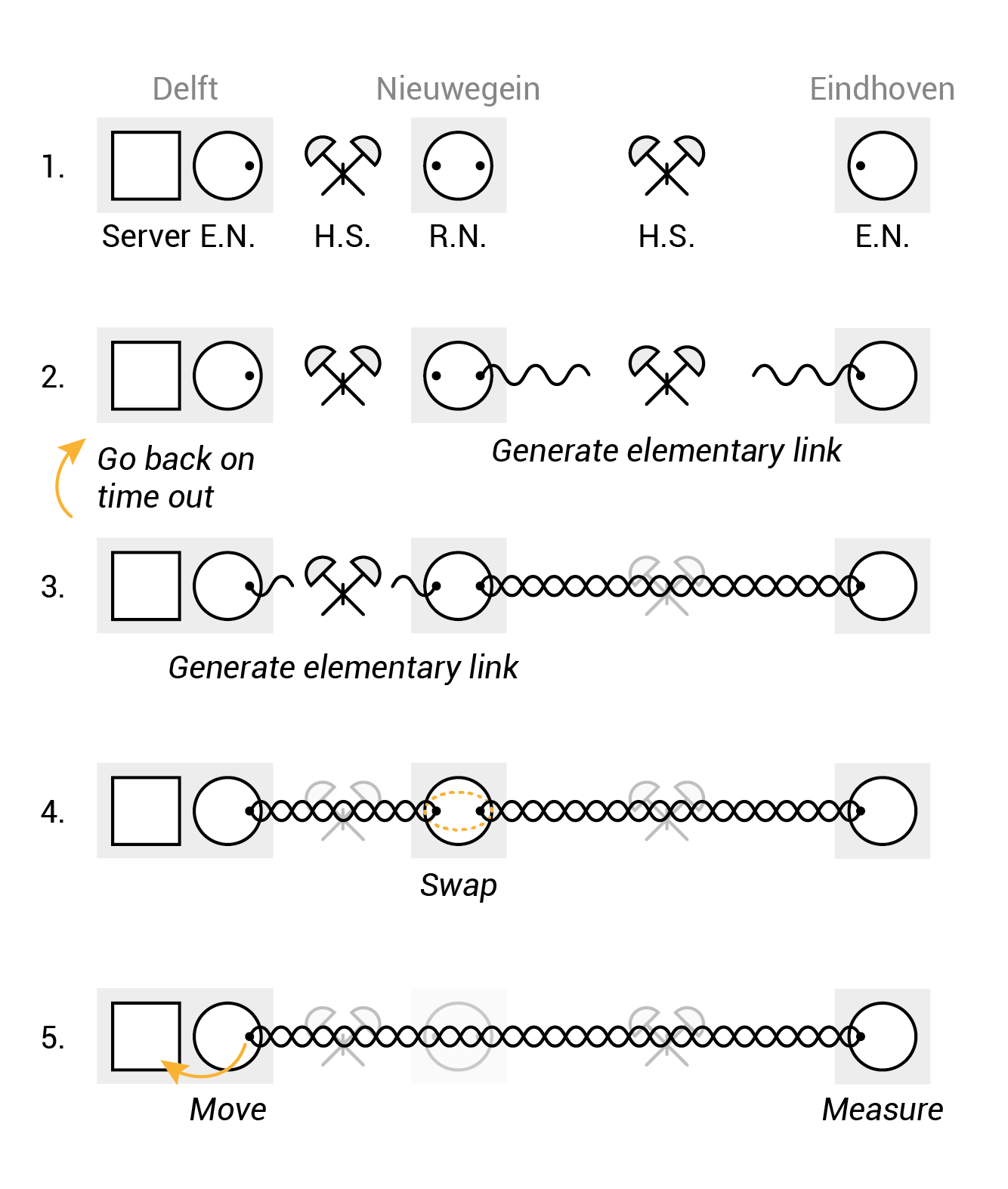}
\caption{Protocol executed in the setup we investigate.
\textbf{1.} No entanglement is shared \textit{a priori}.
E.N. stands for End Node, R.N. stands for Repeater Node and H.S. stands for Heralding Station.
\textbf{2.} Entanglement generation attempts begin along the longer link, which connects the repeater node to the Eindhoven node.
\textbf{3.} After entanglement has been established along the longer link, attempts for entanglement generation along the shorter link start. In case this takes longer than a given cut-off time, the previously generated entanglement is discarded and we go back to \textbf{2}.
\textbf{4.} After entanglement is generated on both links, the repeater node performs an entanglement swap, creating an end-to-end entangled state.
\textbf{5.} The Delft node maps its half of the state to a powerful quantum-computing server, while the Eindhoven node measures its half.
}
\label{fig:protocol_figure}
\end{figure}

For all types of processing nodes, we here assume the repeater to act sequentially~\cite{rozpedek2018parameter} due to hardware restrictions.
That is, it can only generate entanglement with one of the other two nodes at a time.
To minimize the memory requirements at the repeater node (Nieuwegein), we will always first produce entanglement with the farthest node (Eindhoven).
Once this entanglement has been produced, the repeater generates
entanglement with the closest node (Delft).
To combat the effect of memory decoherence, entangled qubits are discarded after a cut-off time~\cite{rozpedek2018parameter}.
This means that if entanglement between Delft and Nieuwegein is not produced within a specific time window following the successful generation of entanglement between Nieuwegein and Eindhoven, all entanglement is discarded and we restart the protocol by regenerating entanglement between Nieuwegein and Eindhoven.
Classical communication is used to initiate entanglement generation between nodes and notify all nodes when swaps or discards are performed.

We consider three types of processing nodes (see Figure~\ref{fig:delft_eindhoven_setup} (c) and (d)): (1) color centers, specifically modeled on NV centers in diamond, (2) ion traps and (3) a general abstract model applicable to all processing nodes.
Let us now provide more specific details on each of these models required for the parameter analysis below.

(1) NV centers are a prominent example of color centers for which significant data is available from quantum-networking experiments~\cite{bernien2013heralded, hensen2015loophole, kalb2017entanglement, humphreys2018deterministic, pompili2021realization, hermans2022qubit}. Here,
the color center's optically-active electronic spin is employed as a communication qubit. The second qubit is given by the long-lived spin state of a Carbon-13 atom, which is coupled to the communication qubit and used as a memory qubit.
Our color-center model accounts for the following:
\begin{itemize}
    \item Restricted topology, with one optically-active communication qubit and one memory qubit (note however that larger registers have been realized, for example in~\cite{bradley2019ten});
    \item Restricted gate set, with arbitrary rotations on the communication qubit, Z-rotations on the memory qubit and a controlled rotation gate between the two qubits;
    \item Depolarizing noise in all gates, bit-flip noise in measurement;
    \item Qubit decoherence in memory modeled through amplitude damping and dephasing channels with decay times $T_1$ and $T_2$ (we consider the experimentally-realized times of $T_1 = 1$ hour (10 hours) and $T_2 = 0.5$ s (1 s) for the communication (memory) qubit~\cite{abobeih2018one, bradley2019ten, hermans2022qubit});
    \item Induced dephasing noise on the memory qubit whenever entanglement generation using the communication qubit is attempted~\cite{kalb2018dephasing, pompili2021realization}.
\end{itemize}
The efficiency of the photonic interface in NV centers is limited to 3\% due to the zero-phonon line (ZPL).
It is likely that executing VBQC using the path we investigate will require overall photon detection probabilities higher than 3\%.
Little data is presently available for other color centers (SiV, SnV).
We hence focus on the NV model, but do allow a higher emission probability, which could be achieved either by using a color center with a more favorable ZPL (65-90\% for SiV~\cite{ruf2021quantum}, 57\% for SnV~\cite{ruf2021quantum}), or by placing the NV in a cavity~\cite{ruf2021resonant}. 
More details about our color-center model, and a validation of the model against experimental data for NV centers, can be found in Appendix \ref{sec:appendix_setup}.

(2) Trapped ions are charged atoms suspended in an electromagnetic trap, the energy levels of which can be used as qubits.
Our trapped-ion model accounts for the following:
\begin{itemize}
    \item Two identical, optically active ions in a trap;
    \item Restricted gate set as described in~\cite{schindlerQuantumInformationProcessor2013}, with arbitrary single-qubit Z rotations, arbitrary collective rotations around axes in the XY plane, and an entangling M\o lmer-Sørensen gate~\cite{molmer1999multiparticle};
    \item Depolarizing noise in all gates, bit-flip noise in measurement;
    \item Qubit decoherence modeled as collective Gaussian dephasing, with a characteristic coherence time~\cite{zwerger2017quantum};
    \item Off-resonant scattering that adds a random delay to the emission time of photons, which is counteracted using a tunable coincidence time window (as captured by a toy model introduced in Appendix \ref{app:time_windows}).
\end{itemize}
More details about our trapped-ion model, and a validation of the model against experimental data, can be found in Appendix \ref{sec:appendix_setup}.

(3) We further investigate an abstract, platform-agnostic processing-node model.
This model accounts for depolarizing noise in all gates and in photon emission, as well as amplitude-damping and phase-damping noise in the memory.
It does not account for any platform-specific restrictions on topology, gate set or noise sources.
Later on, we show that using the abstract model instead of hardware-specific models leads to an inaccurate picture of minimal hardware requirements.
Even so, the abstract model can be valuable to study systems for which hardware-specific models are as of yet unavailable.
Additionally, we note that the smaller number of hardware parameters in the abstract model as compared to the hardware-specific models means that the parameter space can be explored more efficiently, making it easier to, e.g., find minimal hardware parameters.

To entangle two processing nodes, one can use different schemes for entanglement generation, and we here consider the so-called single-click~\cite{cabrillo1999creation} and double-click schemes~\cite{barrett2005efficient}.
Both of these start with two distant nodes generating matter-photon entanglement and sending the photon to a heralding station.
In the single (double)-click protocol, matter-matter entanglement is heralded by the detection of one (two) photons after interference.
The trapped-ion nodes we investigate perform only double-click entanglement generation as single-click entanglement generation has not been realized for the type of trapped-ion devices we consider, i.e., trapped ions in a cavity.
The color-center nodes and abstract nodes perform both single and double click.
Our entanglement-generation models account for the following physical effects:
\begin{itemize}
    \item Emission of the photon in the correct mode, modeled through a loss channel;
    \item Imperfect photon emission modeled through a depolarizing channel;
    \item Capture of the photon into the fiber, modeled through a loss channel;
    \item Photon frequency conversion, modeled through a loss channel (as a first-order approximation, we assume this is a noiseless process);
    \item Photon attenuation in fiber, modeled through a loss channel;
    \item Photon delay in fiber;
    \item Photon detection at the detector, modeled through a loss channel;
    \item Detector dark counts;
    \item Photon arrival at the detector at different times;
    \item Imperfect photon indistinguishability.
\end{itemize}
While photon attenuation losses depend on the characteristics (such as the length) of the fiber that is used to deploy a quantum network, the other losses depend only on the quantum hardware that is used.
For convenience we collect all the hardware-related losses into a single parameter, called the photon detection probability excluding attenuation losses.

The hardware parameters used in our models are based on quantum-networking experiments with NV centers (single-click~\cite{kalb2017entanglement, humphreys2018deterministic, pompili2021realization, hermans2022qubit} and double-click~\cite{bernien2013heralded, hensen2015loophole}), and trapped ions (double-click~\cite{krutyanskiy2023}).
\\

\subsection{Blind Quantum Computation}
Having discussed our modeling of the path between Delft and Eindhoven, we turn to the end nodes.

Both end nodes are processing nodes.
The end node in Eindhoven takes the role of client in the VBQC protocol.
In Delft, there is not only an end node, but also a powerful quantum-computing server.
After entanglement is established by the end node in Delft it transfers its half of the entangled state to this server.
The client in Eindhoven simply measures its half of the entangled state.
The Delft scenario is similar to the setting investigated in~\cite{vardoyan2022quantum}, where the authors consider an architecture in which a node contains two NV centers, one of them used for networking and the other for computing.
Here, we make some simplifying assumptions that allow us to focus on the network path:
we take the state transfer process to be instantaneous and noiseless, and assume that the computing node is always available to receive the state.
Further, we assume that the quantum gates performed by the server are noiseless and instantaneous, and that their qubits are subject to depolarizing noise with memory coherence time $T$ = 100 s.
Because of these assumptions, the requirements we find are limited primarily by imperfections in the network path itself rather than in the computing node.

We investigate hardware requirements on three processing nodes (two end nodes and one repeater node) so that a client in Eindhoven can perform 2-qubit VBQC, a particular case of the protocol described in~\cite{leichtle2021verifying}, using the Delft server.
In this protocol, the client prepares qubits at the server, which are then used to perform either computation or test rounds.
In test rounds, the results of the computation returned by the server are compared to expected results.
The protocol is only robust to noise if the noise does not cause too large an error rate.
The protocol is shown in~\cite{leichtle2021verifying} to remain correct if the \emph{maximal} probability of error in a test round can be upper-bounded by $25\%$.
We prove in Appendix \ref{app:target_metric} that the protocol is still correct if the \emph{average} probability of error in a test round can be upper-bounded by $25\%$.
We further prove in the same section that if the entangled pairs distributed by the network path can be used to perform quantum teleportation at a given rate and quality, the protocol can be executed successfully.
Namely, this is true if the average fidelity at which unknown pure quantum states can be teleported using the entangled pairs distributed by the network path ($F_\text{tel}$) and the entangling rate $R$ satisfy a specific bound.
We note that this bound takes into account potential jitter in the delivery of entanglement (i.e., the fact that the time required to generate entanglement, and hence the time entangled states need to be stored in memory, can fluctuate around its expected value).
We consider two distinct pairs of $F_\text{tel}$ and $R$ that satisfy this bound as our target metrics, namely:
\begin{itemize}
\item
Target 1: $F_\text{tel} = 0.8717$, $R = 0.1$ Hz,
\item
Target 2: $F_\text{tel} = 0.8571$, $R = 0.5$ Hz.
\end{itemize}
The choice of these specific values was motivated by the fact that there is no fidelity $F_\text{tel} \leq 1$ for $R \approx 0.014$ Hz such that the VBQC condition is satisfied, therefore all target rates should satisfy $R > 0.014$ Hz, preferably with some margin to avoid trivial solutions.
Additionally, Target 1 is achievable using either the single-click or double-click protocol and using either one or zero repeaters on the fiber path under consideration, given sufficient hardware improvements.
In contrast, Target 2 is achievable only using the single-click protocol and one repeater (see also Appendices \ref{appendix:sec_repeaterless} and \ref{sec:appendix_single_double}).
This suggests that the difference between the two targets is large enough to lead to significantly different results.

The derivation of this bound assumes that the client prepares qubits at the server by first generating them locally and then transmitting them to the server using quantum teleportation.
We note that alternatively the remote-state-preparation protocol~\cite{bennett2001remote} can be used, which will likely be more feasible in a real experiment as it requires fewer quantum operations by the client.
In Appendix \ref{app:target_metric} we describe a way how the VBQC protocol \cite{leichtle2021verifying} can be performed using remote state preparation.
Note however that we have not investigated the security of the protocol in this case.
We show that under the assumption that local operations are noiseless, quantum teleportation and remote state preparation lead to the exact same requirements on the network path.
Thus, in case the target is met, VBQC can be successfully executed using either quantum teleportation or remote state preparation.
Lastly, we note that there is a linear relation between the average teleportation fidelity $F_\text{tel}$ and the fidelity of the entangled pair~\cite{horodeckiGeneralTeleportationChannel1999}.
\\

\subsection{Minimal Hardware Requirements.}
Here, we aim to find the smallest improvements over current hardware to generate entanglement enabling VBQC.
These are shown at the bottom of Figure~\ref{fig:table_and_figures_requirements} for color centers (b) and trapped ions (c).
In the table in Figure~\ref{fig:table_and_figures_requirements} (a) we show a selection of the actual values for the minimal hardware requirements (the set of parameters representing the smallest improvement over state-of-the-art parameters, see Section~\ref{sec:methods} for details on how we determine this), as well as the absolute minimal requirements (the minimal value for each parameter assuming that every other parameter except for photon loss in fiber is perfect).
All the parameters are explained in Section \ref{sec:methods}, and their state-of-the-art values that we consider are given in Table \ref{tab:baseline_parameters}.
\begin{figure*}[!ht]
\centering
\includegraphics[width=.95\textwidth]{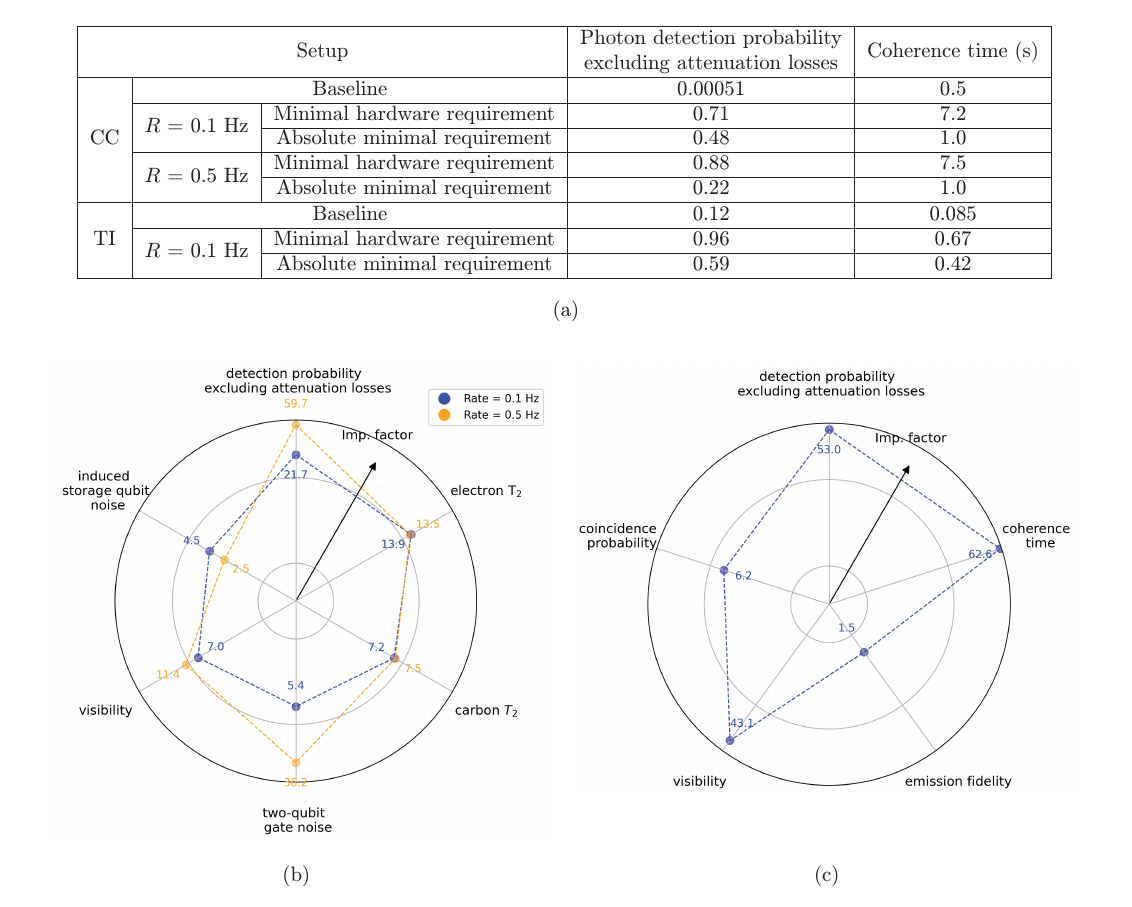}
\caption{
Improvements required to connect the Dutch cities of Delft and Eindhoven using color-center (CC) and trapped-ion (TI) repeaters for an entanglement-generation rate of 0.1 Hz and an average teleportation fidelity of 0.8717 (Target 1) and a rate of 0.5 Hz and average teleportation fidelity of 0.8571 (Target 2).
\textbf{(a)} The values that are required for the photon detection probability excluding attenuation losses and the coherence time.
The baseline parameter values have been demonstrated in state-of-the-art experiments.
The absolute minimal requirements are the required parameter values assuming that there are no other sources of noise or loss with the exception of fiber attenuation.
The coherence-time values in the table are the communication-qubit dephasing time for CC and the collective dephasing time for TI (see Section \ref{sec:methods} for an explanation of these parameters).
The TI requirements are for running a double-click entanglement-generation protocol.
The CC requirements are for running a double-click protocol for Target 1, and a single-click protocol for Target 2.
We note that all the minimal requirements found have a photon detection probability excluding attenuation losses above 30\%, the current state-of-the-art value for frequency conversion~\cite{krutyanskiy2017}.
\textbf{(b-c)} Directions along which hardware must be improved to connect the Dutch cities of Delft and Eindhoven using respectively a CC or a TI  repeater.
% (subfigures \textbf{(b)} and \textbf{(c)}) respectively).
The further away the line is from the center towards a given parameter, the larger improvement that parameter requires.
Improvement is measured in terms of the ``improvement factor'', which tends to infinity as a parameter tends to its perfect value (see Section \ref{sec:methods} for the definition).
In both plots a logarithmic scale is used.
The origin of the plots corresponds to an improvement factor of 1, i.e., no improvement with respect to the state of the art.
\textbf{(b)} (CC), the blue (orange) line corresponds to the minimal requirements for Target 1 (Target 2).
Improvement is depicted for the following parameters, clockwise from the top: photon detection probability excluding attenuation losses in fiber, dephasing time of the communication qubit, dephasing time of the memory qubit, noise in the two-qubit gate, visibility of photon interference and dephasing noise induced on memory qubits when entanglement generation is attempted.
\textbf{(c)} (TI), the line corresponds to the minimal requirements for Target 1.
Improvement is depicted for the following parameters, clockwise from the top: photon detection probability excluding attenuation losses in fiber, qubit collective dephasing coherence time, spin-photon emission fidelity, visibility of photon interference and probability that two emitted photons coincide at the detection station.
All parameters are explained in Section \ref{sec:methods}, and their state-of-the-art values that are being improved upon are given in Table~\ref{tab:baseline_parameters}.
}
\label{fig:table_and_figures_requirements}
\end{figure*}

The minimal color-center hardware requirements for Target 1 (blue line in Figure~\ref{fig:table_and_figures_requirements} (b)) correspond to the usage of a double-click protocol, as we found that this allows for laxer requirements than using a single-click protocol.
On the other hand, the minimal requirements for Target 2 (orange line in Figure~\ref{fig:table_and_figures_requirements} (b)) correspond to the usage of a single-click entanglement-generation protocol.
This is because achieving Target 2 in the setup we studied is not possible at all with a double-click protocol even if every parameter except for photon loss in fiber is perfect.
Therefore, and since we do not model single-click entanglement generation with trapped ions, Figure~\ref{fig:table_and_figures_requirements} (c) depicts only the requirements for trapped ions to achieve Target 1.

We thus find that in the setup we investigated performance targets with relatively higher fidelity and lower rate are better met by using a double-click protocol.
On the other hand, higher rates can only be achieved with single-click protocols.
This was to be expected, as (a) states generated with single-click protocols are inherently imperfect, even with perfect hardware and (b) the entanglement-generation rate of double-click protocols scales poorly with both the distance and the detection probability due to the fact that two photons must be detected to herald success.
\\

\subsection{Absolute Minimal Requirements.}
We now aim to find the minimal parameter values that enable meeting the targets, if the only other imperfection were photon loss in fiber.
These are the absolute minimal requirements, presented in the table at the top of Figure~\ref{fig:table_and_figures_requirements}.
We observe that while there is a gap between them and the minimal hardware requirements, it is perhaps surprisingly small.
For example, the minimal photon detection probability excluding attenuation losses required to achieve Target 1 with color centers is roughly 1.5 times larger than the corresponding absolute minimal requirement.
However, both requirements represent a three order of magnitude increase with respect to the state-of-the-art, which makes a factor of 1.5 seem small in comparison.

We remark on the feasibility of achieving the minimal hardware requirements for color centers.
NV centers, on which we have based the state-of-the-art parameters used in this work, are the color center that has been most extensively used in quantum-networking experiments (see~\cite{ruf2021quantum} for a review).
As discussed in Section~\ref{sec:results}, the efficiency of the photonic interface in this system is limited to 3\% due to the zero-phonon line.
Both targets we investigated place an absolute minimal requirement on the photon detection probability excluding attenuation losses above this value.
Improving the photonic interface of NV centers beyond the limit imposed by the zero-phonon line is only possible through integration of the NV center into a resonant cavity~\cite{ruf2021resonant}.
Alternatively, other color centers with a more efficient photonic interface could be considered as alternatives for long-distance quantum communication~\cite{ruf2021quantum}.
\\

\subsection{Hardware Requirements in Simplified Settings.}
Since we made use of real-life fiber data and elaborate, platform-specific hardware models, the results above would be difficult to obtain analytically.
For instance, collective Gaussian dephasing in ion traps could be challenging to analyze.
Analytical results are however attractive, as they provide a more intuitive picture of the problem at hand.
In order to find them, an approach commonly taken in the literature is to simplify the setup under study so that it becomes analytically tractable.
A usual simplification is to assume what we name the \textit{standard scenario}, in which nodes and heralding stations are equally spaced, and where the fiber attenuation is 0.2 dB km$^{-1}$ throughout.
Another common simplification is to consider simplified physical models for the nodes and the entangled states they generate (see, among others,~\cite{jiang2007optimal, coopmans2022improved, borregaard2020one, sangouard2011quantum}).
In order to investigate how hardware requirements change if such simplifications are used, we now apply our methodology to these two simplified situations and compare the resulting hardware requirements with the ones for our setup.
We hope to understand whether considering these setups leads to similar results, indicating that the simplifying approach is a good one, or if doing so paints an unrealistic picture of the hardware requirements, which would favor our approach.

\subsubsection{Effect of Existing Fiber Networks on Hardware Requirements}
We investigate how the hardware requirements in the standard scenario differ from the fiber-network-based setup.
We thus present in Figure~\ref{fig:standard_vs_surf_double_click_nv} a comparison of the hardware requirements for color centers in the two situations.
In both cases, we consider double-click entanglement generation, targeting an entanglement-generation rate of 0.1 Hz and an average teleportation fidelity of 0.8717.
\begin{figure}[!htpb]
\centering
\includegraphics[width=\columnwidth]{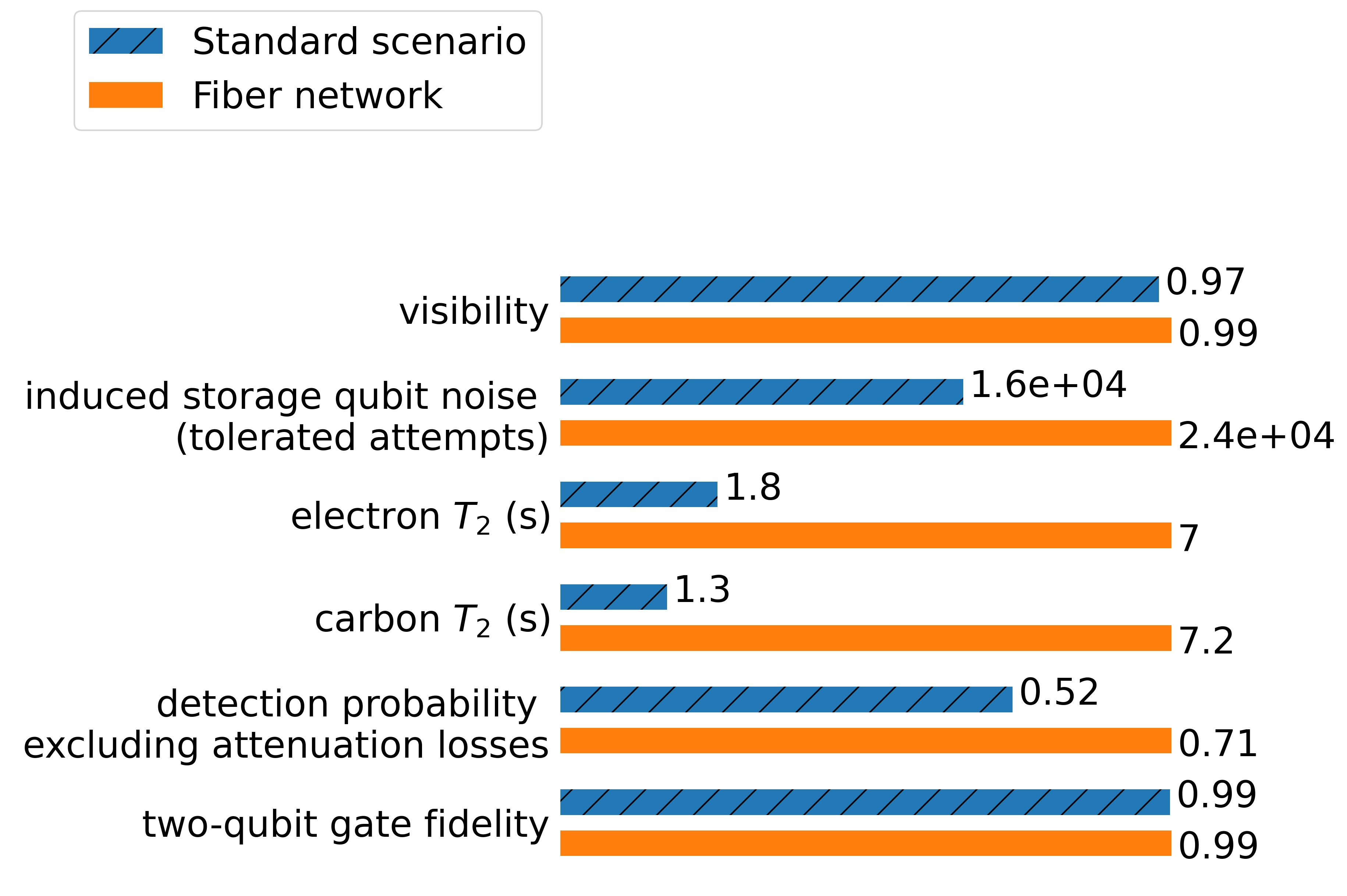}
\caption{Hardware requirements for connecting the Dutch cities of Delft and Eindhoven using a color center repeater performing double-click entanglement generation on an actual fiber network (blue) and assuming the standard scenario (orange, dashed).
Requirements are for achieving an entanglement-generation rate of 0.1 Hz and an average teleportation fidelity of 0.8717.
Parameters shown are, from top to bottom: visibility of photon interference, dephasing noise induced on memory qubits when entanglement generation is attempted, dephasing time of communication qubit, dephasing time of memory qubit, photon detection probability excluding attenuation losses in fiber and two-qubit gate fidelity.
}
\label{fig:standard_vs_surf_double_click_nv}
\end{figure}
Significant improvements over the state-of-the-art are required in both scenarios, but the magnitude of these improvements would be understated in case one were to consider the standard scenario and ignore existing fiber infrastructure.
For example, doing so would lead to underestimating the required coherence time of the memory qubits by a factor of four.
More broadly, we see that the improvement required is larger in the fiber-network scenario for (i) the photon detection probability excluding attenuation losses and (ii) memory parameters (coherence times and tolerance to entanglement-generation attempts).
Both of these results can be explained by the fact that when a real-world fiber network is considered there is more attenuation and the nodes are not evenly spaced.
As a consequence, better photonic interfaces are required to achieve similar rates, and states likely spend a longer time in memory, necessitating longer coherence times.
This emphasizes the need for considering limitations imposed by existing fiber infrastructure when estimating requirements on repeater hardware.

\subsubsection{Effect of Platform-Specific Modeling on Hardware Requirements}
Finally, we look into how the hardware requirements are affected if the processing nodes are modeled in a simplified, platform-agnostic way.
We thus compare the hardware requirements for color-center and trapped-ion repeaters with those for a platform-agnostic abstract model for a quantum repeater.
This is a simple processing-node model that accounts for generic noise sources such as memory decoherence and imperfect photon indistinguishability, but does not take platform-specific considerations such as restricted topologies into account.
For more details on the platform-agnostic abstract model, see Appendix \ref{appendix:sec_abstract_model}.
We consider double-click entanglement generation in the fiber-network-based setup, targeting an entanglement-generation rate of 0.1 Hz and an average teleportation fidelity of 0.8717.

To perform the comparison, we proceed as follows: (i) map the state-of-the-art hardware parameters to abstract-model parameters, (ii) run the optimization process for the platform-specific model and the abstract model in order to find the minimal hardware requirements for both, (iii) map the obtained platform-specific hardware requirements to the abstract model and (iv) compare them to the hardware requirements obtained by running the optimization process for the abstract model.
The results of this comparison can be seen in Figure~\ref{fig:compare_abstract_hardware_requirements}.
\begin{figure*}[!tbp]
    \centering
    \subfloat[\centering Color center.]{{\includegraphics[width=\columnwidth]{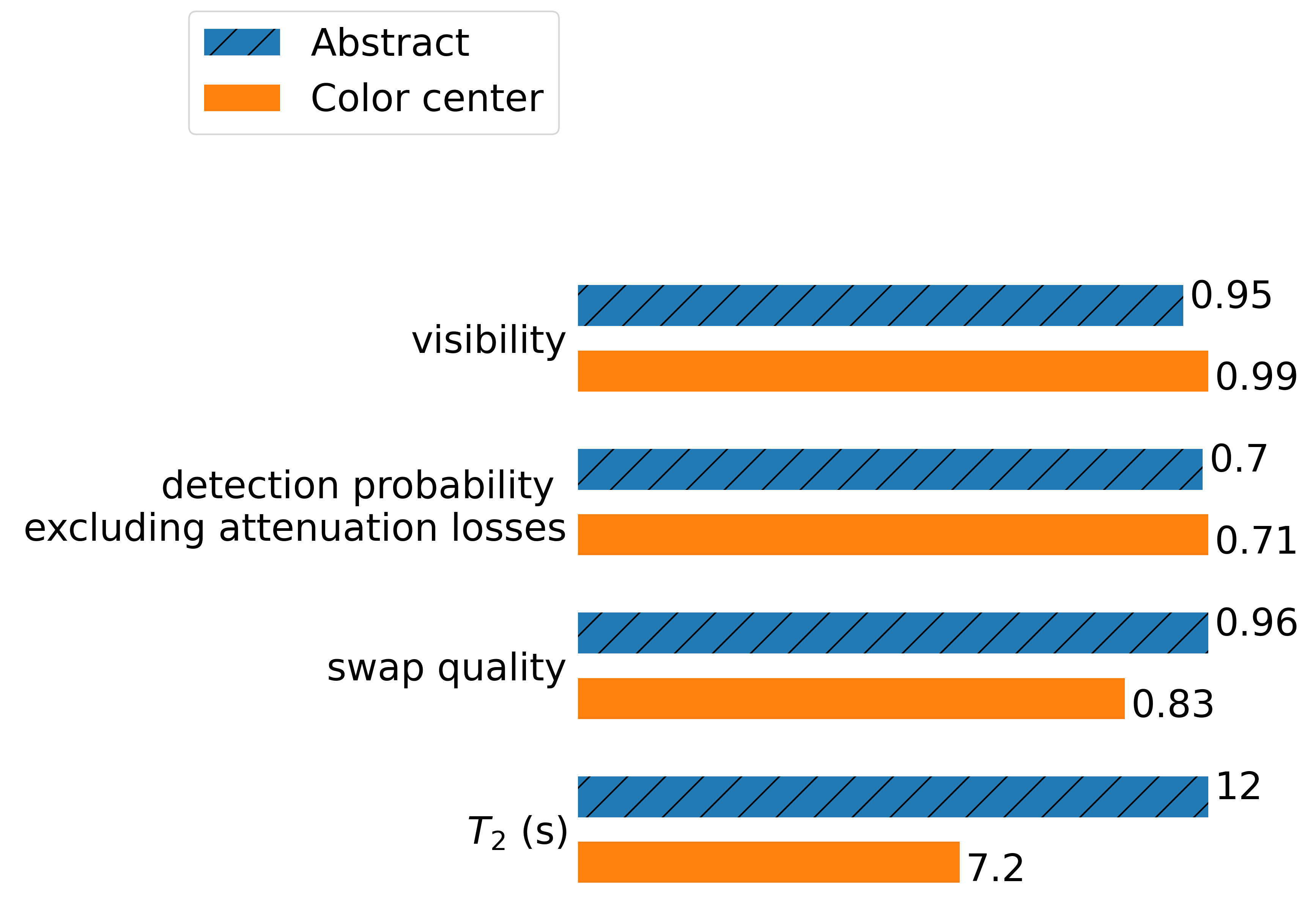}}}
    \subfloat[\centering Ion trap.]{{\includegraphics[width=\columnwidth]{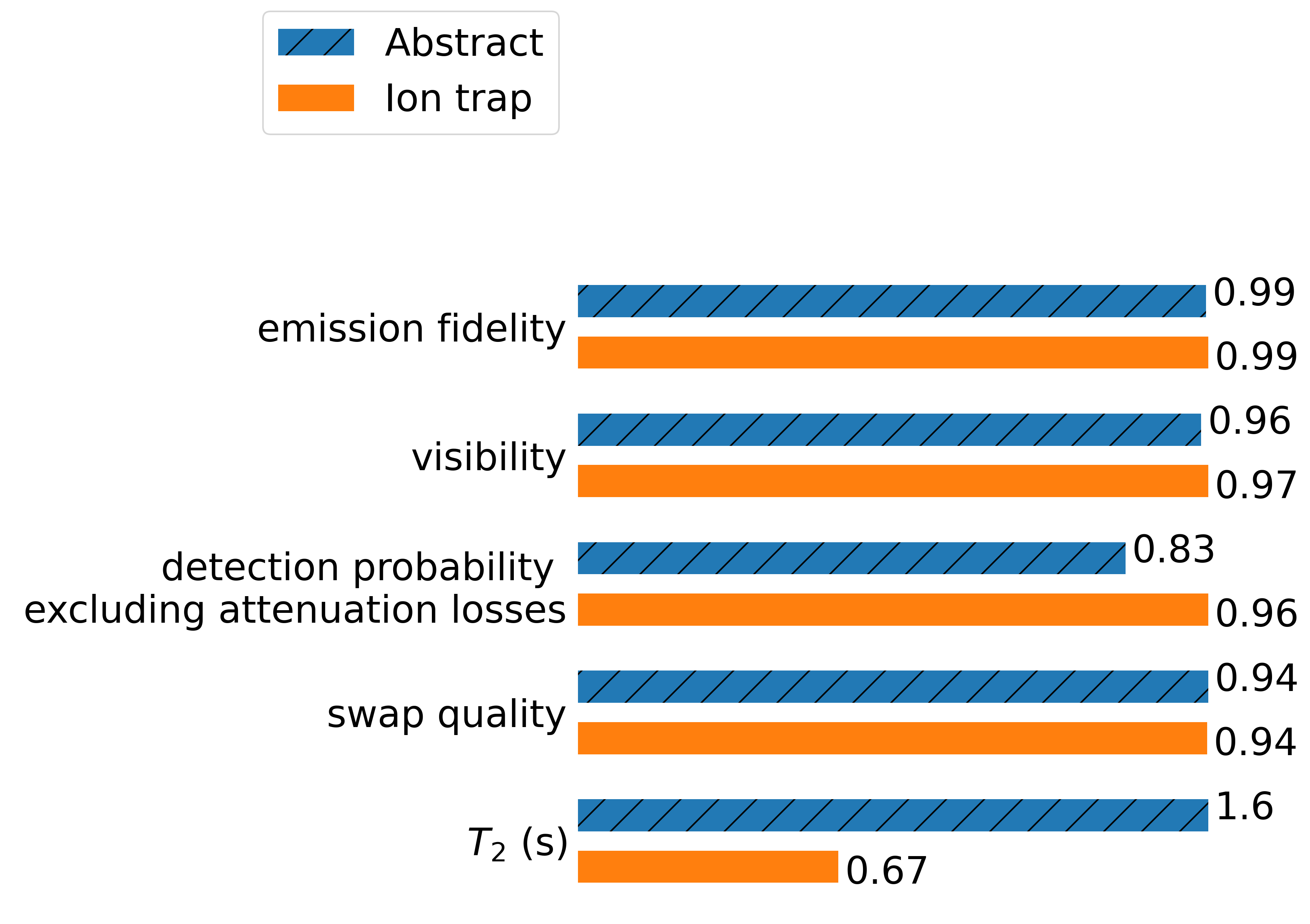}}}
    \caption{
	Comparison of hardware requirements for connecting the Dutch cities of Delft and Eindhoven using a repeater performing double-click entanglement generation considering a simple abstract model and more detailed \textbf{(a)} color center and \textbf{(b)} ion trap models.
	Requirements are for achieving an entanglement-generation rate of 0.1 Hz and an average teleportation fidelity of 0.8717.
	Parameters shown are, from top to bottom: spin-photon emission fidelity (trapped ion only), visibility of photon interference, photon detection probability excluding attenuation losses in fiber, fidelity of entanglement swap and qubit coherence time.
}
    \label{fig:compare_abstract_hardware_requirements}
\end{figure*}
The hardware requirements are significantly different for the abstract model and for the trapped-ion and color-center models.
This can be explained by the greater simplicity of the abstract model.
Take coherence time as an example.
The communication and memory qubits of color centers decohere at different rates, a complexity which is not present in the abstract model.
Therefore, improving the coherence time in the abstract model has a bigger impact than improving a given coherence time in the color center model.
This means that in the abstract model it is comparatively cheaper to achieve the same performance by improving the coherence time rather than other parameters.
The fact that memory noise in trapped ions is modeled differently than in the abstract model (the trapped-ion memory noise is Gaussian, arising from a collective dephasing process. See Equations~\ref{eq:gaussian_dephasing_1} and~\ref{eq:gaussian_dephasing_2}) could also explain the difference in the requirements for the coherence times seen in that case.
\\

\subsection{Entanglement Without a Repeater.}
We note that one of the set of targets we investigated, namely an entanglement-generation rate of 0.1 Hz and an average teleportation fidelity of 0.8717, could also be achieved in the setup we investigated without using a repeater node if a single-click entanglement-generation protocol were employed.
Furthermore, the hardware improvements required would be more modest in this case than if a repeater were used.
For more details on this, see Appendix \ref{appendix:sec_repeaterless}.
\\

\subsection{Outlook.}
In order to design and realize real-world quantum networks, it is important to determine minimal hardware requirements in more complex scenarios such as heterogeneous networks with multiple repeaters and end nodes.
The method presented in this work is well suited for this.
Furthermore, it would be valuable to investigate what limitations the assumptions we have made in our modeling place on our results.
For example, we did not consider the effects of fiber dispersion.
These effects could hamper entanglement generation and hence affect the minimal hardware requirements.
Even though preliminary investigations suggest that these effects might be small, quantifying them would represent a step forward in determining realistic minimal repeater-hardware requirements.
Another interesting open question is what effect the use of entanglement-distillation protocols (see~\cite{dur2007entanglement} for a review) would have on the minimal hardware requirements.

\section{Methods} \label{sec:methods}
In this section we elaborate on our approach for determining the minimal and absolute minimal hardware requirements for processing-node repeaters to generate entangled states enabling VBQC.
\\

\textbf{Conditions on Network Path to Enable VBQC.} In our setup, a client wishes to perform 2-qubit VBQC, a particular case of the protocol described in~\cite{leichtle2021verifying}, on a powerful remote server whose qubits are assumed to suffer from depolarizing noise with coherence time $T = 100$~s.
We further assume that the computation itself is perfect, with the only imperfections arising from the network path used to remotely prepare the qubits.
This protocol is shown in~\cite{leichtle2021verifying} to be robust to noise, remaining correct if the \textit{maximal} probability of error in a test round can be upper-bounded by $25\%$.
We argue in Appendix \ref{app:target_metric} that the protocol is still correct if the \textit{average} probability of error in a test round can be upper-bounded by $25\%$, as long as we assume that the error probabilities are independent and identically distributed across different rounds of the protocol.
This is the case for the setup studied here, as the state of the network is fully reset after entanglement swapping takes place at the repeater node.
This condition, together with the assumption on the server's coherence time, can be used to derive bounds on the required average teleportation fidelity and entanglement-generation rate, as shown in Appendix \ref{app:target_metric}.
\\

\subsection{Average Teleportation Fidelity.}
We use the average teleportation fidelity $F_\text{tel}$ that can be obtained with the teleportation channel $\Lambda_\sigma$ arising from the end-to-end entangled state $\sigma$ generated by the network we investigate as a target metric:
\begin{equation}
F_\text{tel}(\sigma) \equiv \int_{\psi} \expectationvalue{\Lambda_\sigma(\ketbra{\psi})}{\psi} d \psi,
\end{equation}
where the integral is taken over the Haar measure.
See Appendix \ref{sec:appendix_teleportation_fidelity} for more details.
\\

\subsection{Hardware Improvement for VBQC as an Optimization Problem.} We want to find the minimal hardware requirements that achieve a given average teleportation fidelity $F_{\text{target}}$ and entanglement-generation rate $R_{\text{target}}$. 
We restate this as a constrained optimization problem: we wish to minimize the hardware improvement, while ensuring that the performance constraints are met.
These constraints are relaxed through scalarization, resulting in a single-objective problem in which we aim to minimize the sum of the hardware improvement and two penalty terms, one for the rate target and one for the teleportation fidelity target.
The resulting cost function is given by
\begin{align}
\begin{split}
    C &= w_1 \Big(1 + \left(F_{\text{target}} - F_{\text{tel}}\right)^2\Big)\Theta\left(F_{\text{target}} - F_{\text{tel}}\right) \\
    &+ w_2\Big(1 + \left(R_{\text{target}} - R\right)^2\Big)\Theta\left(R_{\text{target}} - R\right) \\
    &+w_3H_{\text{C}}\left(x_1,...,x_{\text{N}}\right),
\end{split}
\label{eq:total_cost_function}
\end{align}
where $H_{\text{C}}$ is the hardware cost associated with parameter set $\{x_1,...,x_{\text{N}}\}$, $w_{\text{i}}$ are the weights of the objectives, $\Theta$ is the Heaviside function and $F_{\text{tel}}$ and $R$ are the average teleportation fidelity and entanglement-generation rate achieved by the parameter set, respectively.
The hardware cost function $H_C$ maps sets of hardware parameters to a cost that represents how large of an improvement over the state of the art the set requires.
To compute this consistently across different parameters we use no-imperfection probabilities, as done in~\cite{coopmans2021netsquid} (where they are called no-error probabilities).
A parameter is improved by a factor $k$, called the \textit{improvement factor}, if its corresponding no-imperfection probability $p_\text{ni}$ becomes $\sqrt[k]{p_\text{ni}}$.
For example, if the error probability of a gate is 40\%, its probability of no-imperfection is 0.6.
After improving it by a factor of 4 the no-imperfection probability becomes $\sqrt[4]{0.6} \approx 0.88$, corresponding to an error probability of approximately 12\%.
The hardware cost associated with a set of hardware parameters is the sum of the respective improvement factors, i.e.,
\begin{equation}
H_{\text{C}}\left(x_1,...,x_{\text{N}}\right) = \sum_{i=1}^N \frac{\ln{\{p_\text{ni}(b_{\text{i}})\}}}{\ln{\{p_\text{ni}(x_{\text{i}})\}}},
\end{equation}
where $p_\text{ni}(x_{\text{i}})$ is the no-imperfection probability corresponding to the value $x_{\text{i}}$ of parameter $i$ and $p_\text{ni}(b_i)$ is the no-imperfection probability corresponding to the baseline value $b_{\text{i}}$ of parameter $i$.
We have here for concreteness used natural logarithms, but the hardware cost is invariant to changes in the logarithms' bases.
We note that these improvement factors are the quantities shown in Figure~\ref{fig:table_and_figures_requirements}.
The weights $w_i$ are chosen such that the first two terms are larger than the last one for near-term parameters, guaranteeing that the set of parameters minimizing $C$ meets performance targets.
We are then effectively restricted to the region of parameter space in which the performance constraints are satisfied, as all points corresponding to near-term parameters in this region have a lower cost than points outside it.
The problem then becomes one of minimizing the hardware cost in this region.
We have verified that the expected values of the average teleportation fidelity and entanglement-generation rate of the parameter sets found meet the constraints, thus enabling VBQC conditional on our assumptions.
Our method guarantees that the set of parameters found is 'minimal' in the sense that making any of the parameters worse would result in the target not being met.
However, we note that there exist many such solutions, and if specific knowledge is available about how hard it is to improve particular parameters, the cost function could be adapted to pick out minimal parameter sets that may be easier to attain.
An example of this is the efficiency of the NV center's photonic interface, which is limited to 3\% due to the ZPL.
Going beyond this limit requires integration into a cavity, which carries with it a host of challenges~\cite{ruf2021resonant, ruf2021quantum}.
One could then modify the cost function to make improving the efficiency of the photonic interface beyond 3\% more expensive than improving other parameters.
However, as it is challenging to accurately estimate the hardness associated with specific improvements and, furthermore, the hardness may depend on the specific expertise available within a given research group, we have refrained from making such estimates.
\\

\subsection{Optimization Parameters.}
Using the methodology described later on in this section, we perform an optimization over both protocol and hardware parameters.
First we enumerate the protocol parameters:
\begin{itemize}
    \item Cut-off time, the time after which a stored qubit is discarded;
    \item Bright-state parameter (single-click entanglement generation only), the fraction of a matter qubit's superposition state that is optically active;
    \item Coincidence time window (double-click entanglement generation with ion traps only), the maximum amount of time between the detection of two photons for which a success is heralded.
          We model the effect of the coincidence time window using a toy model, see Appendix \ref{app:time_windows}.
\end{itemize}
Second, we enumerate the hardware parameters:
\begin{itemize}
    \item The Hong-Ou-Mandel visibility \cite{hongMeasurementSubpicosecondTime1987} is a measure for the indistinguishability of interfering photons and is defined by \cite{bouchard2020}
    \begin{equation}
    1 - \frac {C_\text{min}} {C_\text{max}}.
    \end{equation}
    Here $C_\text{min}$ is the probability (coincidence count rate) that two photons that are interfered on a 50:50 beamsplitter are detected at two different detectors when the indistinguishability is optimized (as is the case when using interference to generate entanglement),
    while $C_\text{max}$ is the same probability when the photons are made distinguishable.
    \item The probability of double excitation is the probability that two photons are emitted instead of one in entanglement generation with color centers;
    \item The induced memory qubit noise is the dephasing suffered by the memory qubit when the communication qubit is used to attempt entanglement generation.
The number given for this parameter in Table~\ref{tab:baseline_parameters} corresponds to the number of electron spin pumping cycles after which the Bloch vector length of the memory qubit in the state $(\ket{0} + \ket{1}) / \sqrt{2}$ in the $X - Y$ plane of the Bloch sphere has shrunk to $1/\text{e}$ when the communication qubit has bright-state parameter $0.5$~\cite{kalb2018dephasing};
    \item The interferometric phase uncertainty is the uncertainty in the phase acquired by the two interfering photons when they travel through the fiber in single-click entanglement generation with color centers;
    \item The photon detection probability excluding attenuation losses is the probability that a photon is detected given that emission was attempted, and assuming that the fiber length is negligible, i.e., considering every form of photon loss (including coupling to fiber) except the length-dependent attenuation loss in fiber;
    \item Every gate is parameterized by a depolarizing-channel fidelity;
    \item  For color centers, $T_1$ and $T_2$ are the characteristic times of the time-dependent amplitude damping and phase damping channels affecting the qubits, and are different for the communication and memory qubits.
The effect of the amplitude (phase) damping channel after time $t$ is given by Equation~\eqref{eq:amplitude_damping_main} (\eqref{eq:phase_damping_main})
\begin{align}
\begin{split}
\rho &\rightarrow \left(\ketbra{0}{0} + \sqrt{e^{-t/T_1}}\ketbra{1}{1}\right) \rho \\& \left(\ketbra{0}{0} + \sqrt{e^{-t/T_1}}\ketbra{1}{1}\right)^{\dagger} \\&+ \sqrt{1 - e^{-t/T_1}}\ketbra{0}{1} \rho \left(\sqrt{1 - e^{-t/T_1}}\ketbra{0}{1}\right)^{\dagger}
\label{eq:amplitude_damping_main}
\end{split}
\end{align}
\begin{align}
\begin{split}
\rho &\rightarrow \left(1 - \frac{1}{2}\left(1 - e^{-t/T_2}e^{-t/(2T_1)}\right)\right)\rho \\ &+ \frac{1}{2}\left(1 - e^{-t/T_2}e^{-t/(2T_1)}\right) Z\rho Z;
\label{eq:phase_damping_main}
\end{split}
\end{align}
    \item For ion traps, the coherence time characterizes the time-dependent collective Gaussian dephasing process that the qubits undergo, which is given by~\cite{zwerger2017quantum}:
\begin{equation}
\rho \to \int_{- \infty}^\infty K_{\text{r}} \rho K_{\text{r}}^\dagger p(r) dr,
\label{eq:gaussian_dephasing_1}
\end{equation}
where
\begin{equation}
K_{\text{r}} =  \exp(-\text{i} r \frac t \tau \sum_{j=1}^n Z_\text{j}),
\end{equation}
$Z_\text{j}$ denotes a Pauli $Z$ acting on qubit $j$, $n$ is the total number of ions in the trap, $\tau$ the coherence time and $t$ the storage time, and
\begin{equation}
p(r) = \frac{1}{\sqrt{2\pi}}\text{e}^{-\text{r}^2/2};
\label{eq:gaussian_dephasing_2}
\end{equation} 
    \item The noise on matter-photon emission is parameterized by a depolarizing-channel fidelity (i.e., the matter-photon state directly after emission is a mixture between a maximally entangled state and a maximally mixed state);
    \item The dark-count probability is the probability that a detection event is registered at a detector without a photon arriving.
\end{itemize}
The state-of-the-art values we use for the hardware parameters are shown in Table~\ref{tab:baseline_parameters}.
For more details on how the effects of the different hardware parameters are included in our models, see Appendix \ref{sec:appendix_setup}.
We note that some of the hardware parameters we consider in fact conceal trade-offs. 
For example, the probability of getting a double excitation when using color centers to emit photons can to an extent be tuned.
In this case, a lower probability of double excitation would come at the cost of getting fewer events.
However, optimizing over all such trade-offs is beyond the scope of this work.
\begin{table*}[!htpb]
\begin{tabular}{|c|c|c|}
\hline
                                                                                                             & \multicolumn{1}{c|}{Color center}                        & \multicolumn{1}{c|}{Ion trap}                            \\ \hline
Visibility                                                                                                   & 0.9~\cite{hermans2022qubit}                 & 0.89~\cite{krutyanskiy2023}                \\ \hline
Probability of double excitation                                                                             & 0.06~\cite{hermans2022qubit}                & -             \\ \hline
\begin{tabular}[c]{@{}c@{}}Induced memory qubit noise\\ (entanglement attempts until dephasing)\end{tabular} & 5300~\cite{hermans2022qubit}                & -             \\ \hline
Interferometric phase uncertainty (rad)                                                                     & 0.21~\cite{hermans2022qubit}                & -             \\ \hline
\begin{tabular}[c]{@{}c@{}}Photon detection probability \\ excluding attenuation losses \end{tabular}                                                                & $5.1\times 10^{-4}$~\cite{hermans2022qubit} & 0.111~\cite{schupp_interface_2021, krutyanskiy2017, krutyanskiy2019light}              \\ \hline
Two-qubit gate fidelity                                                                                              & 0.97~\cite{kalb2017entanglement}  & 0.95~\cite{tirepeater}  \\ \hline
Two-qubit gate duration                                                                                              & 500 $\mu$s~\cite{kalb2017entanglement}    & \newline 107 $\mu$s~\cite{tirepeater}    \\ \hline
Communication $T_1$                                                                                             & 1 h~\cite{abobeih2018one}           & -             \\ \hline
Communication $T_2$                                                                                             & 0.5 s~\cite{hermans2022qubit}         & -             \\ \hline
Memory $T_1$                                                                                                    & 10 h~\cite{bradley2019ten}          & -             \\ \hline
Memory $T_2$                                                                                                    & 1 s~\cite{bradley2019ten}           & -             \\ \hline
Coherence time                                                                                               & -             & 85 ms~\cite{tirepeater}         \\ \hline
Matter-photon emission fidelity                                                                                       & 1~\cite{pfaff2014unconditional}            & 0.99~\cite{stute_tunable_2012}   \\ \hline
Matter-photon emission duration                                                                                       & 3.8 $\mu$s~\cite{pompili2021realization}            & 50 $\mu$s~\cite{tirepeater, schupp_interface_2021}     \\ \hline
Dark count probability                                                                                       & $1.5\times 10^{-7}$~\cite{hermans2022qubit}              & $1.4\times 10^{-5}$~\cite{krutyanskiy2019light} \\ \hline
\end{tabular}
\caption{State-of-the-art color center and trapped-ion hardware parameters.
For the trapped-ion parameters, a detection time window of 17.5 $\mu$s and a coincidence time window of 0.5 $\mu$s are assumed (see Appendix \ref{sec:appendix_setup} for more details).
All fidelities are depolarizing-channel fidelities.
A dash (``-'') indicates that a value would not be well defined (for instance, there is no $T_1$ or $T_2$ time defined for trapped ions, while there is no coherence time defined for color centers).
We note that not all of these parameter values have been realized in a single experiment.}
\label{tab:baseline_parameters}
\end{table*}

\subsection{Evaluating Hardware Quality.} In order to minimize the cost function $C$, we require an efficient way of evaluating the performance attained by each parameter set.
We do this through simulation of end-to-end entanglement generation using NetSquid.
The full density matrix of the states generated, as well as how long their generation took in simulation time are recorded and used to compute the average teleportation fidelity and rate of entanglement generation.
Since entanglement generation is a stochastic process, multiple simulation runs are performed in order to collect representative statistics.
\\

\subsection{Framework for Simulating Quantum Repeaters.}
In our NetSquid simulation framework, we have implemented hardware models for color centers, trapped ions and a platform-agnostic abstract model.
This includes the implementation of different circuits for entanglement swapping and moving states for each platform, conditioned on their respective topologies and gate sets.
Additionally, we have implemented both single and double-click entanglement-generation protocols.
In order to combine these different building blocks that are required to simulate end-to-end entanglement distribution, we define services that each have a well-defined input and output but can have different implementations.
For example, the entanglement-generation service can either use the single-click or double-click protocol, and entanglement swapping can be executed on either color center or trapped-ion hardware.
End-to-end entanglement generation is then orchestrated using a link-layer protocol (inspired on the one proposed in~\cite{dahlberg2019link}) that makes calls to the different services, agnostically of how the services are implemented.
This allows us to use the same protocol for each different configuration of the simulation.
Switching between configurations in our simulation framework then only requires editing a human-readable configuration file.
The modularity of the simulation framework would make it simple to investigate further hardware platforms and protocols.

The link-layer protocol is itself an implementation of the link-layer service defined in~\cite{dahlberg2019link}.
From a user perspective, this simplifies using the simulation as all that needs to be done to generate entanglement is make a call to the well-defined link-layer service, without any knowledge of the protocol that implements the service.
In this work, the link-layer protocol is the one for a single sequential repeater illustrated in Figure \ref{fig:protocol_figure}.
However, the protocols included in our simulation code are able to simulate entanglement generation on chains of an arbitrary number of (sequential) repeaters that use classical communication to negotiate when to generate entanglement and that implement local cut-off times.
\\

\subsection{Finding Minimal Hardware Improvements.} In order to find the sets of parameters minimizing the cost function $C$, we employ the optimization methodology introduced in~\cite{da2021optimizing}, which integrates genetic algorithms and NetSquid simulations.
A genetic algorithm is an iterative optimization method, which initiates by randomly generating a population consisting of many sets of parameters, also known as individuals.
These are then evaluated using the NetSquid simulation and the cost function, and a new population is bred through mutation and crossover of individuals in the previous population.
The process then iterates, with better-performing individuals being more likely to propagate to further iterations.
For further details on the optimization methodology employed, see Appendix \ref{sec:appendix_optimization} and~\cite{da2021optimizing}.

This methodology is computationally intensive, so we execute it on the Snellius supercomputer.
We use one node of the Snellius supercomputer, which contains 128 2.6 GHz cores and a total of 256 GiB of memory.
Based on previously observed data reported in~\cite{da2021optimizing}, we employ a population size of 150 evolving for 200 generations.
The simulation is run 100 times for each set of parameters, as we have empirically determined that this constitutes a good balance between accuracy and computation time.
The time required for the procedure to conclude is hardware, protocol and parameter dependent, but we have observed that 10 wall-clock hours are typically enough.
We stress that this approach is general, modular and freely available~\cite{da2021optimizing}.
\\

\subsection{Finding Absolute Minimal Hardware Requirements.} In order to find these requirements, which are the minimal parameter values enabling meeting the performance targets if the only other imperfection is photon loss in fiber, we perform a sweep of each parameter, starting at the state-of-the-art value and terminating when the targets are met.
For each value of each parameter, we sweep also over the protocol parameters, i.e., the cut-off time, coincidence time window (for double-click entanglement generation with ion traps) and bright-state parameter (for single-click entanglement generation).

\section{Data availability} \label{sec:data_availability}
The data presented in this work have been made available at \url{https://doi.org/10.4121/19746748}.

\section{Code availability} \label{sec:code_availability}
The code that was used to perform the simulations and generate the plots in this paper has been made available at \url{https://gitlab.com/softwarequtech/simulation-code-for-requirements-for-a-processing-node-quantum-repeater-on-a-real-world-fiber-grid}.

\section*{Acknowledgements} \label{sec:acknowledgements}
We thank Dominik Leichtle and Harold Ollivier for useful discussions on VBQC.
We thank Sophie Hermans, Matteo Pompili and Ronald Hanson for useful discussions on state-of-the-art NV-center networking experiments and for sharing their experimental data.
We thank Tracy Northup, Ben Lanyon, Dario Fioretto, Simon Baier, Maria Galli, Viktor Krutyanskii and Markus Teller for useful discussions on modeling trapped-ion devices and for sharing their experimental data and parameters.
We thank Alejandro R.-P. Montblanch, Álvaro Gómez Iñesta, Anders Søndberg Sørensen, Bart van der Vecht, Bethany Davies, Gayane Vardoyan, Mariagrazia Iuliano, and Viktor Krutyanskii for critical reading of the manuscript.
We thank Eva Peet for helping with the design of Figures \ref{fig:delft_eindhoven_setup} and \ref{fig:protocol_figure}.
This work was supported by the QIA project that has received funding from the European Union’s Horizon 2020 research and innovation program under grant Agreement No. 82044.
G.A. was supported by NWO Zwaartekracht QSC 024.003.037.

\section*{Competing interests}
The authors declare no competing financial or non-financial interests.

\section*{Author contributions} \label{sec:author_contributions}
G.A. led the development of hardware models and the simulation of repeater protocols.
F.F.S. led the development and execution of optimizations.
G.A. and F.F.S. devised the target metric and proved the underlying theorems related to verifiable blind quantum computation.
G.A., F.F.S., D.M., T.C., A.D., H.J. and J.R. contributed to the development of the code used in the simulations.
A.T. contributed to the optimal execution of simulations on computing clusters.
G.A., F.F.S. and S.W. wrote the manuscript.
All authors revised the manuscript.
S.W. conceived and supervised the project.

\bibliography{biblio.bib}

\appendix
\input{sections/appendix.tex}

\end{document}

%% file: preamble.tex
\usepackage{amsthm}
\usepackage{mathtools}
\usepackage{physics}
\usepackage{xcolor}
\usepackage{graphicx}
\usepackage[caption=false]{subfig}
\usepackage[left=18mm,right=18mm,top=20mm, bottom=20mm, columnsep=15pt]{geometry} 
\usepackage{adjustbox}
\usepackage{placeins}
\usepackage[T1]{fontenc}
\usepackage{csquotes}
\usepackage{multirow}
\usepackage{braket}
\newtheorem{theorem}{Theorem}
\newtheorem{lemma}{Lemma}
\newtheorem{corollary}{Corollary}
\newtheorem{definition}{Definition}
\newtheorem{protocol}{Protocol}

%% file: sections/appendix.tex
\onecolumngrid

\input{sections/setup_appendix}
\input{sections/target_metric}

\input{sections/double_click_model}
\input{sections/coincidence_time_model}
\input{sections/single_click_model}
\input{sections/optimization_appendix}
\input{sections/simulation_performance_appendix}
\input{sections/protocols}
\input{sections/optimization_results_appendix}

%% file: sections/setup_appendix.tex
\section{Setup}
\label{sec:appendix_setup}
In this appendix, we elaborate on our modeling of the setup we study.
We go over the topology of the fiber network we considered, the protocols employed by the repeater nodes, the modeling of the nodes themselves and of entanglement generation.
\subsection{Fiber network and node placement}
Deployment of quantum networks in the real world will likely make use of existent fiber infrastructure~\cite{rabbie2022designing}.
In order to accurately account for this in our investigation of repeater hardware requirements, we used data of SURF's fiber network in our simulation.
SURF is a network provider for education and research institutions in the Netherlands.
The data we have access to consists of the physical location in which nodes are placed, the length of the fibers connecting them, the measured attenuation of each fiber and their dispersion.
We restricted the placement of quantum nodes and heralding stations to existing nodes in the network, and we assumed that they were connected by the shortest length of fiber possible.
We note that in the case we studied this corresponds also to the least overall attenuation.
Although dispersion was not considered in our models, an investigation of its effects would constitute an interesting extension to this work.
There are four nodes in the shortest connection between Delft and Eindhoven in SURF's network, as depicted in Figure~\ref{fig:delft_eindhoven_shortest_path}.
\begin{figure}[!ht]
\centering
\includegraphics[width=0.6\columnwidth]{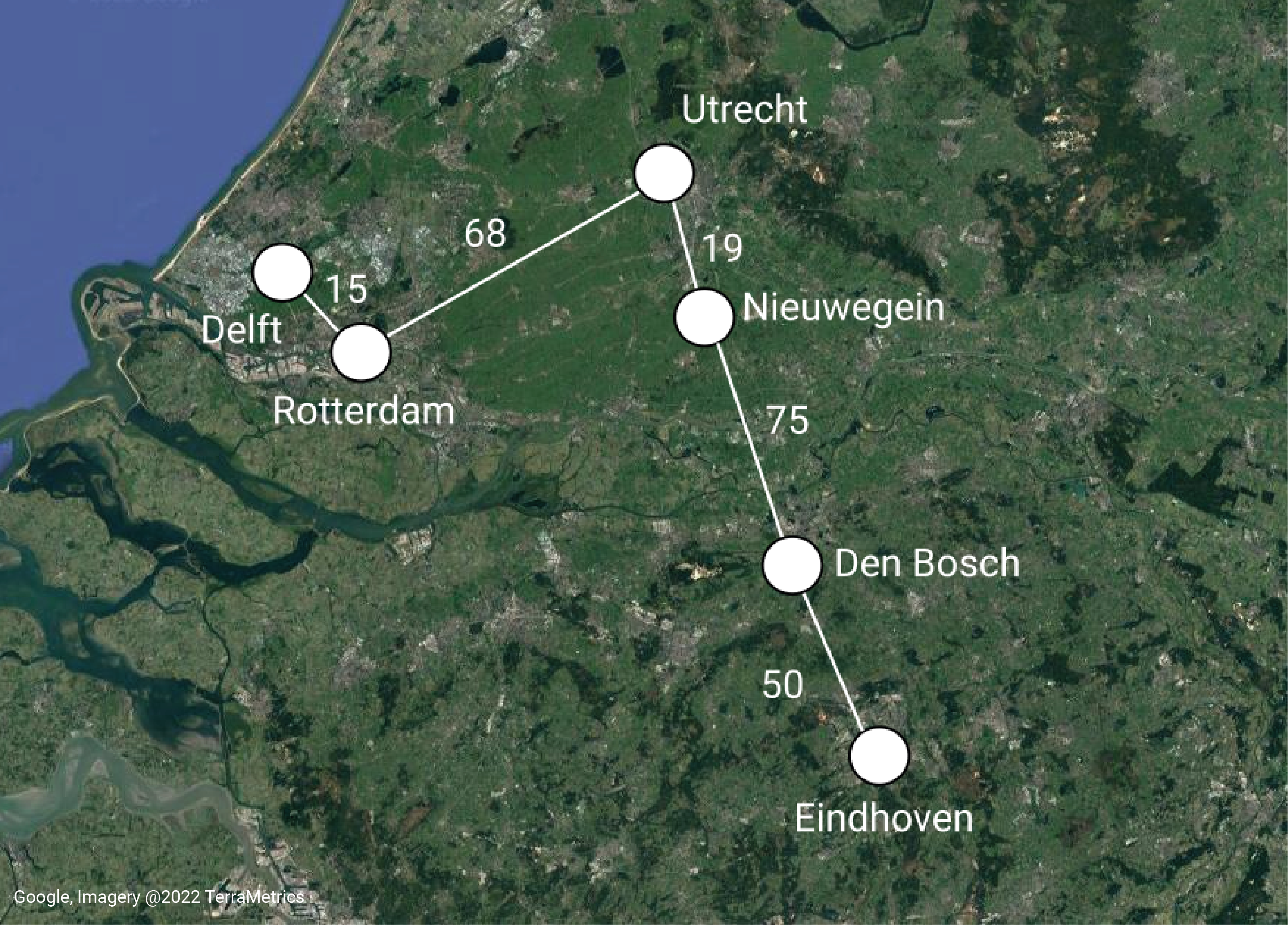}
\caption{Satellite photo of the Netherlands overlaid with depiction of the shortest connection between the Dutch cities of Delft and Eindhoven in SURF's fiber network.
The white circles represent locations where processing nodes and heralding stations can be placed, and are connected to one another through white fibers.
The position of the circles in the figure roughly approximates their actual physical location.
All distances are given in kilometers.}
\label{fig:delft_eindhoven_shortest_path}
\end{figure}
This means we are restricted to placing a single repeater between the end nodes, as a two-repeater setup would require five nodes in total, two for the repeaters and three for the heralding stations.
A single-repeater setup, on the other hand, requires only three nodes, one for the repeater itself and two for heralding stations.
One of the connection's nodes must therefore not be used, and there are two possible choices for how this can be done, as depicted in Figure~\ref{fig:alternative_d_e_paths}.
\begin{figure*}[!tbp]
    \centering
    \subfloat{{\includegraphics[width=0.48\columnwidth]{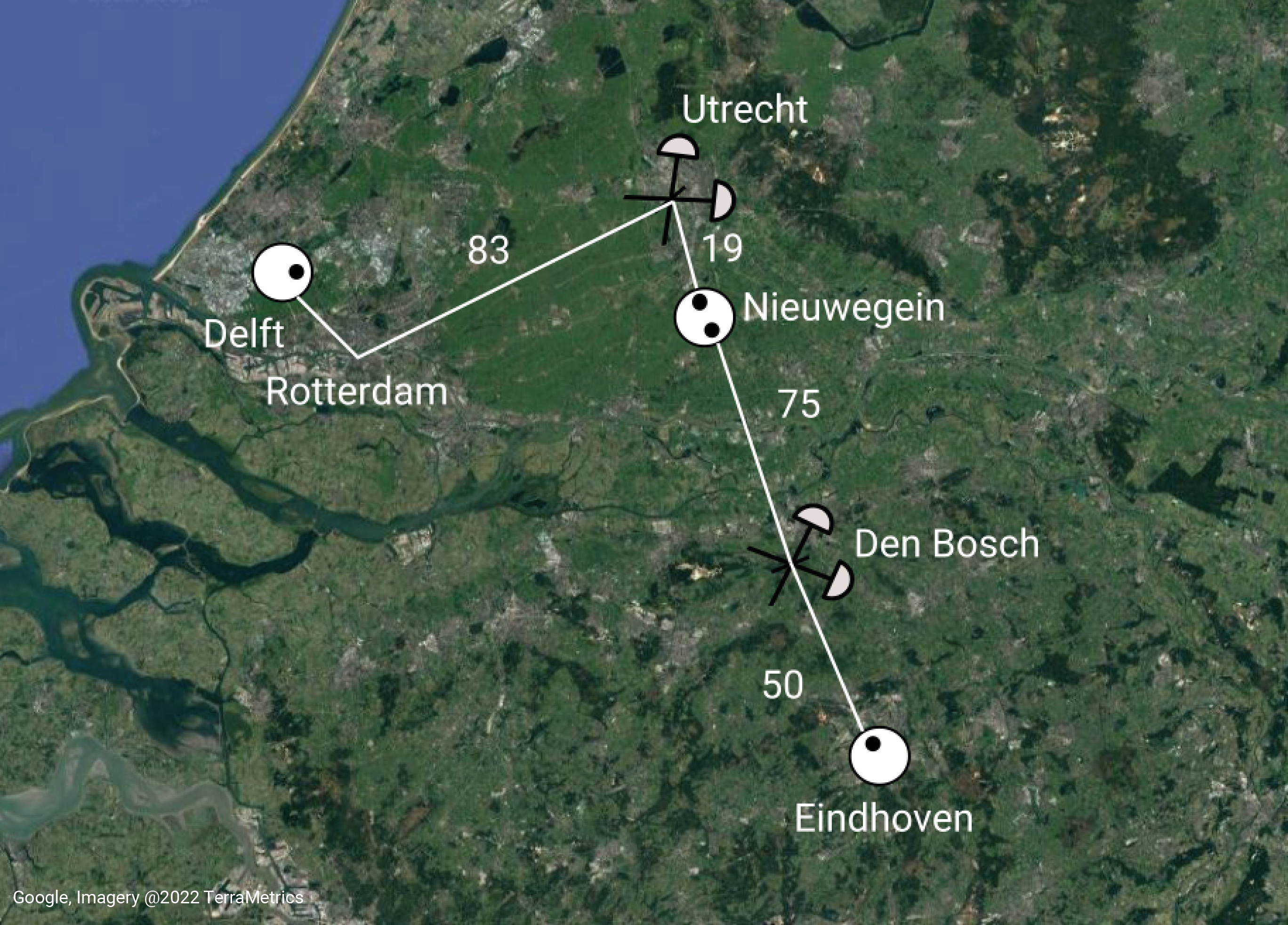}}}
    \qquad
    \subfloat{{\includegraphics[width=0.48\columnwidth]{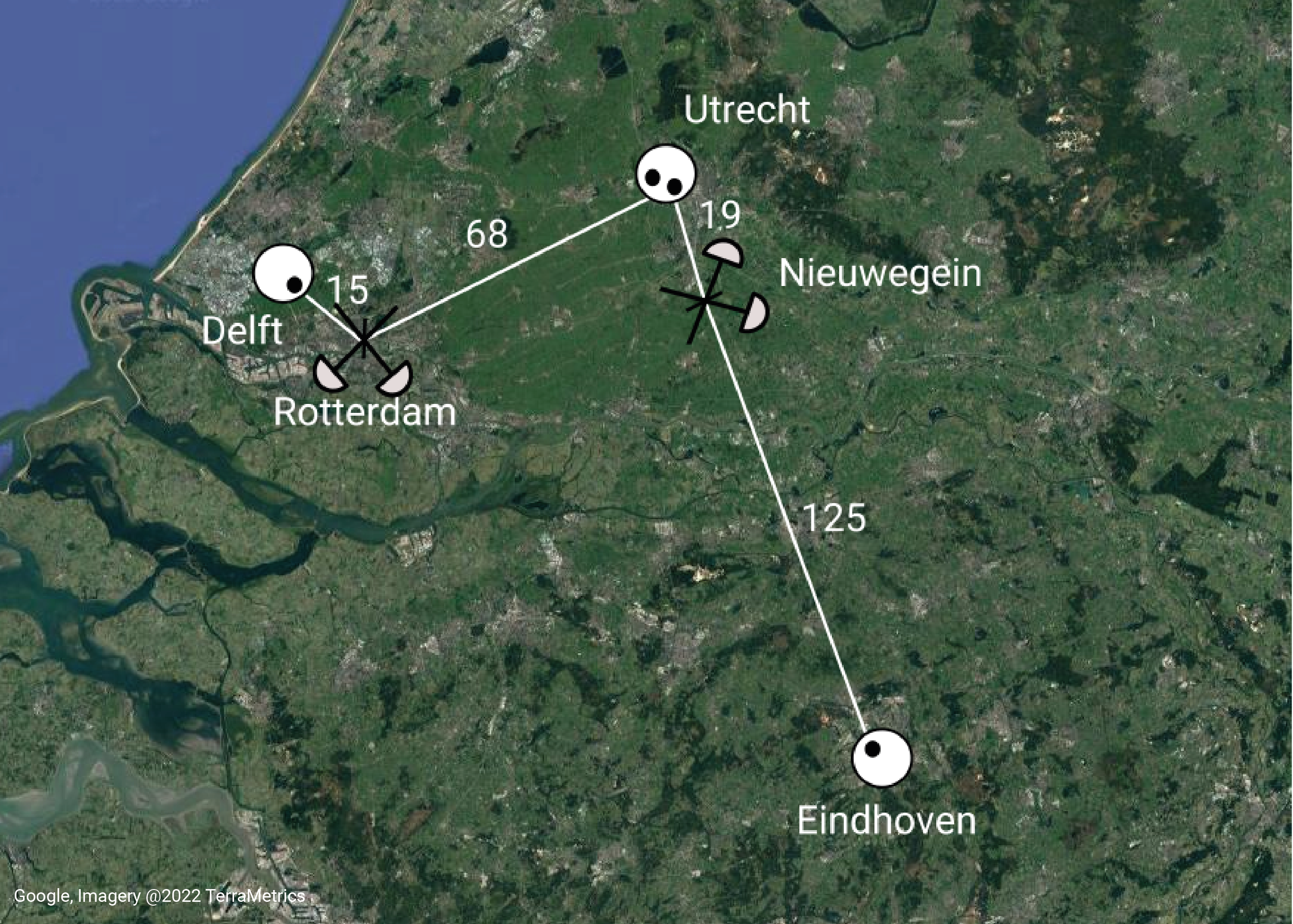}}}
    \caption{Two possible choices for processing node (white circles, black circles within represent qubits) and heralding station placement in the SURF network's shortest connection between the Dutch cities of Delft and Eindhoven.
    In the path on the left, the Rotterdam node is unused, thereby directly connecting the Delft - Rotterdam and Rotterdam - Utrecht links.
    Similarly, the Den Bosch node is unused in the path on the right.
}
    \label{fig:alternative_d_e_paths}
\end{figure*}
We applied our methodology to both of these paths and determined that the one on the left in Figure~\ref{fig:alternative_d_e_paths} requires a smaller improvement over current hardware.
Therefore, all the results presented in the main text pertain to it.
For more details, see Appendix~\ref{appendix:sec_alt_architecture}.
\\ 

\subsection{Repeater protocol}
We now elaborate on the protocol executed by the nodes.
We note that the repeaters we investigate are sequential, which means that they can only generate entanglement with one neighbor at a time.

\begin{enumerate}
    \item A request for end-to-end entanglement generation is placed at one of the end nodes.
    \item This end node sends a classical message through the fiber to the other end node, in order to verify whether it is ready to initiate the entanglement generation protocol.
    \item If that is the case, the second end node sends a confirmation message back, as well as an activation message for the repeater node.
    
    The next step is the generation of elementary link states.
    We begin by generating entanglement on the Eindhoven - Nieuwegein link, which is longest, so as to minimize the time states remain in memory.
    \item The neighboring Eindhoven and Nieuwegein nodes share classical messages sent through the fiber connecting them to ensure that both agree to generate entanglement.
    \item Once they have established agreement, entanglement generation attempts begin and continue until success.
    \item Steps 4 and 5 are repeated by the repeater and the Delft end node.
    \item The repeater performs a Bell-state measurement on the two qubits it holds, thereby creating an entangled state held by the end nodes.
    \item The outcome of this measurement is sent as a classical message to both end nodes.
    \item The end nodes become aware that end-to-end entanglement has been established and perform the appropriate correction on the Bell state.
\end{enumerate}

We also employ a cut-off protocol.
If the generation of the second entangled state lasted longer than a predefined cut-off time, the first state, corresponding to the longer link, is discarded.
Entanglement generation then restarts along the longer link.
\\ 

Such a protocol involving sequential repeaters and a cut-off timer has been studied before, e.g., in \cite{rozpedek2019near}.
The steps above are sufficient to generate one end-to-end entangled state.
If the generation of multiple states had been requested, steps 4-9 would be repeated until enough pairs had been generated.
We note that we do not simulate the application of the Bell-state correction, but instead record which correction should have been applied and handle it in post-processing.
\\ 

Information on how we implemented such a protocol in a scalable and hardware-agnostic fashion can be found in Appendix~\ref{appendix:protocols}.

\subsection{Quantum-computing server}
After the end-to-end entangled state has been generated, we assume the end node in Delft transfers its half of it to a powerful quantum-computing server.
This is a similar setting as the one investigated in~\cite{vardoyan2022quantum}, where the authors consider an architecture in which a node contains two NV centers, one of them used for networking and the other for computing.
We assume that the state transfer process is instantaneous and noiseless and that the server is always available to receive the state.
Additionally, we assume that all quantum gates performed by the server are noiseless and instantaneous, and that qubits stored in the server are subject to depolarizing memory noise with a coherence time of $T$ = 100 s.

\subsection{Processing nodes}
The quantum nodes we investigated are processing nodes, i.e. quantum nodes that are capable of storage and processing of quantum information.
This processing is done through noisy quantum gates.
The specific gate set available to the nodes depends on the particular hardware, but we model the gate noise of all of them with depolarizing channels. Measurements are also noisy, which is captured by a bit-flip channel, i.e. with some probability a $\ket{0}$ ($\ket{1}$) is read as $1$ ($0$).
Furthermore, as already mentioned, all the nodes we investigate are sequential, which means that they can only generate entanglement with one other node at a time.
\\ 

We now elaborate on the details of our modeling for each of the three nodes we study.

\subsection{Color centers}
In Table~\ref{tab:nv_baseline_parameters}, we present the baseline values of all color center hardware parameters relevant to our simulations, as well as references reporting their experimental demonstration.
\\

\begin{table}[!htpb]
\begin{tabular}{|c|c|c|}
\hline
Parameter                                      & Noise & Duration/Time                \\ \hline
Visibility                                     & 0.9~\cite{hermans2022qubit} & -                  \\ \hline
Probability of double excitation               & 0.06~\cite{hermans2022qubit} & -                 \\ \hline
$N_{1/\text{e}}$:
Nuclear dephasing during electron initialization  & 5300~\cite{hermans2022qubit} & -                \\ \hline
Dark count probability                              & $1.5\times 10^{-7}$~\cite{hermans2022qubit} & -  \\ \hline
$\sigma_{\text{phase}}$:
Interferometric phase uncertainty (rad)                        & 0.21~\cite{hermans2022qubit}        & -         \\ \hline
Photon detection probability excluding attenuation losses                                      & $5.1 \times 10^{-4}$~\cite{hermans2022qubit} & - \\ \hline
Spin-photon emission & $F$ = 1~\cite{pfaff2014unconditional} & 3.8 $\mu$s~\cite{pfaff2014unconditional} \\ \hline
Electron readout                                & $F$=0.93(0), 0.995(1)~\cite{hermans2022qubit} &  3.7 $\mu$s~\cite{humphreys2018deterministic}  \\ \hline
Carbon initialization                 & $F$=0.99~\cite{bradley2019ten}       &  300 $\mu$s~\cite{cramer2016repeated}       \\ \hline
Carbon Z-rotation                     & $F$=0.999~\cite{taminiau2014universal}&  20 $\mu$s~\cite{taminiau2014universal}              \\ \hline
Electron-carbon controlled X-rotation & $F$=0.97~\cite{kalb2017entanglement}  &  500 $\mu$s~\cite{kalb2017entanglement}             \\ \hline
Electron initialization               & $F$=0.995~\cite{bradley2019ten}    &  2 $\mu$s~\cite{reiserer2016robust}          \\ \hline
Electron single-qubit gate            & $F$=0.995~\cite{hermans2022qubit}     & 5 ns~\cite{kalb2017entanglement}               \\ \hline
Electron $T_1$                                   &   -   & 1 hour~\cite{abobeih2018one}         \\ \hline 
Electron $T_2$                                    & -      &    0.5 s~\cite{hermans2022qubit}     \\ \hline
Carbon $T_1$                                      & -    &   10 hours~\cite{bradley2019ten}      \\ \hline
Carbon $T_2$                                      & -         &  1 s~\cite{bradley2019ten}       \\ \hline
\end{tabular}
\caption{Baseline color center hardware parameters.}
\label{tab:nv_baseline_parameters}
\end{table}

Color center nodes are modeled with a star topology, with the communication qubit in the middle.
The memory qubits can all interact with the communication qubit, but not with one another.
The communication qubit owes its name to the fact that it is optically active, which means it can be used for light-matter entanglement generation.
The spin states of the memory qubits are long-lived, so they are typically used for information storage.
We model memory decoherence in color center qubits through amplitude damping and phase damping channels with $T_1$ and $T_2$ lifetimes.
The effect of the amplitude (phase) damping channel after time $t$ is given by Equation~\eqref{eq:amplitude_damping} (\eqref{eq:phase_damping}).
\begin{equation}
\rho \rightarrow \left(\ketbra{0}{0} + \sqrt{e^{-t/T_1}}\ketbra{1}{1}\right) \rho \left(\ketbra{0}{0} + \sqrt{e^{-t/T_1}}\ketbra{1}{1}\right)^{\dagger} + \sqrt{1 - e^{-t/T_1}}\ketbra{0}{1} \rho \left(\sqrt{1 - e^{-t/T_1}}\ketbra{0}{1}\right)^{\dagger}
\label{eq:amplitude_damping}
\end{equation}
\begin{equation}
\rho \rightarrow \left(1 - \frac{1}{2}\left(1 - e^{-t/T_2}e^{-t/(2T_1)}\right)\right)\rho + \frac{1}{2}\left(1 - e^{-t/T_2}e^{-t/(2T_1)}\right) Z\rho Z
\label{eq:phase_damping}
\end{equation}
\\ 

The $T_1$ and $T_2$ lifetimes of the communication qubit are different from those of the memory qubits.
An entangling gate is available in the form of a controlled X-rotation between the communication qubit and each memory qubit.
Furthermore, arbitrary single-qubit rotations can be implemented on the communication qubit.
\\ 

The constrained topology and gate set of the color center place some limitations on the quantum circuits to be executed.
First of all, the typical Bell-state-measurement circuit must be adapted, as depicted in Figure 17 (d) of the Supplementary Information of~\cite{coopmans2021netsquid}.
Furthermore, since only the communication qubit can be used to generate light-matter entanglement, the repeater node must move its half of the first entangled state it generates from the communication qubit to a memory qubit in order to free it up to generate the second entangled state.
The circuit for this move operation can be seen in Figure 17 (a) of the Supplementary Information of~\cite{coopmans2021netsquid}.
\\ 

Finally, we note that it might be advantageous for the end node in Eindhoven, which generates entanglement with the repeater first and then has to wait, to map its half of the elementary link entangled state from the communication qubit to the memory qubit while it waits for the repeater node to generate entanglement with the Delft node.
Note that this has nothing to do with the fact that the end node in Delft transfers its qubit to the powerful quantum-computing server after end-to-end entanglement is established.
There is however a trade-off: while mapping the state means that it will be held in a qubit with a longer coherence time, it will also undergo extra decoherence due to the noise in the gates that constitute the circuit for the move operation.
We have investigated this trade-off by applying our methodology to the two situations, and found that not mapping requires a smaller improvement over current hardware.
Therefore, the color center results shown in the main text pertain to the situation in which the Eindhoven node does not map its half of the entangled state to a memory qubit.
We must however note that this finding is specific to both the topology we study and the baseline hardware quality we consider.
For more details on the comparison between mapping and not mapping, see Appendix~\ref{appendix:sec_move_no_move}.
\\

\input{sections/nv_model}

\subsection{Trapped ions}
\label{sec:setup_ti}
In Table~\ref{tab:ti_baseline_parameters}, we present the baseline values of all trapped-ion hardware parameters relevant to our simulations, as well as references to the articles reporting their experimental demonstration.
\\

\begin{table}[!htpb]
\begin{tabular}{|c|c|c|}
\hline
Parameter                                      & Noise & Duration/Time                \\ \hline
Visibility                                     & 0.89 \cite{krutyanskiy2023}  & -                  \\ \hline
Dark count probability                              & $1.4\times 10^{-5}$ \cite{krutyanskiy2019light}& -  \\ \hline
Photon detection probability excluding attenuation losses                                      & 0.111 \cite{schupp_interface_2021, krutyanskiy2017, krutyanskiy2019light} & - \\ \hline
Ion-photon emission & $F$ = 0.99 \cite{stute_tunable_2012}  & 50 $\mu$s \cite{tirepeater, schupp_interface_2021} \\ \hline
Readout                                & $F$=0.999(0), 0.99985(1) \cite{myerson2008} &  1.5 ms \cite{krutyanskiy2023}  \\ \hline
Initialization                 & $F$=0.999 \cite{roos2006}      &  36 $\mu$s \cite{tirepeater}      \\ \hline
Z-rotation                     & $F$=0.99 \cite{tiprivate}&  26.6 $\mu$s \cite{tiprivate}              \\ \hline
M\o lmer-Sørensen gate & $F$=0.95 \cite{tirepeater} &  107 $\mu$s \cite{tirepeater}             \\ \hline
Coherence time                                   &   -   & 85 ms \cite{tirepeater}        \\ \hline 
\end{tabular}
\caption{Baseline trapped-ion hardware parameters.
A detection time window of 17.5 $\mu$s is assumed.
For the visibility, a coincidence time window of 0.5 $\mu$s is assumed (see Appendix \ref{sec:setup_ti} for further explanation).
The photon detection probability excluding attenuation losses includes a 30\% efficiency factor for quantum frequency conversion \cite{krutyanskiy2017}.
It is based on a detection efficiency of ~0.43 for a 46(1) MHz drive laser and a detection time window of 17.5 $\mu$s \cite{schupp_interface_2021}.
However, the number from \cite{schupp_interface_2021} is based on a detector efficiency of 0.87(2) for photons at 854 nm.
The detection efficiency at telecom frequency would instead be ~0.75 using superconducting nanowire detectors \cite{krutyanskiy2019light},
giving an additional conversion factor of ~0.75/0.87.
The dark count probability is based on a 0.8 Hz dark count rate for telecom superconducting nanowire detectors \cite{krutyanskiy2019light} multiplied by 17.5 $\mu$s.
The ion-photon emission fidelity has been corrected for the 1.5\% infidelity due to dark counts in \cite{stute_tunable_2012}.
The initialization duration includes time for cooling sequences and repumping (3 ms of cooling for 230 photon generation attempts on one ion, with 40 $\mu$s for repumping and optical pumping in 30 of the attempts and 20 $\mu$s in 200 of the attempts, averaging at $\sim36\mu$s per attempt \cite{tiprivate}).
}
\label{tab:ti_baseline_parameters}
\end{table}

In this work, we present for the first time a NetSquid model for trapped-ion nodes in quantum networks.
The trapped-ion nodes we model are based on the state of the art for trapped ions in a cavity, which consists of $^{40}$Ca$^+$ ions in a linear Paul trap
\cite{barrosDeterministicSinglephotonSource2009, casaboneHeraldedEntanglementTwo2013, kellerContinuousGenerationSingle2004, krutyanskiy2019light, meraner_indistinguishable_2020, schindlerQuantumInformationProcessor2013, schupp_interface_2021, stuteIonPhotonQuantum2012, stute_tunable_2012, walker2020, walker2018, zwerger2017quantum}
(we note that promising results have also been achieved for trapped ions without cavities, these systems are however not considered in this work
\cite{stephensonHighRateHighFidelityEntanglement2020, vanleent2021, crocker2019a, inlek2017, nadlinger2021}.
In our model they have all-to-all connectivity, their qubits all have the same coherence time and can all be used to generate light-matter entanglement.
However, the node can only generate entanglement with one remote node at a time.
\\

Decoherence in $n$ trapped-ion qubits is modeled through a collective Gaussian dephasing channel that has the following effect on the $n$-qubit state $\rho$ \cite{zwerger2017quantum}:
\begin{equation}
\label{eq:collective_gaussian_dephasing}
\rho \to \int_{- \infty}^\infty K_r  \rho K_r^\dagger p(r) dr,
\end{equation}
where
\begin{equation}
K_r =  \exp(-\text{i} r \frac t \tau \sum_{j=1}^n Z_j),
\end{equation}
$Z_j$ denotes a Pauli $Z$ acting on qubit $j$, $\tau$ the coherence time and $t$ the storage time, and
\begin{equation}
p(r) = \frac{1}{\sqrt{2\pi}}\text{e}^{-r^2/2}.
\end{equation} 
This can be read as follows: the qubits dephase because they undergo Z-rotations at an unknown constant rate of $-2r$ per coherence time $\tau$.
This is modeled by sampling the Gaussian distribution for the dephasing rate, $p(r)$, for each ion trap each time its state is reset.
The qubits are then time-evolved by applying unitary rotations in accordance with the sampled value for $r$.
The baseline value $\tau = 85$ ms included in Table \ref{tab:ti_baseline_parameters} is obtained from \cite{tirepeater}.
However, the value for the coherence time reported there is $62 \pm 3$ ms.
The reason for this discrepancy is a difference in convention.
To see this, we can evaluate Equation \eqref{eq:collective_gaussian_dephasing} for $n=1$, i.e., for a single qubit.
In that case, we find
\begin{equation}
\rho \to \lambda \rho + (1 - \lambda) Z \rho Z,
\end{equation}
where
\begin{equation}
\lambda = \frac 1 2 \left( 1 - \text{e}^{- 2 \left( \frac t \tau \right)^2} \right).
\end{equation}
The single-qubit dephasing model used in \cite{tirepeater} instead has
\begin{equation}
\rho \to \lambda' \rho + (1 - \lambda') Z \rho Z,
\end{equation}
where
\begin{equation}
\lambda' = \frac 1 2 \left( 1 - \text{e}^{- \left( \frac t {\tau'} \right)^2} \right).
\end{equation}
Here, $\tau'$ is the coherence time in their model.
The models are exactly equivalent for $\tau = \sqrt 2 \tau'$.
Therefore, the reported value $\tau' = 62 \pm 3$ ms corresponds to $\tau = 88 \pm 4$ ms.
The value we use, $\tau = 85$ ms, represents a conservative interpretation of the result presented in \cite{tirepeater}.
Our model for the storage of quantum states in ionic qubits has been validated against experimental data from \cite{tirepeater}.
In this experiment, ion-photon entanglement is created with one ion in a two-ion device.
Next, ion-photon entanglement is created with the other ion every 330$\mu$s.
Our simulation results are compared to the experimental results in Figure \ref{fig:ti_memory}.
\\
\begin{figure}[!ht]
\centering
\includegraphics[width=0.6\columnwidth]{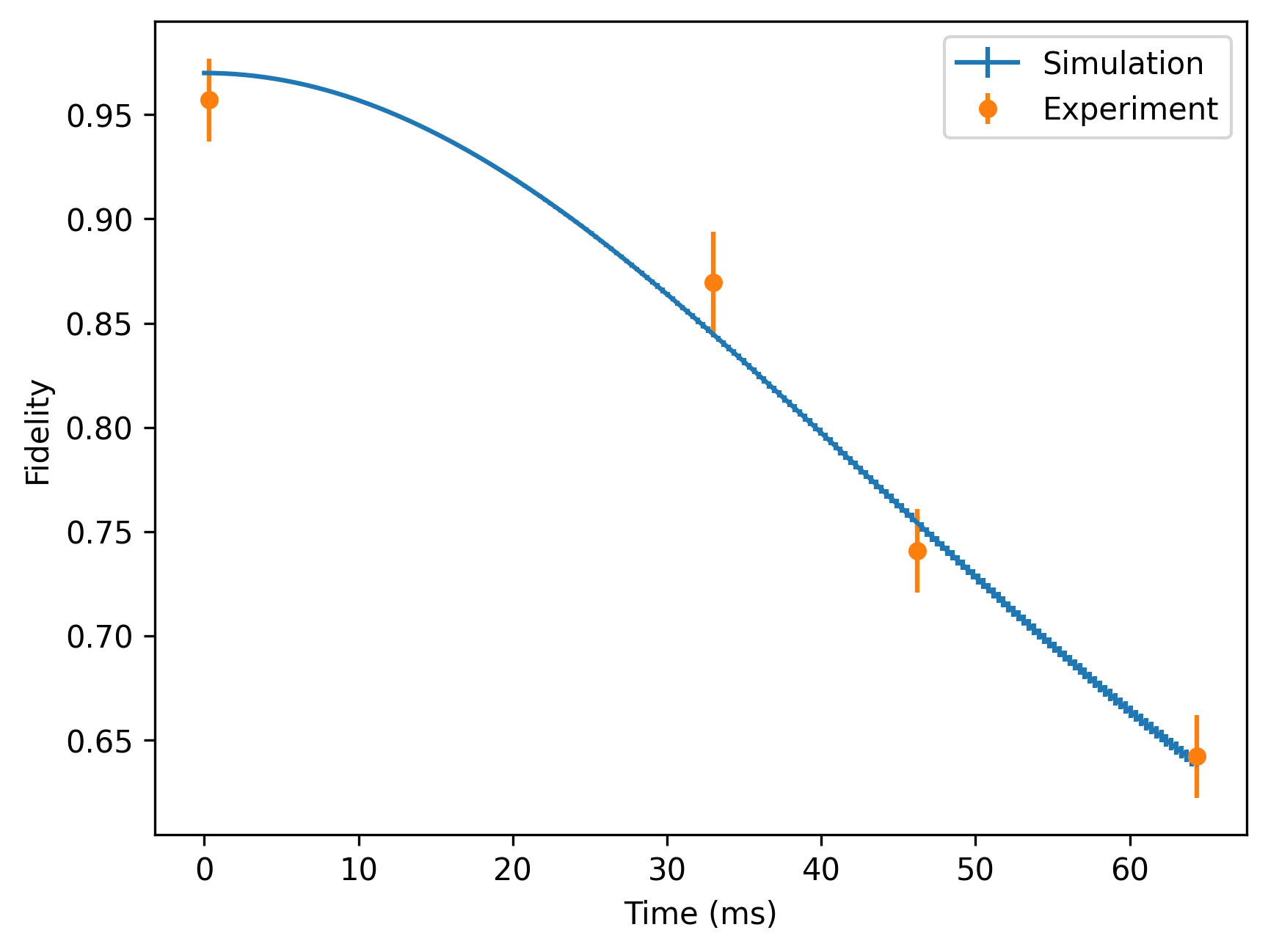}
\caption{
Validation of our trapped-ion decoherence model against an experiment in \cite{tirepeater}.
In the experiment, a trap with two ions first emits a photon entangled with the first ion, and then keeps emitting new photons entangled with the second ion every 330 $\mu$s.
The figure shows the evolution of the fidelity to the perfect Bell state of the state shared by the first ion and the photon entangled to it as a function of time.
Error bars of the simulation results represent the standard error of the mean and are sometimes hard to distinguish because of their size.
The simulation has been conducted using the baseline coherence time $\tau = 85$ ms, which is the value obtained from \cite{tirepeater}.
The ion-photon emission fidelity has been set to $F=0.97$ to tune the fidelity at time zero such that good agreement between the simulation and the experiment is obtained.
All other parameters have been set to their perfect values.
}
\label{fig:ti_memory}
\end{figure}

The entangling gate available to the trapped-ion qubits as we model them is the Mølmer-Sørensen gate~\cite{molmer1999multiparticle}.
The gate set also includes arbitrary single-qubit Z-rotations and collective rotations around a tunable axis in the XY plane~\cite{schindlerQuantumInformationProcessor2013}.
The Bell-state-measurement circuit is implemented as a Z-rotation of angle $\pi/4$ for one qubit and -$\pi/4$ for the other, a Mølmer-Sørensen gate and a measurement of both qubits in the computational basis.
All gates are modeled as a perfect gate followed by depolarizing channels on all partaking qubits.
\\

Just as for color centers, entanglement generation through the Barrett-Kok protocol is modeled using the model introduced in Appendix \ref{appendix:sec_double_click_model}.
A difference with color centers, however, is that the photons emitted by ions are typically temporally impure due to off-resonant scattering \cite{meraner_indistinguishable_2020}, resulting in low Hong-Ou-Mandel visibility and hence entangled-state fidelity.
This can be counteracted by using a stringent detection time window and by imposing a coincidence time window.
A click pattern is then only heralded as a success in case both clicks fall within the detection time window and the time between the two clicks does not exceed the coincidence time window.
The detection time window and coincidence time window can be tuned to increase the visibility, but at the cost of having a smaller success probability.
In order to account for the effect of the detection time window, we employ a toy model for the temporal state of photons emitted from trapped-ion devices.
This toy model does not accurately represent the true state of the emitted photons, but as we show in Figure \ref{fig:ti_coincidence_time}, it can be used to capture the trade-off between success probability and visibility well.
In this toy model, we model photons as mixtures of pure photons emitted at different times.
The pure photons have one-sided exponential wavefunctions, and the emission time is also distributed according to a one-sided exponential.
Under these assumptions, the detection probability, coincidence probability and visibility can all be exactly calculated as a function of the in total two parameters that describe these two exponentials.
These calculations are performed in Appendix \ref{app:time_windows}, and the results can be used in conjunction with the model in Appendix \ref{appendix:sec_double_click_model} to calculate the success probability and state.
To show that this model can be used to capture the success probability and visibility with good accuracy, we have performed a joint least-square procedure for the detection-time probability density function, the coincidence probability and the visibility to match the two free parameters to the data presented in \cite{meraner_indistinguishable_2020}.
This data has been produced by emitting two photons from the same trapped-ion device, frequency converting these photons, and then making them interfere.
In Figure \ref{fig:ti_coincidence_time} (a), we show the resulting theoretical results and compare them to the experimental results.
\\

\begin{figure}[!ht]
\subfloat[][]{
\includegraphics[width=0.49\textwidth]{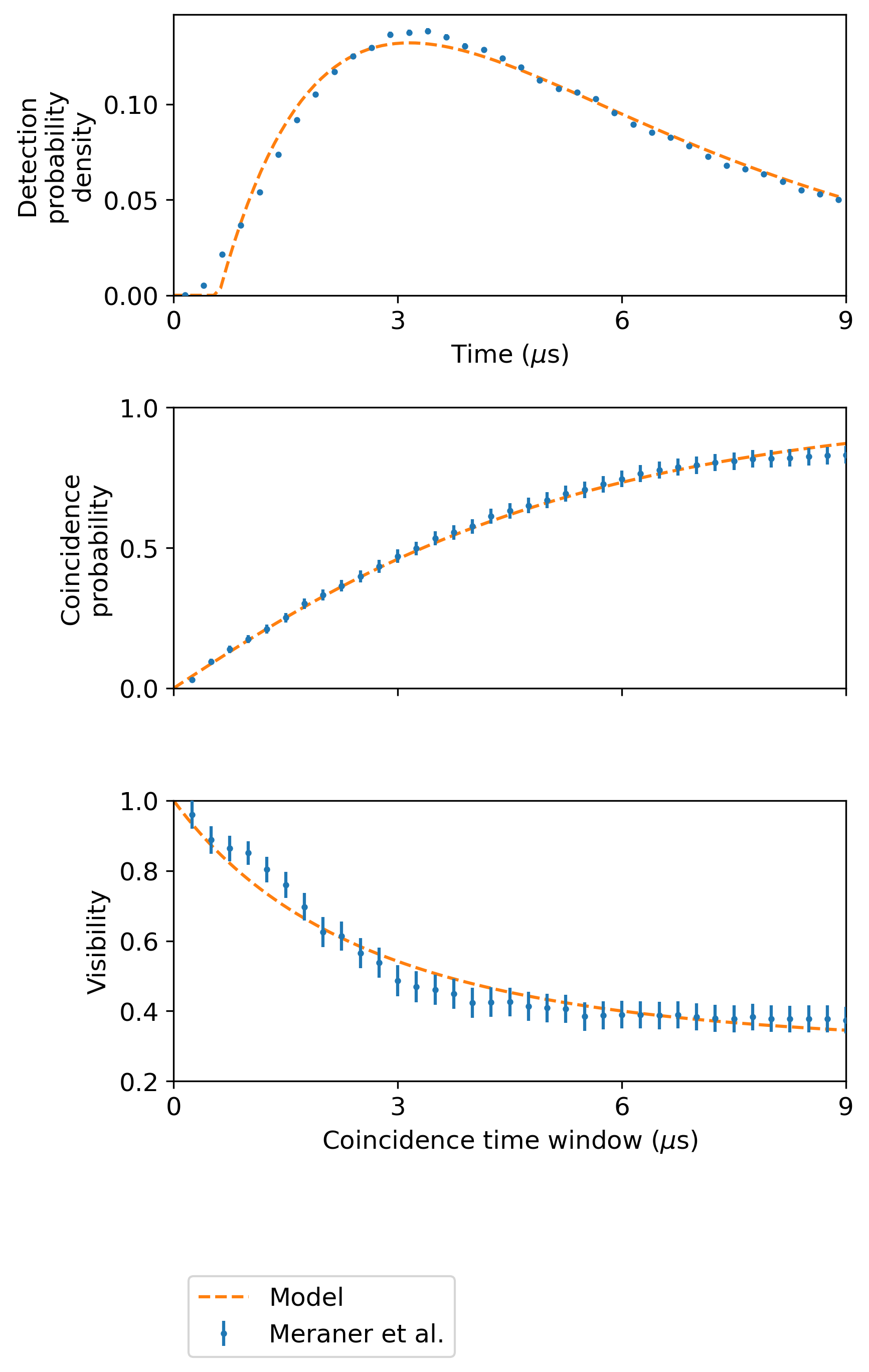}
}
\subfloat[][]{
\includegraphics[width=0.49\textwidth]{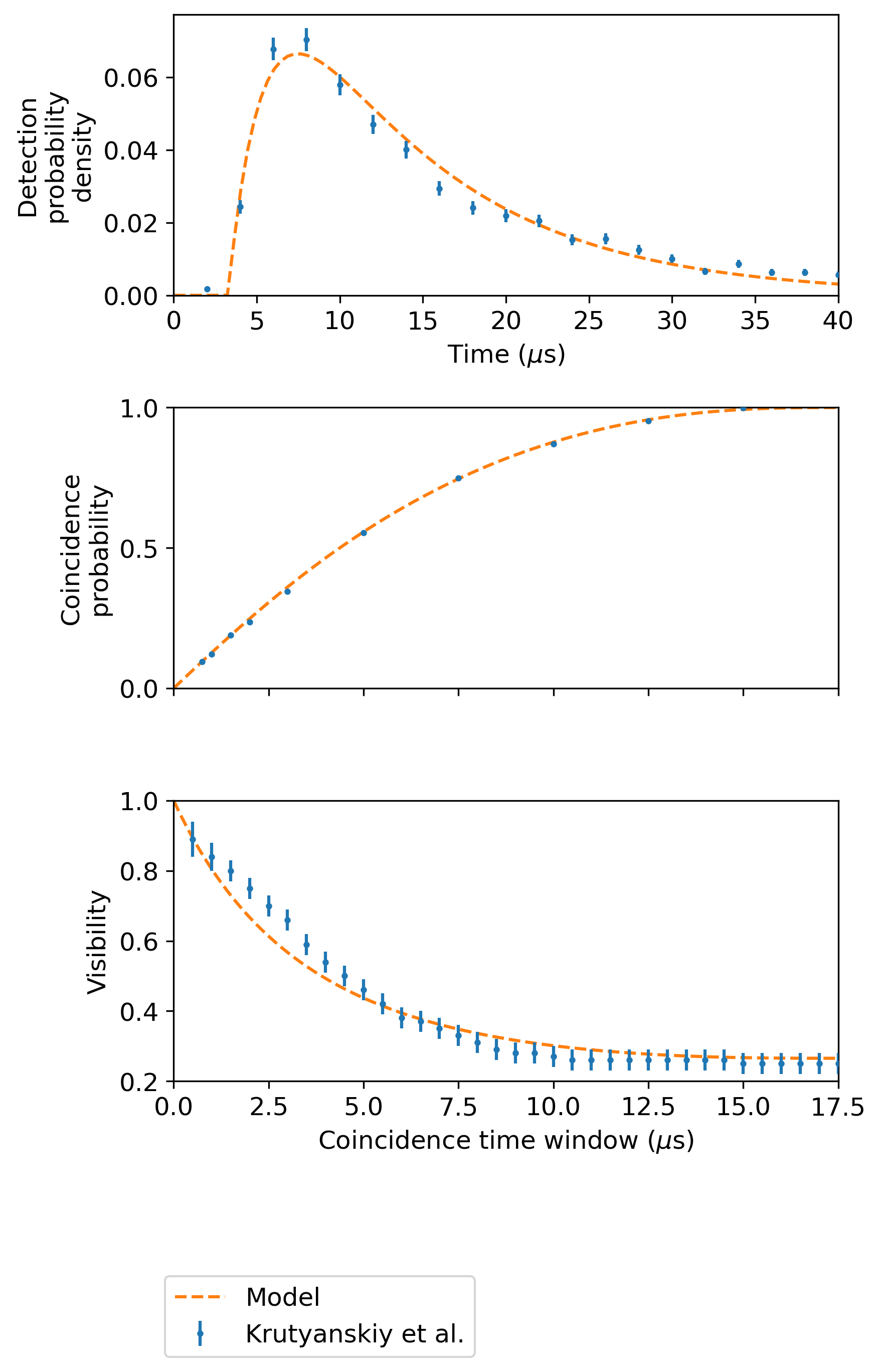}
}
\centering
\caption{
Comparison between data from two different experiments and the toy model introduced in Appendix \ref{app:time_windows}.
In both experiments, the detection probability density and coincidence probability were not conditioned on the successful detection of two photons.
To account for this, we have multiplied both the detection-probability-density data and the coincidence-probability-density data by a different overall scaling factor.
Both the scaling factors, the parameters of the two exponentials describing the photon state and an offset for the detection probability density have been determined using a least-squares procedure.
The least-squares procedure has been performed jointly for the three data sets corresponding to the same experiment by summing the square errors of all three.
Here, the largest weight has been given to the detection probability density ($10^6$), the second largest to the visibility ($10^5$), and the smallest to the coincidence probability (1).
\textbf{(a)} Comparison to data from Meraner et al. \cite{meraner_indistinguishable_2020}.
In the experiment, the Rabi pulse was terminated after approximately 9 $\mu$s, therefore we have only compared the first nine $\mu$s.
Because of the terminated pulse the detection probability density falls to zero at approximately 12 $\mu$s.
Therefore, effectively the entire wave packet is detected.
To reproduce this in our model, we have not implemented a detection time window (or equivalently, have set the detection time window to infinite).
The fitted half-life times of the exponentials representing the wave function and emission time are 2.40 $\mu$s and 2.76 $\mu$s respectively.
\textbf{(b)} Comparison to data from Krutyanskiy et al. \cite{krutyanskiy2023}
We base the modeling for the visibility and coincidence probability of ion traps in this paper on the fit shown here.
A detection time window of 17.5 $\mu$s was used in the experiment.
The detection-probability-density data used here corresponds to ``node A'' from \cite{krutyanskiy2023}.
The fitted half-life times of the exponentials representing the wave function and emission time are 3.01 $\mu$s and 6.79 $\mu$s respectively.
}
\label{fig:ti_coincidence_time}
\end{figure}

Instead of basing the parameters we use in our simulations on \cite{meraner_indistinguishable_2020}, we base them on data for the interference of photons emitted by two distinct ion traps \cite{krutyanskiy2023}, as this more accurately represents the scenario we investigate in this study.
We determine again the two parameters that describe the two exponentials by fitting to the data using the exact same method as above.
The half-life time of the fitted exponentials representing the wave function and the emission time were found to be 3.01 $\mu$s and 6.79 $\mu$s respectively, with the fits and the data shown in Figure \ref{fig:ti_coincidence_time} (b).
This data has been taken using a detection time window of 17.5 $\mu$s.
Therefore, for consistency, we use a fixed detection time window of 17.5 $\mu$s throughout our simulations, and hence the parameters shown in Table \ref{tab:ti_baseline_parameters} (such as, e.g., the photon detection probability excluding attenuation losses and the dark count probability) all assume a detection time window of 17.5 $\mu$s.
On the other hand, the coincidence time window is treated as a freely tunable parameter, allowing for a trade-off between rate and fidelity.
The value for the visibility reported in Tables \ref{tab:baseline_parameters} and \ref{tab:ti_baseline_parameters} was obtained from the model in Appendix \ref{app:time_windows} using the fitted parameters reported above, a detection time window of 17.5 $\mu$s and a coincidence time window of 0.5 $\mu$s.
\\

We note that \cite{meraner_indistinguishable_2020} includes a physically-motivated theoretical model for the trade-off between coincidence probability and visibility as a function of the detection and coincidence time windows.
We have not used their model here as it requires numerical integration to evaluate, while our model can be rapidly evaluated using an analytical closed-form expression.
Additionally, our goal here is not to predict the behaviour of a specific physical system but to accurately represent the trade-off between rate and fidelity without overfitting to experimental data.
Finally, as our toy model does not attempt to closely capture the physics of any individual system, it can be considered to be system agnostic.
It could thus be fitted to different types of photon sources, giving it a potentially broader scope of application.

\subsection{Abstract nodes}
The purpose of the abstract nodes is to provide a general model for processing nodes.
Therefore, their modeling is kept simple and platform-agnostic: there is all-to-all connectivity between the qubits, all of them can be used to generate light-matter entanglement, they all have the same coherence time properties and all quantum gates are available.
The Bell-state measurement circuit implemented by abstract nodes is the usual one: a controlled-NOT gate, followed by a Hadamard on the control qubit and a measurement of both qubits in the computational basis.
We note that this model and its NetSquid implementation are not novel, having first been introduced in~\cite{da2021optimizing}. 
\\

\input{sections/abstract_model}

\subsection{Entanglement generation}
\label{sec:setup_ent_gen}
For near-term parameters, the success probability of entanglement generation is very low.
This means that many entanglement generation attempts are required, and that a simulation of this process would spend most of its time simulating failed attempts.
This is computationally very inefficient, so we instead perform entangled state insertion, through a process we call \textit{magic}~\cite{netsquid-magic}.
This process was first introduced in \cite{dahlberg2019link}.
\\

Magic works as follows: once two nodes have decided to generate entanglement together, we sample from a geometric distribution in order to determine how many attempts would have been required to succeed.
The success probability of this geometric distribution is limited by the product of the probabilities of emitting the photon in the correct mode, capturing it into the fiber, frequency-converting it, transmitting it through the fiber and detecting it at the detector.
Furthermore, imperfections such as the imperfect indistinguishability of interfering photons and detector dark counts also impact the success probability.
Their effect depends on whether a single-click or double-click protocol is used.
\\

The elapsed time for the entanglement generation process is given by the product of the sampled number of required attempts and the duration of one attempt, which is in turn given by the sum of the emission time and the photon travel time.
\\

The state generated is given by an analytical model which is different for single and double-click entanglement generation.
For more details, see Appendix~\ref{appendix:sec_double_click_model}.

%% file: sections/nv_model.tex
The color center hardware model we employed builds on previous work~\cite{rozpedek2019near,dahlberg2019link}, and its NetSquid implementation has been validated against experiments~\cite{coopmans2021netsquid}.
This includes the model for the processor as well as for the entangled states generated through a single-click protocol.
The main novelty introduced in this work regarding color center modelling is a model for the entangled states generated through the Barrett-Kok protocol~\cite{barrett2005efficient}.
This is essentially the model introduced in Appendix~\ref{appendix:sec_double_click_model}, with the addition of induced dephasing noise.
This addition accounts for the fact that every entanglement generation attempt induces dephasing noise on the color center's memory qubits~\cite{kalb2018dephasing}.
We simulate this effect using a dephasing channel.
The dephasing probability $p$, accumulated after possibly multiple entanglement generation attempts, is given by Equation~\eqref{eq:accumulated_dephasing_prob}.
\begin{equation}
p = \frac{1 - (1 - 2p_{\text{single}})^k}{2}.
\label{eq:accumulated_dephasing_prob}
\end{equation}
In this equation, $p_{\text{single}}$ is the probability of dephasing after a entanglement generation attempt and $k$ is the number of required entanglement generation attempts.
In our simulations, we apply a dephasing channel of parameter $p$ twice after entanglement has been successfully generated, to reflect the fact that each attempt requires the emission of two photons.
$p_{\text{single}}$ can be related to $N_{1/\text{e}}$, the number of electron spin pumping cycles after which the Bloch vector length of a nuclear spin in the state $(\ket{0} + \ket{1})/\sqrt{2}$ in the $X - Y$ plane of the Bloch sphere has shrunk to $1/\text{e}$ when the electron spin state has bright-state parameter $\alpha = 0.5$, through Equation~\eqref{eq:prob_dephasing_n1e}.
\begin{equation}
p_{\text{single}} = (1 - \alpha)\left(1 - e^{-1/N_{1/\text{e}}}\right).
\label{eq:prob_dephasing_n1e}
\end{equation}
$N_{1/\text{e}}$ can in turn be experimentally determined, and $N_{1/\text{e}} = 5300$ for state-of-the-art color center experiments~\cite{pompili2021realization}.
\\

The double-click model is the only component of our color center simulations that had not yet been compared to experimental data.
With this in mind, we validated it against the experiment reported in~\cite{bernien2013heralded}.
There, the authors demonstrated heralded entanglement generation between two color centers separated by three meters using the Barrett-Kok protocol.
After establishing entanglement, measurements of the two entangled qubits were performed to investigate whether the outcomes were correlated as expected.
This was repeated for the X and Z bases, and for the states $\ket{\Psi^\pm} = 1/\sqrt{2}\left(\ket{01} \pm \ket{10}\right)$.
We replicated this setup using our color center NetSquid model and ran the experiment 10000 times per measurement basis in order to gather relevant statistics.
The results of this validation are shown in Figure~\ref{fig:validate_nv_double_click}.
\begin{figure}[!ht]\centering
\subfloat[$\ket{\Psi^+}$ measured in X basis.]{\includegraphics[width=0.45\textwidth]{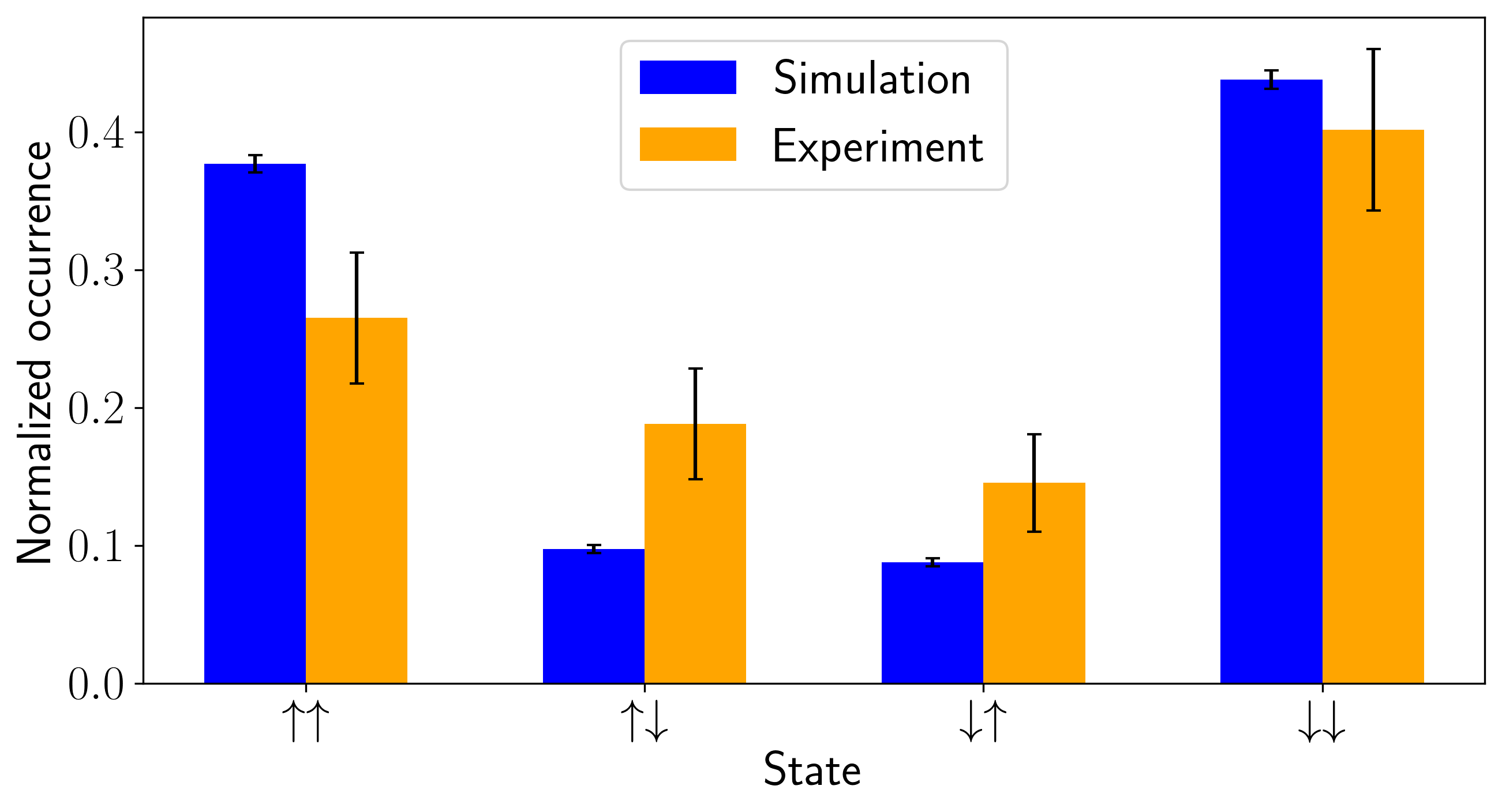}}
\subfloat[$\ket{\Psi^-}$ measured in X basis.]{\includegraphics[width=0.45\textwidth]{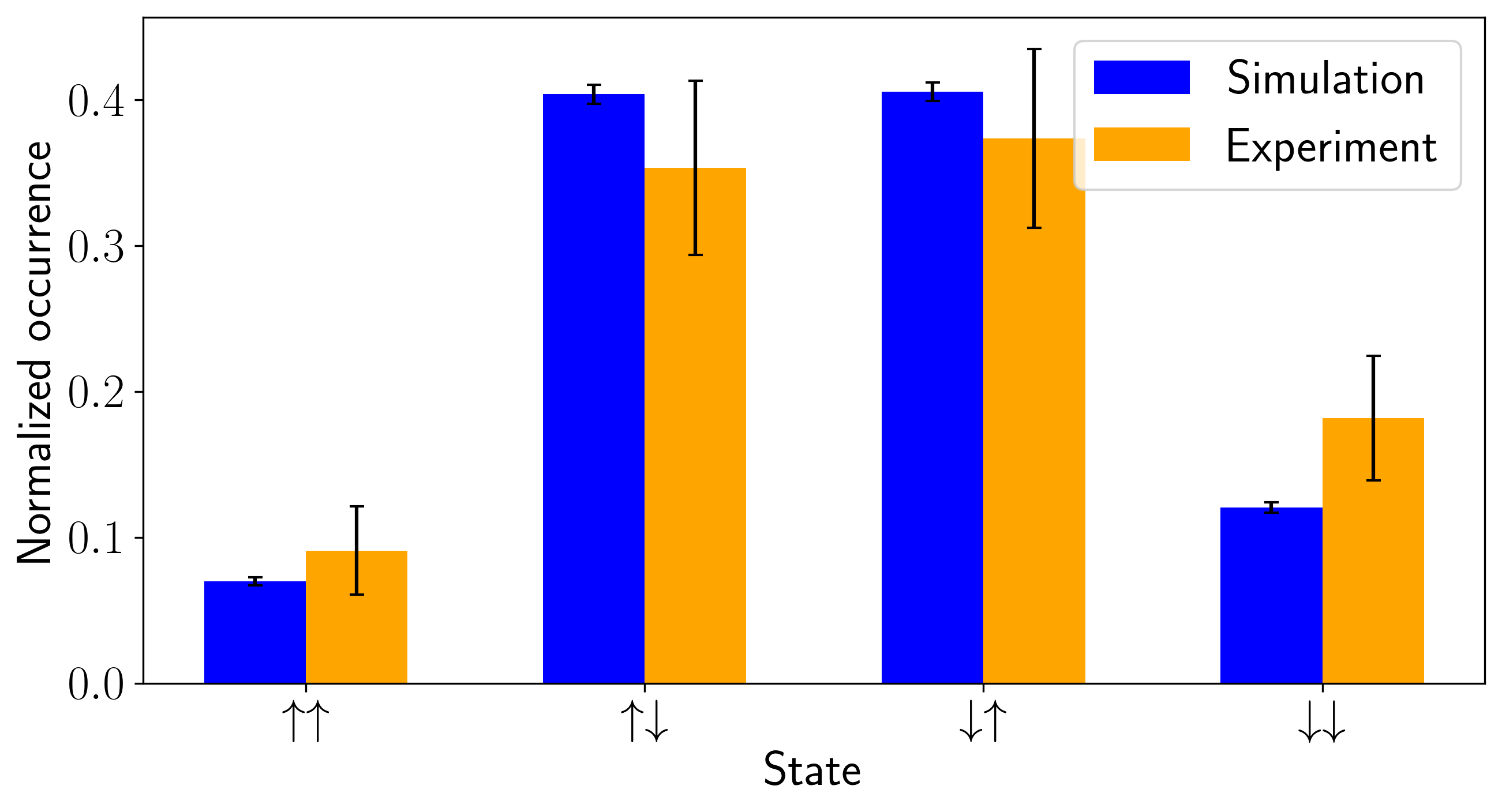}}
\\
\noindent
\subfloat[$\ket{\Psi^+}$ measured in Z basis.]{\includegraphics[width=0.45\textwidth]{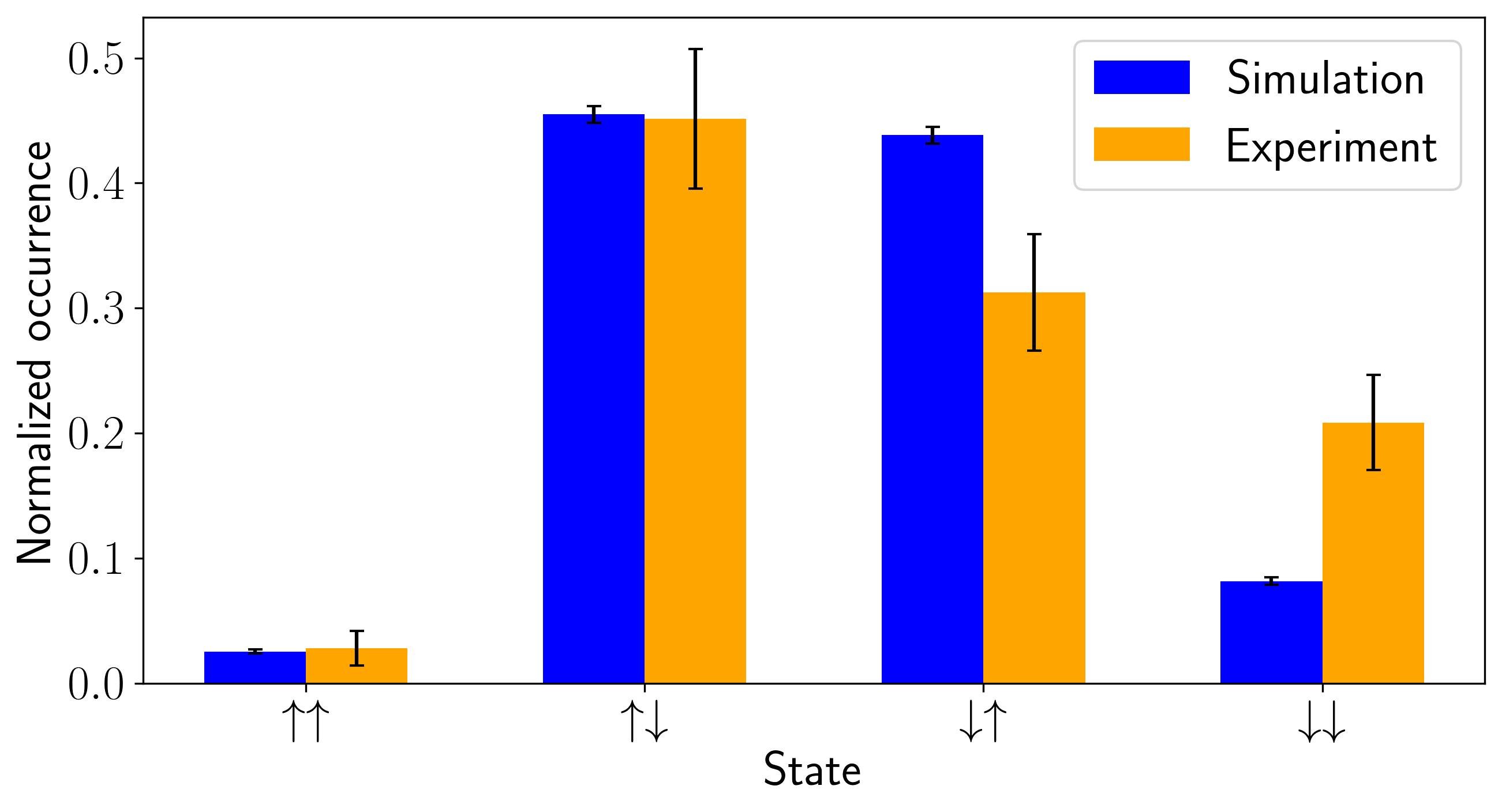}}
\subfloat[$\ket{\Psi^-}$ measured in Z basis.]{\includegraphics[width=0.45\textwidth]{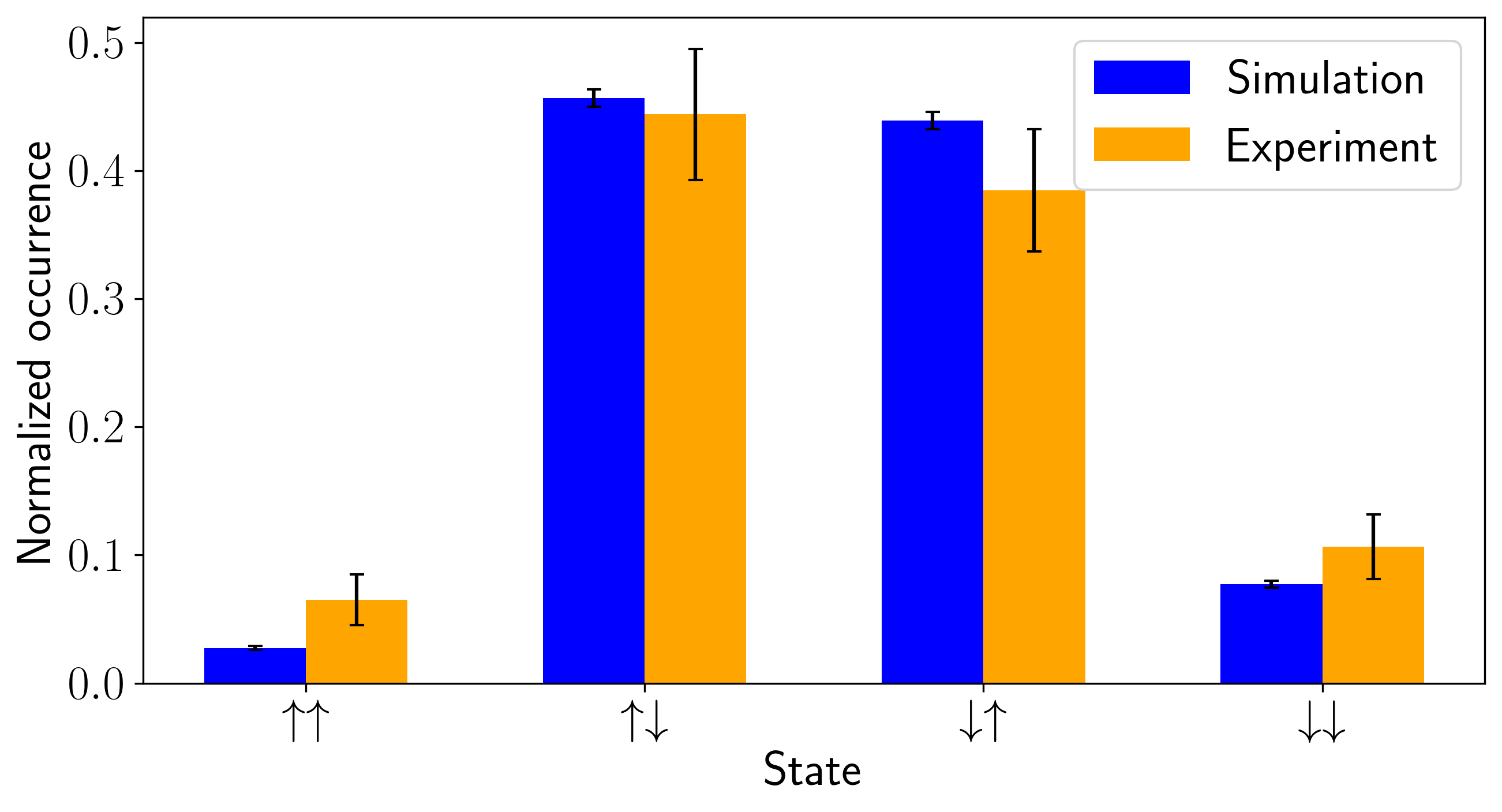}}
\caption{Comparison of measurement outcomes of the entangled state generated using the Barrett-Kok protocol in the experiment described in~\cite{bernien2013heralded} and our simulation of the same scenario.
The plots on the left (right) correspond to the case in which the state $\ket{\Psi^+} = 1/\sqrt{2}\left(\ket{01} + \ket{10}\right)$ ($\ket{\Psi^-} = 1/\sqrt{2}\left(\ket{01} - \ket{10}\right)$) is generated.
The plots above (below) show the outcomes when measuring in the X (Z) basis.
The error bars depict the standard error of the mean.
Anti-correlation of the spin states is expected for every plot except for the one in the top left, for which we expect to see a correlation.
The smaller dimension of the simulation error bars can be attributed to the number of executions of the protocol, which was of the order of 10000 per plot.
This is two orders of magnitude more than what was performed experimentally.}
\label{fig:validate_nv_double_click}
\end{figure}
The results obtained with our simulation model broadly replicate the experimental results, although they do not lie within the statistical error bars.
Overall, the simulation results are closer to the ideal case of perfect (anti-)correlation.
This can be explained by the fact that our model for the double-click states is quite simple and hardware-agnostic, ignoring noise sources such as the probability of double photon emission.
Further, the number of experimental data points is small, of the order of a total of 100 events for each of the plots in the figure.
Nonetheless, considering the simplicity of the model, we believe that the level of agreement is satisfactory.

%% file: sections/abstract_model.tex
\label{appendix:sec_abstract_model}
In order to quantify the level of accuracy that is sacrificed by considering a model with a higher degree of abstraction, we compare the performance of a single abstract-node repeater in the Delft-Eindhoven path to the equivalent color center and trapped ion setups.
To do so, we require a method of converting hardware parameters from the more in-depth models to the abstract model.
We therefore start by introducing this mapping.

\subsubsection{Color center to abstract model mapping}\label{sec:map_nv_abstract}
The emission fidelity, visibility, dark count probability and probability of photon detection excluding attenuation losses are mapped without change from the color center model to the abstract model.
An entanglement swap in an NV platform consists of one-qubit gates on both carbon and electron, two-qubit gates and measurement and initialization of the electron (see Figure 17 in Supplementary Note 5 of~\cite{coopmans2021netsquid} for an image of the circuit).
Imperfections in gates and initialization are modelled by depolarizing channels in the NV model, while the measurement error is modelled by probabilistic bit flips.
In mapping to the abstract model we approximate the measurement error as a depolarizing channel.
All the errors associated to the operations in the circuit are then multiplied together to obtain a single parameter $s_{\text{q}}$.
$1-s_{\text{q}}$ is used to parameterize a depolarizing channel applied after a perfect Bell-state measurement.
The action of this depolarizing channel on a given state $\rho$ as a function of $s_{\text{q}}$ is given by Equation~\eqref{eq:depolarize_swap_quality}, from which we can see that $s_{\text{q}}$ is a measure of the quality of an entanglement swap.
\begin{equation}
    \phi(\rho, s_{\text{q}}) = \left(\frac{1 + 3s_{\text{q}}}{4}\right)\rho + \frac{1 - s_{\text{q}}}{4}(X\rho X + Y\rho Y+ Z\rho Z).
\label{eq:depolarize_swap_quality}
\end{equation}
In our color center model, the coherence times of the carbon spins are different from those of the electron spin.
This subtlety is lost in the abstract model, where we take the coherence time of all qubits to be the same as the carbon spins'. 
Other dephasing processes such as induced dephasing~\cite{kalb2018dephasing}, which are present in our color center model, are ignored in the abstract model.
In Table~\ref{tab:abstract_from_nv_baseline_parameters}, we present the abstract model parameters obtained from the color center baseline hardware parameters as shown in Table~\ref{tab:nv_baseline_parameters}.
\\

\begin{table}[!htpb]
\begin{tabular}{|c|c|c|}
\hline
Parameter                                      & Noise & Duration/Time                \\ \hline
Visibility                                     & 0.9 & -                  \\ \hline
Dark count probability                              & $1.5\times 10^{-7}$ & -  \\ \hline
Photon detection probability excluding attenuation losses                                      & $5.1 \times 10^{-4}$ & - \\ \hline
Spin-photon emission & $F$ = 1 & 3.8 $\mu$s \\ \hline
Swap quality                                & 0.83 &  503.7 $\mu$s  \\ \hline
$T_1$                                      & -    &   10 hours      \\ \hline
$T_2$                                      & -         &  1 s       \\ \hline
\end{tabular}
\caption{Baseline abstract model hardware parameters mapped from color center baseline shown in Table~\ref{tab:nv_baseline_parameters}.}
\label{tab:abstract_from_nv_baseline_parameters}
\end{table}

Having introduced the process by which we map color center parameters to the abstract model, we now proceed with the results of validating the abstract model against the NV model.
To do so, the following steps were taken: (i) define the values of the baseline hardware parameters for the more in-depth model and map them to the abstract model following the procedure described above, thus obtaining the corresponding abstract model baseline, (ii) run the simulation, (iii) improve both baselines using the improvement factor technique introduced in Section~\ref{sec:methods} and (iv) repeat steps (ii) and (iii) for improvement factors in the desired range. 
\\

This analysis is done both for single and double-click entanglement generation, as we simulated color center repeaters running both protocols.
\subsubsection{Color center validation}\label{sec:validate_nv_abstract}
In Figure~\ref{fig:validate_abstract_nv} we show the results of the validation for the abstract model against the NV model, for single-click (top) and double-click (bottom) entanglement generation.
\\
\begin{figure}[!htpb]
    \centering
    \subfloat[\centering Single-click.]{{\includegraphics[width=\columnwidth]{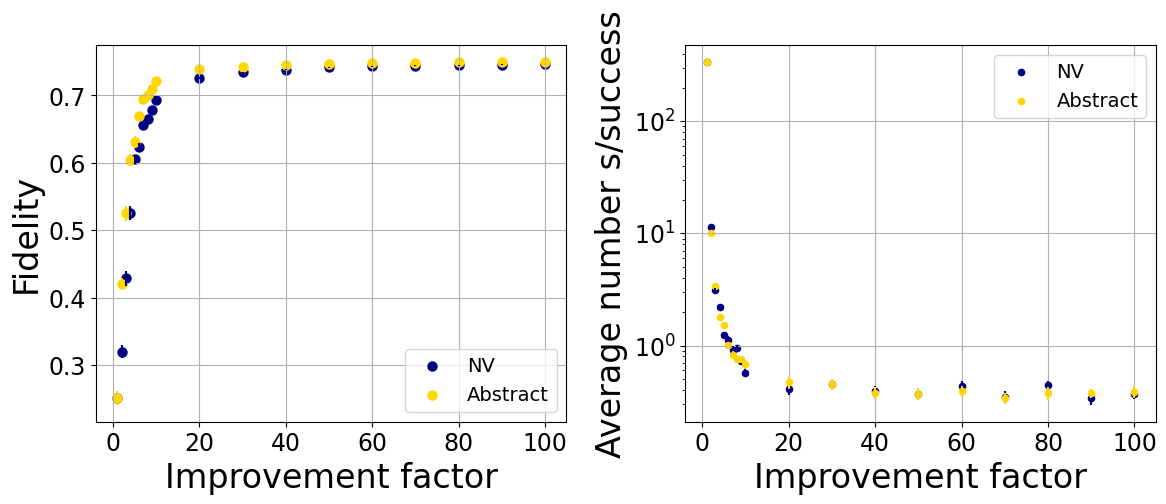}}}
    \qquad
    \subfloat[\centering Double-click.]{{\includegraphics[width=\columnwidth]{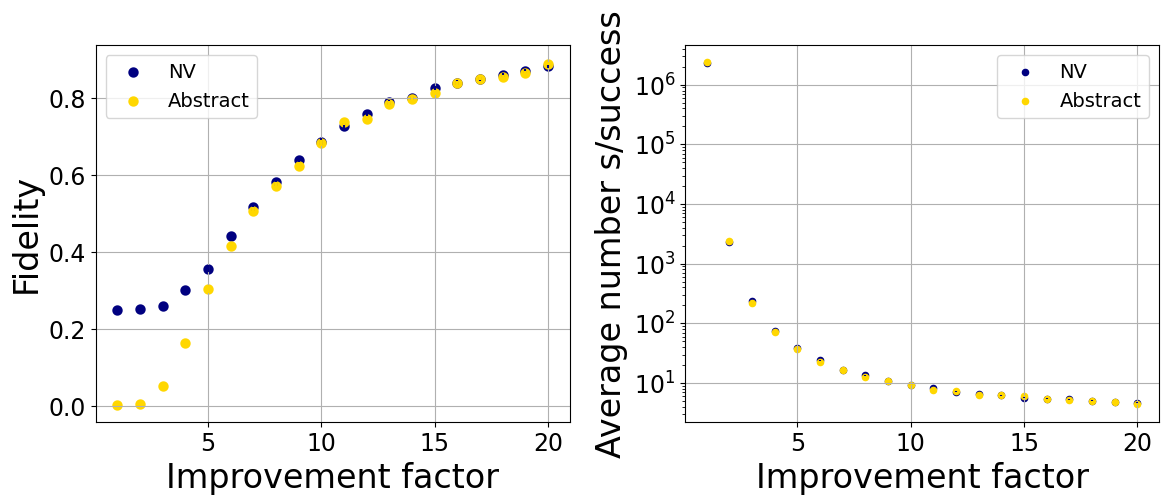}}}
    \caption{Performance of color center nodes and abstract nodes on the Delft - Eindhoven setup, with single-click (top) and double-click (bottom) entanglement generation.
The leftmost point on both plots corresponds to the baseline hardware values.
The points to the right were obtained by uniformly improving the hardware over this baseline.
The error bars represent the standard error of the mean and are often smaller than the markers.
``Average number s/success'' is the average number of seconds per entangled pair that is succesfully distributed.
}
    \label{fig:validate_abstract_nv}
\end{figure}

The agreement is similar for both protocols.
The rate of entanglement generation, shown on the plots on the right side, is identical for both models.
The only source of difference timing-wise is in how long it takes to perform an entanglement swap, with color center taking slightly longer due to its more complex circuit.
However, the low success probability of generating entanglement means that many attempts are required, rendering the time devoted to local operations negligible.
Since the time taken per entanglement generation attempt is equal in both models, it is to be expected that the rate is identical.
\\

For small improvement factors, there is a sizeable gap in the average teleportation fidelity achievable in each model, as shown on the plots on the left side.
This fidelity is significantly larger for the abstract model.
We conjecture that this is due to sources of noise that are present in the NV model but not in the abstract model.
These include induced dephasing noise, probability of double photon excitation and deviations in interferometric phase, the last two being single-click specific.
As parameters improve, the magnitude of these noise sources drops, and so does the gap between the fidelity achieved by the two setups.
\\

Overall, the abstract model captures the behavior of the more in-depth NV model reasonably well.
However, it does result in a more optimistic picture regarding the parameter quality required to achieve certain fidelity targets.
For example, in the abstract model with double-click entanglement generation, an improvement factor of $5$ suffices to reach an average teleportation fidelity of $0.7$.
This same target requires an improvement factor of $7$ in the NV model.
This supports the need for detailed hardware models, which take platform-specific limitations and sources of noise into account.
\subsubsection{Trapped ion to abstract model mapping}
The visibility, dark count probability, photon detection probability excluding attenuation losses and spin-photon emission parameters are mapped without change from the trapped ion model to the abstract model.
The process by which the swap quality parameter is obtained is identical to the one described in Appendix~\ref{sec:map_nv_abstract}.
There is a notable difference between how memory decoherence is accounted for in the two models.
In our trapped ion model, states stored in memory suffer from collective dephasing and no relaxation is considered.
Our abstract model, on the other hand, considers a $T_1$, $T_2$ memory noise model, as empirically it has been found to fit well to a large variety of physical systems.
When mapping from trapped ion parameters to abstract model parameters, we take the abstract model's $T_2$ to be given by the collective dephasing coherence time of the trapped ion and we set $T_1$ to infinity, i.e. we consider no relaxation in the abstract model.
The collective dephasing affecting trapped ion qubits follows a Gaussian shape, whereas the dephasing in the abstract model follows a simple exponential.
In Table~\ref{tab:abstract_from_ti_baseline_parameters}, we present the abstract model parameters obtained from the trapped ion baseline hardware parameters as shown in Table~\ref{tab:ti_baseline_parameters}.
\begin{table}[!htpb]
\begin{tabular}{|c|c|c|}
\hline
Parameter                                      & Noise & Duration/Time                \\ \hline
Visibility                                     & 0.89 & -                  \\ \hline
Dark count probability                              & $1.5\times 10^{-5}$ & -  \\ \hline
Photon detection probability excluding attenuation losses                                      & 0.0288 & - \\ \hline
Spin-photon emission & $F$ = 0.99 & 50 $\mu$s \\ \hline
Swap quality                                & 0.94 &  1.91 ms \\ \hline
$T_1$                                      & -    &   -      \\ \hline
$T_2$                                      & -         &  6 ms       \\ \hline
\end{tabular}
\caption{Baseline abstract model hardware parameters mapped from trapped ion baseline shown in Table~\ref{tab:ti_baseline_parameters}.}
\label{tab:abstract_from_ti_baseline_parameters}
\end{table}

\subsubsection{Trapped ion validation}
In this section, we investigate how well the simpler abstract model captures the behavior of the trapped ion model.
We do this considering only double-click entanglement generation, as this was the only entanglement generation protocol we considered when performing trapped ion simulations.
\\

\begin{figure}[!htpb]
\centering
\includegraphics[width=\columnwidth]{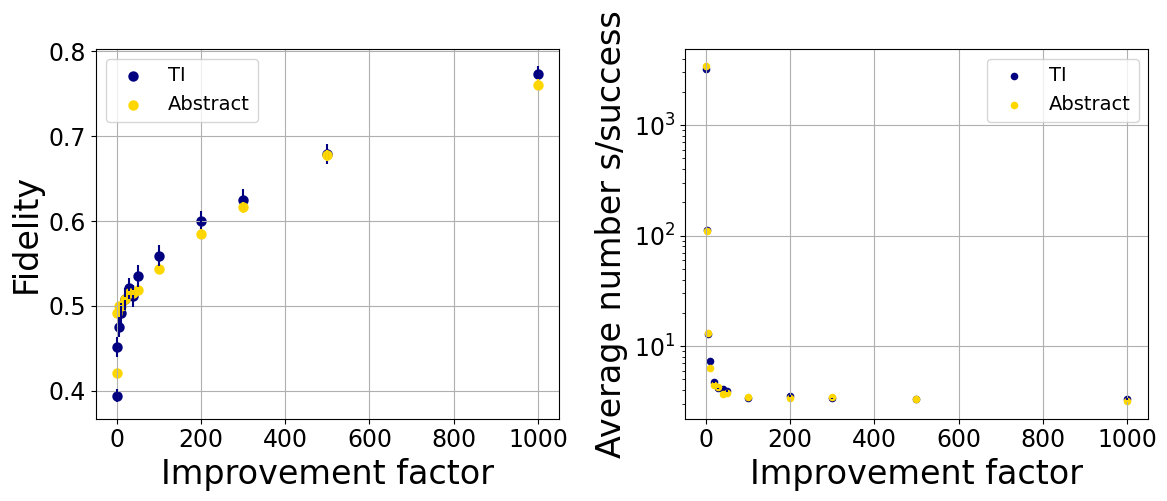}
\caption{Performance of trapped ion nodes and abstract nodes on the Delft - Eindhoven setup, double-click entanglement generation.
The leftmost point on both plots corresponds to the baseline hardware values.
The points to the right were obtained by uniformly improving the hardware over this baseline.
The error bars represent the standard error of the mean.}
\label{fig:validate_abstract_ti}
\end{figure}

In Figure~\ref{fig:validate_abstract_ti} we show the results of the validation of the abstract model against the trapped ion model.
The agreement in terms of the entanglement generation rate is perfect, with the rates overlapping for all values of the improvement factor.
This is to be expected, since the end-to-end entanglement generation time is dominated by the time spent attempting to generate elementary links, and each attempt takes the same amount of time in both models.
The average teleportation fidelity follows the same trend for both models, starting at very low values for current hardware parameters and quickly rising as hardware parameters are improved.
We note that for low improvement factors, the abstract model achieves a higher fidelity.
The opposite seems to be true for high improvement factors, although there the difference is small and does not exceed one error bar.
This can be explained by the Gaussian nature of the trapped ion dephasing.
In the trapped ion model, the probability of a state stored in memory dephasing over a given period of time $t$ is $1 - e^{-t^2/T^2}$, with $T$ being the coherence time.
In the abstract model, this probability is $1 - e^{-t/T}$.
This means that for $t/T < 1$, the probability of error for trapped ions is smaller, while the opposite is true for $t/T > 1$.
At low values of the improvement factor, the success probability of entanglement generation is small, as are coherence times.
Therefore, the time a state is expected to stay in memory is likely larger than the coherence time, and we expect that the error rate is higher in the trapped ion model.
As parameters improve, it becomes more likely that states remain in memory for periods of time smaller than the coherence time, which is the regime in which the error rate is higher in the abstract model.
This in line with what is observed in Figure~\ref{fig:validate_abstract_ti}.
Overall, the agreement is better than what was observed in Appendix~\ref{sec:validate_nv_abstract}.
There, owing to noise sources present in the color center model that were ignored in the abstract model, the latter performed better than the former.
No noise sources were ignored when mapping from the trapped ion model to the abstract model, so this better agreement was to be expected.
We conclude that the abstract model captures the behavior of the more detailed trapped ion model almost perfectly in the setup we considered.

%% file: sections/target_metric.tex
\section{Target metric} \label{app:target_metric}

In this appendix, we explain the target metric used in this paper.
As discussed in the main text, there are two conditions on end-to-end entanglement distribution that define the target.
The first is on the average fidelity with which qubits can be teleported using the generated entangled states,
and the second is on the rate at which such states are generated.
The target values for the teleportation fidelity and entangling rate are chosen such that the quantum link would be able to support Verifiable Blind Quantum Computation (VBQC) \cite{leichtle2021verifying} when the server consists of a powerful quantum computer with a coherence time of 100 seconds.
We show that if the targets are met, the client would be able to execute VBQC by preparing states at the powerful quantum computer using either quantum teleportation or remote state preparation
(for remote state preparation, see Appendix \ref{sec:RSP}).
\\

The following results presented in this appendix are novel:
\begin{itemize}
\item
the constraint equation that, when solved, guarantees VBQC is feasible (Theorems \ref{thm:vbqc_constraint} and \ref{thm:rsp_works});
\item
the extension of the noise robustness theorem in \cite{leichtle2021verifying} to guarantee that VBQC is feasible when the \emph{average} error probability can be bounded instead of the \emph{maximum} error probability, assuming that the error probabilities across different rounds are independent and identically distributed (Theorem~\ref{thm:vbqc_relaxation} and Appendix~\ref{app:noise_robustness});
\item
a modified version of the VBQC protocol \cite{leichtle2021verifying} that is based on remote state preparation instead of qubit transmission (Protocol \ref{prot:vbqc_with_rsp}) and a proof that, in the absence of local noise, it is equivalent to the original protocol where some effective quantum channel is used for qubit transmission (thereby guaranteeing that the correctness of the original protocol is inherited; we note that we have not otherwise investigated the security of this protocol) (Theorem \ref{thm:rsp_channel}).
\end{itemize}

\subsection{Teleportation fidelity}
\label{sec:appendix_teleportation_fidelity}
We consider the following quantum-teleportation protocol \cite{leichtle2021verifying}.
A one-qubit information state $\rho$ is teleported using a two-qubit resource state $\sigma$ shared by two parties.
A Bell-state measurement is performed between the qubit holding the information state and one of the qubits in the resource state.
If the outcome of the measurement corresponds to Bell state
\begin{equation}
\ket{\Phi_{ij}} \equiv X^iZ^j \ket {\Phi^+}
\end{equation}
with $\ket {\Phi^+} = \tfrac 1 {\sqrt 2} (\ket{00} + \ket{11})$ then the Pauli correction $X^i Z^j$ is performed on the one remaining qubit.
Executing this protocol results in transmitting the information state through the teleportation channel $\Lambda_\sigma$.
\begin{definition} \label{def:teleportation_channel}
Teleportation channel.
The teleportation channel associated with the two-qubit state $\sigma$ is given by the single-qubit quantum channel
\begin{equation}
\Lambda_\sigma\Big(\rho\Big) \equiv \sum_{i,j} \left( X^iZ^j \otimes \bra{\Phi_{ij}} \right) \left( \sigma \otimes \rho \right) \left( X^iZ^j \otimes \ket{\Phi_{ij}} \right).
\end{equation}
\end{definition}
We note that if $\sigma = \ketbra{\Phi_{00}}$ then $\Lambda_\sigma$ is the identity map.
The average teleportation fidelity corresponding to the resource state $\sigma$ is given by $F_\textnormal{tel}(\sigma)$.
\begin{definition}
Average teleportation fidelity.
The average teleportation fidelity associated with the two-qubit state $\sigma$ is given by
\begin{equation} \label{eq:F_tel}
F_\textnormal{tel}(\sigma) \equiv \int_{\psi} \expectationvalue{\Lambda_\sigma(\ketbra{\psi})}{\psi} d \psi,
\end{equation}
where the integral is over the Haar measure.
\end{definition}
We note that by the Haar measure, we here mean the uniform measure over single-qubit quantum states, i.e. the uniform measure on the unit sphere in $\mathcal C^2$.
It is the unique measure that is invariant under unitary transformations \cite{watrous2018}.
\\

Finally, we note that if the sender and receiver agree on a unitary $U$,
then teleportation can also be executed as follows.
First, the sender applies $U$ to the information state.
Second, the sender teleports the resulting information state to the receiver.
Last, the receiver applies the unitary $U^\dagger$ to undo the original unitary and obtain the information state.
The qubit is then transmitted through a rotated teleportation channel.
\begin{definition} \label{def:rotated_teleportation_channel}
Rotated teleportation channel.
The rotated teleportation channel associated with the two-qubit state $\sigma$ and the unitary $U$ is given by
\begin{equation}
\Lambda_{\sigma, U}(\rho) = U^\dagger \Lambda_\sigma(U \rho U^\dagger) U
\end{equation}
\end{definition}
We remark that the average teleportation fidelity is not affected by the introduction of the unitary $U$ because of the invariance of the Haar measure, i.e.
\begin{equation}
F_\textnormal{tel}(\sigma) = \int_{\psi} d \psi \expectationvalue{\Lambda_{\sigma,U}(\ketbra{\psi})}{\psi},
\end{equation}
Using a unitary to turn a teleportation channel into a rotated teleportation channel can be advantageous when not every state on the Bloch sphere needs to be transmitted with equal fidelity, and $\sigma$ is such that not all states can be transmitted with equal fidelity.
By applying the unitary $U$, the Bloch sphere can potentially be rotated in such a way to make states for which high-fidelity transmission is desirable coincide with states that can be transmitted at high fidelity.

\subsection{Requirements from VBQC}

We consider the scenario where two nodes are connected using the one-repeater quantum connection studied in this work.
These two nodes use the entanglement generated by this quantum connection to perform VBQC.
Specifically, the first node (the client) utilizes VBQC to execute a two-qubit computation on the quantum processor of the second node (the server) in a verified and blind fashion.
It is assumed that the server is able to execute gates without noise and has a coherence time of 100 seconds.
Our target metric is chosen such that it guarantees that the quantum connection is able to support this protocol.
\\

A single round of the VBQC protocol involves the preparation of two qubits by the client at the server,
and the execution of a series of quantum gates and measurements on those qubits by the server.
The client can use the remote-state-preparation protocol \cite{bennet2001remote} to use one entangled state to prepare one qubit at the server.
Some rounds are computation rounds, the results of which are sent classically by the server to the client.
All other rounds are test rounds.
In a test round, some of the qubits transmitted to the server are traps;
if the server tries to measure these qubits or performs another operation than the one specified by the client, this will become apparent from the returned computation results.
However, tampering by the server is indistinguishable from noise.
Only if noise is within certain bounds can the protocol be performed successfully.
\\

This defines minimum requirements on the quantum connection used by the client to prepare the qubits at the server.
First, the fidelity at which states can be prepared needs to be large enough.
Second, the rate at which they can be prepared needs to be large enough as well.
The reason for this is that after the first qubit is prepared at the server, it will undergo memory decoherence while waiting for the second qubit to be prepared.
\\

Specifically, we consider the case of depolarizing memory.
\begin{definition} \label{def:depolar_mem}
Depolarizing memory.
If a single-qubit quantum state $\rho$ is stored in a depolarizing memory with coherence time $T$ for a time $t$,
it is subjected to a depolarizing channel
\begin{equation}
\mathcal D_\textnormal{p}(\rho) = p \rho + (1 - p) \frac{\mathbf 1} 2
\end{equation}
where the depolarizing parameter $p$ is given by
\begin{equation}
p = e^{- \frac t T}.
\end{equation}
\end{definition}
The minimum requirements are then defined by the following theorem.
\begin{theorem} \label{thm:vbqc_constraint}
Requirements on entanglement generation for VBQC.
Assume a quantum link generates the two-qubit state $\sigma$ between a client and a server with average rate $R$, and that the distribution times are independent and identically distributed.
Furthermore, assume that qubits at the server are stored in a depolarizing memory with coherence time $T$. % (Definition \ref{def:depolar_mem}).
Lastly, assume that all local operations are noiseless and instantaneous.
If the client prepares qubits at the server using the rotated teleportation channel $\Lambda_{\sigma, U}$ for some unitary $U$,
then a unitary $U$ exists such that the VBQC protocol proposed in \cite{leichtle2021verifying}, for a two-qubit deterministic quantum computation, can be executed in a way that is composably secure with exponentially small $\epsilon$ if
\begin{equation} \label{eq:vbqc_bound_rate_and_fidelity}
F_\textnormal{tel} (\sigma) > \frac 1 2 \Big( 1 + \frac{1}{\sqrt 2} e^{\frac 1 {2RT}} \Big).
\end{equation}
\end{theorem}

Practically speaking, Theorem \ref{thm:vbqc_constraint} means that VBQC with two qubits and no failure probability that is inherent to the computation is feasible in case Equation \eqref{eq:vbqc_bound_rate_and_fidelity} holds.
A requirement is that the state $\sigma$ is the same for each delivery of entanglement, and that the distribution times are independent and identically distributed.
We note that this is the case for the one-repeater setup studied in this paper.
After entanglement swapping at the repeater node takes place, the end-to-end entangled state is removed from the end nodes and the state of the network path is fully reset, making each entanglement delivery completely independent from the last, with identical distributions for both delivery times and errors.
In general, the state that is delivered will depend on the amount of time entangled qubits are stored before entanglement swapping takes place at the repeater node, resulting in a state that is not the same each round.
However, if the processing of the entangled state is not conditioned on the amount of storage time, the final state will effectively look like a constant mixture over all values that the storage time can take.
\\

In this paper, we consider two different sets of target teleportation fidelity and target rate,
namely $(F_\textnormal{tel}, R) = (0.8717, 0.1 \textnormal{ Hz})$ and $(0.8571, 0.5 \textnormal{ Hz})$.
Both of these have been chosen to satisfy Equation \eqref{eq:vbqc_bound_rate_and_fidelity} for $T = 100$ seconds.
\\

\subsection{Proving Theorem \ref{thm:vbqc_constraint}}

In \cite{leichtle2021verifying}, it is shown that the VBQC protocol is composably secure with exponentially small $\epsilon$ in case the noise is such that the failure probability of each individual test round can be upper bounded.
Key to proving Theorem \ref{thm:vbqc_constraint} is a relaxation of this condition:
two-qubit VBQC is also feasible if instead the \textit{average} failure probability of test rounds can be upper bounded, in case the failure probabilities are independent and identically distributed.
This is stated in the following theorem.
\begin{theorem} \label{thm:vbqc_relaxation}
(Local correctness of VBQC protocol on Noisy Devices)
Let $p$ denote the inherent error probability of the quantum computation, which is executed using a $k$-colorable graph state.
Assume that, for every test round, the probability that at least one of the trap-measurement outcomes is incorrect is a random variable.
Furthermore, assume that these are independent and identically distributed for all test rounds.
Let $q$ be the expected value of these random variables.
The VBQC protocol presented in \cite{leichtle2021verifying} is $\epsilon_\textnormal{cor}$-locally-correct with exponentially low $\epsilon_\textnormal{cor}$
if $q < (1/k) (2p-1)/(2p-2)$.
\end{theorem}
Theorem~\ref{thm:vbqc_relaxation} is proven in Appendix~\ref{app:noise_robustness} and allows us to derive the following lemma.
\begin{lemma} \label{lem:vbqc_bound_average_failure_prob}
Two-qubit VBQC for deterministic computations is composably secure with exponential $\epsilon$ if the probabilities that the trap-measurement outcome is incorrect are independent and identically distributed for all test rounds and the average probability that the trap-measurement outcome in a single test round is incorrect, $q$, satisfies $q < 1/4$.
\end{lemma}
\begin{proof}
First, we note that all two-qubit graph states are at least two-colorable, i.e., $k \leq 2$.
Second, we note that for deterministic computations the inherent error probability of the computation is zero, i.e., $p=0$.
Then, from Theorem \ref{thm:vbqc_relaxation},
it follows that if $q<1/4$ is true, then the VBQC protocol is $\epsilon_\textnormal{cor}$-locally-correct with exponentially low $\epsilon_\textnormal{cor}$.
Additionally, as shown in \cite{leichtle2021verifying}, the VBQC protocol is $\epsilon_\textnormal{bl}$-local-blind and $\epsilon_\textnormal{ver}$-local-verifiable with $\epsilon_\textnormal{ind}$-independent-verification, with $\epsilon_\textnormal{bl}$, $\epsilon_\textnormal{ver}$ and $\epsilon_\textnormal{ind}$ exponentially low.
Therefore, as in \cite{leichtle2021verifying}, it follows that the protocol is composably secure with exponential $\epsilon$.
\end{proof}

% We note that in an experiment it might be more feasible to prepare the server-side qubits using remote state preparation \cite{},
% as this protocol requires fewer quantum resources and a smaller classical communication overhead.
% However, we note that under the assumption that all local operations involved in the teleportation protocol are noiseless,
% the requirements imposed on the quantum connection by VBQC are exactly the same no matter whether quantum teleportation of remote state preparation is used.

During a test round, the client randomly designates one of the two qubits that it prepares at the server the ``dummy'' qubit and the other the ``trap'' qubit.
The client that remotely prepares the dummy qubit in $\ket d$, where $d$ is chosen uniformly at random by the client from $\{0, 1\}$.
It prepares the trap qubit in the state $\ket+_{\theta}$ defined by
\begin{equation}\label{eq:define_+-_theta}
\ket{\pm_\theta} = \frac 1 {\sqrt 2} (\ket 0 \pm e^{i\theta} \ket 1),
\end{equation}
where the client chooses $\theta$ uniformly at random from $\Theta \equiv \{i\pi/4\}_{0 \leq i \leq 7}$.
That is, the trap qubit will be in one of eight equidistant quantum states on the equator of the Bloch sphere.
The server will perform a CZ gate between the two qubits, measure them in the basis $\{ \ket{+_\delta}, \ket{-_\delta} \}$, and send the measurement outcomes to the client.
Here, $\delta = \theta + r \pi$, where $r$ is chosen uniformly at random by the client from $\{0, 1\}$.
The test round is declared a success if the measurement on the trap qubit yields $d \oplus r$ and a failure otherwise.
If the server is honest and there is no noise, test rounds are always successful.
Otherwise, we show that the following Lemma holds:
\begin{lemma} \label{lem:vbqc_failure_prob_single_round}
If, during a test round of two-qubit VBQC, the trap qubit is prepared with fidelity $F_\textnormal{trap}$ and the dummy qubit is prepared with fidelity $F_\textnormal{dummy}$,
then the probability that the measurement outcome on the trap qubit is incorrect is given by
\begin{equation} \label{eq:p_fail_single_test}
p_\textnormal{fail} = F_\textnormal{dummy} (1 - F_\textnormal{trap}) + F_\textnormal{trap} (1 - F_\textnormal{dummy}).
\end{equation}
\end{lemma}
\begin{proof}
Consider the case $d=r=0$.
In that case, we can write
\begin{equation}
\rho_\textnormal{dummy, server} = F_\textnormal{dummy} \ketbra 0 + (1 - F_\textnormal{dummy}) \ketbra 1 + a \ket 0 \bra 1 + a^* \ket 1 \bra 0
\end{equation}
for some constant $a$ and
\begin{equation}
\rho_\textnormal{trap, server} = F_\textnormal{trap} \ketbra{+_\theta} + (1 - F_\textnormal{trap}) \ketbra{-_\theta} + b \ket{+_\theta} \bra{-_\theta} + b^* \ket{-_\theta} \bra{+_\theta}
\end{equation}
for some constant $b$.
Here, we have made use of the fact that both $\{\ket 0, \ket 1\}$ and $\{\ket{+_\theta}, \ket{-_\theta}\}$ are complete bases for the single-qubit Hilbert space.
% The complete state received by the server is then $\rho_\textnormal{dummy, server} \rho_\textnormal{trap, server}$.
\\

After receiving both states, the server will perform a CZ gate between the two qubits,
and then measure the trap qubit in the $\{\ket{+_\theta}, \ket{-_\theta}\}$ basis.
Whether the test round is successful or not depends on whether the expected outcome $d\oplus r = 0$, i.e. $\ket{+_\theta}$, is obtained from this measurement.
In order to get the measurement statistics on the trap qubit, we can first trace out the dummy qubit.
With that in mind, let's look at what happens with the term
\begin{equation}
a \textnormal{CZ} \ket 0 \bra 1 \rho_\textnormal{trap, server} \textnormal{CZ}+ a^* \textnormal{CZ} \ket 1 \bra 0 \rho_\textnormal{trap, server} \textnormal{CZ} = a \ket 0 \bra 1 \rho_\textnormal{trap, server} Z + a^* \ket 1 \bra 0 Z \rho_\textnormal{trap, server}.
\end{equation}
After the CZ has been performed, the off-diagonal terms of $\rho_\textnormal{dummy, server}$ are still off diagonal.
These will vanish when tracing out the dummy qubit and can therefore be safely ignored.
Therefore, we make the substitution
\begin{equation}
\rho_\textnormal{dummy, server} \to F_\textnormal{dummy} \ketbra 0 + (1 - F_\textnormal{dummy}) \ketbra 1.
\end{equation}
Then, the effect of the CZ is easy to evaluate, giving
\begin{equation}
\rho_\textnormal{after CZ} = F_\textnormal{dummy} \ketbra 0 \rho_\textnormal{trap, server} + (1 - F_\textnormal{dummy}) \ketbra 1 Z \rho_\textnormal{trap, server} Z
\end{equation}
which, after tracing out the dummy qubit, gives
\begin{equation}
\begin{aligned}
\rho_\textnormal{trap, after CZ} =& \Big ( F_\textnormal{dummy} F_\textnormal{trap} + (1 - F_\textnormal{dummy})(1 - F_\textnormal{trap}) \Big ) \ketbra{+_\theta} \\
&+ \Big ( F_\textnormal{dummy} (1 - F_\textnormal{trap}) + F_\textnormal{trap} (1 - F_\textnormal{dummy}) \Big )  \ketbra{-_\theta} \\
&+ c \ket{+_\theta} \bra{-_\theta} + c^* \ket{-_\theta} \bra{+_\theta},
\end{aligned}
\end{equation}
where $c$ is a function of $b$, $F_\textnormal{dummy}$ and $F_\textnormal{trap}$.
Applying a POVM with elements $\ketbra{+_\theta}$ and $\ketbra{-_\theta}$ then gives a failure probability of the test round of
\begin{equation} \label{eq:p_test_failed_dummy_trap}
p_\textnormal{fail} = \Trace \Big( \ketbra{-_\theta}\rho_\textnormal{trap, after CZ} \Big) = F_\textnormal{dummy} (1 - F_\textnormal{trap}) + F_\textnormal{trap} (1 - F_\textnormal{dummy}).
\end{equation}
This calculation can be repeated for all three cases where $d=r=0$ is false,
each time giving the exact same outcome.
\end{proof}

We now have a formula for the probability that a test round fails, given by Equation \eqref{eq:p_test_failed_dummy_trap}.
However, this formula depends on the fidelity with which specific states are transmitted over the teleportation channel.
These states are randomly chosen during each test round
($\ket 0$ or $\ket 1$ for the dummy qubit, $\ket{+_\theta}$ for the trap qubit).
This means that, in general, the failure probability is not constant per round.
Before we are able to use Lemma \ref{lem:vbqc_bound_average_failure_prob}, we need to know something about the average failure probability per round.
Additionally, we need to account for decoherence in the server's memory while waiting for the second qubit to be prepared at the server.
Both are accounted for in the following lemma.
\begin{lemma} \label{lem:vbqc_failure_prob_average}
Assume a quantum link generates the two-qubit state $\sigma$ between a client and a server with average rate $R$,
and that the distribution times are independent and identically distributed.
Additionally assume that a unitary $U$ has been chosen
such that dummy qubits can be transmitted through a rotated teleportation channel with average fidelity
\begin{equation} \label{eq:F_dummy}
\bar F_\textnormal{dummy} \equiv \frac 1 2 \bigg( \expectationvalue{\Lambda_{\sigma,U}(\ketbra{0})}{0} + \expectationvalue{\Lambda_{\sigma,U}(\ketbra{1})}{1} \bigg)
\end{equation}
and trap qubits with average fidelity
\begin{equation} \label{eq:F_trap}
\bar F_\textnormal{trap} \equiv \frac 1 8 \sum_{\theta \in \Theta}  \expectationvalue{\Lambda_{\sigma,U}(\ketbra{+_{\theta}})}{+_{\theta}}.
\end{equation}
Assume that the condition
\begin{equation} \label{eq:F_dummy_trap_constraint}
\bar F_\textnormal{dummy} (1 - \bar F_\textnormal{trap}) + \bar F_\textnormal{trap} (1 - \bar F_\textnormal{dummy}) \leq \frac 1 2
\end{equation}
holds.
Furthermore, assume that qubits received by the server are stored in depolarizing quantum memory with coherence time $T$.
Lastly, assume that all local operations are noiseless and instantaneous.
In that case, for two-qubit VBQC, the average test-round failure probability is bounded by
\begin{equation} \label{eq:q_bound_dummy_trap}
q \leq e^{-\frac 1 {RT}} \bigg[\bar F_\textnormal{dummy} (1 - \bar F_\textnormal{trap}) + \bar F_\textnormal{trap}(1-\bar F_\textnormal{dummy})\bigg] + \frac 1 2 (1 - e^{- \frac 1 {RT}}).
\end{equation}
\end{lemma}
\begin{proof}
Let $\Delta t$ be the time between the generation of the first and second entangled state.
Then, the first qubit is stored for time $\Delta t$ in depolarizing memory until the second qubit is prepared at the server.
If the qubit was prepared at the server with fidelity $F$, the depolarizing noise will have the effect
\begin{equation}
F \to e^{- \frac {\Delta t} T} F +  \frac 1 2 (1 - e^{- \frac {\Delta t} T}).  % = e^{- \frac {\Delta t} T} (F - \frac 1 2) + \frac 1 2.
\end{equation}
We note that Equation \eqref{eq:p_fail_single_test} is symmetric under interchange of $F_\textnormal{dummy}$ and $F_\textnormal{trap}$.
Therefore, we can assume that the dummy qubit is prepared first without loss of generality.
Writing $F_\textnormal{dummy}$ and $F_\textnormal{trap}$ for the fidelities with which the qubits are teleported to the server (i.e. excluding the effect of memory decoherence),
it follows that
\begin{equation}
p_\textnormal{fail} = e^{- \frac {\Delta t} T} \bigg[F_\textnormal{dummy} (1 - F_\textnormal{trap}) + F_\textnormal{trap} (1 - F_\textnormal{dummy})\bigg] + \frac 1 2 (1 - e^{- \frac {\Delta t} T}).
\end{equation}
Now, to calculate the average failure probability $q \equiv \expectationvalue{p_\textnormal{fail}}$, we note that $F_\textnormal{dummy}$, $F_\textnormal{trap}$ and $\Delta t$ are all independent random variables;
the first depends on the choice of $d$ (i.e. whether to prepare $\ket 0$ or $\ket 1$),
the second depends on the choice of $\theta$ (i.e. which $\ket{+_{\theta}}$ to prepare),
and the last depends on the probability distribution for the entanglement delivery time.
This allows us to write
\begin{equation} \label{eq:q_in_terms_of_expected_values}
q = \expectationvalue{e^{- \frac {\Delta t} T}} \bigg[\bar F_\textnormal{dummy} (1 - \bar F_\textnormal{trap}) + \bar F_\textnormal{trap} (1 - \bar F_\textnormal{dummy})\bigg] + \frac 1 2 (1 - \expectationvalue{e^{-\frac {\Delta t} T}}).
\end{equation}
Because the exponential function is convex, Jensen's inequality \cite{jensen} gives
\begin{equation}
\expectationvalue{e^{-\frac {\Delta t} T}} \geq e^{-\frac {\expectationvalue{\Delta t}} T}.
\end{equation}
The times between the distribution of two entangled states are by assumption all independent and identically distributed, i.e., they are all copies of the same $\Delta t$.
The (average) entangling rate is therefore simply equal to
\begin{equation}
R = \frac 1 {\expectationvalue{\Delta t}},
\end{equation}
and therefore we find
\begin{equation} \label{eq:jensen}
\expectationvalue{e^{-\frac {\Delta t} T}} \geq e^{-\frac 1{RT}}.
\end{equation}
In case Equation \eqref{eq:F_dummy_trap_constraint} holds,
Equation \eqref{eq:jensen} can be combined with Equation \eqref{eq:q_in_terms_of_expected_values} to obtain Equation \eqref{eq:q_bound_dummy_trap}.
\end{proof}
We note that the use of Jensen's inequality above accounts for any kind of potential jitter in the delivery of entangled qubits to the server.
Whatever the distribution on the waiting time $\Delta t$ looks like and at however irregular intervals entanglement is delivered, Jensen's inequality will guarantee that Equation \eqref{eq:q_bound_dummy_trap} holds.
\\

Now, we want to use the average teleportation fidelity $F_\textnormal{tel}$ instead of the quantities $\bar F_\textnormal{dummy}$ and $\bar F_\textnormal{trap}$ to bound $q$.
The final building block towards obtaining such a bound and proving Theorem \ref{thm:vbqc_constraint} is the following lemma.
\begin{lemma} \label{lem:there_exists_unitary}
There exists a unitary $U$ such that
\begin{equation}
\bar F_\textnormal{dummy} = \bar F_\textnormal{trap} = F_\textnormal{tel},
\end{equation}
where $\bar F_\textnormal{dummy}$ is defined in Equation \eqref{eq:F_dummy}, $\bar F_\textnormal{trap}$ in Equation \eqref{eq:F_trap} and $F_\textnormal{tel}$ in Equation \eqref{eq:F_tel} (with $\sigma$ left implicit).
\end{lemma}

\begin{proof}
While $\bar F_\textnormal{dummy}$ and $\bar F_\textnormal{trap}$ are fidelity averages over specific subsets of the Bloch sphere, $F_\textnormal{tel}$ is an average over the entire Bloch sphere.
This allows us to find the relationship
\begin{equation} \label{eq:F_tel_in_terms_of_dummy_and_trap}
F_\textnormal{tel} = \frac 1 3 \bar F_\textnormal{dummy} + \frac 2 3 \bar F_\textnormal{trap}.
\end{equation}
To see how this relationship follows, we first note that the average fidelity over the entire Bloch sphere can be written as an average over any six states that form a regular octahedron on the Bloch sphere \cite{bowdreyFidelitySingleQubit2002}.
One example of such a octahedron is given by the six eigenstates of the Pauli operators, which gives
\begin{equation} \label{eq:F_tel_average_over_Pauli_states}
\begin{aligned}
F_\textnormal{tel} =& \frac 1 6 \Bigg(\expectationvalue{\Lambda_{\sigma, U} \Big( \ketbra{0} \Big)}{0}\Big) + \expectationvalue{\Lambda_{\sigma, U}\Big(\ketbra{1}\Big)}{1}\\
&+ \expectationvalue{\Lambda_{\sigma, U}\Big(\ketbra{+_0}\Big)}{+_0} + \expectationvalue{\Lambda_{\sigma, U}\Big(\ketbra{-_0}\Big)}{-_0} \\
&+ \expectationvalue{\Lambda_{\sigma, U}\Big(\ketbra{+_{\frac \pi 2}}\Big)}{+_{\frac \pi 2}} + \expectationvalue{\Lambda_{\sigma, U}\Big(\ketbra{-_{\frac \pi 2}}\Big)}{-_{\frac \pi 2}} \Bigg)\\
=&  \frac 1 6 \Bigg( \expectationvalue{\Lambda_{\sigma, U} \Big( \ketbra{0} \Big)}{0}\Big) + \expectationvalue{\Lambda_{\sigma, U}\Big(\ketbra{1}\Big)}{1} + \sum_{i=0}^4 \expectationvalue{\Lambda_{\sigma, U}\Big(\ketbra{+_{\frac{i\pi}{2}}}\Big)}{+_{\frac{i\pi}{2}}} \Bigg).
\end{aligned}
\end{equation}
Another such octahedron is obtained by rotating these six eigenstates around the Z axis by an angle of $\pi / 4$.
This gives the relation
\begin{equation} \label{eq:F_tel_average_over_rotated_Pauli_states}
F_\textnormal{tel} = \frac 1 6 \Bigg( \expectationvalue{\Lambda_{\sigma, U} \Big( \ketbra{0} \Big)}{0}\Big) + \expectationvalue{\Lambda_{\sigma, U}\Big(\ketbra{1}\Big)}{1} + \sum_{i=0}^4 \expectationvalue{\Lambda_{\sigma, U}\Big(\ketbra{+_{\frac{(2i+1)\pi}{4}}}\Big)}{+_{\frac{(2i + 1)\pi}{4}}} \Bigg).
\end{equation}
Adding Equations \eqref{eq:F_tel_average_over_Pauli_states} and \eqref{eq:F_tel_average_over_rotated_Pauli_states} together and dividing by two then gives
\begin{equation}
F_\textnormal{tel} = \frac 1 6 \Bigg( \expectationvalue{\Lambda_{\sigma, U} \Big( \ketbra{0} \Big)}{0}\Big) + \expectationvalue{\Lambda_{\sigma, U}\Big(\ketbra{1}\Big)}{1}\Bigg) + \frac 1 {12} \Bigg( \sum_{\theta \in \Theta} \expectationvalue{\Lambda_{\sigma, U}\Big(\ketbra{+_{\theta}}\Big)}{+_\theta} \Bigg),
\end{equation}
which is equivalent to Equation \eqref{eq:F_tel_in_terms_of_dummy_and_trap}.
\\

While the unitary $U$  will leave the average over the entire Bloch sphere, $F_\textnormal{tel}$, invariant,
the same does not hold for $\bar F_\textnormal{dummy}$.
The unitary rotates the Bloch sphere and thus effectively turns $\bar F_\textnormal{dummy}$ into an average over any pair of antipodal points on the Bloch sphere.
Each pair of antipodal points can be described using only one of the two points.
The average over pairs of antipodal points can therefore be described as a function $f$ with as domain one half of the Bloch sphere.
This function $f$ maps each point on that half of the Bloch sphere to the average fidelity of that point and its antipodal point.
Now, $\bar F_\textnormal{dummy}$ can be chosen to correspond to any of the values in $f$'s range.
Additionally, the average of $f$ over its domain is equal to the average fidelity over all points on the entire Bloch sphere, i.e. $F_\textnormal{tel}$.
By the mean value theorem, we can conclude that there is a value in the range of the function that equals the average of the function.
That is, there exists a choice for the unitary $U$ such that $\bar F_\textnormal{dummy} = F_\textnormal{tel}$.
Then, Equation \eqref{eq:F_tel_in_terms_of_dummy_and_trap} implies that if $\bar F_\textnormal{dummy} = F_\textnormal{tel}$, then $\bar F_\textnormal{trap} = F_\textnormal{tel}$.
\end{proof}

Theorem \ref{thm:vbqc_constraint} is then finally proven by combining Lemmas \ref{lem:vbqc_bound_average_failure_prob}, \ref{lem:vbqc_failure_prob_average}, and \ref{lem:there_exists_unitary}.

\subsection{Proving Theorem \ref{thm:vbqc_relaxation}} \label{sec:vbqc_proof}

\input{sections/vbqc_noise_robustness_proof}

\subsection{Remote state preparation}
\label{sec:RSP}

Here, we introduce a modified version of the VBQC protocol \cite{leichtle2021verifying} in which the client ``sends'' qubits to the server using remote state preparation (Protocol \ref{prot:vbqc_with_rsp}).
Remote state preparation is experimentally simpler than teleportation.
Therefore, it is likely that early VBQC demonstrations will be more feasible when using remote state preparation than when using teleportation.
We show that, when local operations are noiseless, the modified protocol is equivalent to the protocol introduced in \cite{leichtle2021verifying} but using some specific effective quantum channel to send qubits from the client to the server.
This result is expressed in Theorem \ref{thm:rsp_channel}.
Therefore, the correctness property carries over from the protocol in \cite{leichtle2021verifying} to the modified protocol, showing that it is indeed possible to use remote state preparation to execute VBQC.
Additionally, we show that the conclusions about the feasibility of VBQC found above (Theorem \ref{thm:vbqc_constraint}) also hold for the modified protocol.
That is, when the rate and fidelity of entanglement generation are good enough to support VBQC through quantum teleportation with noiseless local operations,
they are also good enough to support VBQC through remote state preparation with noiseless local operations.
This result is expressed in Theorem \ref{thm:rsp_works}.
\\

As preliminaries to proving the above, we first introduce two definitions.
\begin{definition}
U-NOT operation.
The U-NOT operation $\Upsilon$ is defined as \cite{buzek1999}
\begin{equation}
\Upsilon\left(\alpha \ket 0 + \beta \ket 1\right) = \beta^* \ket 0 - \alpha^* \ket 1.
\end{equation}
That is, $\Upsilon$ maps any qubit state to a state that is orthogonal to it.
\end{definition}
We note that the U-NOT operation $\Upsilon$ is anti-unitary and hence cannot be physically implemented \cite{buzek1999}.
It maps all states on the Bloch sphere to their antipodal points, which cannot be realized with rotations only.
However, mapping a specific point on the Bloch sphere to its antipodal point can always be achieved by rotating the Bloch sphere by $\pi$ around any axis that is orthogonal to the axis intersecting the point.
Such a mapping is provided by the following definition.
\begin{definition}
$\ket{\psi}$-NOT operations.
The family of $\ket{\psi}$-NOT operations $\mathcal A_{\phi, \ket{\psi}}$ is defined by
\begin{equation}
\mathcal{A}_{\phi, \ket{\psi}} \equiv e^{-i\phi} \Upsilon \left( \ket \psi \right) \bra \psi + e^{i\phi} \ket \psi \left( \Upsilon \left(\ket \psi \right) \right)^\dagger
\end{equation}
\end{definition}
The parameter $\phi$ in $\mathcal A_{\phi, \ket \psi}$ represents the freedom in choosing which axis to use for the $\pi$ rotation that maps $\ket \psi$ to $\Upsilon(\ket \psi)$ and vice versa.
We note that $\mathcal A_{\phi, \ket \psi}^\dagger = \mathcal A_{\phi, \ket \psi}$.
Now, we define a modified version of the VBQC protocol that makes use of remote state preparation instead of quantum teleportation.

\begin{protocol}
\label{prot:vbqc_with_rsp}
VBQC with remote state preparation.
This protocol is the same as the VBQC protocol presented in \cite{leichtle2021verifying},
except for the following.
\begin{itemize}
\item
Before starting the protocol, the client and server agree on a one-qubit unitary operation $U$.

\item
Whenever the client would send a qubit $v$ in the state $\ket \psi$ to the server, it instead measures its half of a two-qubit resource state shared with the server in the basis $\{U\ket \psi, \Upsilon\left(U\ket\psi\right)\}$.
The outcome of this measurement is stored at the client as $c_v$, with $c_v = 0$ corresponding to outcome $U \ket \psi$ and $c_v = 1$ corresponding to outcome $\Upsilon \left(U \ket \psi \right)$.
The server applies the operation $U^\dagger$ to its local entangled qubit.
This qubit held by the server is now considered the qubit as received from the client.

\item
In a computation round, the measurement outcome $\delta_v$ obtained from qubit $v$ is bit flipped by the client in case $c_v = 1$.
That is,
\begin{equation}
\delta_v \to \delta_v \oplus c_v \hspace{1 cm} \textnormal{(computation round)}.
\end{equation}

\item
In a test round, for each trap qubit $v$, the measurement outcome $\delta_v$ is bit flipped by the client in case $c_v = 1$, and once more for every neighboring dummy qubit $w$ for which $c_w = 1$.
That is,
\begin{equation}
\delta_v \to \delta_v \oplus c_v \oplus \bigoplus_{w \in N_G(v)} c_w \hspace{1 cm} \textnormal{(test round)}.
\end{equation}
Here, $G$ is the computation graph used in the VBQC protocol and $N_G(v)$ is the neighbourhood of qubit $v$ in graph $G$.
\end{itemize}

The outcomes $c_v$ are never shared with the server.
\end{protocol}

\begin{lemma}
\label{lem:effective_channel}
Effective remote-state-preparation channel.
Let $\ket \psi$ be some pure single-qubit state and let $\sigma$ be some two-qubit density matrix shared by Alice and Bob.
Let $\phi_{\ket \psi}$ be some function mapping the single-qubit state $\ket \psi$ to a real number.
Furthermore, let $U$ be some single-qubit unitary operation.
If the first of two qubits holding the state $\sigma$ is measured in the basis  $\{U \ket \psi, \Upsilon(U \ket \psi)\}$ with measurement outcome $c$
($c=0$ corresponding to $U\ket \psi$, $c=1$ corresponding to $\Upsilon(U\ket \psi)$)
after which the operation $U^\dagger \mathcal A^c_{\phi_{\ket \psi}, U \ket \psi}$ is applied to the second qubit and the first qubit is traced out,
then this is equivalent to sending a qubit in the state $\ket \psi$ through the rotated effective remote-state-preparation $\Lambda_{\phi_{\ket \psi} ,\sigma,U}$ channel given by
\begin{equation} \label{eq:rsp_channel_rotated}
\Lambda_{\phi_{\ket \psi}, \sigma,U} (\ket{\psi}) = U^\dagger \Lambda_{\phi_{\ket \psi}, \sigma}(U\ket \psi) U,
\end{equation}
where $\Lambda_{\phi_{\ket \psi}, \sigma}$ is the effective remote-state-preparation channel given by
\begin{equation} \label{eq:rsp_channel}
\Lambda_{\phi_{\ket \psi}, \sigma} (\ket{\psi}) = \Big (\bra \psi \otimes \mathbf 1 \Big ) \sigma \Big(\ket \psi \otimes \mathbf 1 \Big) + \Big(\bra \psi \otimes \mathbf 1 \Big) \Big( \mathcal A_{\phi_{\ket \psi}, \ket \psi} \otimes \mathcal A_{\phi_{\ket \psi}, \ket \psi} \Big) \sigma \Big(\mathcal A_{\phi_{\ket \psi}, \ket \psi} \otimes \mathcal A _{\phi_{\ket \psi}, \ket \psi} \Big) \Big(\ket \psi \otimes \mathbf 1 \Big).
\end{equation}
\end{lemma}
\begin{proof}
In case the state $U \ket \psi$ is measured on the first qubit, i.e., $c = 0$, the unnormalized post-measurement state after tracing out the first qubit and applying $U^\dagger \mathcal A^0_{\phi_{\ket \psi}, U\ket\psi} = U^\dagger$ is
\begin{equation}
\rho'_{c=0} = U^\dagger \Big (\bra \psi U^\dagger \otimes \mathbf 1 \Big ) \sigma \Big(U\ket \psi \otimes \mathbf 1 \Big) U.
\end{equation}
This measurement outcome is obtained with probability $p_{c=0}=\Tr{\rho_{c=0}}$, and the corresponding normalized state is $\rho_{c=0} = \rho'_{c=0} / p_{c=0}$.
In case the state $\Upsilon(U \ket \psi)$ (which is equal op to global phase to $\mathcal A_{\phi_{\ket \psi}, U \ket \psi} U \ket \psi$) is measured, i.e., $c=1$, the unnormalized state after tracing out the first qubit and applying $U^\dagger \mathcal A^1_{\phi_{\ket \psi}, U\ket\psi} = U^\dagger \mathcal A_{\phi_{\ket \psi}, U\ket\psi}$ is instead
\begin{equation}
\begin{aligned}
\rho'_{c=1} &= U^\dagger \mathcal A_{\phi_{\ket \psi}, U \ket \psi}  \Big (\bra \psi U^\dagger \mathcal A_{\phi_{\ket \psi}, U \ket \psi} \otimes \mathbf 1 \Big ) \sigma \Big(\mathcal A_{\phi_{\ket \psi}, U \ket \psi} U\ket \psi \otimes \mathbf 1 \Big) \mathcal A_{\phi_{\ket \psi}, U \ket \psi} U \\
&= U^\dagger  \Big (\bra \psi U^\dagger\otimes \mathbf 1 \Big )  \Big(\mathcal A_{\phi_{\ket \psi}, U \ket \psi} \otimes \mathcal A_{\phi_{\ket \psi}, U \ket \psi} \Big)  \sigma \Big(\mathcal A_{\phi_{\ket \psi}, U \ket \psi} \otimes \mathcal A_{\phi_{\ket \psi}, U \ket \psi} \Big) \Big( U\ket \psi \otimes \mathbf 1 \Big)U
\end{aligned}
\end{equation}
with measurement probability $p_{c=1}=\Tr{\rho_{c=1}}$ and normalized state $\rho_{c=1} = \rho'_{c=1} / p_{c=1}$.
The resulting state can be described as a mixture between the states corresponding to the different measurement outcomes weighted by their respective probabilities, i.e.,
\begin{equation}
\rho = p_{c=0}\rho_{c=0} + p_{c=1}\rho_{c=1} = \rho'_{c=0} + \rho'_{c=1} = \Lambda_{\phi_{\ket \psi}, \sigma,U}(\ket \psi).
\end{equation}
\end{proof}

We note that the effective remote-state-preparation channel is not a true quantum channel, i.e., it is not a completely positive trace-preserving (CPTP) map between density matrices.
In fact, it is only defined for pure states, and can not (straightforwardly) be rephrased as a linear operator on a density matrix.
However, the output state is a valid density matrix with trace 1, as it should be as it is the result of a measurement on and unitary evolution of the resource state $\sigma$.

\begin{theorem}
\label{thm:rsp_channel}
Equivalence of VBQC with remote state preparation.
Assume all local operations at both the server and the client are noiseless.
Then, there exists a function $\phi_{\ket \psi}$ that maps single-qubit states to real numbers such that
Protocol \ref{prot:vbqc_with_rsp} is equivalent to the unaltered VBQC protocol described in \cite{leichtle2021verifying}
using the rotated effective remote-state-preparation channel $\Lambda_{\phi_{\ket \psi}, \sigma,U}$ to send qubits in pure states from the client to the server.
Here, $\sigma$ is the resource state used in Protocol \ref{prot:vbqc_with_rsp}.
\end{theorem}
\begin{proof}
In Protocol \ref{prot:vbqc_with_rsp}, the client performs bit flips on the measurement outcomes received from the server.
Every measurement the server performs is in a basis of the form $\{\ket {+_\theta}, \ket {-_\theta}\}$ (defined in Equation \eqref{eq:define_+-_theta}).
These states are mapped to each other by the Pauli $Z$ operator, which is a $\ket{+_\theta}$-NOT operation
\begin{equation} \label{eq:A_and_Z}
Z = \mathcal A_{-\theta, \ket{+_\theta}}.
\end{equation}
Therefore, each measurement performed by the server in Protocol \ref{prot:vbqc_with_rsp} of which the result is bit flipped in case some number $c$ is equal to one (i.e., $\delta \to \delta \oplus c$ where $\delta$ is the measurement result)
can effectively be replaced by a unitary operation $Z^c$ followed by a measurement of which the result is not bit flipped.
It is thus as if the server applies the operation $Z^c$, even though the server never actually learns the value of $c$.
This equivalence is essential to the proof.
\\

First, we show that a computation round in Protocol \ref{prot:vbqc_with_rsp} is equivalent to a computation round in the unaltered VBQC protocol when sending the qubits using $\Lambda_{\phi_{\ket \psi}, \sigma,U}$ in case a specific condition on $\phi_{\ket \psi}$ holds.
In a computation round, for each of the qubits held by the server, it first performs the unitary operation $U^\dagger$.
Then, it executes a number of CZ gates between the qubit and some other qubits.
We remind the reader that CZ gates are symmetric in the two partaking qubits; we can thus always choose which qubit we consider the control qubit and which we consider the target qubit as we find convenient.
These gates are followed by a measurement in the basis $\{\ket{+_\theta}, \ket{-_\theta}\}$, where the angle $\theta$ is specified by the client.
The outcome $\delta$ of the measurement is bit flipped by the client according to $\delta \to \delta \oplus c_v$.
In shorthand, we will write the sequence as: $U^\dagger$, CZs, measurement, bit flip.
We will show that this sequence is equivalent to a sequence that we can apply Lemma \ref{lem:effective_channel} to.
As a first step, we use the equivalence stated in the first paragraph of this proof to replace the measurement followed by a bit flip by a measurement preceded by the operation $Z^{c_v}$.
The sequence is thus equivalent to the sequence: $U^\dagger$, CZs, $Z^{c_v}$, measurement.
As a second step, because $Z$ commutes with CZ, we rewrite the sequence as: $U^\dagger$, $Z^{c_v}$, CZs, measurement.
\\

Now, using Equation \eqref{eq:A_and_Z}, the sequence can be rewritten as follows: $U^\dagger$, $\mathcal A_{-\theta, \ket {+_{\theta}}}^{c_v}$, CZs, measurement.
To enable us to move the operator $U^\dagger$ in this sequence, we represent the unitary $U$ in general matrix form
\begin{equation} \label{eq:general_unitary}
U =
\begin{bmatrix}
a & b \\
-e^{i\varphi} b^* & e^{i\varphi} a^* \\
\end{bmatrix},
\end{equation}
where $|a|^2 + |b|^2 = 1$ and $\varphi \in [0, 2\pi)$.
This can be used to verify that
\begin{equation}
\Upsilon(U\ket \psi) = e^{-i\varphi} U \Upsilon(\ket \psi).
\end{equation}
Therefore, for every $U$, there exists a $\varphi$ such that for every $\phi$ and every $\ket \psi$
\begin{equation} \label{eq:U_and_A}
U^\dagger \mathcal A_{\phi - \varphi, U \ket \psi} = U^\dagger \Big[ e^{-i(\phi - \varphi)} e^{-i\varphi} U \Upsilon(\ket \psi) \bra \psi U^\dagger + e^{i(\phi - \varphi)} U \ket \psi \left(e^{-i\varphi} U \Upsilon(\ket \psi) \right)^\dagger\Big] = \mathcal A_{\phi, \ket \psi} U^\dagger.
\end{equation}
From this, we conclude that there exists a $\varphi$ (determined by $U$) such that the sequence on qubit $v$ is equivalent to: $\mathcal A_{-(\theta_v + \varphi), U \ket{+_{\theta_v}}}^{c_v}$, $U^\dagger$, CZs, measurement.
At this point, we are able to invoke Lemma \ref{lem:effective_channel}.
From this lemma, it follows that the client performing its measurement followed by the server applying the above sequence is equivalent to the the client sending the state $\ket{+_{\theta_v}}$ through a channel $\Lambda_{\phi_{\ket \psi}, \sigma, U}$ for which $\phi_{\ket {+_\theta}} = -\theta - \varphi$, after which the server applies the sequence: CZs, measurement.
This is exactly the sequence of operations in the unaltered VBQC protocol.
Therefore it follows that a computation round in Protocol \ref{prot:vbqc_with_rsp} is equivalent to a computation round in the unaltered VBQC protocol where the channel $\Lambda_{\phi_{\ket \psi}, \sigma, U}$ is used to send qubits from the client to the server in case the condition $\phi_{\ket {+_\theta}} = - \theta - \varphi $ is met.
\\

It now remains to show the same equivalence between the two protocols for test rounds.
For the trap qubit $v$, we can again replace the measurement followed by $c_v \oplus \bigoplus_{w \in N_G(v)} c_w \equiv \bar c$ bit flips by a measurement without bit flips preceded by the operator $Z ^ {\bar c}$.
The sequence of operations on the trap then becomes:
$U^\dagger$, CZ gates with dummy qubits, $Z^{\bar c}$, and then a measurement.
Now, the identity
\begin{equation}
\textnormal{CZ}(\mathbf 1 \otimes Z) = (X  \otimes \mathbf 1) \textnormal{CZ}( X  \otimes \mathbf 1)
\end{equation}
can be used to move every bit flip due to a measurement outcome in the preparation of a dummy qubit by the client to the corresponding qubit at the server.
That is, each $Z^{c_w}$ for $w \in N_G(v)$ is moved to the qubit $w$.
What remains at the trap qubit $v$ itself is then exactly the same sequence of operations as in a computation round.
From what we have shown above for computation rounds, it follows that we can treat trap qubits in test rounds of Protocol \ref{prot:vbqc_with_rsp} as if they are trap qubits in test rounds of the unaltered VBQC protocol, where the qubits are sent from the client to the server using the channel $\Lambda_{\phi_{\ket \psi}, \sigma,U}$ if $\phi_{\ket{+_\theta}} = - \theta - \varphi$.
% (where $\varphi$ is the same quantity as above and depends on the choice of $U$).
It then remains only to show that the the equivalence holds for the dummy qubits.
\\

Now, we focus on one of the dummy qubits, which we denote $w$.
Consider the scenario where the client attempts to send the qubit $w$ in the state $\ket d$, where $d \in \{0, 1\}$, to the server as in Protocol \ref{prot:vbqc_with_rsp}.
This qubit is the server's half of the resource state $\sigma$. 
The client measures its half of $\sigma$ in the basis $\{U\ket d, \Upsilon(U\ket d)\}$, with measurement outcome $c_w$.\
At the server, the following sequence of operations is applied to the qubit $w$: $U^\dagger$, $\prod_{u \in N_G(w)} \textnormal{CZ}_{w,u}$, measurement in the basis $\{\ket{+_\theta}, \ket{-_\theta}\}$ for some $\theta$.
Let us first consider the case where all $u \in N_G(w)$ are trap qubits.
Then, by moving the effects of bit flips from trap qubits to dummy qubits as described above, every $\textnormal{CZ}_{w,u}$ is effectively replaced by $(X^{c_w} \otimes \mathbf 1) \textnormal{CZ}_{w,u} (X^{c_w} \otimes \mathbf 1)$.
Because $X^2 = \mathbf 1$, this has the effect of transforming the sequence into the following: $U^\dagger$, $X^{c_w}$, $\prod_{u \in N_G(w)} \textnormal{CZ}_{w,u}$, $X^{c_w}$, measurement.
The second occurrence of $X^{c_w}$ changes the outcome of the measurement on the dummy qubit.
However, the measurement outcome of the dummy qubits is of no consequence in the VBQC protocol (the outcome is sent by the server to the client and then discarded by the client).
Therefore, we can effectively remove the second occurrence of $X^{c_w}$ from the sequence.
For the first occurrence, we note that $X$ is both a $\ket 1$-NOT gate and a $\ket 0$-NOT gate,
\begin{equation}
X = \mathcal A_{\ket 1} = - \mathcal A_{\ket 0}.
\end{equation}
Therefore, up to a global phase in case $d=0$, the sequence becomes equivalent to: $U^\dagger$, $\mathcal A_{\ket d}$, $\prod_{u \in N_G(w)} \textnormal{CZ}_{w,u}$, measurement.
We note that the unitary $U$ is here the same as for the trap qubit (it is the same for all qubits in Protocol \ref{prot:vbqc_with_rsp}).
Therefore, we can invoke Equation \eqref{eq:U_and_A} again to rewrite the sequence as: $\mathcal A_{-\varphi, U\ket d}$, $U^\dagger$, $\prod_{u \in N_G(w)} \textnormal{CZ}_{w,u}$, measurement.
It then immediately follows from Lemma \ref{lem:effective_channel} that this is equivalent to the client sending the qubit $w$ in the pure state $\ket d$ using a quantum channel $\Lambda_{\phi_{\ket \psi}, \sigma, U}$ for which $\phi_{\ket d} = - \varphi$.
After the server receives the qubit through this effective channel, the remaining sequence is: $\prod_{u \in N_G(w)} \textnormal{CZ}_{w,u}$, measurement.
This is the same as in the unaltered VBQC protocol, and therefore we can treat dummy qubits in test rounds of Protocol \ref{prot:vbqc_with_rsp} as if they are dummy qubits in the unaltered VBQC protocol that are transmitted using $\Lambda_{\phi_{\ket \psi}, \sigma,U}$ with the condition $\phi_{\ket 0} = \phi_{\ket 1} = -\varphi$, provided they are only adjacent to trap qubits in the computation graph $G$.
\\

As final part of our proof, we show that the above derivation for dummy qubits still holds in case they are adjacent to other dummy qubits in the computation graph.
Every CZ with a trap qubit results in two $X^{c_w}$s.
When the dummy qubit is only adjacent to trap qubits, $X^{c_w}$s resulting from neighboring CZs then cancel out in the middle (because $X^2=\mathbf 1$), such that only operators at the beginning and ending of the entire sequence remain.
However, a CZ with another dummy qubit does not give any $X^{c_w}$s.
$X^{c_w}$s from CZs with trap qubits that enclose one or more CZs with dummy qubits can then no longer cancel against one another.
A way out is offered by the following identity:
\begin{equation}
(\mathbf 1 \otimes X) \textnormal{CZ} = \textnormal{CZ} (Z \otimes X).
\end{equation}
This means that $X$ can be commuted through CZs at the cost of inducing a $Z$ at the other qubit partaking in the CZ.
Now, if a $Z$ is induced on a dummy qubit, it can be commuted through all CZs the dummy partakes in and placed in front of the measurement.
Here, it results in an effective bit flip on the measurement outcome.
Since again the measurement outcomes at the dummy qubits are inconsequential, the operator can safely be ignored.
This means that $X^{c_w}$s can safely commute through all the CZs with other dummy qubits, allowing them to cancel out as before and get again to a sequence where there is one $X^{c_w}$ before all the CZs and one after.
The sequence then is the same as when the dummy qubit would not be adjacent to other dummy qubits, and the same conclusion derived in the above paragraph holds.
\\

Combining all the above, we conclude that Protocol \ref{prot:vbqc_with_rsp} is equivalent to the VBQC protocol \cite{leichtle2021verifying} using the channel $\Lambda_{\phi_{\ket \psi}, \sigma, U}$ to send pure state from the client to the server.
This holds for any function $\phi_{\ket \psi}$ that satisfies
\begin{align}
\phi_{\ket {+_\theta}} &= -\theta - \varphi,\\
\phi_{\ket {d}} &= -\varphi,
\end{align}
for any $d \in \{0, 1\}$, for any $\theta \in [0, 2\pi)$, and where $\varphi$ depends on the choice of unitary $U$ in Protocol \ref{prot:vbqc_with_rsp} (it is the parameter appearing in Equation \eqref{eq:general_unitary}).
There exists an infinite number of functions satisfying this condition (note that it is not required that the function is continuous; in fact it does not matter in the least how the function behaves away from $\ket d$ and $\ket {+_\theta}$ as these are the only states that are ever sent through the channel), and therefore the theorem is proven.
\end{proof}

\begin{lemma} \label{lem:RSP_and_teleportation_fidelity}
Equivalence of remote state preparation and quantum teleportation.
The average fidelity of the effective remote-state-preparation channel (Equation \eqref{eq:rsp_channel}) corresponding to the two-qubit state $\sigma$,
\begin{equation}
F_\textnormal{RSP}(\sigma) \equiv \int_{\psi} d \psi \expectationvalue{\Lambda_{\phi_{\ket \psi}, \sigma} (\ket{\psi})} {\psi},
\end{equation}
is independent of the function $\phi_{\ket \psi}$.
Furthermore, it is equal to the average teleportation fidelity corresponding to the same state $\sigma$ (Equation \eqref{eq:F_tel}).
That is,
\begin{equation}
F_\textnormal{RSP}(\sigma) = F_\textnormal{tel}(\sigma)
\end{equation}
\end{lemma}
\begin{proof}
First we rewrite the average teleportation fidelity defined in Equation \eqref{eq:F_tel} as
\begin{equation}
F_\textnormal{tel}(\sigma) = \sum_{i,j} \int_\psi d \psi \Big(\bra \psi \otimes \bra{\Phi_{00}}\Big) \Big( X^iZ^j \otimes X^iZ^j \otimes \mathbf 1 \Big) (\sigma \otimes \ketbra \psi) \Big( X^iZ^j \otimes X^iZ^j \otimes \mathbf 1 \Big) \Big(\ket \psi \otimes \ket{\Phi_{00}}\Big).
\end{equation}
Then we use the property
\begin{equation}
\bra{\Phi_{00}}(\mathbf 1 \otimes \ket \psi) = \frac 1 {\sqrt 2} \bra \psi
\end{equation}
to find
\begin{equation}
F_\textnormal{tel}(\sigma) = \frac 1 2 \sum_{i,j} \int_\psi d\psi \Big(\bra \psi \otimes \bra \psi \Big) \Big( X^iZ^j \otimes X^iZ^j \Big) \sigma \Big( X^iZ^j \otimes X^iZ^j  \Big) \Big(\ket \psi \otimes \ket \psi\Big).
\end{equation}
Since the Haar measure is invariant under unitaries, the $X^iZ^j$ can be absorbed into the state $\ket \psi$, giving
\begin{equation} \label{eq:F_tel_evaluated}
F_\textnormal{tel}(\sigma) = 2 \int_\psi d\psi \Big(\bra \psi \otimes \bra \psi \Big) \sigma \Big(\ket \psi \otimes \ket \psi\Big).
\end{equation}
Similarly we can rewrite $F_\textnormal{RSP}(\sigma)$ as
\begin{equation} \label{eq:F_RSP_unprocessed}
F_\textnormal{RSP}(\sigma) = \int_\psi d\psi \Big(\bra \psi \otimes \bra \psi \Big) \sigma \Big(\ket \psi \otimes \ket \psi \Big) + \int_\psi d \psi \Big(\bra \psi \otimes \bra \psi \Big) \Big( \mathcal A_{\phi_{\ket \psi}, \ket \psi} \otimes \mathcal A_{\phi_{\ket\psi}, \ket \psi} \Big) \sigma \Big( \mathcal A_{\phi_{\ket \psi}, \ket \psi} \otimes \mathcal A_{\phi_{\ket \psi}, \ket \psi} \Big) \Big(\ket \psi \otimes \ket \psi \Big).
\end{equation}
The second term can be rewritten as
\begin{equation}
\begin{aligned}
\int_\psi d\psi & \Big( e^{i\phi_{\ket \psi}} \left( \Upsilon(\ket \psi) \right)^\dagger \otimes e^{i\phi_{\ket \psi}} \left( \Upsilon(\ket \psi) \right)^\dagger \Big) \sigma \Big( e^{-i\phi_{\ket \psi}} \Upsilon(\ket \psi) \otimes e^{-i\phi_{\ket \psi}} \Upsilon(\ket \psi) \Big) \\
&= \int_\psi d\psi \Big( \left( \Upsilon(\ket \psi) \right)^\dagger \otimes \left( \Upsilon(\ket \psi) \right)^\dagger \Big) \sigma \Big( \Upsilon(\ket \psi) \otimes \Upsilon(\ket \psi) \Big) \\
&= \int_\psi d\psi \Big( \bra \psi \otimes \psi \Big) \sigma \Big( \ket \psi \otimes \ket \psi \Big).
\end{aligned}
\end{equation}
The last step here follows from the fact that an integral over all antipodal points on the Bloch sphere is itself just an integral over all points on the Bloch sphere.
We thus find
\begin{equation} \label{eq:F_RSP_evaluated}
F_\textnormal{RSP}(\sigma) = 2 \int_\psi d \psi \Big(\bra \psi \otimes \bra \psi \Big) \sigma \Big(\ket \psi \otimes \ket \psi\Big).
\end{equation}
\end{proof}

\begin{theorem}
\label{thm:rsp_works}
Requirements on entanglement generation for VBQC through remote state preparation.
Assume a quantum link generates the two-qubit state $\sigma$ between a client and a server with average rate $R$.
Furthermore, assume that qubits at the server are stored in a depolarizing memory with coherence time $T$.
Lastly, assume that all local operations are noiseless and instantaneous.
Then, a unitary $U$ exists such that Protocol \ref{prot:vbqc_with_rsp} can be executed to realize the VBQC protocol \cite{leichtle2021verifying} for a two-qubit deterministic quantum computation in a way that is composably secure with exponentially small $\epsilon$ if
\begin{equation} \label{eq:vbqc_bound_rate_and_fidelity}
F_\textnormal{tel}(\sigma) > \frac 1 2 \Big( 1 + \frac{1}{\sqrt 2} e^{\frac 1 {2RT}} \Big).
\end{equation}
\end{theorem}
\begin{proof}
By Theorem \ref{thm:rsp_channel}, there exists a function $\phi_{\ket \psi}$ such that Protocol \ref{prot:vbqc_with_rsp} is equivalent to the VBQC protocol as presented in \cite{leichtle2021verifying} where qubits are transmitted using the channel $\Lambda_{\phi_{\ket \psi},\sigma,U}$.
Therefore, we can simply repeat the proof of Theorem \ref{thm:vbqc_constraint} but with the channel $\Lambda_{\sigma, U}$ replaced by $\Lambda_{\phi_{\ket \psi},\sigma,U}$
This results then exactly in Equation \eqref{eq:vbqc_bound_rate_and_fidelity}, but with $F_\textnormal{tel}(\sigma)$ replaced by the average fidelity over $\Lambda_{\phi_{\ket \psi}, \sigma}$, i.e., $F_\textnormal{RSP}(\sigma)$.
equation \eqref{eq:vbqc_bound_rate_and_fidelity} then follows directly from Lemma \ref{lem:RSP_and_teleportation_fidelity}.
\\

We note that in order to repeat the proof of Theorem \ref{thm:vbqc_constraint} two properties of the effective remote-state-preparation channel need to hold.
Specifically, they need to hold in order to reproduce Lemma \ref{lem:there_exists_unitary}.
These are properties that hold for any linear CPTP map.
$\Lambda_{\phi_{\ket \psi}, \sigma,U}$ however is not linear, but the properties can still be shown to hold.
First, the average fidelity of the channel is invariant under unitary transformations.
That is,
\begin{equation}
\int_\psi d \psi \expectationvalue{\Lambda_{\phi_{\ket \psi},\sigma,U}(\ket \psi)}{\psi} = F_\textnormal{RSP}(\sigma)
\end{equation}
for any unitary $U$.
This follows most evidently from Equation \eqref{eq:F_RSP_evaluated}, where the effect of including a unitary $U$ would be just to replace $\ket \psi \to U \ket \psi$, which leaves the Haar measure invariant.
\\

Second, it needs to be shown that $F_\textnormal{RSP}(\sigma)$ can be evaluated by evaluating the fidelity of $\Lambda_{\textnormal{RSP},\sigma}$ only at six states on the Bloch sphere forming a regular octahedron.
To this end, we use the fact that six states forming a regular octahedron are the union of three mutually unbiased bases and hence form a complex projective 2-design \cite{klappenecker2005}.
Therefore an integral over the Bloch sphere of which the integrand is a second-order polynomial in $\ketbra \psi$ can be replaced by an average over those six states.
We note that this cannot be applied to Equation \eqref{eq:F_RSP_unprocessed} directly, as the dependence of $\mathcal A_{\phi_{\ket \psi}, \ket \psi}$ on $\ket \psi$ means that the integrand is not necessarily a second-order polynomial.
However, it can be applied directly to Equation \eqref{eq:F_RSP_evaluated} to express $F_\textnormal{RSP}(\sigma)$ as an average over the six states.
Below, we show that the resulting expression is the same as taking the average over the six states directly in Equation \eqref{eq:F_RSP_unprocessed}.
\\

An octahedron is made up out of three pairs of antipodal points, so we denote the set of six states $\{\ket {\psi_i}, \Upsilon(\ket {\psi_i})\}$ for $i=0,1,2$.
Then, we can write \eqref{eq:F_RSP_evaluated} as
\begin{equation}
F_\textnormal{RSP}(\sigma) = \frac 1 3 \left( \sum_i \Big(\bra {\psi_i} \otimes \bra {\psi_i} \Big) \sigma \Big(\ket {\psi_i} \otimes \ket {\psi_i} \Big) + \sum_i \Big( (\Upsilon(\ket{\psi_i}))^\dagger \otimes(\Upsilon(\ket{\psi_i}))^\dagger \Big) \sigma \Big( (\Upsilon(\ket{\psi_i}) \otimes (\Upsilon(\ket{\psi_i})\Big) \right).
\end{equation}
It now remains to show that this is the same expression as what one would get from directly averaging the channel fidelity over these six states.
This direct average can be written as
\begin{equation}
\begin{aligned}
&\frac 1 6 \left( \sum_i \bra {\psi_i} \Lambda_{\phi_{\ket \psi}, \sigma}(\ket{\psi_i}) \ket {\psi_i} + \sum_i  (\Upsilon(\ket{\psi_i}))^\dagger \Lambda_{\phi_{\ket \psi}, \sigma}\left(\Upsilon(\ket{\psi_i})\right) \Upsilon(\ket{\psi_i}) \right)\\
=& \frac 1 6 \left( \sum_i \Big(\bra {\psi_i} \otimes \bra {\psi_i} \Big) \sigma \Big(\ket {\psi_i} \otimes \ket {\psi_i} \Big) + \sum_i \Big( (\Upsilon(\ket{\psi_i}))^\dagger \otimes(\Upsilon(\ket{\psi_i}))^\dagger \Big) \sigma \Big( (\Upsilon(\ket{\psi_i}) \otimes (\Upsilon(\ket{\psi_i})\Big) \right) \\
&+ \frac 1 6 \Bigg( \sum_i \Big(\bra {\psi_i} \otimes \bra {\psi_i} \Big) \Big( \mathcal A_{\phi_{\ket {\psi_i}}, \ket {\psi_i}} \otimes \mathcal A_{\phi_{\ket {\psi_i}}, \ket {\psi_i}} \Big)\sigma \Big( \mathcal A_{\phi_{\ket {\psi_i}}, \ket {\psi_i}} \otimes \mathcal A_{\phi_{\ket {\psi_i}}, \ket {\psi_i}} \Big)\Big(\ket {\psi_i} \otimes \ket {\psi_i} \Big) \\
&+ \sum_i \Big( (\Upsilon(\ket{\psi_i}))^\dagger \otimes(\Upsilon(\ket{\psi_i}))^\dagger \Big) \Big( \mathcal A_{\phi_{\ket {\psi_i}}, \ket {\psi_i}} \otimes \mathcal A_{\phi_{\ket {\psi_i}}, \ket {\psi_i}} \Big) \sigma  \Big( \mathcal A_{\phi_{\ket {\psi_i}}, \ket {\psi_i}} \otimes \mathcal A_{\phi_{\ket {\psi_i}}, \ket {\psi_i}} \Big)\Big( (\Upsilon(\ket{\psi_i}) \otimes (\Upsilon(\ket{\psi_i})\Big) \Bigg) \\
=& \frac 1 6 \left( \sum_i \Big(\bra {\psi_i} \otimes \bra {\psi_i} \Big) \sigma \Big(\ket {\psi_i} \otimes \ket {\psi_i} \Big) + \sum_i \Big( (\Upsilon(\ket{\psi_i}))^\dagger \otimes(\Upsilon(\ket{\psi_i}))^\dagger \Big) \sigma \Big( (\Upsilon(\ket{\psi_i}) \otimes (\Upsilon(\ket{\psi_i})\Big) \right) \\
&+ \frac 1 6 \Bigg( \sum_i \Big( e^{-i\phi_{\ket {\psi_i}}} (\Upsilon(\ket {\psi_i}))^\dagger \otimes e^{-i\phi_{\ket {\psi_i}}} (\Upsilon(\ket {\psi_i}))^\dagger \Big) \sigma \Big( e^{i\phi_{\ket {\psi_i}}} \Upsilon(\ket {\psi_i}) \otimes e^{i\phi_{\ket {\psi_i}}} \Upsilon(\ket {\psi_i}) \Big) \\
&+ \sum_i \Big( e^{i\phi_{\ket {\psi_i}}} \bra{\psi_i} \otimes e^{i\phi_{\ket {\psi_i}}} \bra{\psi_i} \Big) \sigma \Big( e^{-i\phi_{\ket {\psi_i}}} \ket{\psi_i} \otimes e^{-i\phi_{\ket {\psi_i}}} \ket{\psi_i} \Big) \Bigg) \\
=& \frac 1 3 \left( \sum_i \Big(\bra {\psi_i} \otimes \bra {\psi_i} \Big) \sigma \Big(\ket {\psi_i} \otimes \ket {\psi_i} \Big) + \sum_i \Big( (\Upsilon(\ket{\psi_i}))^\dagger \otimes(\Upsilon(\ket{\psi_i}))^\dagger \Big) \sigma \Big( (\Upsilon(\ket{\psi_i}) \otimes (\Upsilon(\ket{\psi_i})\Big) \right).
\end{aligned}
\end{equation}
Therefore, we conclude that taking the average of the fidelity over a regular octahedron of $\Lambda_{\phi_{\ket \psi}, \sigma}$ is equivalent to taking the average over the entire Bloch sphere using the Haar measure.
We note that the above argument also holds for $\Lambda_{\phi_{\ket \psi},\sigma,U}$ for any unitary $U$.
\end{proof}

%% file: sections/vbqc_noise_robustness_proof.tex
\label{app:noise_robustness}
In Appendix F of~\cite{leichtle2021verifying}, the authors show that their VBQC protocol is robust to noise, assuming that the probability of error in each round can be upper-bounded by some maximum probability of error $p_{\text{max}}$.
More specifically, they show that the protocol can be configured in such a way that it is $\epsilon_\text{cor}$-locally-correct with exponentially small $\epsilon_\text{cor}$.
\\

Here we argue that if we assume that the error probabilities are independent and identically distributed across different rounds of the protocol, then the error probability in each round is effectively equal to the average probability of error.
It then suffices that this average be bounded to obtain local correctness per the result of~\cite{leichtle2021verifying}, as the error probability becomes constant and the maximum error probability is equal to the average error probability.
We hereby prove Theorem \ref{thm:vbqc_relaxation}.
\\

We assume that for each round, there is a ``true'' probability of error.
This true probability of error is a random variable, with a second-order probability distribution determining what values it takes and with what probabilities~\cite{hansson2008we}.
Let $p_{\text{error}_i}$ be the probability of there being an error in round $i$, i.e., the value taken by the true probability of error in round $i$, drawn from the second-order probability distribution.
By the law of total probability, this can be written as:
\begin{equation}
    p_{\text{error}_i} = \int P\left(\text{error}\text{|}p=p_\text{e}\right)P\left(p = p_\text{e}\right)dp_\text{e},
\end{equation}
where $P\left(\text{error}\text{|}p=p_\text{e}\right)$ is the probability that there is an error given that the true probability of error takes the value $p_\text{e}$ and $P\left(p = p_\text{e}\right)$ is the probability density that this happens.
By definition, $P\left(\text{error}\text{|}p=p_\text{e}\right) = p_\text{e}$, therefore we can rewrite the equation as:
\begin{equation}
p_{\text{error}_i} = \int p_\text{e} P\left(p = p_\text{e}\right)dp_\text{e} = \overline{p_\text{e}},
\end{equation}
with $\overline{p_\text{e}}$ being the expected value of the second-order probability distribution from which each round's probability of error is sampled.
The second-order probability distribution can then be ignored, and the probability that an error occurs in a given round is simply given by a first-order probability.
It follows that the probability of error in every round is $\overline{p_\text{e}}$, i.e., the average probability of error, so it suffices that the average probability of error be bounded.

%% file: sections/double_click_model.tex
\section{Double-click model}
\label{appendix:sec_double_click_model}

In this appendix, we derive an analytical model for the entangled states created on elementary links when using the double-click protocol,
also known as the  Barrett-Kok protocol \cite{barrettEfficientHighfidelityQuantum2005}.
This model is used as one of the building blocks of our NetSquid simulations, as mentioned in Appendix \ref{appendix:protocols}.
To the best of our knowledge, the analytical model is a novel result.

\subsection{Model assumptions}

The double-click protocol is a protocol for heralded entanglement generation on an elementary link.
First, at each of the two nodes sharing the elementary link (designated A and B), a photon is emitted.
This photon can be in one of two different photonic modes.
For concreteness, we will here assume these two modes are horizontal and vertical polarization ($\ket H$ and $\ket V$, respectively), as is the case for the trapped-ion systems we consider in this work.
However, depending on the hardware platform that is used, they could just as well be some other modes, e.g., different temporal modes (``early'' and ``late''), as is the case for the color-center systems we consider.
Our model does not incorporate any effects specific to the type of modes that are used, and therefore the assumption that the modes are polarization modes is made without loss of generality.
The photon is emitted such that the mode that it is in is maximally entangled with the state of the emitter,
i.e. such that the emitter - photon state after emission is (up to normalization) $\ket{0H} + \ket{1V}$.
Then, the photons emitted at both nodes are sent to a midpoint station.\\

At the midpoint station, the photons from the two different nodes are interfered on a non-polarizing beam splitter.
The two output modes are then passed through a polarizing beam splitter, of which each output mode is impinged on a single-photon detector.
There are thus four single-photon detectors, two corresponding to horizontal polarization, and two corresponding to vertical polarization.
This setup is illustrated in Figure \ref{fig:double_click_setup}.
If a single photon is detected at one of the ``horizontal'' detectors and one at the ``vertical'' detectors,
assuming photons in the same polarization emitted at different nodes are indistinguishable,
the photons are projected on the state $\ket{HV} \pm \ket{VH}$.
This results in the emitters being in the maximally entangled state $\ket{\Psi^{\pm}} = \ket{01} \pm \ket{10}$.
The $+$ state is obtained if the two detectors clicking are located behind the same polarizing beam splitter,
while the $-$ state is obtained if they are located behind different polarizing beam splitters.
Note: if a different type of modes is used, this setup may look slightly different.
For example, in case temporal modes are used, there is no need for polarizing beam splitters and using only two single-photon detectors is sufficient as the different modes can be distinguished based on the time at which they are detected.
\\

\begin{figure}[h]
\includegraphics[width=.7\textwidth]{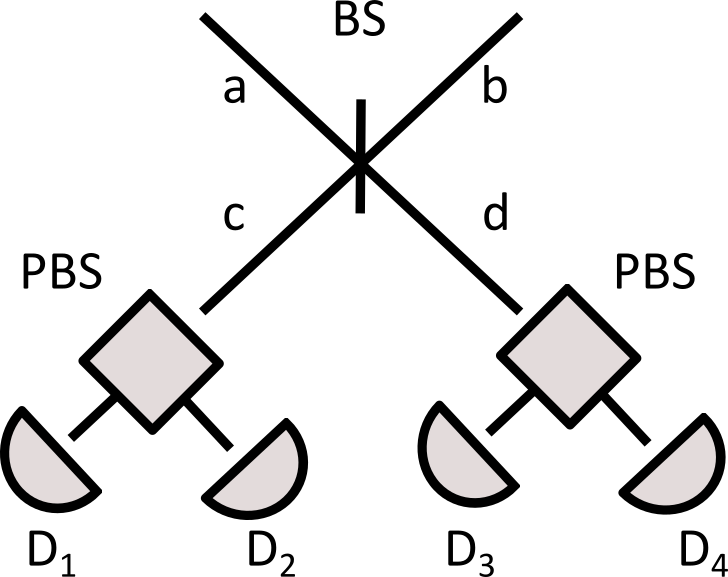}
\centering
\caption{
    Setup of midpoint station in double-click entanglement generation using polarization-encoded photons.
    Two photonic modes (a and b) are interfered on a non-polarizing 50-50 beam splitter (BS).
    The output modes (c and d) are then each led into a separate polarizing beam splitter (PBS).
    Each of the two output modes of each of the two polarizing beam splitters is caught at one of four detectors (D1, D2, D3 and D4).
    }
\label{fig:double_click_setup}
\end{figure}

In our simulations, we use an analytical model to describe the success probability and post-measurement state of the double-click scheme in the presence of several imperfections.
The imperfections included in our model are
\begin{itemize}
    \item{
        Photon loss.
        Due to nonunit collection efficiency of emitters,
        attenuation losses in optical fiber
        and inefficiency of single-photon detectors,
        there is often only a small probability that an emitted photon is not lost before it partakes in the midpoint measurement.
        This is captured by the parameters $p_\textnormal{A}$ and $p_\textnormal{A}$,
        where $p_\textnormal{A}$ ($p_\textnormal{A}$)  denotes the detection probability given that a photon is emitted at node $A$ ($B$).
        These account both for attenuation losses and for the photon detection probability excluding attenuation losses.
    }
    \item{
        Imperfect indistinguishability.
        We assume photons emitted by the different nodes with the same polarization are not perfectly indistinguishable.
        This is captured using the Hong-Ou-Mandel visibility $V$ \cite{hongMeasurementSubpicosecondTime1987, bouchard2020}.
        % which is the absolute value squared of the overlap between two interfering wave packets.
        We assume the visibility is the same between two horizontally polarized photons as between two vertically polarized photons.
    }
    \item{
        Non-photon-number-resolving detectors.
        In our model, we distinguish between the case of photon-number-resolving detectors (case NR) and non-photon-number-resolving detectors (case NNR).
        If the used detectors are NR, when there are two or more photons at the same detector during a single midpoint measurement, all photons are registered individually.
        However, if detectors are NNR, they cannot distinguish between one or more photons.
        This model does not account for the case when photons can sometimes, but not always, be distinguished.
        Such behavior occurs in reality when e.g. two photons can only be resolved if the time between their detections is large enough.
    }
    \item{
        Detector dark counts.
        Sometimes, single-photon detectors report the presence of a photon when there is none.
        We model this using a fixed dark-count probability, $p_\textnormal{dc}$.
        During a midpoint measurement, each single-photon detector gives a single dark count with probability $p_\textnormal{dc}$,
        and gives none with probability $1 - p_\textnormal{dc}$.
        Note that, in reality, for NR detectors, there is also a nonzero probability for multiple dark counts to occur during a single midpoint measurement in the same detector.
        Therefore, for NR detectors, treating dark counts this way will only lead to an approximation.
        The approximation can be expected to be accurate if the probability of multiple dark counts is negligible.
        For NNR detectors, this way of treating dark counts does not lead to an approximation but is perfectly accurate;
        multiple dark counts cannot be distinguished from one dark count, and therefore the probability of two or more dark counts and the probability of one dark count can be safely absorbed into one number, which is $p_\textnormal{dc}$.
    }
    \item{
        Imperfect emission.
        It is possible that, directly after emission,
        the emitter and photon are not in the maximally entangled state $\ket{\phi} = \tfrac 1 {\sqrt{2}} (\ket{0H} + \ket{1V})$.
        To capture this, the state is modelled as a Werner state of the form $ \rho_\textnormal{emit} = q\ketbra{\phi} + (1 - q) \frac{\mathbf 1} 4$.
        For each node, the parameter $q$ is chosen such that $F_\textnormal{em A}$ ($F_\textnormal{em B}$) is the emission fidelity $q + (1 - q) / 4 = \tfrac 1 4 (1 + 3 q)$ at node A (B).
    }
\end{itemize}

\subsection{POVMs}

To derive an analytical model, we notice that the midpoint station effectively implements a single-click midpoint measurement on each of the two different photonic modes (horizontal and vertical) separately.
To make use of this, we write the photonic states as Fock states on the two different modes,
such that $\ket H = \ket 1 _H \ket 0 _V$ and $\ket V = \ket 0 _H \ket 1 _V$.
Distinguishing also between photons arriving from side A and side B,
this allows us to write the pre-measurement state as a state in the Hilbert space
that is obtained from taking the tensor product between the Hilbert spaces of the emitters and the horizontally and vertically polarized photons.
That is, $\mathcal H _\textnormal{pre-measurement} = \mathcal H_\textnormal{A} \otimes \mathcal H_\textnormal{B} \otimes \mathcal H_{H_\textnormal{A}} \otimes \mathcal H_{H_\textnormal{B}} \otimes \mathcal H_{V_\textnormal{A}} \otimes \mathcal H_{V_\textnormal{B}}$.
Since we are not interested in the post-measurement state of the photons,
we can model the measurement as a POVM
\footnote{
    Note that we are interested in the post-measurement state of the emitters.
    However, as long as the state of the photons is traced out immediately after the measurement,
    a POVM is sufficient to accurately determine the post-measurement state.
}.
The POVM elements of the double-click midpoint station can then be derived from the single-click measurement operators as
\begin{equation}
M_\textnormal{double click, $ijkl$}, = \mathbf 1 _\textnormal{A} \otimes \mathbf 1_\textnormal{B} \otimes [M_\textnormal{single click, $ij$}]_{H_\textnormal{A} H_\textnormal{B}} \otimes [M_\textnormal{single click, $kl$}]_{V_\textnormal{A} V_\textnormal{B}}.
\end{equation}
Here, $M_\textnormal{single click, $ij$}$ is the POVM element corresponding to $i$ clicks in the first detector and $j$ clicks in the second detector of a single-click setup.
Thus, keeping in line with the naming of Figure \ref{fig:double_click_setup}, $M_\textnormal{double click, $ijkl$}$ is the POVM element corresponding to $i$ clicks in detector 1, $j$ clicks in detector 3, $k$ clicks in detector 2, and $l$ clicks in detector 4,
such that detectors 1 and 3, and 2 and 4 correspond to the same polarization,
and detector 1 and 2, and 3 and 4, correspond to the same polarizing beam splitter.
The single-click POVM elements can be obtained from Section D.5.2 of the supplemental material of \cite{dahlberg2019link}.
In doing so, we identify the square absolute value of the overlap of the two photon wave functions, $|\mu|^2$ in \cite{dahlberg2019link}, with the Hong-Ou-Mandel visibility $V$.
The reason for this is that these two are the same when both photons are in a pure state \cite{bouchard2020}
(we note that the photons in our model are only mixed in the polarization degree of freedom, the wave packets themselves are pure and therefore we can safely make the substitution).
We modify the single-click POVM elements from \cite{dahlberg2019link} to account for dark counts as follows (dropping the ``single click'' subscript):
\begin{equation}
\begin{aligned}
M_{10}' &= M_{10} (1 - p_\textnormal{dc}) ^ 2 + M_{00} p_\textnormal{dc} (1 - p_\textnormal{dc}), \\
M_{20}' &= M_{20} (1 - p_\textnormal{dc})+ M_{10} p_\textnormal{dc} (1 - p_\textnormal{dc}),
\end{aligned}
\end{equation}
and similarly for $M_{01}'$ and $M_{02}'$.
Note that we have absorbed the POVM element $M_{30}' = M_{20}p_\textnormal{dc} (1 - p_\textnormal{dc})$ into the POVM element $M_{20}'$,
since for neither NR and NNR detectors will the occurrence of two and the occurrence of three detections be discriminated;
for NNR detectors, the different detection events cannot be resolved,
while for NR detectors, both the presence of two and of three detections will lead to heralded failure.
Other POVM elements ($M_{00}'$, $M_{11}'$, $M_{21}'$, ...) are not needed for our analysis,
since having no detection in one of the polarizations, or having two detections at different detectors for one of the modes, is always heralded as a failure.\\

The double-click protocol heralds two different measurement outcomes as success, namely outcome ``detectors behind same polarizing beam splitter'' and outcome ``detectors behind different polarizing beam splitters''.
These two outcomes are henceforth abbreviated ``same PBS'' and ``different PBS''.
To determine the probability of each occurring and the corresponding post-measurement states,
we need to write down the POVM elements corresponding to these two outcomes.
Here, we note that in the case NR, the presence of multiple detections in a single detector is always heralded as a failure,
while in the case NNR, multiple detections cannot be distinguished from a single detection.
This gives the POVM elements (only writing the part acting on $\mathcal H_{H_\textnormal{A}} \otimes \mathcal H_{H_\textnormal{B}} \otimes \mathcal H_{V_\textnormal{A}} \otimes \mathcal H_{V_\textnormal{B}}$)
\begin{equation}
\begin{aligned}
M_\textnormal{same PBS, NR} &= M_{01}' \otimes M_{01}' + M_{10}' \otimes M_{10}', \\
M_\textnormal{different PBS, NR} &= M_{01}' \otimes M_{10}' + M_{10}' \otimes M_{01}', \\
M_\textnormal{same PBS, NNR} &= \sum_{n, m=1, 2}\Big( M_{0n}' \otimes M_{0m}' + M_{n0}' \otimes M_{m0}'\Big),\\
M_\textnormal{different PBS, NNR} &= \sum_{n, m=1, 2}\Big( M_{0n}' \otimes M_{m0}' + M_{0n}' \otimes M_{m0}'\Big).
\end{aligned}
\end{equation}

\subsection{Results without coincidence window}

To derive formulas for the success probability and post-measurement state,
we explicitly calculate the probabilities and post-measurement states of the above POVM elements on the six-qubit space using the symbolic-mathematics Python package SymPy \cite{sympy}.
The corresponding code can be found in the repository holding our simulation code~\cite{delft_eindhoven_code}.
The results are obtained by first initializing Werner states for each node and applying amplitude-damping channels with loss parameter $1 - p_\textnormal{A}$ on the $\mathcal {H}_{H_\textnormal{A}}$ and $\mathcal {H}_{V_\textnormal{A}}$ subspaces and $1 - p_\textnormal{B}$ on the $\mathcal{H}_{H_\textnormal{B}}$ and $\mathcal{H}_{V_\textnormal{B}}$ subspaces.
Then the probability and post-measurement state for both the ``same PBS''  and ``different PBS'' measurement outcomes are calculated from this pre-measurement state in both the cases NR and NNR.
The result can be written as
\begin{equation}
\begin{aligned}
p_\textnormal{double click} =& p_\textnormal{T} + p_{\textnormal{F}1} + p_{\textnormal{F}2} + p_{\textnormal{F}3} + p_{\textnormal{F}4},\\
\rho_\textnormal{double click} =& q_\textnormal{em} \Big( p_\textnormal{T} \ketbra{\Psi^{\pm}} + p_{\textnormal{F}1} \frac{\ketbra{01} + \ketbra{10}} 2 + p_{\textnormal{F}2} \frac{\ketbra{00} + \ketbra{11}} 2 \Big) \\
&+ \Big( (1 - q_\textnormal{em})(p_\textnormal{T} + p_{\textnormal{F}1} + p_{\textnormal{F}2}) + p_{\textnormal{F}3} + p_{\textnormal{F}4} \Big) \frac{\mathbf 1}{4},
\end{aligned}
\end{equation}
where $p_\textnormal{double click}$ is the success probability and $\rho_\textnormal{double click}$ is the unnormalized post-measurement state.
The different constants are defined as
\begin{equation}
\begin{aligned}
q_\textnormal{em} =& \frac 1 9 (4 F_\textnormal{em $A$} - 1) (4 F_\textnormal{em $B$} - 1),\\
p_\textnormal{T} =&
\begin{cases}
    \frac 1 2 p_\textnormal{A} p_\textnormal{B} V (1 - p_\textnormal{dc})^4 \qquad &\textnormal{if NR,}\\
    \frac 1 2 p_\textnormal{A} p_\textnormal{B} V (1 - p_\textnormal{dc})^2 \qquad &\textnormal{if NNR,}\\
\end{cases} \\
p_{\textnormal{F}1} =&
\begin{cases}
    \frac 1 2 p_\textnormal{A} p_\textnormal{B} (1 - V) (1 - p_\textnormal{dc})^4 \qquad &\textnormal{if NR,}\\
    \frac 1 2 p_\textnormal{A} p_\textnormal{B} (1 - V) (1 - p_\textnormal{dc})^2 \qquad &\textnormal{if NNR,}\\
\end{cases} \\
p_{\textnormal{F}2} =&
\begin{cases}
    0 \qquad &\textnormal{if NR,}\\
    \frac 1 2 p_\textnormal{A} p_\textnormal{B} (1 + V) p_\textnormal{dc} (1 - p_\textnormal{dc})^2 \qquad &\textnormal{if NNR,}\\
\end{cases} \\
p_{\textnormal{F}3} =&
\begin{cases}
    2 [p_\textnormal{A} (1 - p_\textnormal{B}) + (1 - p_\textnormal{A}) p_\textnormal{B}] p_\textnormal{dc} (1 - p_\textnormal{dc})^3 \qquad &\textnormal{if NR,}\\
    2 [p_\textnormal{A} (1 - p_\textnormal{B}) + (1 - p_\textnormal{A}) p_\textnormal{B}] p_\textnormal{dc} (1 - p_\textnormal{dc})^2 \qquad &\textnormal{if NNR,}\\
\end{cases} \\
p_{\textnormal{F}4} =& 4 (1 - p_\textnormal{A}) (1 - p_\textnormal{B}) p_\textnormal{dc}^2 (1 - p_\textnormal{dc})^2.\\
\end{aligned}
\end{equation}
Furthermore, the Bell states are defined by
\begin{equation}
\ket{\Psi^\pm} = \frac 1 {\sqrt 2} (\ket {01} \pm \ket {10}).
\end{equation}

The different terms in the equations can be interpreted as corresponding to different possible detection cases,
with $p_i$ being the probability of case $i$ occurring,
and the density matrix that it multiplies with the state that is created in that case.
The different cases are as follows.
\begin{itemize}
    \item{
        \textbf{case T}.
        This is a ``true'' heralded success.
        That is, two photons were detected at the midpoint station (probability $p_\textnormal{A} p_\textnormal{B}$)
        in different polarizations (probability $\tfrac 1 2$),
        and they behaved as indistinguishable photons (i.e. they interfered) (probability $V$).
        Finally, there cannot have been dark counts in any of the detectors,
        except for the detectors at which the photons were detected in the case of non-number-resolving detectors (as this doesn't change the outcome).
        The resulting density matrix is one corresponding to the Bell state $\ket {\Psi^\pm}$ (+ for both detections at the same polarizing beam splitter, - for both detections at different polarizing beam splitters).
    }
    \item{
        \textbf{case F1}.
        This is the first ``false'' heralded success
        (i.e. a false positive; a ``success'' detection pattern is observed without there being a maximally entangled state).
        Again, two photons arrived at the midpoint station and were detected in different polarizations (probability $\tfrac 1 2 p_\textnormal{A} p_\textnormal{B}$).
        However, they did not behave as indistinguishable photons (i.e. they did not interfere) (probability $(1 - V)$).
        Since the photons are in different polarizations, the post-measurement state will be classically anticorrelated $\tfrac 1 2 (\ketbra{01} + \ketbra{10})$.
    }
    \item{
        \textbf{case F2}.
        This is the second ``false'' heralded success.
        Two photons arrived at the midpoint station (probability $p_\textnormal{A} p_\textnormal{B}$) and are detected at the exact same detector.
        Additionally, a dark count occurs, causing a click pattern that is heralded as a success.
        For this, both photons need to be detected in the same polarization (probability $\tfrac 1 2$) and end up at the same detector.
        If they behave as indistinguishable photons (probability $V$), they will bunch together due to Hong-Ou-Mandel interference and will be guaranteed to go to the same detector.
        If they do not behave as indistinguishable photons (probability $1 - V$), there is a $\tfrac 1 2$ probability that they happen to go to the same detector.
        Combining, this gives a factor $V + \tfrac 1 2 (1 - V) = \tfrac 1 2 (1 + V)$.
        Since the photons are detected with the same polarization, the post-measurement state will be classically correlated $\tfrac 1 2 (\ketbra{00} + \ketbra{11})$.
        Note that this case cannot occur when detectors are NR,
        as detecting both photons at the same detector is then heralded as a failure.
    }
    \item{
        \textbf{case F3}.
        This is the third ``false'' heralded success.
        Only one photon arrives at the midpoint station (probability $p_\textnormal{A} (1 - p_\textnormal{B}) + (1 - p_\textnormal{A}) p_\textnormal{B}$),
        and a dark count makes the detector click pattern look like a success.
        For this, either of the two detectors corresponding to the polarization the photon is not detected in must undergo a dark count, while the remaining detectors do not undergo dark counts,
        which occurs with probability $2 p_\textnormal{dc} (1 - p_\textnormal{dc})^n$,
        where $n=3$ in case NR and $n=2$ in case NNR (because then it doesn't matter whether there is a dark count in the detector that detects the photon).
        There is no information about correlation between the photons,
        therefore the post-measurement state is maximally mixed.
    }
    \item{
        \textbf{case F4}.
        This is the fourth and final ``false'' heralded success.
        No photons arrive at the midpoint station (probability $(1 - p_\textnormal{A})(1 - p_\textnormal{B})$),
        and the detector click pattern is created solely by dark counts.
        Since there are four distinct click patterns resulting in a heralded success,
        these dark counts occur with probability $4 p_\textnormal{dc}^2 (1 - p_\textnormal{dc})^2$.
        There is no information about correlation between the photons,
        therefore the post-measurement state is maximally mixed.
    }
\end{itemize}

Finally, to understand the role of the parameter $q_\textnormal{em}$,
note that when either of the initial emitter - photon states is maximally mixed instead of entangled,
there is no correlation between the emitter and the photon.
Therefore, whatever detection event takes place at the midpoint station,
no information about correlation between the emitters is revealed.
The post-measurement state is thus a maximally mixed state in this case.
The probability that both nodes send an entangled photon instead of a maximally mixed state is exactly $q_\textnormal{em}$,
and thus the probability that the post-measurement state is maximally mixed regardless which of the above cases takes place is $1 - q_\textnormal{em}$.

\subsection{Results with coincidence window}

When performing the double-click protocol,
click patterns can be accepted as a success or instead rejected based on the time at which the two detector clicks are registered.
A first reason for this is that each round of the double-click protocol  only lasts a finite amount of time.
That is, there is a detection time window corresponding to each round,
and only clicks occurring during the detection time window can result in a heralded success for that specific round.
If there is a nonzero probability that photons are detected outside of the detection time window,
e.g. because their wave functions are stretched very long,
this can be captured in the model by adjusting the detection probabilities appropriately ($p_\textnormal{A}$ and $p_\textnormal{B}$). \\

However, there can also be a second reason.
Sometimes, it is desirable to implement a coincidence time window.
In this case, when two clicks occur within the correct detectors and within the detection time window,
a success is only heralded if the time between the two clicks is smaller than the coincidence time window.
While this lowers the success probability of the double-click protocol,
it can increase the Hong-Ou-Mandel visibility $V$
(thereby increasing the fidelity of entangled states created using the protocol). \\

To account for protocols that implement a coincidence time window,
we here introduce three new parameters into our model.
\begin{itemize}
    \item
    $p_\textnormal{ph-ph}$,
    the probability that two photon detections that occur within the detection time window occur less than one coincidence time window away from each other.

    \item
    $p_\textnormal{ph-dc}$,
    the probability that a photon detection and a dark count that occur within the detection time window occur less than one coincidence time window away from each other.

    \item
    $p_\textnormal{dc-dc}$,
    the probability that two dark counts that occur within the detection time window occur less than one coincidence time window away from each other.
\end{itemize}
These parameters will be functions of photon-detection-time probability-density functions and the coincidence time window (we calculate them for a simplified model of the photon state in Appendix \ref{app:time_windows})
Then, we make the following adjustments to the above results to account for the coincidence time window:
\begin{equation}
\begin{aligned}
p_\textnormal{T} &\to p_\textnormal{ph-ph} p_\textnormal{T},\\
p_{\textnormal{F}1} &\to p_\textnormal{ph-ph} p_{\textnormal{F}1},\\
p_{\textnormal{F}2} &\to p_\textnormal{ph-dc} p_{\textnormal{F}2},\\
p_{\textnormal{F}3} &\to p_\textnormal{ph-dc} p_{\textnormal{F}3},\\
p_{\textnormal{F}4} &\to p_\textnormal{dc-dc} p_{\textnormal{F}4}.
\end{aligned}
\end{equation}
The reason for this is as follows.
$p_\textnormal{T}$ corresponds to an event where two photons are detected, leading to a heralded success.
When using a coincidence time window, the two photons are only close enough in time to lead to a heralded success with probability $p_\textnormal{ph-ph}$.
The same logic holds for $p_{\textnormal{F}1}$.
Probability $p_{\textnormal{F}3}$ corresponds to a photon detection and a dark count leading to a heralded success;
that now only happens if the photon detection and dark count are within one coincidence time window,
which is exactly $p_\textnormal{ph-dc}$.
And probability $p_{\textnormal{F}4}$ corresponds to a heralded success due to two dark counts.
These dark counts also should not be separated by too much time,
giving a factor $p_\textnormal{dc-dc}$.
Less straightforward to adjust is $p_{\textnormal{F}2}$ in the NNR case.
It corresponds to an event where two photons are detected within the same detector,
but they are not independently resolved.
The probability that the time stamp assigned to this detection is within a coincidence window from a dark count occurring in another detector,
may not be exactly $p_\textnormal{ph-dc}$.
However, we do expect it to be a reasonable approximation,
and therefore we use $p_\textnormal{ph-dc}$ to avoid introducing a fourth new parameter to the model.

%% file: sections/coincidence_time_model.tex
\section{Effect of detection and coincidence time windows}
\label{app:time_windows}

In the double-click protocol, success is declared only if there are clicks in two detectors that measure different polarization modes.
These clicks typically occur at random times, and a prerequisite for success is that certain conditions on the detection times are met.
First, in any practical experiment, detection time windows have to be of finite duration.
If a click only occurs after the detection time window closes, it is effectively not detected.
Thus, success is only declared if two clicks occur within the detection time window.
Second, it is sometimes beneficial to also condition success on the time difference between the two clicks.
In that case, a success is only declared if the time between the clicks does not exceed the coincidence time window.
This can help boost the Hong-Ou-Mandel visibility of the photon interference and thereby increase the fidelity of entangled states.
\\

In Appendix \ref{appendix:sec_double_click_model}, we present a model that allows for the calculation of the success probability of the double-click protocol and the two-qubit state that it creates.
The coincidence probabilities between two photons, two dark counts and a photon and a dark count are free parameters in this model, just as the visibility and the photon detection probability.
To accurately account for the detection time window and coincidence time window in this model, these parameters need to be given appropriate values.
In this appendix, we introduce a simplified model for the photon state that allows us to calculate the required values.
We use this simplified model to simulate double-click entanglement generation with trapped-ion devices, as described in Appendix \ref{sec:setup_ti}.
To the best of our knowledge, this is a novel result.
\\

\begin{definition}
Detection time window.
If a detection time window of duration $T > 0$ is used in the double-click protocol,
success is only heralded if both detector clicks occur within the time interval $[0, T]$.
\end{definition}

\begin{definition}
Coincidence time window.
If a coincidence time window of duration $\tau > 0$ is used in the double-click protocol,
success is only heralded if the time between both detector clicks does not exceed $\tau$.
\end{definition}

\begin{definition}
Photon state described by $(p_\textnormal{em}(t), \psi_{t_0}(t))$.
% The photon wave packet described by the tuple of functions $(p_\textnormal{em}(t), \psi(t))$
% emission-time probability-density-function $p_\textnormal{em}(t)$ and pure wave function $\psi(t)$ % tuple of functions $(p_\textnormal{em}(t), $\psi(t))$,
% where
Let $p_\textnormal{em}(t)$ be a function such that
\begin{equation}
\int_0^\infty dt p_\textnormal{em}(t) = 1
\end{equation}
and let $\psi_{t_0}(t)$ be a function such that
\begin{equation}
\int_{t_0}^\infty dt |\psi_{t_0}(t)|^2 = 1.
\end{equation}
Then, $p_\textnormal{em}(t)$ can be interpreted as a probability density function for the photon emission time,
and $\psi(t)$ can be interpreted as the temporal wave function of a photon emitted at $t=0$.
The tuple $(p_\textnormal{em}(t), \psi(t))$ then describes a mixed photon state
\begin{equation}
\rho = \int_0^\infty dt_0 p_\textnormal{em}(t_0) \ketbra{\psi_{t_0}}
\end{equation}
where
\begin{equation}
\ket{\psi_{t_0}} = \int_{t_0}^\infty dt \psi_{t_0}(t) a^\dagger_t \ket 0
\end{equation}
with $a^\dagger_t$ the photon's creation operator at time $t$.
\end{definition}

The temporal impurity of a state described by $(p_\textnormal{em}(t), \psi_{t_0}(t))$ (if $p_\textnormal{em}(t)$ is not a delta function) can reduce the Hong-Ou-Mandel visibility of photons.
The reason for this is that photons that are emitted at very different times have small overlap.
If two photons are detected close together, they were probably not emitted at very different times (depending on their distributions).
Using a coincidence time window is then effectively applying a temporal purification to the photons, allowing for an increase in visibility.

\begin{definition}
Double-exponential photon state $(a, b)$.
The double-exponential photon state described by $(a, b)$,
where both $a$ and $b$ are constants with dimension time$^{-1}$,
is the photon state described by $(p_\textnormal{em}(t), \psi_{t_0}(t))$ where
\begin{equation}
p_\textnormal{em}(t) = a \textnormal{e}^{-at} \Theta(t)
\end{equation}
and
\begin{equation}
\psi_{t_0}(t) =
\sqrt{2b} \textnormal{e}^{-b(t - t_0)} \Theta(t-t_0).
\end{equation}
Here, $\Theta(t)$ is the Heaviside step function.
That is, both the emission-time probability density function and the pure photon wavefunctions are one-sided exponentials.
\end{definition}

In this appendix, we model all photons emitted by processing nodes as having a double-exponential state.
The pure wave functions of photons emitted by spontaneous decay of an excited state to a ground state in a two-level system are described well as one-sided exponentials \cite{kambs2018}.
An example of a system where photons are emitted this way is NV centers \cite{pompili2021realization}.
Similarly, pure wave functions of photons emitted using cavity-enhanced Raman transitions (using a constant Rabi pulse), as is the case for the trapped-ion systems we study in this paper, also look approximately exponential \cite{fioretto2020}.
We note that such trapped-ion systems are exactly the use case in this paper for the simplified model presented here.
For solid-state sources such as color centers, temporal impurity of photons is not a limiting factor \cite{kambs2018}.
However, for cavity-enhanced Raman transitions, off-resonant scattering causes the photon to only be emitted at a random time after a trajectory through the ion-state manifold \cite{meraner_indistinguishable_2020, fioretto2020}.
We model the resulting temporal impurity using the function $p_\textnormal{em}(t)$.
We note that we do not expect this function to be exponential for cavity-enhanced Raman transitions.
For instance, the function should include a $\delta(0)$ delta-function contribution to account for the probability that not a single off-resonant scattering takes place.
However, in the toy model presented here, we will assume $p_\textnormal{em}(t)$ is a one-sided exponential so that we have a model with a small number of parameters in which exact closed-form expressions can be obtained for the relevant quantities.
As shown in Appendix \ref{sec:setup_ti}, this model can be fitted well to experimental data for interference between photons emitted by ion-cavity systems.
\\

\begin{lemma} \label{lem:detection_time_pdf}
Detection-time probability density function.
Consider the case where a photon with double-exponential state $(a, b)$ is emitted directly on a photon detector.
Assume this photon detector is perfect except that it has a possibly nonunit detection efficiency $\eta$ (with a flat response).
The probability density function for the photon being detected at time $t$ is given by
\begin{equation}
p(t) = \frac{2ab \eta}{a - 2 b} \left( \textnormal{e}^{-2bt} - \textnormal{e}^{-at} \right) \Theta(t).
\end{equation}
This probability density function may be subnormalized, as it is also possible that no photon is detected.
\end{lemma}
\begin{proof}
A perfect detector implements a POVM with operators $E_t = a^\dagger_t \ketbra 0 a_t$.
Instead, a detector with efficiency factor $\eta$ implements a POVM with operators $E'_t = \eta E_t$ and $F = 1 - \eta$,
where $F$ corresponds to no photon detection taking place.
The probability density that the photon is detected at time $t$ is then the probability density corresponding to the POVM operator $E'_t$, given by
\begin{equation}
p(t) = \Trace (E'_t \rho).
\end{equation}
For a photon state described by $(p_\textnormal{em}(t), \psi_{t_0}(t))$, the density matrix is
\begin{equation}
\rho = \int_0^\infty dt_0 \int_{t_0}^\infty dt_1 \int_{t_0}^\infty dt_2 p_\textnormal{em}(t_0) \psi_{t_0}(t_1) \psi^*_{t_0}(t_1) a^\dagger_{t_1} \ketbra 0 a_{t_2}
\end{equation}
This can be evaluated using the cyclic property of the trace to give
\begin{equation}
\label{eq:detection_time_pdf_for_general_state}
\begin{aligned}
p(t) &=\eta \int_0^\infty dt_0 \int_{t_0}^\infty dt_1 \int_{t_0}^\infty dt_2  p_\textnormal{em}(t_0) \psi_{t_0}(t_1) \psi^*_{t_0}(t_1)\bra 0 a_t a^\dagger_{t_1} \ketbra 0 a_{t_2} a^\dagger_t \ket 0 \\
&=\eta \int_0^\infty dt_0 \int_{t_0}^\infty dt_1 \int_{t_0}^\infty dt_2  p_\textnormal{em}(t_0) \psi_{t_0}(t_1) \psi^*_{t_0}(t_1) \delta(t-t_1) \delta(t-t_2) \\
&=\eta \int_0^\infty dt_0 p_\textnormal{em}(t_0)|\psi_{t_0}(t)|^2.
\end{aligned}
\end{equation}
For a double-exponential photon state $(a, b)$, this becomes
\begin{equation}
\begin{aligned}
p(t) &= 2ab\eta \textnormal{e}^{-2bt} \int_0^\infty dt_0 \textnormal{e}^{-(a - 2 b)t_0} \Theta(t - t_0) \\
&= 2ab\eta \textnormal{e}^{-2bt} \Theta(t) \int_0^t dt_0 \textnormal{e}^{-(a - 2 b)t_0} \\
&= 2ab\eta \textnormal{e}^{-2b} \Theta(t) \frac 1 {a - 2b} (1 - \textnormal{e}^{-(a - 2b) t}) \\
&= \frac {2ab\eta}{a - 2b} \left(\textnormal{e}^{-2bt} - \textnormal{e}^{-at} \right) \Theta(t).
\end{aligned}
\end{equation}
\end{proof}

\begin{definition}
\label{def:coin_prob}
Coincidence probability.
When using a detection time window $T$ and coincidence time window of $\tau$ in the double-click protocol,
the coincidence probability is the probability that given that there are two clicks within the detection time window, 
the clicks are also within one coincidence time window.
\end{definition}

Our goal now is to find the coincidence probability for two double-exponential photons.
This requires us to calculate the probability that two photons arrive within a time $\tau$ of one another,
conditioned on each of the photons being successfully detected within the time interval $[0, T]$.
To this end, we calculate the probability density function for the detection time of a double-exponential photon conditioned on the photon being successfully detected.
This requires us to calculate the detection probability of the photon, i.e., the probability that it is successfully detected within the detection time window.
The detection probability is also an important result in itself, as it is required by the model presented in Appendix \ref{appendix:sec_double_click_model}
(it takes the role of $p_A$ and $p_B$ in this model).

\begin{theorem} \label{thm:prob_photon_in_detection_time_window}
Detection probability.
If a detection time window of duration $T$ is used, then the probability that a photon with double-exponential state $(a, b)$ is detected within the time window is given by
\begin{equation} \label{eq:prob_photon_in_detection_time_window}
p_\textnormal{det}(T) = \eta \bigg[1 - \frac a {a - 2b} \textnormal{e}^{-2bT} + \frac {2b} {a - 2b} \textnormal{e}^{-aT}\bigg].
\end{equation}
\end{theorem}
\begin{proof}
$p_\textnormal{det}(T)$ is given by the probability that the photon is detected in the time interval $[0, T]$.
This probability can be calculated from the probability density function
\begin{equation}
\begin{aligned}
p_\textnormal{det}(T) &= \int_0^T dt p(t) \\
&= \frac{2ab\eta}{a-2b} \int_0^T dt \left(\textnormal{e}^{-2bt} - \textnormal{e}^{-at} \right) \\
&= \frac{2ab\eta}{a-2b} \bigg[ \frac 1 {2b} \left( 1 - \textnormal{e}^{-2bT} \right) - \frac 1 a \left(1 - \textnormal{e}^{-aT} \right) \bigg] \\
&= \frac \eta {a - 2b} \bigg[ a \left(1 - \textnormal{e}^{-2bT} \right) - 2b \left(1 - \textnormal{e}^{-aT} \right) \bigg] \\
&= \eta \bigg[1 - \frac a {a - 2b} \textnormal{e}^{-2bT} + \frac {2b} {a - 2b} \textnormal{e}^{-aT} \bigg].
\end{aligned}
\end{equation}
\end{proof}

\begin{corollary} \label{cor:efficiency_no_detection_time_window}
When no detection time window is used, i.e., when the duration of the detection time window $T \to \infty$, then the photon detection probability is equal to the detector's detection efficiency $\eta$.
\end{corollary}
\begin{proof}
When we take $T \to \infty$ in Equation \eqref{eq:prob_photon_in_detection_time_window}, we find $p_\textnormal{det}(T) \to \eta$.
\end{proof}
This corresponds to the situation when the entire photon is within the detection time window and the only reason why the photon would not be detected is detector inefficiency.
One can think of $p_\textnormal{det}(T) / \eta$ as the ``additional efficiency factor'' due to not capturing the entire photon in the detection time window.

\begin{lemma} \label{lem:detection_time_pdf_conditional}
Conditional detection-time probability density function.
Consider the case where a photon with double-exponential state $(a, b)$ is emitted directly on a photon detector.
Assume this photon detector is perfect except that it has a possibly nonunit detection efficiency $\eta$ (with a flat response).
The probability density function for the photon being detected at time $t$, if the photon is in the double-exponential state $(a, b)$, is given by
\begin{equation}
\label{eq:detection_time_pdf_conditional}
\begin{aligned}
p_T(t) = \Theta(t) \Theta(T-t) \frac{p(t)}{p_\textnormal{det}(T)}.
\end{aligned}
\end{equation}
Unlike $p(t)$, this probability density function is always normalized.
\end{lemma}
\begin{proof}
Let $X$ be the continuous random variable corresponding to the detection time of the photon.
Let it take the value $-1$ if no photon is detected, such that the corresponding probability density function $f_X(x)$ is normalized and $X$ is well-defined as a random variable.
It follows from Corollary \ref{cor:efficiency_no_detection_time_window} that the probability that this happens is $1 - \eta$.
Therefore, the probability density function can then be writen as
\begin{equation}
f_X(x) = (1 - \eta) \delta(x+1) + \Theta(x) p(x).
\end{equation}
Additionally, we define a discrete random variable $Y_T$ which takes the value 1 if the photon is detected within the detection time window $T$ and $0$ if not.
By Theorem \ref{thm:prob_photon_in_detection_time_window}, the probability distribution of $Y_T$ is given by
\begin{equation}
f_{Y_T}(y) =
\begin{cases}
p_\textnormal{det}(T) \, &\textnormal{if $y = 1$} \\
1 - p_\textnormal{det}(T) \, &\textnormal{if $y = 0$.}
\end{cases}
\end{equation}
Note that $X$ and $Y_T$ are not independent random variables.
Now, the conditional probability density function that we are interested in is
\begin{equation} \label{eq:conditional_probability_X_Y}
p_T(x) \equiv f_{X|Y_T}(x, 1) = \frac{f_{X, Y_T}(x, 1)}{f_{Y_T}(1)} = \frac{f_{X, Y_T}(x, 1)}{p_\textnormal{det}(T)}.
\end{equation}
Here, $f_{X, Y_T} (x, y)$ is the mixed joint density of the continuous random variable $X$ and the discrete random variable $Y_T$.
When $X$ takes a value between 0 and $T$, then $Y_T$ takes the value 1 with unit probability.
Otherwise, $Y_T$ takes the value 0 with unit probability.
Therefore,
\begin{equation}
f_{X, Y_T}(x, 1) = \Theta(x) \Theta(T-x) f_X(x) = \Theta(x) \Theta(T-x) p(x).
\end{equation}
Substituting this into Equation \eqref{eq:conditional_probability_X_Y} allows us then finally to write
\begin{equation}
p_T(t) = \Theta(t) \Theta(T-t) \frac{p(t)}{p_\textnormal{det}(T)}.
\end{equation}
\end{proof}

\begin{theorem} \label{thm:coin_prob_ph_ph}
Coincidence probability of two photons.
The coincidence probability for two photon detections,
if both photons are in the double-exponential state $(a, b)$, a detection time window of duration $T$ is used and the coincidence time window of duration $\tau$ is used,
is given by
\begin{equation} \label{eq:coin_prob_ph_ph}
\begin{aligned}
p_\textnormal{ph-ph}(T, \tau) \left(\frac{p_\textnormal{det}(T)}{\eta}\right)^2 =&\frac{a^2}{a^2 - 4b^2}  (1 - \textnormal{e}^{-2b\tau}) -  \frac{4b^2}{a^2 - 4b^2}(1 - \textnormal{e}^{-a\tau})\\
&+ \frac{a^2}{(a - 2b)^2}(1 - \textnormal{e}^{2b\tau})  \textnormal{e}^{-4bT} + \frac{4b^2}{(a-2b)^2}(1 - \textnormal{e}^{a\tau})  \textnormal{e}^{-2aT} \\
&- \frac{4ab}{(a - 2b)^2}\left( 1 - \frac{a \textnormal{e}^{2b\tau} + 2b \textnormal{e}^{a\tau}}{a + 2b} \right) \textnormal{e}^{-(a + 2b) T} .
\end{aligned}
\end{equation}
\end{theorem}
\begin{proof}
By definition, the coincidence probability is given by
\begin{equation}
\label{eq:coin_prob_ph_ph_def}
p_\textnormal{ph-ph}(T, \tau) = \iint_{|t_1 - t_2| \leq \tau} dt_1 dt_2 p_T(t_1) p_T(t_2).
\end{equation}
By Lemma \ref{lem:detection_time_pdf_conditional}, this implies
\begin{equation}
p_\textnormal{ph-ph}(T, \tau) \left(\frac{p_\textnormal{det}(T)}{\eta}\right)^2 = \frac 1 {\eta^2} \iint_{|t_1 - t_2| \leq \tau} dt_1 dt_2 p(t_1) p(t_2).
\end{equation}
The region of integration is $|t_1 - t_2| \leq \tau$, i.e., $-\tau \leq t_1 - t_2 \leq \tau$.
The integrand is symmetric under the interchange of $t_1$ and $t_2$.
Therefore, the region $0 \leq t_1 - t_2 \leq \tau$ will give exactly the same contribution as $-\tau \leq t_1 - t_2 \leq 0$.
This has the following physical interpretation: the probability of photon 2 arriving a time $\Delta t$ after photon 1 is the same as the probability of photon 1 arriving $\Delta t$ after photon 2.
This can be used to simplify the integral somewhat, giving
\begin{equation}
p_\textnormal{ph-ph}(T, \tau) \left(\frac{p_\textnormal{det}(T)}{\eta}\right)^2 = \frac 2 {\eta^2} \iint_{0 \leq t_1 - t_2 \leq \tau} dt_1 dt_2 p(t_1) p(t_2).
\end{equation}
It follows from Lemma \ref{lem:detection_time_pdf} that each $p(t)$ carries an overall factor $\Theta(t) \Theta(T)$.
This can be absorbed into the integration limits to give
\begin{equation} \label{eq:coin_prob_ph_ph_two_parts}
\begin{aligned}
p_\textnormal{ph-ph}(T, \tau) \left(\frac{p_\textnormal{det}(T)}{\eta}\right)^2 &= \frac 2 {\eta^2} \int_0^T dt_1 \int_{t_1}^{\textnormal{min}(t_1 + \tau, T)} dt_2 p(t_1) p(t_2) \\
&=  2\Bigg( \frac 1 {\eta^2} \int_0^{T - \tau} dt_1 \int_{t_1}^{t_1 + \tau} dt_2 p(t_1) p(t_2) + \frac 1 {\eta^2} \int_{T-\tau}^T dt_1 \int_{t_1}^T dt_2 p(t_1) p(t_2) \Bigg).
\end{aligned}
\end{equation}
We calculate these two integrals one by one, using Lemma \ref{lem:detection_time_pdf}.
The first is
\begin{equation}
\begin{aligned}
\left( \frac{a-2b}{2ab} \right)^2 &\left( \frac 1 {\eta^2} \int_0^{T - \tau} dt_1 \int_{t_1}^{t_1 + \tau} dt_2 p(t_1) p(t_2) \right)\\
=& \int_0^{T-\tau} dt_1 (\textnormal{e}^{-2bt_1} - \textnormal{e}^{-at_1}) \int_{t_1}^{t_1 + \tau}(\textnormal{e}^{-2bt_2} - \textnormal{e}^{-at_2})\\
=& \int_0^{T-\tau} dt_1 (\textnormal{e}^{-2bt_1} - \textnormal{e}^{-at_1}) \left( \frac 1 {2b} \textnormal{e}^{-2bt_1} (1 - \textnormal{e}^{-2b\tau}) - \frac 1 a \textnormal{e}^{-at_1}(1 - \textnormal{e}^{-a\tau}) \right)\\
=& \frac{1 - \textnormal{e}^{-2b\tau}}{2b} \int_0^{T-\tau} dt_1 (\textnormal{e}^{-4bt_1} - \textnormal{e}^{-(a+2b)t_2}) - \frac{1 - \textnormal{e}^{-a\tau}}{a} \int_0^{T-\tau} dt_1 (\textnormal{e}^{-(a+2b)t_1} - \textnormal{e}^{-2at_1})\\
=& \frac{1 - \textnormal{e}^{-2b\tau}}{2b} \left( \frac{1 - \textnormal{e}^{4b\tau} \textnormal{e}^{-4bT}}{4b} - \frac{1 - \textnormal{e}^{(a+2b)\tau} \textnormal{e}^{-(a+2b)T}}{a+2b}\right)\\
&- \frac{1 - \textnormal{e}^{-a\tau}}{a} \left( \frac{1 - \textnormal{e}^{(a+2b)\tau} \textnormal{e}^{-(a+2b)T}}{a+2b} - \frac{1 - \textnormal{e}^{2a\tau} \textnormal{e}^{-2aT}}{2a}\right)\\
=& \frac {a-2b} {a + 2b} \left( \frac  {1 - \textnormal{e}^{-2b\tau}} {8b^2} - \frac{1 - \textnormal{e}^{-a\tau}}{2a^2}\right) +\frac{- \textnormal{e}^{2a\tau} + \textnormal{e}^{a\tau} }{2a^2} \textnormal{e}^{-2aT} \\
&+  \frac{\textnormal{e}^{2b\tau} - \textnormal{e}^{4b\tau}}{8b^2} \textnormal{e}^{-4bT} + \frac 1 {a+2b} \left( \frac{\textnormal{e}^{(a+2b)\tau} - \textnormal{e}^{a\tau}}{2b} + \frac{\textnormal{e}^{(a+2b)\tau} - \textnormal{e}^{2b\tau}}{a} \right) \textnormal{e}^{-(a+2b)T}.
\end{aligned}
\end{equation}
In the last step, terms with the same exponents of $T$ were collected.
Resolving the prefactor gives
\begin{equation} \label{eq:coin_prob_ph_ph_part_1}
\begin{aligned}
&\frac 1 {\eta^2} \int_0^{T - \tau} dt_1 \int_{t_1}^{t_1 + \tau} dt_2 p(t_1) p(t_2) \\
=& \frac {a^2} {2(a^2 - 4b^2)}  (1 - \textnormal{e}^{-2b\tau}) -  \frac {2b^2} {2(a^2 - 4b^2)} (1 - \textnormal{e}^{-a\tau}) + \frac{2b^2}{(a-2b)^2} (- \textnormal{e}^{2a\tau} + \textnormal{e}^{a\tau}) \textnormal{e}^{-2aT} \\
&+  \frac {a^2}{2(a-2b)^2} (\textnormal{e}^{2b\tau} - \textnormal{e}^{4b\tau}) \textnormal{e}^{-4bT} + \frac {2ab} {(a-2b)^2(a+2b)} \left( a (\textnormal{e}^{(a+2b)\tau} - \textnormal{e}^{a\tau}) + 2b (\textnormal{e}^{(a+2b)\tau} - \textnormal{e}^{2b\tau}) \right) \textnormal{e}^{-(a+2b)T}.
\end{aligned}
\end{equation}
The second term is
\begin{equation}
\begin{aligned}
\left( \frac{a-2b}{2ab} \right)^2 &\left( \frac 1 {\eta^2} \int_{T - \tau}^T dt_1 \int_{t_1}^{T} dt_2 p(t_1) p(t_2) \right)\\
=& \int_{T-\tau}^T dt_1 (\textnormal{e}^{-2bt_1} - \textnormal{e}^{-at_1}) \int_{t_1}^T(\textnormal{e}^{-2bt_2} - \textnormal{e}^{-at_2})\\
=& \int_{T-\tau}^T dt_1 (\textnormal{e}^{-2bt_1} - \textnormal{e}^{-at_1}) \left(\frac 1 {2b} (\textnormal{e}^{-2bt_1} - \textnormal{e}^{-2bT}) - \frac 1 a (\textnormal{e}^{-at_1} - \textnormal{e}^{-aT}) \right)\\
=& \frac 1 {2ab} \left( 2b \textnormal{e}^{-aT} - a \textnormal{e}^{-2bT} \right) \int_{T-\tau}^T dt_1  \left(\textnormal{e}^{-2bt_1} - \textnormal{e}^{-at_1} \right) + \frac 1 {2b} \int_{T-\tau}^T dt_1  \textnormal{e}^{-4bt_1}\\
&+ \frac 1 a \int_{T-\tau}^T dt_1  \textnormal{e}^{-2at_1} - \frac{a + 2 b}{2ab}\int_{T-\tau}^T dt_1  \textnormal{e}^{-(a+2b)t_1}\\
=& \frac 1 {2ab} \left( 2b \textnormal{e}^{-aT} - a \textnormal{e}^{-2bT} \right) \left[\frac{\textnormal{e}^{-2bT}}{2b} (\textnormal{e}^{2b\tau} - 1) - \frac{\textnormal{e}^{-aT}}{a} (\textnormal{e}^{a\tau} - 1) \right] + \frac {\textnormal{e}^{-4bT}} {8b^2} (\textnormal{e}^{4b\tau} - 1)\\
&+ \frac {\textnormal{e}^{-2aT}} {2a^2} (\textnormal{e}^{2a\tau} - 1) - \frac{\textnormal{e}^{-(a+2b)T}}{2ab} (\textnormal{e}^{(a+2b)\tau} - 1)\\
=& \frac 1 {2a^2} \left(\textnormal{e}^{2a\tau} - 2 \textnormal{e}^{a\tau} + 1 \right) \textnormal{e}^{-2aT} + \frac 1 {8b^2} \left(\textnormal{e}^{4b\tau} - 2\textnormal{e}^{2b\tau} + 1 \right) \textnormal{e}^{-4bT}\\
&+ \frac 1 {2ab} \left(\textnormal{e}^{a\tau} + \textnormal{e}^{2b\tau} - \textnormal{e}^{(a+2b)\tau} - 1 \right) \textnormal{e}^{-(a+2b)T}.
\end{aligned}
\end{equation}
Again resolving the prefactor, we find
\begin{equation} \label{eq:coin_prob_ph_ph_part_2}
\begin{aligned}
&\frac 1 {\eta^2} \int_{T - \tau}^T dt_1 \int_{t_1}^{T} dt_2 p(t_1) p(t_2) \\
=& \frac{2b^2}{(a-2b)^2}  \left(\textnormal{e}^{2a\tau} - 2\textnormal{e}^{a\tau} + 1\right) \textnormal{e}^{-2aT} + \frac{a^2}{2(a-2b)^2} \left( \textnormal{e}^{4b\tau} -  2 \textnormal{e}^{2b\tau} + 1 \right) \textnormal{e}^{-4bT}\\
&+ \frac{2ab}{(a-2b)^2} \left(\textnormal{e}^{a\tau} + \textnormal{e}^{2b\tau} - \textnormal{e}^{(a+2b)\tau} - 1 \right) \textnormal{e}^{-(a+2b)T}.
\end{aligned}
\end{equation}
Finally, substituting Equations \eqref{eq:coin_prob_ph_ph_part_1} and \eqref{eq:coin_prob_ph_ph_part_2} into Equation \eqref{eq:coin_prob_ph_ph_two_parts} yields Equation \eqref{eq:coin_prob_ph_ph}.
\end{proof}

\begin{theorem}
Coincidence probability of two dark counts.
The coincidence probability for two detector dark counts if the detection time window is $T$ and the coincidence time window is $\tau$ is given by
\begin{equation}
p_\textnormal{dc-dc}(T, \tau) = 1 - \left( \frac{T - \tau}{T} \right)^2,
\end{equation}
assuming that dark counts occur uniformly throughout the detection time window.
\end{theorem}
\begin{proof}
Given that there is a dark count within the detection time window $T$, the probability density function for the time at which is occurs is given by
\begin{equation} \label{eq:pdf_dark_count}
d_T(t) = \frac 1 T \Theta(t) \Theta(T-t)
\end{equation}
because we assume the dark counts are uniformly distributed.
The coincidence probability is then given by
\begin{equation}
\begin{aligned}
p_\textnormal{dc-dc}(T, \tau) =& \iint_{|t_1 - t_2| \leq \tau} dt_1 dt_2 d_T(t_1) d_T(t_2) \\
&= \frac{2}{T^2} \int_0^T dt_1 \int_{t_1}^{\textnormal{min}(t_1 + \tau, T)} dt_2\\
&= \frac{2}{T^2} \Bigg(\int_0^{T - \tau} dt_1 \int_{t_1}^{t_1 + \tau} dt_2  + \int_{T-\tau}^T dt_1 \int_{t_1}^T dt_2 \Bigg)\\
&= \frac{2}{T^2} \Bigg(\int_0^{T - \tau} dt_1 \tau +  \int_{T-\tau}^T dt_1 (T-t_1)  \Bigg) \\
&= \frac{2}{T^2} \left( (T-\tau) \tau + T\tau - \frac 1 2 T^2 + \frac 1 2 (T-\tau)^2 \right) \\
&= \frac 1 {T^2} (2T\tau - \tau^2)\\
&= 1 - \left(\frac{T-\tau}{T} \right)^2.
\end{aligned}
\end{equation}
\end{proof}

\begin{theorem} \label{thm:coin_prob_ph_dc}
Coincidence probability of a photon and a dark count.
The coincidence probability for one photon detection and one dark count, if the photon is in the double-exponential state $(a, b)$, the detection time window is $T$ and the coincidence time window is $\tau$ is given by
\begin{equation} \label{eq:coin_prob_ph_dc}
\begin{aligned}
p_\textnormal{ph-dc}(T, \tau) \frac{p_\textnormal{det}(T)}{\eta} = \frac a {2b(a-2b)T} &\Bigg[ 1 + 2b \tau - \textnormal{e}^{-2b\tau} + \textnormal{e}^{-2bT} \left(  1 - 2b \tau - \textnormal{e}^{2b\tau} \right) \Bigg] \\
- \frac {2b} {a(a-2b)T} &\Bigg[ 1 + a \tau - \textnormal{e}^{-a\tau} + \textnormal{e}^{-aT} \left(1 - a \tau - \textnormal{e}^{a\tau} \right) \Bigg].
\end{aligned}
\end{equation}
assuming that dark counts occur uniformly throughout the detection time window.
\end{theorem}
\begin{proof}
For the photon, we again have the probability density function $p_T(t)$ as given by Lemma \ref{lem:detection_time_pdf_conditional},
while for the dark count we have the probability density function $d_T(t)$ as given by Equation \eqref{eq:pdf_dark_count}.
The coincidence probability is then given by
\begin{equation}
p_\textnormal{ph-dc}(T, \tau) = \iint_{|t_1 - t_2| \leq \tau} dt_1 dt_2 p_T(t_1) d_T(t_2).
\end{equation}
When calculating the other coincidence probabilities, we were able to use a symmetry argument to simplify the integral.
However, because $p_T(t) \neq d_T(t)$, we are unable to do so here.
Assuming for the moment that $\tau \leq \tfrac 1 2 T$ and the fact that both probability density functions are proportional to $\Theta(t) \Theta(T-t)$,
we can write \begin{equation}
p_\textnormal{ph-dc}(T, \tau) = \left(\int_0^{\tau} dt_1 \int_0^{t_1 + \tau} dt_2 + \int_\tau^{T-\tau} dt_1 \int_{t_1 - \tau}^{t_1 + \tau} dt_2 + \int_{T-\tau}^T dt_1 \int_{t_1-\tau}^T dt_2 \right) p_T(t_1) d_T(t_2).
\end{equation}
This becomes
\begin{equation}
\begin{aligned}
&p_\textnormal{ph-dc}(T, \tau) \frac{p_\textnormal{det}(T)T}{\eta} \frac{a-2b}{2ab}\\
=& \left(\int_0^{\tau} dt_1 \int_0^{t_1 + \tau} dt_2 + \int_\tau^{T-\tau} dt_1 \int_{t_1 - \tau}^{t_1 + \tau} dt_2 + \int_{T-\tau}^T dt_1 \int_{t_1-\tau}^T dt_2 \right) (\textnormal{e}^{-2bt_1} - \textnormal{e}^{-at_1})\\
=& \left(\int_0^{\tau} dt_1 (t_1 + \tau)  + \int_\tau^{T-\tau} dt_1 2\tau  + \int_{T-\tau}^T dt_1 (T + \tau -t_1) \right) (\textnormal{e}^{-2bt_1} - \textnormal{e}^{-at_1}).
\end{aligned}
\end{equation}
We calculate these three terms individually.
Before doing this, we note that we can use integration by parts to calculate
\begin{equation}
\int_x^y dt \textnormal{e}^{-zt} t =  - \frac 1 z \left[ t \textnormal{e}^{-zt}\right]^{t=y}_{t=x} + \frac 1 z \int_x^y \textnormal{e}^{-zt}dt = \left[ \frac {\textnormal{e}^{-zt}}{z} (t + \frac 1 z) \right]^{t=x}_{t=y}.
\end{equation}
Then, the first term becomes
\begin{equation}
\begin{aligned}
\int_0^{\tau} dt_1 (t_1 + \tau)  (\textnormal{e}^{-2bt_1} - \textnormal{e}^{-at_1}) &= \left[ \left(\tau + \frac 1 {2b} + t_1\right) \frac{ \textnormal{e}^{-2bt_1} }{2b} \right]^\tau_0 - \left[ \left(\tau + \frac 1 {a} + t_1\right) \frac{ \textnormal{e}^{-at_1} }{a} \right]^\tau_0 \\
&= \frac 1 {4b^2} + \frac \tau {2b} - \frac 1 {a^2} - \frac \tau a - \left( \frac 1 {2b} + 2 \tau \right) \frac{\textnormal{e}^{-2b\tau}}{2b} + \left( \frac 1 a + 2 \tau \right) \frac{\textnormal{e}^{-a\tau}}{a}.
\end{aligned}
\end{equation}
The second yields
\begin{equation}
2\tau \int_\tau^{T-\tau} dt_1 (\textnormal{e}^{-2bt_1} - \textnormal{e}^{-at_1}) = 2\tau \frac{\textnormal{e}^{-2b\tau}}{2b} - 2\tau \frac{\textnormal{e}^{2b(\tau - T)}}{2b} - 2\tau \frac{\textnormal{e}^{-a\tau}}{a} + 2\tau \frac{\textnormal{e}^{a(\tau-T)}}{a}.
\end{equation}
The final one yields
\begin{equation}
\begin{aligned}
\int_{T-\tau}^T dt_1 (T + \tau -t_1) (\textnormal{e}^{-2bt_1} - \textnormal{e}^{-at_1}) &= \left[ \left( T + \tau - \frac 1 {2b} - t_1 \right) \frac{\textnormal{e}^{-2bt_1}}{2b} \right]^{T-\tau}_T - \left[ \left( T + \tau - \frac 1 a - t_1 \right) \frac{\textnormal{e}^{-at}}{a}  \right]^{T-\tau}_T \\
&= \left[ \frac 1 {2b} - \tau + \textnormal{e}^{2b\tau} \left( 2 \tau - \frac 1 {2b} \right) \right] \frac{\textnormal{e}^{-2bT}}{2b} - \left[ \frac 1 a - \tau + \textnormal{e}^{a\tau} \left(2\tau - \frac 1 a \right) \right] \frac{\textnormal{e}^{-aT}}{a}.
\end{aligned}
\end{equation}
We note that the second term cancels fully against the first and the third.
When adding all together, we find
\begin{equation}
\begin{aligned}
p_\textnormal{ph-dc}(T, \tau) \frac{p_\textnormal{det}(T)T}{\eta} \frac{a-2b}{2ab} = \frac 1 {2b} &\Bigg[ \frac 1 {2b} + \tau - \frac{\textnormal{e}^{-2b\tau}}{2b} + \textnormal{e}^{-2bT} \left( \frac 1 {2b} - \tau - \frac{\textnormal{e}^{2b\tau}}{2b} \right) \Bigg] \\
- \frac 1 a &\Bigg[ \frac 1 a + \tau - \frac{\textnormal{e}^{-a\tau}} a + \textnormal{e}^{-aT} \left( \frac 1 a - \tau - \frac{\textnormal{e}^{a\tau}}{a} \right) \Bigg]
\end{aligned}
\end{equation}
and thus
\begin{equation}
\begin{aligned}
p_\textnormal{ph-dc}(T, \tau) \frac{p_\textnormal{det}(T)}{\eta} = \frac a {2b(a-2b)T} &\Bigg[ 1 + 2b \tau - \textnormal{e}^{-2b\tau} + \textnormal{e}^{-2bT} \left(  1 - 2b \tau - \textnormal{e}^{2b\tau} \right) \Bigg] \\
- \frac {2b} {a(a-2b)T} &\Bigg[ 1 + a \tau - \textnormal{e}^{-a\tau} + \textnormal{e}^{-aT} \left(1 - a \tau - \textnormal{e}^{a\tau} \right) \Bigg].
\end{aligned}
\end{equation}
This procedure can be repeated when making the assumption $\tau \geq \tfrac T 2$.
In that case, the exact same formula is found.
%\Guus{Should we perform this calculation explicitly?}
Therefore, Equation \eqref{eq:coin_prob_ph_dc} is valid for any $0 \leq \tau \leq T$.
\end{proof}

\begin{definition} \label{def:visibility}
Visibility.
When using a detection time window $T$ and coincidence time window of $\tau$ in the double-click protocol,
the Hong-Ou-Mandel visibility is defined as
\begin{equation}
\label{eq:visibility_def}
% V(T, \tau) = 1 - \frac{ \textnormal{probability photons are detected at different detectors in case they are in the same mode}} {\textnormal{probability photons are detected at different detectors in case they are in different modes}}.
V(T, \tau) = 1 - \frac{ P(\textnormal{two photons detected at different detectors $|$ same mode})} { P(\textnormal{two photons detected at different detectors $|$ different modes}) }.
\end{equation}
Here,  $P(\textnormal{two photons detected at different detectors $|$ same mode})$ is the probability that if both incoming photons are in the same mode, they will be both be detected, and these detection events occur at different detectors.
On the other hand, $P(\textnormal{two photons detected at different detectors $|$ different modes})$ is the probability that if both incoming photons are in different modes (e.g., different polarizations), they will both be detected, and these detection events occur at different detectors.
Dark counts are not considered photon detections for the definitions of these probabilities.
\end{definition}
Note that when the two photons are in the same mode, they are able to interfere.
Then, if the two photons are pure and have the same temporal profile, they are perfectly indistinguishable and will never be detected at different detectors because of the Hong-Ou-Mandel effect \cite{hongMeasurementSubpicosecondTime1987}.
On the other hand, if the two photons are in different modes (e.g., different polarizations), they are not able to interfere.
% For perfectly indinstinguishable photons, the visibility is always one.
We note that the definition here given is in line with the definition given for the Hong-Ou-Mandel visibility in the main text.

\begin{theorem} \label{thm:visibility}
Visibility. % of two photons with double-exponential states $(a, b)$ given detection time window $T$ and coincidence time window $\tau$.
The Hong-Ou-Mandel visibility for a double-click protocol with two photons that are both in the double-exponential state $(a, b)$
% if there is a detection time window $T$ and a coincidence time window $\tau$,
is given by
\begin{equation} \label{eq:visibility}
\begin{aligned}
V(T, \tau) \left( \frac{ p_\textnormal{det}(T) } {\eta} \right) ^2 p_\textnormal{ph-ph}(T, \tau) =& \frac a {a + 2b} (1 - \textnormal{e}^{-2b\tau}) + \frac{2ab^2}{(a-2b)^2(a - b)}(1 - \textnormal{e}^{2(a-b)\tau})  \textnormal{e}^{-2aT} \\
&+ \frac{a^2}{(a - 2b)^2}(1 - \textnormal{e}^{2b\tau})  \textnormal{e}^{-4bT} - \frac{16ab^2}{(a-2b)^2(a+2b)} (1 - \textnormal{e}^{a\tau}) \textnormal{e}^{-(a + 2b)T} .
\end{aligned}
\end{equation}
\end{theorem}
\begin{proof}

First, we evaluate $P(\textnormal{two photons detected at different detectors | different modes})$.
Because the photons do not interfere, this probability is just the probability that both photons are detected within the detection time window and within one coincidence time window,
multiplied by a factor of $\frac 1 2$ as the probability for both photons going to different detectors is the same as the probability for both photons going to the same detector.
The probability for a single photon falling within the detection time window is the detection probability (Theorem \ref{thm:prob_photon_in_detection_time_window}),
and the probability of both photons being detected within a single coincidence time window is the photon-photon coincidence probability (Theorem \ref{thm:coin_prob_ph_ph}).
Thus,
\begin{equation}
P(\textnormal{two photons detected at different detectors | different modes}) = \frac 1 2 p_\textnormal{det}(T)^2 p_\textnormal{ph-ph}(T, \tau).
\end{equation}
The second probability can be evaluated as \cite{meraner_indistinguishable_2020}
\begin{equation}
\begin{aligned}
P(&\textnormal{two photons detected at different detectors | same mode}) \\
&= \frac {\eta^2} 4 \int_0^\infty dt_1 \int_0^\infty dt_2 \iint_{|t_1' - t_2'| \leq \tau} dt_1' dt_2' p_\textnormal{em}(t_1) p_\textnormal{em}(t_2) \left| \psi_{t_1}(t_1') \psi_{t_2}(t_2') -  \psi_{t_1}(t_2') \psi_{t_2}(t_1') \right|^2.
\end{aligned}
\end{equation}
From combining Equations \eqref{eq:detection_time_pdf_for_general_state}, \eqref{eq:detection_time_pdf_conditional} and \eqref{eq:coin_prob_ph_ph_def}, we see that
\begin{equation}
\eta^2 \int_0^\infty dt_1 \int_0^\infty dt_2 \iint_{|t_1'-t_2'| \leq \tau} p_\textnormal{em}(t_2) p_\textnormal{em}(t_2) |\psi_{t_1}(t_1')|^2 |\psi_{t_2}(t_2')|^2 = p_\textnormal{det}(T)^2 p_\textnormal{ph-ph}(T, \tau).
\end{equation}
We use this, together with the fact that for double-exponential photons it holds that $\psi^*_{t_0}(t) = \psi_{t_0}(t)$, to find
\begin{equation}
\begin{aligned}
P(&\textnormal{two photons detected at different detectors | same mode}) = \frac 1 2 p_\textnormal{det}(T)^2 p_\textnormal{ph-ph}(T, \tau) \\
&- \frac {\eta^2} 2 \int_0^\infty dt_1 \int_0^\infty dt_2 \iint_{|t_1' - t_2'| \leq \tau} dt_1' dt_2' p_\textnormal{em}(t_1) p_\textnormal{em}(t_2) \psi_{t_1}(t_1') \psi_{t_1}(t_2') \psi_{t_2}(t_2') \psi_{t_2}(t_1').
\end{aligned}
\end{equation}
We can then work out Equation \eqref{eq:visibility_def} to find
\begin{equation}
\begin{aligned}
&V(T, \tau) \left(\frac{p_\textnormal{det}(T)}{\eta}\right)^2 p_\textnormal{ph-ph}(T, \tau) \\
&= \int_0^\infty dt_1 \int_0^\infty dt_2 \int_{0}^T dt_1' \int_{0}^T dt_2' \Theta(|t_1' - t_2'| - \tau)  p_\textnormal{em}(t_1) p_\textnormal{em}(t_2) \psi_{t_1}(t_1') \psi_{t_1}(t_2') \psi_{t_2}(t_2') \psi_{t_2}(t_1').
\end{aligned}
\end{equation}
The integrand is symmetric under interchange of $t_1'$ and $t_2'$.
This allows us to consider only the region $0 \leq t_2' - t_1' \leq \tau$, giving
\begin{equation}
\begin{aligned}
&V(T, \tau) \left(\frac{p_\textnormal{det}(T)}{\eta}\right)^2 p_\textnormal{ph-ph}(T, \tau) \\
&= 2 \int_0^\infty dt_1 \int_0^\infty dt_2 \int_{0}^T dt_1' \int_{t_1'}^{\textnormal{min}(t_1'+\tau, T)} dt_2' p_\textnormal{em}(t_1) p_\textnormal{em}(t_2) \psi_{t_1}(t_1') \psi_{t_1}(t_2') \psi_{t_2}(t_2') \psi_{t_2}(t_1').
\end{aligned}
\end{equation}
Now, we notice that each $\psi_{t_0}(t)$ is proportional to $\Theta(t-t_0)$.
This can be absorbed into the limit of integration for $t_1$ and $t_2$, yielding
\begin{equation}
\begin{aligned}
&V(T, \tau) \left(\frac{p_\textnormal{det}(T)}{\eta}\right)^2 p_\textnormal{ph-ph}(T, \tau) \\
&= 2 \int_{0}^T dt_1' \int_0^{t_1'} dt_1 \int_0^{t_1'} dt_2 \int_{t_1'}^{\textnormal{min}(t_1'+\tau, T)}dt_2' p_\textnormal{em}(t_1) p_\textnormal{em}(t_2) \psi_{t_1}(t_1') \psi_{t_1}(t_2') \psi_{t_2}(t_2') \psi_{t_2}(t_1') \\
&= 8a^2b^2 \int_0^T dt_1' \textnormal{e}^{-2bt_1'} \int_0^{t_1'} dt_1 \textnormal{e}^{-(a-2b)t_1} \int_0^{t_1'} dt_2 \textnormal{e}^{-(a-2b)t_2} \int_{t_1'}^{\textnormal{min}(t_1'+\tau, T)} dt_2' \textnormal{e}^{-2bt_2'} \\
&= \frac{4a^2b}{(a-2b)^2} \int_0^T dt_1' \textnormal{e}^{-2bt_1'} (\textnormal{e}^{-2bt_1'} - \textnormal{e}^{-2b\textnormal{min}(t_1'+\tau, T)}) (1 - \textnormal{e}^{-(a-2b)t_1'})^2 \\
&= \frac{4a^2b}{(a-2b)^2} \left[ \int_0^T \textnormal{e}^{-2bt} - \textnormal{e}^{-2b\tau} \int_0^{T-\tau} \textnormal{e}^{-2bt} - \textnormal{e}^{-2bT} \int_{T-\tau}^T \right] \left( \textnormal{e}^{-2bt} - 2 \textnormal{e}^{-at} + \textnormal{e}^{-2(a-b)t} \right) dt.
\end{aligned}
\end{equation}
In the last step, we split up the integration region into a part where $t_1' + \tau$ is smaller and a part where $T$ is smaller.
Furthermore, for brevity, we renamed $t_1'$ to $t$.
We first calculate the first integral to find
\begin{equation}
\begin{aligned}
\int_0^T dt \left( \textnormal{e}^{-4bt} - 2\textnormal{e}^{-(a+2b)t} + \textnormal{e}^{-2at} \right) &= \frac 1 {4b} (1 - \textnormal{e}^{-4bT}) - \frac 2 {a+2b} (1 - \textnormal{e}^{-(a+2b)T}) + \frac 1 {2a} (1 - \textnormal{e}^{-2aT})\\
&= \frac 1 {4b} - \frac 2 {a + 2b} + \frac 1 {2a} - \frac 1 {2a} \textnormal{e}^{-2aT} - \frac 1 {4b} \textnormal{e}^{-4bT} + \frac 2 {a+2b} \textnormal{e}^{-(a+2b)T}.
\end{aligned}
\end{equation}
The second yields
\begin{equation}
\begin{aligned}
\textnormal{e}^{-2b\tau}& \int_0^{T-\tau} dt \left(- \textnormal{e}^{-4bt} + 2\textnormal{e}^{-(a+2b)t} - \textnormal{e}^{-2at} \right) \\
&= \textnormal{e}^{-2b\tau} \left(- \frac 1 {4b} (1 - \textnormal{e}^{-4b(T-\tau)}) + \frac 2 {a+2b} (1 - \textnormal{e}^{-(a+2b)(T-\tau)}) - \frac 1 {2a} (1 - \textnormal{e}^{-2a(T -\tau)}) \right) \\
&= \left( - \frac 1 {4b} + \frac 2 {a + 2b} - \frac 1 {2a} \right) \textnormal{e}^{-2b\tau} + \frac {\textnormal{e}^{2(a-b)\tau}} {2a} \textnormal{e}^{-2aT} + \frac {\textnormal{e}^{2b\tau}} {4b} \textnormal{e}^{-4bT} - \frac {2\textnormal{e}^{a\tau}} {a+2b} \textnormal{e}^{-(a+2b)T}.
\end{aligned}
\end{equation}
The final one yields
\begin{equation}
\begin{aligned}
\textnormal{e}^{-2bT}& \int_{T-\tau}^T dt \left(- \textnormal{e}^{-2bt} + 2\textnormal{e}^{-at} - \textnormal{e}^{-2(a-b)t} \right) \\
&= - \frac{\textnormal{e}^{2b\tau} - 1}{2b} \textnormal{e}^{-4bT} + \frac {2(\textnormal{e}^{a\tau} - 1)}{a} \textnormal{e}^{-(a + 2b)T} - \frac{\textnormal{e}^{2(a-b)\tau} - 1}{2(a - b)} \textnormal{e}^{-2aT}.
\end{aligned}
\end{equation}
Now, it is just a matter of adding these three terms together and taking the prefactor into account.
We collect terms separately for each different exponent with a $T$.
The part of the expression that is independent of $T$ yields
\begin{equation}
\begin{aligned}
(1 - \textnormal{e}^{-2b\tau}) \frac{4a^2b}{(a-2b)^2} \left(\frac 1 {4b} - \frac 2 {a + 2b} + \frac 1 {2a} \right) &= \frac a {(a-2b)^2} \left(a - \frac{8ab}{a+2b} + 2b\right) \\
&=(1 - \textnormal{e}^{-2b\tau}) \frac{a}{(a-2b)^2 (a+2b)} \left( a (a+2b) - 8ab + 2b(a+2b) \right) \\
&= (1 - \textnormal{e}^{-2b\tau}) \frac a {a+2b}.
\end{aligned}
\end{equation}
For the terms proportional to $\textnormal{e}^{-2aT}$ we find
\begin{equation}
\begin{aligned}
\frac{4a^2b}{(a-2b)^2} (1 - \textnormal{e}^{2(a-b)\tau}) \left( \frac 1 {2a} - \frac 1 {2(a-b)} \right) \textnormal{e}^{-2aT} &= \frac{2ab}{(a-2b)^2} (1 - \textnormal{e}^{2(a-b)\tau}) \left( \frac a {a-b} - 1 \right) \textnormal{e}^{-2aT} \\
&= \frac{2ab}{(a-2b)^2} (1 - \textnormal{e}^{2(a-b)\tau}) \frac b {a-b} \textnormal{e}^{-2aT} \\
&= \frac{2ab^2}{(a-2b)^2 (a-b)} (1 - \textnormal{e}^{2(a-b)\tau}) \frac b {a-b} \textnormal{e}^{-2aT}.
\end{aligned}
\end{equation}
For the terms proportional to $\textnormal{e}^{-4bT}$ we find
\begin{equation}
\frac{4a^2b}{(a-2b)^2}  (1 - \textnormal{e}^{2b\tau}) \left( \frac 1 {4b} - \frac 1 {2b} \right) \textnormal{e}^{-4bT} = \frac{a^2}{2(a-2b)^2} (1 - \textnormal{e}^{2b\tau}) \textnormal{e}^{-4bT}.
\end{equation}
Finally, for the terms proportional to $\textnormal{e}^{-(a+2b)T}$ we find
\begin{equation}
\begin{aligned}
\frac{4a^2b}{(a-2b)^2} (1 - \textnormal{e}^{a\tau}) \left(\frac 2 {a+2b} - \frac 2 a \right) \textnormal{e}^{-(a+2b)T} &= \frac{8ab}{(a-2b)^2(a+2b)} (1 - \textnormal{e}^{a\tau}) \left(a - (a + 2b) \right) \textnormal{e}^{-(a+2b)T} \\
&=  \frac{-16ab^2}{(a-2b)^2(a+2b)} (1 - \textnormal{e}^{a\tau}) \textnormal{e}^{-(a+2b)T}
\end{aligned}
\end{equation}
Adding these four different contributions together then yields Equation \eqref{eq:visibility}.
\end{proof}

We note that Lemma \ref{lem:detection_time_pdf_conditional} and Theorems \ref{thm:coin_prob_ph_ph} and \ref{thm:visibility} are compared to experimental results obtained with a trapped-ion device in Figure \ref{fig:ti_coincidence_time}.

%% file: sections/single_click_model.tex
\section{Single-click model}
\label{appendix:sec_single_click_model}
In this appendix, we present an analytical model for the entangled states created on elementary links when using a single-click entanglement generation protocol~\cite{cabrillo1999creation}.
This model is used as one of the building blocks of our NetSquid simulations, as mentioned in Appendix~\ref{appendix:protocols}, and based on the models previously introduced in~\cite{kalb2017entanglement, humphreys2018deterministic, pompili2021realization}.
The novelty of the model presented here lies in combining features of the three previous models~\cite{kalb2017entanglement, humphreys2018deterministic, pompili2021realization}
and in additionally considering the possibility that non-number-resolving detectors may be used (the three cited papers assume the use of number-resolving detectors).

\subsection{Model assumptions}
We model a single-click protocol for entanglement generation on an elementary link between two nodes, which we designate A and B.
The protocol starts with the preparation of an optically-active matter qubit at each of the nodes in the following state:
\begin{equation}
    \ket{\psi_\text{m}} = \sqrt{\alpha}\ket{\uparrow} + \sqrt{1-\alpha}\ket{\downarrow},
\end{equation}
where the subscript $\text{m}$ stands for matter, $\ket{\uparrow}$ is a bright state, i.e., a state that rapidly decays via photon emission after being excited, and $\alpha$ is the bright-state parameter, i.e. the fraction of the matter qubit's state that is in $\ket{\uparrow}$.
Excitation and subsequent radiative decay of $\ket{\uparrow}$ entangles the state of the matter qubit with the presence $\ket{1}$ or absence $\ket{0}$ of a photon (subscript $\text{p}$):
\begin{equation}
    \ket{\psi_\text{m}, \psi_\text{p}} = \sqrt{\alpha} \ket{\uparrow}\ket{1} + \sqrt{1 - \alpha} \ket{\downarrow}\ket{0}
\end{equation}
The photons are then sent to a heralding station where they are interfered.
Detection of a single photon heralds the generation of a matter-matter entangled state.
\\

In our analytical model, we account for the following imperfections when computing the success probability and entangled state generated with the protocol:
\begin{itemize}
    \item Double excitation of the matter qubit.
    Resonant laser light incident on the optically-active matter qubit triggers its excitation.
    It is possible that this excitation happens two times as the laser shines on the matter qubit, leading to the emission of two photons.
    This happens with probability $p_{\text{dexc}}$.
    We note that the excitation could in theory also happen multiple times, but, as detailed in Appendix~\ref{sec:single_click_results}, the effect this would have on the state would be the same as if two photons were emitted, so we can absorb the probability of more than one excitation into one quantity.
    \item Photon phase uncertainty.
    The photons interefering at the midpoint acquire a phase during transmission over the fiber, and the difference between the phases of the two interefering photons influences the entangled state that is generated~\cite{kalb2017entanglement}.
    The dephasing probability of the state $p_{\text{phase}}$ can be computed from the standard deviation of the difference between the acquired phases $\sigma_{\text{phase}}$~\cite{humphreys2018deterministic}:
        \begin{equation}
            p_{\text{phase}} = \frac{1}{2}\left(1 - \text{e}^{-\sigma^2_{\text{phase}} / 2}\right).
        \end{equation}
\end{itemize}
Furthermore, we also account for photon loss, imperfect indistinguishability, non-photon-number-resolving detectors and detector dark counts, as described in~\ref{appendix:sec_double_click_model}.
Finally, we account for the possibility of asymmetry in the placement of the heralding station, the attenuation of the fibers connecting the nodes to the heralding station and the bright-state parameters of the nodes.
\subsection{Results}
\label{sec:single_click_results}
Here we present the derivation of the entangled matter-matter state generated in our model of the single-click protocol.
We split this derivation into four situations, as done in~\cite{kalb2017entanglement, humphreys2018deterministic, pompili2021realization}.
Each of them corresponds to one of the different configurations of the states of the matter qubits that can result in a heralded success.
\begin{enumerate}
    \item Both matter qubits are in the bright state.
        In this case, both matter qubits emit a photon, so this situation is heralded as a success if:
        \begin{enumerate}
            \item Only one of the emitted photons survives.
            For NR detectors, there cannot be dark counts in either of the detectors
            (as otherwise two photons would be detected in the detector which also saw the actual emitted photon, and the event would be heralded as a failure).
            For NNR detectors, the requirement is just that there is no dark count in the detector that did not detect the emitted photon.
            The probability of this case happening is:
            \begin{equation}
                p_{1 \text{a}} = 
                \begin{cases}
                \alpha_\text{A} \alpha_\text{B}(1 - p_\text{dc})^2 (p_\text{A} (1 - p_\text{B}) + p_\text{B} (1 - p_\text{A})) \qquad &\text{if NR,}\\
                \alpha_\text{A} \alpha_\text{B}(1 - p_\text{dc}) (p_\text{A} (1 - p_\text{B}) + p_\text{B} (1 - p_\text{A})) \qquad &\text{if NNR.}\\
                \end{cases}
            \end{equation}
            \item No emitted photon is detected, and there is a dark count in one of the detectors.
            This case is the same irrespective of whether or not the detectors are NR.
            There is a factor of two because this can happen in either detector.
            The probability of this case happening is:
            \begin{equation}
                p_{1 \text{b}} = 2\alpha_\text{A} \alpha_\text{B} (1-p_\text{A})(1-p_\text{B})(1-p_{\text{dc}})p_{\text{dc}}.
            \end{equation}
            \item Both emitted photons make it to the midpoint and are detected, but they bunch and go into the same detector.
            Furthermore, there is no dark count in the other detector. There is a factor of two because this can happen in either detector.
            This is heralded as a failure if the detectors are NR.
            The probability of this case happening is:
            \begin{equation}
                p_{1 \text{c}} = 
                \begin{cases}
                0 \qquad &\text{if NR,}\\
                \alpha_\text{A} \alpha_\text{B} p_\text{A} p_\text{B} p_{\text{same dets}}(1 - p_{\text{dc}}) \qquad &\text{if NNR,}\\
                \end{cases}
            \end{equation}
            where $p_{\text{same dets}}$ is the probability that the two photons go to the same detector, as derived in case F2 of Appendix~\ref{appendix:sec_double_click_model}:
            \begin{equation}
                p_{\text{same dets}} = 1 - \frac{1 - V}{2}.
            \end{equation}
        \end{enumerate}
    \item Matter qubit A is in the bright state, matter qubit B is not.
        In this case, only matter qubit A emits a photon, so this situation is heralded as a success if:
        \begin{enumerate}
            \item The emitted photon survives.
            For NR detectors, there cannot be dark counts in either of the detectors
            (as otherwise two photons would be detected in the detector which also saw the actual emitted photon, and the event would be heralded as a failure). For NNR detectors, the requirement is just that there is no dark count in the detector that did not detect the emitted photon.
            \begin{equation}
                p_{2 \text{a}} = 
                \begin{cases}
                \alpha_\text{A} (1 - \alpha_\text{B})(1 - p_\text{dc})^2 p_\text{A} \qquad &\text{if NR,}\\
                \alpha_\text{A} (1 - \alpha_\text{B})(1 - p_\text{dc}) p_\text{A} \qquad &\text{if NNR.}\\
                \end{cases}
            \end{equation}
            \item The emitted photon does not survive, and there is a dark count in one of the detectors.
            This case is the same irrespective of whether the detectors are NR.
            There is a factor of two because this can happen in either detector.
            \begin{equation}
                p_{2 \text{b}} = 2 \alpha_\text{A} (1 - \alpha_\text{B})(1-p_\text{A})(1-p_{\text{dc}})p_{\text{dc}}.
            \end{equation}
        \end{enumerate}
    \item Matter qubit B is in the bright state, matter qubit A is not.
        In this case, only matter qubit B emits a photon.
        This identical to case 2, interchanging A and B.
    \item Neither of the matter qubits are in the bright state.
        No photon is emitted.
        In this case, we only get a success if there is a dark count in one of the detectors, but not the other.
        This case is the same irrespective of whether the detectors are NR.
        There is a factor of two because this can happen in either detector.
        \begin{equation}
            p_4 = 2 (1-\alpha_\text{A})(1-\alpha_\text{B})(1-p_{\text{dc}})p_{\text{dc}}.
        \end{equation}
\end{enumerate}

The overall success probability of the protocol $p_{\text{suc}}$ is then given by adding up the probability that each of the cases above happens, $p_{\text{suc}} = p_1 + p_2 + p_3 + p_4$, with $p_1 = p_{1 \text{a}} + p_{1 \text{b}} + p_{1 \text{c}}$, $p_2 = p_{2 \text{a}} + p_{2 \text{b}}$ and $p_3 = p_{3 \text{a}} + p_{3 \text{b}}$.
\\

The unnormalized density matrix of the generated state $\rho$ can then be obtained by taking the model introduced in~\cite{humphreys2018deterministic} and replacing the probabilities appropriately.
The result is the following:
\begin{equation}
    \rho = 
    \begin{pmatrix}
    p_1 & 0 & 0 & 0 \\
    0 & p_2 & \pm \sqrt{V  p_2  p_3} & 0 \\
    0 & \pm \sqrt{V  p_2  p_3} & p_3 & 0 \\
    0 & 0 & 0 & p_4
    \end{pmatrix},
\end{equation}
where the sign depends on which detector clicked.
\\

Two more dephasing channels are then applied in succession to the state in order to account for the effect of double photon excitation and photonic phase drift.
\\

The first channel, corresponding to double excitation, is applied to both matter qubits, with probability $p_{\text{dexc}}/2$.
The light pulse used to excite the bright state to a short-lived excited state is not instantaneous, so there is a chance that the matter qubit decays back down to the original state and be re-excited before the pulse is complete.
The first emitted photon will be lost to the environment because it is impossible to distinguish it from the laser light used to excite the matter qubit~\cite{humphreys2018deterministic}.
It must then be traced out, resulting in a loss of coherence between the two matter qubit states.
However, detection of the second emitted photon will falsely herald entanglement, so we apply a dephasing channel with probability $p_{\text{dexc}}/2$ to account for the possibility that more than one photon is emitted.
\\

The second one, corresponding to the photonic phase drift, is applied to only one of them, with probability $p_{\text{phase}}$.
The difference in the phases acquired by the two interfering photons results in a phase difference between the two components of the resulting Bell state~\cite{humphreys2018deterministic}.
Applying a dephasing channel to only one of the matter qubits, with the correct probability given by $p_{\text{phase}}$, has the same effect.

%% file: sections/optimization_appendix.tex
\section{Optimization method}
\label{sec:appendix_optimization}

In this appendix we provide more details regarding our optimization methodology.
As mentioned in the main text, this methodology is based on genetic algorithms, which come in several different flavors.
Our particular implementation is heavily based on the one introduced in~\cite{da2021optimizing}, to which we refer the interested reader.
The only novelty introduced here is the use of a different termination criterion, which is explained in detail in the following section.
We note also that the code for our implementation, together with the tools required for integration with NetSquid simulations, is publicly accessible at~\cite{smart-stopos}.
\\

\subsection{Termination criteria for genetic algorithms}
The matter of choosing termination criteria for genetic algorithms (and, more generally, evolutionary algorithms) has been the object of some study (see, e.g.,~\cite{jain2001termination} for a review).
If the algorithm is terminated too soon, good solutions might remain undiscovered.
On the other hand, running the algorithm for too long in case good solutions have already been found leads to wasting computational resources.
Typically-used termination criteria can be split into two groups~\cite{jain2001termination}:
\begin{enumerate}
    \item Direct termination criteria: these can be obtained directly from the optimization, without any extra data analysis.
    Examples include setting a maximum number of generations for the optimization or imposing a threshold value on the value of the best solution's cost;
    \item Derived termination criteria: these are \textit{a posteriori} criteria, requiring that some data analysis be performed on the outcome of the optimization.
    Examples include setting a threshold on the standard deviation of the costs of the individuals in the population or on the gap between the best and worst individuals in a given generation.
\end{enumerate}
The authors of~\cite{jain2001termination} applied an evolutionary algorithm to a particular cost function with different termination criteria.
They found that the only reliable termination criteria fitting into the groups above were one in which the algorithm terminated after a fixed, predetermined number of generations, which we name GEN, and one in which the best solution had not varied by more than a predetermined value after a predetermined number of generations, which we name VAR.
For all the other criteria tested, the algorithm did not terminate even though the optimal solution had already been found.
GEN and VAR both have the drawback of depending on hyperparameters for which a good choice can only be made with knowledge of the problem at hand.
By this we mean that the number of generations or accepted variation in the best solution per generation that guarantee termination are problem-dependent.
\\ 

With this in mind, we opted to employ VAR as the termination criterion for our optimization runs.
We made this choice as using VAR results in a more systematic, performance-dependent process for the decision of terminating the optimization.
By this we mean that even though GEN guarantees termination (by definition) it does so in an arbitrary way by deciding to stop the optimization without any regard for its evolution.
As suggested in~\cite{jain2001termination} we terminated the algorithm if the best solution's cost averaged over the past fifteen generations had not varied by more than a given value.
In contrast with the work of~\cite{jain2001termination}, we measured the variation in percentual terms.
For each setup, we ran the optimization process ten times, each for two hundred generations.
Then, to determine what the tolerance for the variation should be, we swept across its values, starting at $1\%$ and with a step of $1\%$.
The chosen tolerance was the one that guaranteed termination for all ten of the optimization runs, and the best solution (i.e., the ones showed in this work), was then the best cost found across the ten different runs, up until termination.
We note that the tolerance can be different for different setups.
We further note that this offline implementation of VAR is not good for saving computational resources, as the optimization must anyway be run for a large number of generations, with some of them being discarded.
It was however simpler to integrate into our workflow, which weighed heavily since we were more constrained in working hours than in computing hours.
\subsection{Cost function}
Our goal with this work was to find the minimal requirements for a quantum repeater enabling Verifiable Blind Quantum Computation between two nodes separated by fiber of length 226.5 km.
This implies solving a multi-objective optimization problem, as we want to minimize hardware-parameter improvement while simultaneously ensuring that performance targets are met.
There are various ways of approaching such problems, one of them being scalarization.
This consists of adding the cost functions corresponding to different objectives together, so that effectively only one scalar quantity has to be optimized.
Through this process, we arrive at the cost function $C$ introduced in Section~\ref{sec:methods}, which we reproduce in Equation~\eqref{eq:tot_cost_appendix}.
\begin{equation}
    C = w_1 \Big(1 + \left(F_{\text{min}} - F\right)^2\Big)\Theta\left(F_{\text{min}} - F\right)
    + w_2\Big(1 + \left(R_{\text{min}} - R\right)^2\Big)\Theta\left(R_{\text{min}} - R\right)
    + w_3H_{\text{C}}\left(x_{1_\text{c}},...,x_{N_\text{c}}\right)
    \label{eq:tot_cost_appendix}
\end{equation}
We recall that $H_{\text{C}}$ is the hardware cost, $w_i$ are the weights of the objectives, $\Theta$ is the Heaviside function and $F$ and $R$ are the average teleportation fidelity and entanglement generation rate achieved by the parameter set, respectively.
$F_{\text{min}}$ and $R_{\text{min}}$ are the minimal performance requirements.
Scalarization conveniently transforms multi-objective optimization into their much simpler single-objective counterparts, but it does so by stowing away the problem in defining the weights $w_1$, $w_2$ and $w_3$ assigned to each of the objectives.
Different choices in the weights can lead to different outcomes from the optimization procedure.
In this work, just as in~\cite{da2021optimizing}, we wanted the performance targets to be hard requirements, i.e., a set of hardware parameters that did not fulfill them should not be assigned a low cost.
To ensure this, we picked $w_1$, $w_2 \gg w_3$, such that $w_1 \Big(1 + \left(F_{\text{min}} - F\right)^2\Big)\Theta\left(F_{\text{min}} - F\right)$, $w_2\Big(1 + \left(R_{\text{min}} - R\right)^2\Big)\Theta\left(R_{\text{min}} - R\right) \gg w_3H_{\text{C}}\left(x_{1_\text{c}},...,x_{N_\text{c}}\right)$.
We set $w_1 = w_2 = 1 \times 10^{20}$ and $w_3 = 1$.
No particular heuristic was used to select these numbers.
They were picked because they ensure that the cost assigned to not meeting the performance targets is much higher than the hardware cost, effectively making the performance targets hard requirements. 
\\ 

As mentioned in the main text, we picked the hardware cost function because it reflects the concept of progressive hardness, i.e., that parameters become harder to improve as they approach their perfect value.
Furthermore, it satisfies a composability property regarding the probability of no-imperfection.
To see this, consider a parameter's probability of no-imperfection $p$ that can be expressed as the product of two other parameters' probabilities of no-imperfection, $p = p_a p_b$.
$p$ could for example be the probability that a photon emitted in the correct mode is collected into a fiber, while $p_a$ and $p_b$ are the probabilities that the photon is emitted with the right wavelength and collected into the fiber, respectively.
Improving $p$ by a factor of $k$ takes it to $\sqrt[k]{p} = \sqrt[k]{p_a p_b} = \sqrt[k]{p_a}\sqrt[k]{p_b}$, which is equivalent to improving $p_a$ and $p_b$ separately by the same factor $k$.
Therefore, hardware improvement as measured by this function is invariant to the granularity at which parameters are considered.
\\ 

The last aspect we would like to highlight regarding the cost function is its squared difference terms, i.e., $1 + (F_{\text{min}} - F)^2$ and $1 + (R_{\text{min}} - R)^2$.
These were introduced in~\cite{labay2021genetic} and are used to steer the algorithm towards sets of hardware parameters that are more likely to meet the performance targets.
They do this by ensuring that parameter sets which fail to meet the targets by a large margin are assigned a higher cost, being therefore less likely to progress into further generations.
\\ 

\subsection{Probabilities of no-imperfection}
For some parameters, such as the probability that a photon is not lost when coupling to a fiber, the conversion to probability of no-imperfection is obvious.
For others, such as coherence times, this is not so.
Therefore, we show in Table~\ref{tab:probabilities_no_error} the probability of no-imperfection for all parameters considered in our hardware models.
\begin{table}[!ht]
\begin{tabular}{|c|c|}
\hline
Parameter                                                             & Probability of no-imperfection                      \\ \hline
Photon detection probability excluding attenuation losses $p_{\text{det}}$                         & $p_{\text{det}}$                             \\ \hline
Probability of double excitation $p_{\text{dexc}}$                    & $1 - p_{\text{dexc}}$                        \\ \hline
Gate depolarizing probability $p_{\text{dep}}$                        & $1 - p_{\text{dep}}$                         \\ \hline
Number of entanglement generation attempts before dephasing $N_{1/\text{e}}$ & $(1 + \text{e}^{-1/N_{1/\text{e}}})/2$                     \\ \hline
$T_1$                                                                 & $e^{-t/T_1}$                                 \\ \hline
$T_2$                                                                 & $e^{-t/T_2}$                                 \\ \hline
Ion coherence time $T_\text{c}$                                              & $e^{- t/T_\text{c}^2}$                              \\ \hline
Emission fidelity $F_{\text{em}}$                                     & $1/3\left(4F_{\text{em}} - 1\right)$         \\ \hline
Swap quality $s_\text{q}$                                                    & $s_\text{q}$                                        \\ \hline
Visibility $V$                                                        & $V$                                          \\ \hline
Dark count probability $p_{\text{dc}}$                                       & $1 - p_{\text{dc}}$                                 \\ \hline
\end{tabular}
\caption{Probabilities of no-imperfection for parameters we optimized over in this work.
Some parameters were merged for brevity, e.g. the probability of no-imperfection presented for $T_2$ holds for the abstract model $T_2$, the carbon spin $T_2$ and the electron spin $T_2$.
In the probability of no-imperfection for each of the coherence times, $t$ is the time spent in memory.}
\label{tab:probabilities_no_error}
\end{table}
We proceed with the derivation of the probability of no-imperfection for some of the less obvious cases.
\\ 

As mentioned in the main text, and explained in detail in Supplementary Note 4 A c of~\cite{coopmans2021netsquid}, the initialization of an color center's electron spin state induces dephasing of its carbon spin states through their hyperfine coupling.
This is typically modelled as the carbon spin states dephasing with some probability each time entanglement generation is attempted ~\cite{kalb2018dephasing}.
This probability can be related to $N_{1/\text{e}}$ as $p = 1/2 \left(1 - \text{e}^{-1/N_{1/\text{e}}}\right)$.
The corresponding probability of no-imperfection is then $p_{\text{ne}} = 1 - p = (1 + \text{e}^{-1/N_{1/\text{e}}})/2$.
\\ 

$T_1$ represents the timescale over which qubit relaxation occurs, with the probability of amplitude damping occurring over a period of time $t$ being given by $p_{\text{ad}} = 1 - \text{e}^{-t/T_1}$.
The associated probability of no-imperfection is $\text{e}^{-t/T_1}$.
Improving $T_1$ by a factor of $k$ then corresponds to improving the probability of no-imperfection to $\sqrt[k]{\text{e}^{-t/T_1}}$.
Some algebra reveals that this is equivalent to multiplying $T_1$ by a factor of $k$, and that this holds irrespective of the chosen timescale.
\\ 

$T_2$ represents the timescale over which qubit dephasing occurs, with the probability of a $Z$ error occurring over a period of time $t$ being given by $p_\text{Z} = (1 - \text{e}^{-t/T_2})/2$.
The associated probability of no-imperfection is $p_{\text{ne}} = \frac{1+\text{e}^{-t/T_2}}{2}$.
To first order, this can be written as $p_{\text{ne}} = \text{e}^{-t/2T_2}$, and some algebra again reveals that with this approximation improving $T_2$ by a factor of $k$ is equivalent to multiplying it by the same factor.
\\ 

The ion coherence time $T_\text{c}$ also represents a timescale for dephasing, but in this the case the probability of a $Z$ error occurring is given by $1 - \frac{1}{2}\left(1 + \text{e}^{-2t^2/T^2}\right)$.
To first order, the probability of no-imperfection can thus be written as $p_{\text{ne}} = \text{e}^{-t^2/T^2}$.
In this case, improving $T_\text{c}$ by a factor of $k$ is equivalent to multiplying it by $\sqrt{k}$.
\\ 

We model noise in photon emission as a depolarizing channel of fidelity $F_{\text{em}}$.
The action of the depolarizing channel on a perfect Bell state $\ket{\Phi^+}$ can be written as follows:
\begin{equation*}
\ket{\Phi^+}\bra{\Phi^+} \rightarrow \left(1 - p_{\text{dep}}\right)\ket{\Phi^+}\bra{\Phi^+} + p_{\text{dep}} \frac{\mathbb{I}}{4},
\end{equation*}
where $p_{\text{dep}}$ is the associated depolarizing probability and $\mathbb{I}$ is the identity matrix.
This can be rewritten as:
\begin{equation*}
\ket{\Phi^+}\bra{\Phi^+} \rightarrow \left(1 - p_{\text{dep}}\right)\ket{\Phi^+}\bra{\Phi^+} + p_{\text{dep}} \frac{1}{4} \left(\ket{\Phi^+}\bra{\Phi^+} + \ket{\Phi^-}\bra{\Phi^-} + \ket{\Psi^+}\bra{\Psi^+} + \ket{\Psi^-}\bra{\Psi^-}\right).
\end{equation*}
Since the Bell states are orthogonal to each other, it follows that $p = \tfrac 4 3 (1 - F_{\text{em}})$ and that the corresponding probability of no-imperfection is $\tfrac 1 3 \left(4F_{\text{em}} - 1\right)$.
\\ 

The derivation of the probability of no-imperfection for the remaining parameters should be self-evident and is therefore omitted here.

\subsection{Optimizing over tunable parameters}
As discussed in Section~\ref{sec:methods}, the entanglement generation and distribution protocols employed in our simulations include parameters that can be freely varied.
We name these \textit{tunable} parameters.
They affect the behavior and performance of the setups we investigated, and as a consequence also the minimal hardware requirements.
The tunable parameters should thus be chosen such that the best possible performance is extracted from a given set of hardware parameters, minimizing the cost function.
The values of the tunable parameters that allow for this are the optimal values.
This is however not trivial, as different sets of hardware parameters perform best with different tunable parameters.
To illustrate this, we again go over the tunable parameters considered in our simulations.
\\

We start with the \textit{cut-off time}.
This is the maximum duration for which a state can be held in memory before being discarded.
For details on the implementation of a cut-off timer in our simulations, see Appendix~\ref{appendix:protocols}.
If the cut-off time is very short, states will not be held in memory for long, and therefore the end-to-end fidelity will be high.
On the other hand, states will also be frequently discarded and regenerated, which means that establishment of end-to-end entanglement will take longer.
In contrast, a very long cut-off is equivalent to no cut-off, in the sense that states are never discarded.
This maximizes the entanglement generation rate at the expense of lower state fidelity.
\\ 

The second tunable parameter is the \textit{bright-state parameter}, which is relevant in the single-click entanglement generation protocol.
This is the fraction of the superposition that is in the optically-active state, and therefore corresponds to the probability that a photon is emitted.
A larger bright-state parameter corresponds to a higher probability of entanglement generation, but at the expense of a lower fidelity, as it also introduces a component orthogonal to the Bell basis in the generated entangled state.
For more details on single-click entanglement generation see Appendix~\ref{appendix:sec_single_click_model}.
\\ 

The final tunable parameter is the \textit{coincidence time window}, which is part of our trapped ion double-click entanglement generation model.
Two detection events arising from the correct detectors are only heralded as a success if the time elapsed between the events is smaller than the coincidence time window.
It acts as a temporal filter, lowering the protocol's success probability but increasing the visibility and hence the fidelity of the generated entangled states.
For more details on our modeling of a coincidence time window, see Appendix~\ref{app:time_windows}.
\\ 

These three parameters can be used to trade-off rate against fidelity, and their optimal values are different for different sets of hardware parameters.
For example, if the coherence time is short and the detection probability is high, it will likely be beneficial to have a short cut-off time.
The opposite is true if the coherence time is long and the detection probability is low.
\\ 

In order to find good values for the tunable parameters, we included them as parameters to be optimized by the genetic-algorithm-based optimization machinery.
We imposed that the values the cut-off time can take are in the interval between $0.1 \, T_\text{C}$ and $T_\text{C}$, where $T_\text{C}$ is the coherence time (collective dephasing coherence time for trapped ions, $T_2$ for abstract nodes and carbon $T_2$ for NV centers).
The expected entanglement generation time grows exponentially as the cut-off time is reduced, so the lower bound was imposed to prevent the simulation taking unreasonably long to run.
Furthermore, we anyway expect that a too low cut-off time would not allow the rate target to be met, so we can be reasonably sure that no cheap hardware requirements are missed by imposing this constraint.
The upper bound is imposed as we observed that not imposing it made it hard for the algorithm to converge, due to the reduced sensitivity of the target metrics to high values of the cut-off time.
As discussed above, employing a very long cut-off time is effectively equivalent to not employing one at all.
Therefore, in that regime the choice of cut-off time becomes irrelevant, and the set of parameters minimizing the cost function is chosen independently of it.
We have empirically observed that the cut-off time tends to converge to around $65$\% of the relevant coherence time, which is fairly distant from both bounds we imposed.
A back-of-the-envelope calculation can also be performed to argue that it is unlikely that allowing for cut-off times which are larger than the memory's coherence time would be useful.
We do this by computing the end-to-end fidelity in a single sequential-repeater setup under the following assumptions:
\begin{itemize}
    \item The cut-off time is equal to the memory dephasing time;
    \item There are no other noise sources.
\end{itemize}
The worst case scenario in this setup in terms of fidelity occurs when the second entangled state takes exactly cut-off time seconds to be generated, resulting in both qubits of the first entangled pair to dephase for a time equal to their dephasing time.
The dephasing probability is in this case given by $p_Z = \frac{1 - \text{e}^{-2}}{2}$.
Assuming that the state that had been generated was $\ket{\Phi^+}$, the post-dephasing state is a mixture of $\ket{\Phi^+}$ and $\ket{\Phi^-}$:
\begin{equation}
\rho = (1 - p_Z) \ket{\Phi^+}\bra{\Phi^+} + p_Z \ket{\Phi^-}\bra{\Phi^-}.
\end{equation} 
This has a fidelity of $0.57$ with the target Bell state $\ket{\Phi^+}$, corresponding to a teleportation fidelity of $0.71$.
This value is much lower than our lowest teleportation fidelity target, 0.8571, even with no noise sources besides dephasing noise on the memory.
It is then unlikely that picking even higher cut-off times would lead to finding better solutions to our optimization problem.  
\\ 

In the single-repeater setup we investigated, there are four bright-state parameters to be chosen, corresponding to the four different fiber segments between processing nodes and heralding stations.
We imposed that $\alpha p_{\text{det}}$ had to be equal for all of them, with $\alpha$ is the bright-state parameter and $p_{\text{det}}$ the probability that a photon is not lost in the fiber connecting the node to the midpoint station.
This condition guarantees balanced entanglement-generation success probabilities across all segments, which is a good heuristic for segments connecting to the same heralding station, as it maximizes the fidelity of the generated states~\cite{pompili2021realization}.
Imposing it also for segments connecting to different heralding stations was done in order to reduce the size of the search space.
\\ 

There are also two coincidence time window parameters to be chosen, corresponding to the two elementary links.
We imposed that they must have the same value in order to make the search space smaller.

%% file: sections/simulation_performance_appendix.tex
\section{Simulation performance}
\label{sec:appendix_simulation_performance}
Each execution of our quantum-network simulations simulates the delivery of $n$ end-to-end entangled states.
When the protocols running on the end nodes learn through classical communication between nodes that $n$ states were successfully distributed, they abort and the simulation terminates.
In Figure~\ref{fig:runtime_scaling_nruns}, we show how the runtime of our simulation of the Delft - Eindhoven setup scales with the number $n$.
\begin{figure}[!ht]
\centering
\includegraphics[width=0.75\columnwidth]{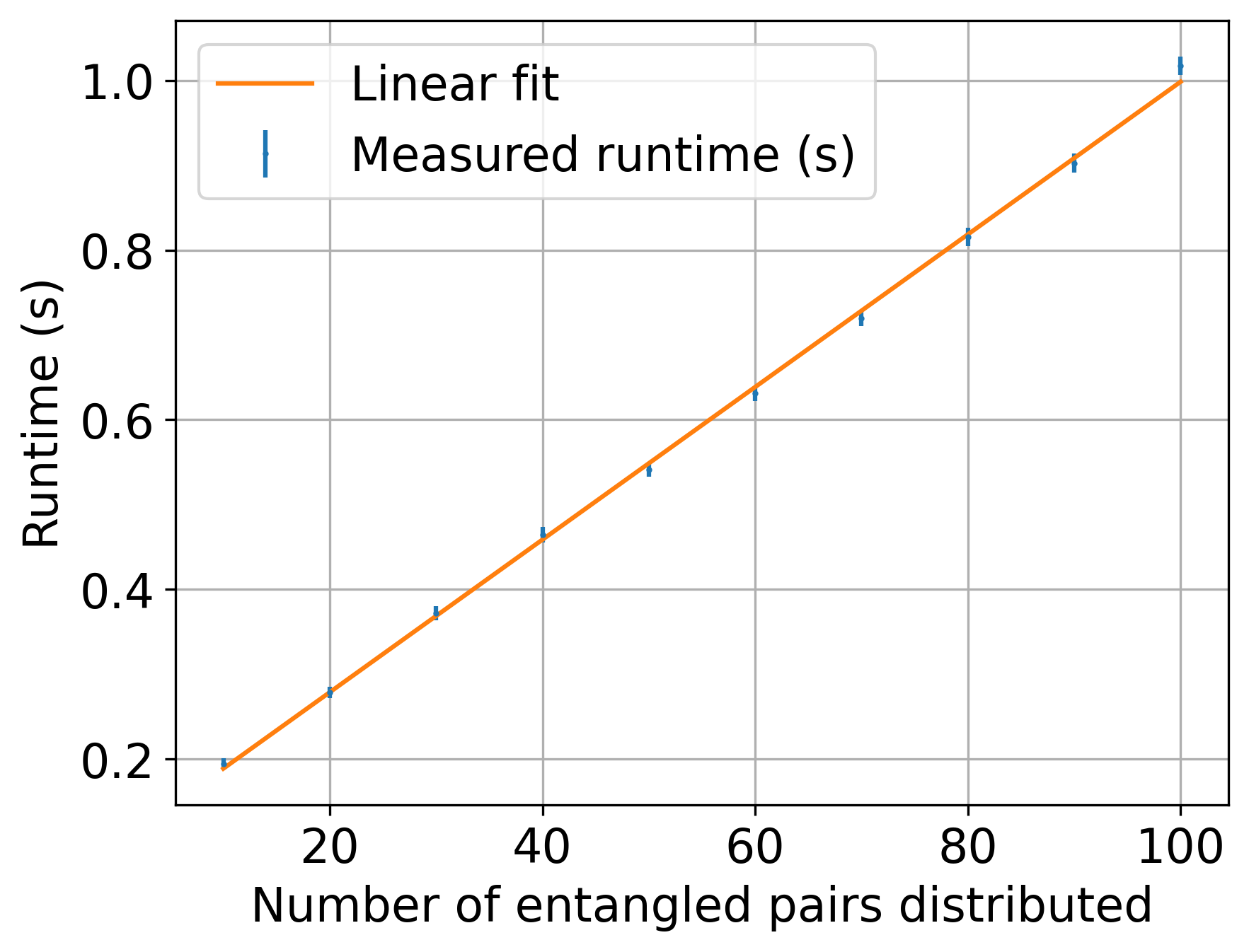}
\caption{Performance of our simulation of the Delft - Eindhoven setup with abstract model nodes and a cut-off timer using a machine running 40 Intel Xeon Gold cores at 2.1 GHz and 192 GB of RAM.
The runtime scales linearly with the number of entangled pairs being distributed.
Distributing 100 times, which we have empirically determined is enough to evaluate the performance of a given parameter set with reasonable accuracy, takes roughly one second.
The data point corresponding to $n$ pairs was obtained by running the corresponding simulation 500 times.
The error bars represent the standard error of the mean.}
\label{fig:runtime_scaling_nruns}
\end{figure}
As expected, the scaling is linear.
A simulation with $n=100$, which we have empirically determined is enough to evaluate the performance of a given parameter set with reasonable accuracy, takes roughly 1 s.
To be more concrete, when running a color-center double-click simulation using the minimal hardware parameters presented in Section~\ref{sec:discussion}, we find that after distributing 100 pairs a teleportation fidelity $F_\text{tel}$ of 0.8774 $\pm$ 0.0035 and a rate of 0.106 Hz $\pm$ 0.003 are obtained.
\\

We note that Figure~\ref{fig:runtime_scaling_nruns} was obtained by running the simulation without a cut-off.
Although the runtime still grows linearly with the number of distributed entagled pairs with a cut-off, its inclusion does mean that the simulation runtime grows exponentially as the the cut-off time becomes shorter.
This is because the expected number of necessary entanglement generation attempts also grows exponentially, as seen in Figure~\ref{fig:runtime_scaling_cutoff}.
\\

\begin{figure}[!ht]
\centering
\includegraphics[width=0.75\columnwidth]{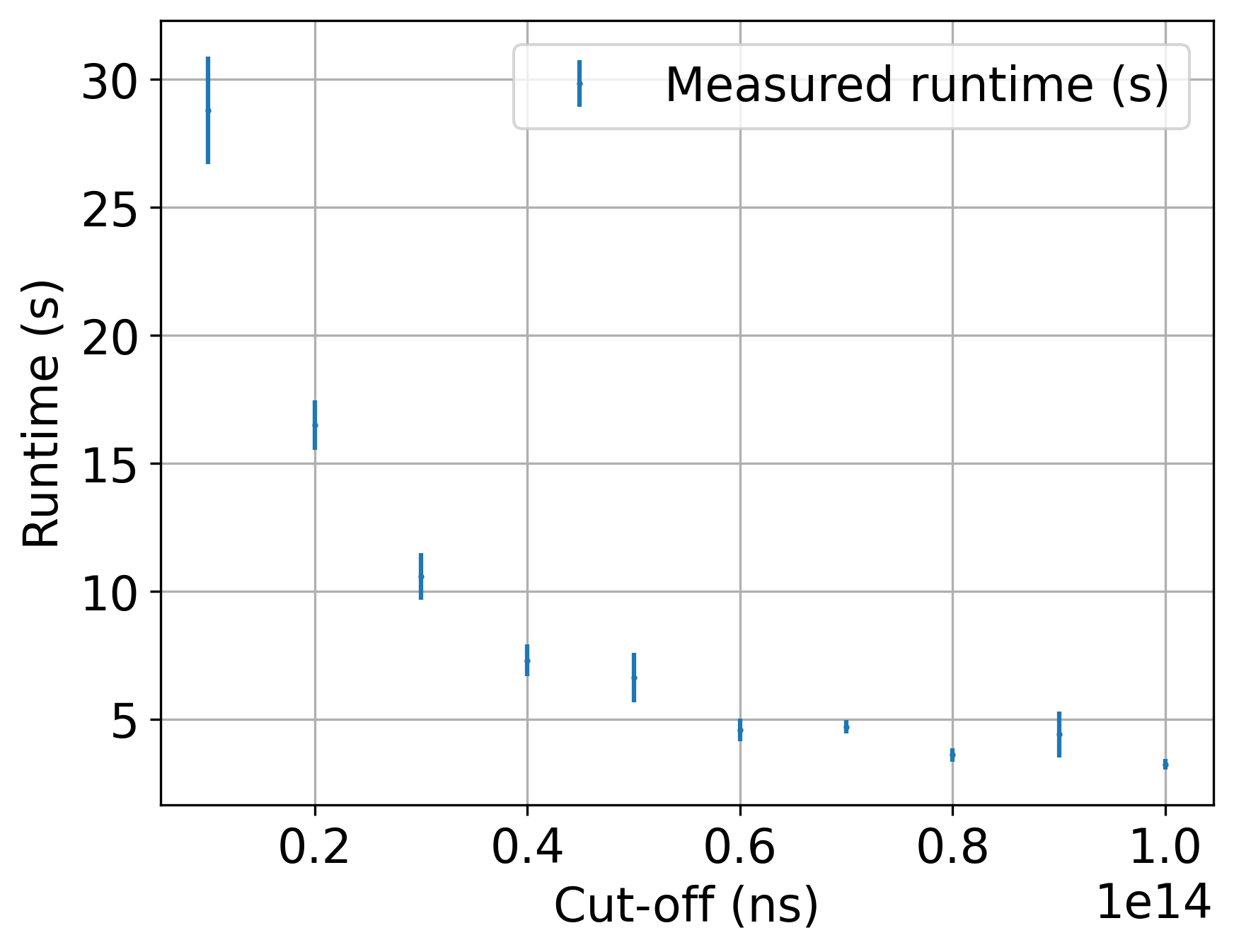}
\caption{Performance of our simulation of the Delft - Eindhoven setup with abstract model nodes and a cut-off timer using a laptop running a quad-core Intel i7-8665U processor at 1.9 GHz and 8 GB of RAM.
The runtime scales exponentially as the cut-off time is reduced.
The data point corresponding to $n$ pairs was obtained by running the corresponding simulation 20 times.
The error bars represent the standard error of the mean.}
\label{fig:runtime_scaling_cutoff}
\end{figure}

As discussed in Section~\ref{sec:methods} of the main text and in Appendix~\ref{sec:appendix_optimization}, the optimization methodology we employ requires running our simulation for many different sets of parameters.
We now estimate a lower bound on the time required to perform optimization in one setup.
We run the optimization algorithm for 200 iterations.
In each of these, there are 200 different parameter sets, and the distribution of 100 entanled pairs is simulated for each.
The computing nodes in the high-performance-computing cluster we use have 128 cores, which means that the simulation for 128 of the 200 parameter sets can be executed in parallel.
Assuming that there is no cut-off, or that it is large enough not to significantly impact the simulation runtime, this means that we can expect 1 generation to be run in roughly 2.5 seconds.
The data processing and file input and output required to generate new sets of parameters take a comparable amount of time, making $T = 200 \times 5$ s, roughly seventeen minutes, a good estimate for the time required to perform optimization for one setup.
We must however stress that this is a very optimistic lower bound, because as Figure~\ref{fig:runtime_scaling_cutoff} makes clear, the use of a cut-off has a huge impact on the runtime of the simulation.
We have observed that optmization of most of the setups we studied required 10 to 20 hours to terminate. 

%% file: sections/protocols.tex
\section{Framework for simulating quantum repeaters}
\label{appendix:protocols}

In this appendix, we discuss the framework that we use to evaluate the performance of quantum repeaters.
This framework is presented in this work for the first time.
\\

The code that we have used to simulate all the quantum networks in this paper is publicly available~\cite{delft_eindhoven_code}.
The repository contains code that has a much broader applicability than simulating the networks of up to three nodes presented here.
In fact, the simulations can be used to assess the performance of quantum-repeater chains with any number of nodes, and any spacing between nodes.
The currently supported types of nodes are those containing NV centers, ion traps or abstract quantum processors,
and the currently supported types of entanglement generation between neighboring nodes are the single-click and double-click protocols.
The simulation code depends on a number of other public repositories \cite{netsquid-netconf, netsquid-magic, netsquid-nv, netsquid-physlayer, netsquid-trappedions, netsquid-simulationtools}, all of which were developed in tandem with the code for this paper and will be explained in more detail below.
\\

\subsection{Services}

The primary functional unit of our quantum-network simulations is the ``service'', which is defined by an input, an output, and its intended function.
An example of a service that can be defined on a node is the measurement service.
It takes as input a request to measure a qubit, and the intended function is that the qubit is measured.
As output, the service returns the measurement outcome.
\\

A service is distinct from its implementation, which is a protocol.
Protocols make sure the intended function is fulfilled and generate the appropriate output.
Different protocols can fulfill the same function.
For example, in the case of the measurement service, a protocol that simulates a direct measurement of the required qubit (e.g., a fluorescence measurement) could be activated.
Another possible implementation would be a protocol that first swaps the quantum state of the required qubit to some different physical qubit (that perhaps allows for higher-fidelity measurements), and then simulates a measurement of that qubit.
The distinction between service and its implementation is illustrated in Figure~\ref{fig:services_figure}, which emphasizes that the same high-level functionality can be implemented using different physical systems.
\begin{figure}[!htpb]
    \centering
    \includegraphics[width=0.5\columnwidth]{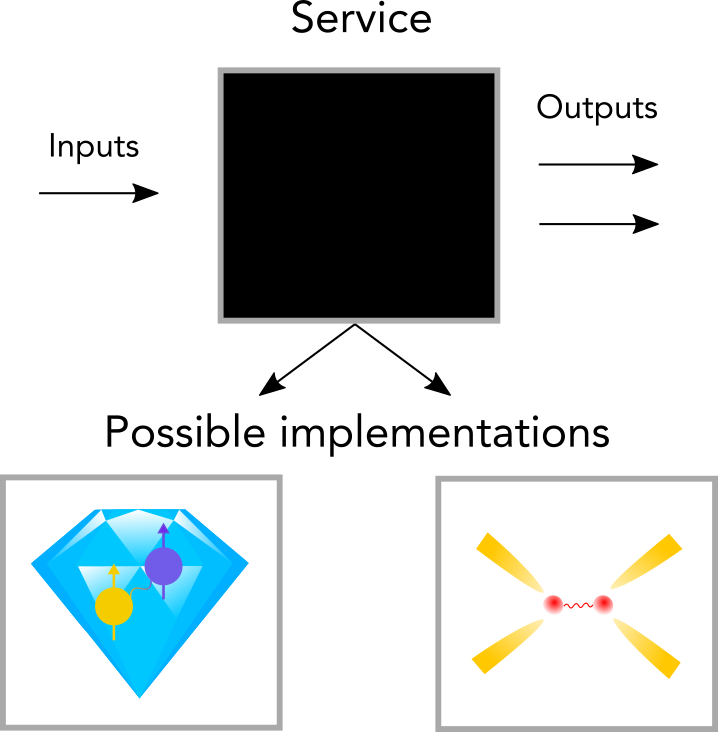}
    \caption{The black box represents a service, defined by a set of inputs, a set of outputs and some promised functionality.
    The protocols interacting with the service need not know how this functionality is implemented.
    Therefore, different implementations can be swapped in and out.
    In the figure, color centers and trapped ions are depicted to emphasize that the same high-level functionality can be executed by different physical systems.}
    \label{fig:services_figure}
\end{figure}
\\

Treating services and their implementation separately has two distinct advantages for our simulations.
First, it allows us to easily run the same protocols on different types of simulated hardware.
Take as example performing an entanglement swap in the broader context of a repeater chain.
To do so, the repeater protocol will place a request with the local entanglement-swap service.
The repeater protocol does not need to know how the swap is implemented.
On an abstract quantum processor, it can be implemented using a CNOT gate, while on an ion trap, it can be implemented using a Mølmer–Sørensen gate.
Second, it allows for a modular stack of protocols, where protocols implementing a specific service can easily be interchanged.
In the example of the repeater protocol, requests are made of an entanglement-generation service before the swap can be performed.
If the protocol runs on an NV node, entanglement could either be generated using a single-click or double-click protocol.
Switching between these two modes is easily realized by changing the protocol that implements the entanglement-generation service.
Again, the repeater-node protocol does not need to be adapted.
\\

The main interface of the repeater chain itself is also defined by a service.
The service implemented by the repeater chain is a link-layer service \cite{dahlberg2019link, pompiliExperimentalDemonstrationEntanglement2021},
which provides robust entanglement generation between the end nodes of the chain.
These requests should be put on the end nodes of the chain, which activates a protocol that uses a messaging service to put requests on the SWAP-ASAP repeater services defined on the repeater nodes of the network.
When the end-node protocols confirm they share entanglement (using a protocol that tracks entanglement in the network based on the classical communication shared by nodes),
an appropriate output message is returned by the service.
This is the cue that we use in our simulations to collect the density matrix of the created state and the time it took to create it.
\\

\subsection{SWAP-ASAP protocol}
A SWAP-ASAP repeater chain is one in which repeater nodes perform an entanglement swap as soon as they hold two entangled qubits that were generated with different neighbors.
This is in contrast to e.g. nested repeater schemes \cite{briegelQuantumRepeatersRole1998, duanLongdistanceQuantumCommunication2001}.
We have implemented two different SWAP-ASAP repeater protocols.
The first is suitable for repeater chains of any length and node spacing,
and for repeater nodes that can generate entanglement with either one or both neighbors at the same time.
The second, on the other hand, has been tailored more specifically to the one-repeater scenario studied in this paper.
It assumes that entanglement generation is limited to a single neighbor at a time.
First a request is issued to generate entanglement over a single connection.
Once that has finished, a request is issued for the second link,
and a swap is executed as soon as entanglement is confirmed.
In case the connections are not of equal length, entanglement generation takes place on the longer link first.
The reason for this is that the longer connection is expected to be the connection on which entanglement generation takes longer.
By finishing the longer link first, the total time that entanglement needs to be stored in quantum memory is minimized.
This protocol is the one used to generate the results reported in this paper.
\\

To generate entanglement over elementary links, the repeater protocols issue requests with the entanglement service.
In the protocol that we use to implement this service, these requests are queued.
The number of requests that are processed simultaneously is hardware-dependent, and is a free parameter in our simulations.
When handling a request for entanglement the protocol will, before doing anything else, issue a request to an agreement service.
This service is in charge of synchronizing neighboring nodes that want to generate entanglement together.
This is needed as typically both nodes need to be actively involved in generating entanglement for a state to be created between the two.
In our simulations, we use an implementation of the agreement service where even-numbered nodes in the chain always initiate entanglement generation.
These nodes will send a classical message to their neighbors when a request for agreement is made,
and then wait for those nodes to send a classical reply indicating readiness, after which entanglement generation can start.
On the other hand, when a request is made on an odd node, it will check whether a classical message has been received by the neighboring even node in the past.
If so, it will reply indicating readiness.
Otherwise, the request for agreement will be rejected.
In that case, the entanglement service can try to process the next request in the queue,
and see if agreement can be reached with this node again at some later time.
\\

In case agreement is reached between two nodes, the entanglement protocols of the nodes will start entanglement generation.
In our simulations, this is done using analytical models that decide after how much time an entangled state should be created between the nodes,
and what this state should look like.
This process is known as \textit{magic}~\cite{netsquid-magic} and is further discussed in Appendix \ref{sec:setup_ent_gen}.
\\

Finally, there is a cut-off protocol active on repeater nodes.
It discards qubits that have been stored in quantum memories for too long.
The exact amount of time after which states are discarded is called the cut-off time,
and is a tunable parameter that allows for a trade-off between end-to-end entangling rate and fidelity.
Every node runs an entanglement-tracking protocol that keeps track of both any local entangled qubits,
and what entangled states currently exist in the network at large.
Whenever the entanglement service registers a new qubit at the entanglement tracker,
the cut-off protocol starts a timer.
When the timer goes off, the cut-off protocol checks whether the entangled qubit still exists in local memory.
If so, the entanglement tracker is told to discard the qubit.
The entanglement tracker will also communicate classically with the entanglement trackers of other nodes in the network to inform them that the qubit has been discarded.
If an entanglement tracker learns that a qubit has been discarded that was entangled with one of its local qubits,
it responds by discarding that qubit as well.
Links corresponding to discarded qubits must be regenerated.
We note that the cut-off protocol does not run on the end nodes of the repeater chain.
This is to prevent the possibility of one end node believing end-to-end entanglement has been achieved,
while the other end node has in actuality discarded its qubit (but the classical message has not yet reached the first end node).

\subsection{Configuring quantum networks}

In our simulations, quantum networks are made up of nodes.
Each node represents a single physical location, and contains an object that we refer to as ``driver''.
This object provides a mapping between services and protocols that implement those services.
The driver allows access to services without knowledge of their implementations.
Each node has its own driver.
Apart from drivers, nodes hold components that represent quantum hardware,
which allow for the storage and/or manipulation of quantum states.
The protocols running on the node can use this quantum hardware to implement specific services.
The nodes in our simulations are ready-made packages with both driver and hardware included.
In order to use them in a quantum network, they just need to be initialized (thereby specifying their parameters) and connected to other nodes.
\\

The simulations performed for this paper contain three different types of nodes.
The first is the NV node.
It holds an NV quantum processor, which is imported from the Python package NetSquid-NV \cite{netsquid-nv}.
The second is the ion-trap node.
This node holds an ion-trap quantum processor, imported from the Python package NetSquid-TrappedIons \cite{netsquid-trappedions}.
Finally, there is the abstract node, which contains an abstract quantum processor imported from the Python package NetSquid-AbstractModel \cite{netsquid-abstractmodel}.
On initialization, each of these takes hardware parameters specific to the type of hardware being simulated,
and a number of parameters used to configure the protocols used at the node.
For example, the cut-off time needs to be specified,
and in case of single-click heralded entanglement generation, the bright-state parameter as well.
\\

Nodes are connected by two types of connections.
These connections are themselves also ready-made packages, and can be found in the Python package NetSquid-PhysLayer \cite{netsquid-physlayer}.
The first type is the classical connection, which represents optical fiber that can be used to send classical messages.
The second type is the heralded connection.
It represents a midpoint station connected to two nodes by optical fiber,
where optical Bell-state measurements can be performed on incoming photons.
Such a connection can be used to perform heralded entanglement generation.
As discussed above, we do not simulate the process of heralded entanglement generation itself,
but instead use analytical models to magically create entangled states.
However, the heralded connections still perform an important role as placeholders.
Parameters passed to the heralded connection when configuring the network
are later retrieved by the analytical models to decide how long it should take before a state is created,
and what that state should be exactly.
One key parameter specified in the heralded connection
is whether single-click or double-click heralded entanglement distribution is used.
In the simulations presented in this paper
neighboring nodes are always connected by both a classical connection and a heralded connection.
\\

To put together nodes and connections for the creation of quantum networks,
and to configure their parameters,
we make use of the Python package NetSquid-NetConf \cite{netsquid-netconf}.
The tools provided in this package allow for the writing of human-readable configuration files.
These configuration files contain entries for all the different nodes and connections in the network.
Their type is specified (such as ``NV node'' or ``heralded connection''),
as well as their parameters and how they are connected.
These configuration files can also be used to vary some of the parameters,
allowing to e.g. perform a parameter scan over one of them and observe its effect on the network performance.
\\

%% file: sections/optimization_results_appendix.tex
\section{Extra optimization results}
\label{sec:appendix_optimization_results}
In this appendix, we present results of optimizations we performed that were not presented in the main text, but might still be of interest.

\input{sections/move_no_move_appendix}
\input{sections/alternative_architecture_appendix}
\input{sections/repeaterless_appendix}
\input{sections/single_double_click_appendix}
\input{sections/setup_costs_appendix}

%% file: sections/move_no_move_appendix.tex
\subsection{To move or not to move}
\label{appendix:sec_move_no_move}
As mentioned in the main text, the communication qubit of color centers has typically shorter coherence times than the memory qubits.
For the baseline hardware parameters we investigated, the communication qubit had $T_1 = 1$ hours~\cite{abobeih2018one} and $T_2 = 0.5$ s~\cite{hermans2022qubit}, whereas the memory qubits had $T_1 = 10$ hours and $T_2 = 1$ s~\cite{bradley2019ten}.
It might then be worthwhile for the end node that generates entanglement with the repeater first, i.e., the Eindhoven node, to move its half of the entangled state to memory while waiting for end-to-end entanglement to be established, even though that comes at the cost of more noise being introduced in this operation.
A diagram of the circuit used for this operation can be found in Supplementary Note 5 B of~\cite{coopmans2021netsquid}.
To investigate this, we applied our methodology to two single-repeater color-center setups performing double-click entanglement generation.
In one of the setups, which we name "move scenario", once the first elementary link is established, the end node performs the move operation while the waiting for the second link to be established.
In the other setup, which we name "no-move scenario", the state is kept in the electron spin until end-to-end entanglement is established.
The hardware requirements for these two scenarios are shown in Figure~\ref{fig:circular_plot_move_no_move}.
\begin{figure}[!htpb]
\centering
\includegraphics[width=0.5\columnwidth]{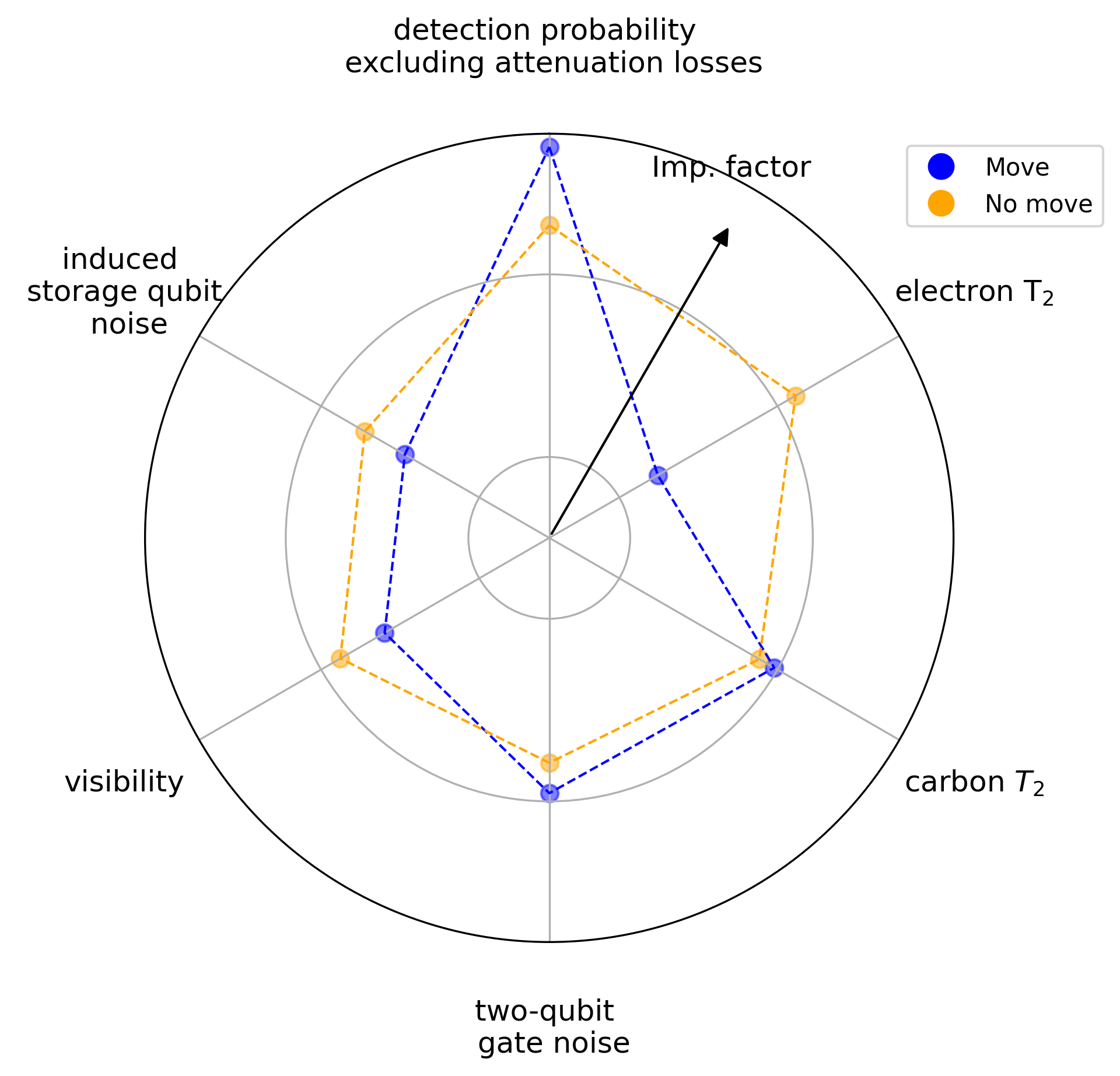}
\caption{Directions along which color-center hardware must be improved to achieve entanglement generation rate $R = 0.1$ Hz and teleportation fidelity $F_\text{T} = 0.8717$, enabling VBQC between Delft and Eindhoven, assuming that a double-click entanglement generation protocol is used.
The blue (orange) line corresponds to the direction of hardware improvement in case the Eindhoven end node (does not) move their half of the entangled state to the memory qubit.
Note the use of a logarithmic scale.
}
\label{fig:circular_plot_move_no_move}
\end{figure}
The move scenario requires that the two-qubit gate be significantly improved, which is to be expected as the move operation requires the application of two of these gates~\cite{coopmans2021netsquid}.
On the other hand, the move scenario does not require an improvement on the electron spin's coherence time, in contrast with the no-move scenario.
This is also not surprising, as in the move scenario entanglement is not stored in the electron spin for a significant amount of time.
\\

The overall cost associated to the no-move scenario is slightly lower than the cost of the move scenario, so all the NV center results presented in the main text were obtained in the no-move scenario.
We stress that this finding, although relevant for our particular case study, is not general.
It might be that different baselines, different goals or different setups would lead to laxer hardware requirements for the move scenario.

%% file: sections/alternative_architecture_appendix.tex
\subsection{Architecture comparison}\label{appendix:sec_alt_architecture}
As discussed in detail in Appendix~\ref{sec:appendix_setup}, the fiber network we study contains four nodes in the shortest path connecting the Dutch cities of Delft and Eindhoven.
This means that there is some freedom in how to place the two heralding stations and repeater node required for a single-repeater setup, as shown in Figure~\ref{fig:alternative_d_e_paths}.
In order to decide how to make this placement, we determined the minimal hardware requirements for achieving an entanglement generation rate $R = 0.1$ Hz and a teleportation fidelity $F_\text{T} = 0.8717$, enabling VBQC between Delft and Eindhoven, for both possibilites.
These requirements are shown in Figure~\ref{fig:circular_plot_architecture_comparison}.
\begin{figure}[!h]
\centering
\includegraphics[width=0.5\columnwidth]{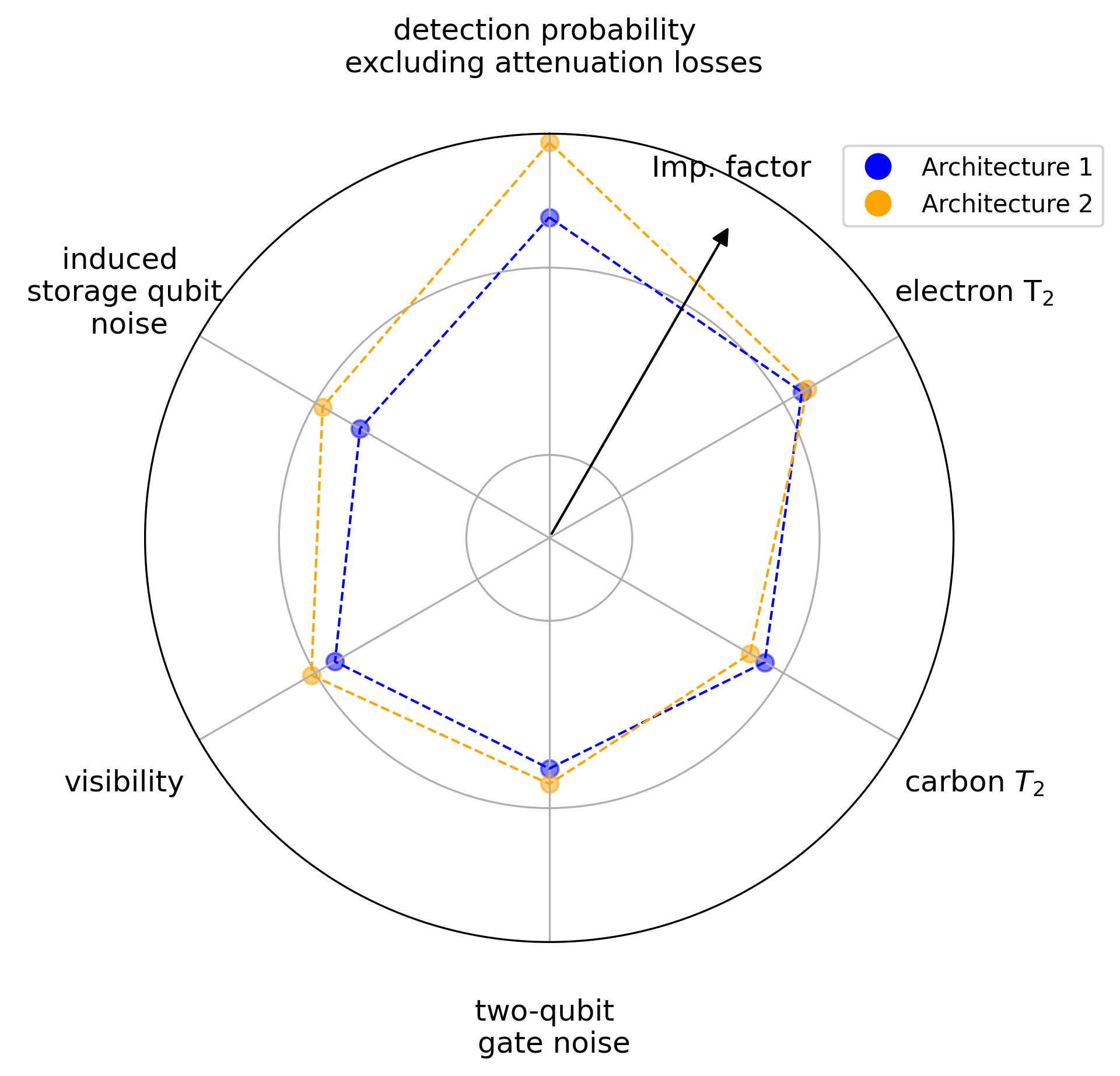}
\caption{Directions along which color-center hardware must be improved to achieve entanglement generation rate $R = 0.1$ Hz and teleportation fidelity $F_\text{T} = 0.8717$, enabling VBQC between Delft and Eindhoven, assuming that a double-click entanglement generation protocol is used.
The blue (orange) line corresponds to the direction of hardware improvement for the architecture shown on the left (right) in Figure~\ref{fig:alternative_d_e_paths}.
Note the use of a logarithmic scale.
}
\label{fig:circular_plot_architecture_comparison}
\end{figure}
The requirements are qualitatively similar for both architectures, with the photon detection probability excluding attenuation losses and induced noise on memory qubits (see Appendix~\ref{sec:appendix_setup} for details on our modeling of color-center based repeaters) being the parameteres requiring the most improvement.
The architecture on the left in Figure~\ref{fig:alternative_d_e_paths} required more modest improvements overall, so this was the architecture considered in our work.

%% file: sections/repeaterless_appendix.tex
\subsection{Connecting Delft and Eindhoven without a repeater}\label{appendix:sec_repeaterless}
The main contribution of this work was the investigation of the hardware requirements for enabling 2-qubit VBQC between two cities separated by 226.5 km of optical fiber using a single repeater node.
We investigated two sets of performance targets compatible with this goal, namely (i) $R = 0.1$ Hz, $F_\text{T} = 0.8717$ and (ii) $R = 0.5$ Hz, $F_\text{T} = 0.8571$.
While (ii) is impossible to achieve via direct transmission, i.e., without a repeater, due to fiber loss, this is not the case for (i) if a single-click entanglement generation protocol is employed.
In Figure~\ref{fig:circular_plot_double_click_vs_repeaterless} we show directions along which color-center hardware would have to be improved to meet (i) without using a repeater.
For comparison, we also reproduce the improvement directions for color-center hardware to meet the same targets with a repeater employing double-click entanglement generation, because this was the repeater setup requiring smallest improvements as measured by our cost function.
\\

\begin{figure}[!h]
\centering
\includegraphics[width=0.5\columnwidth]{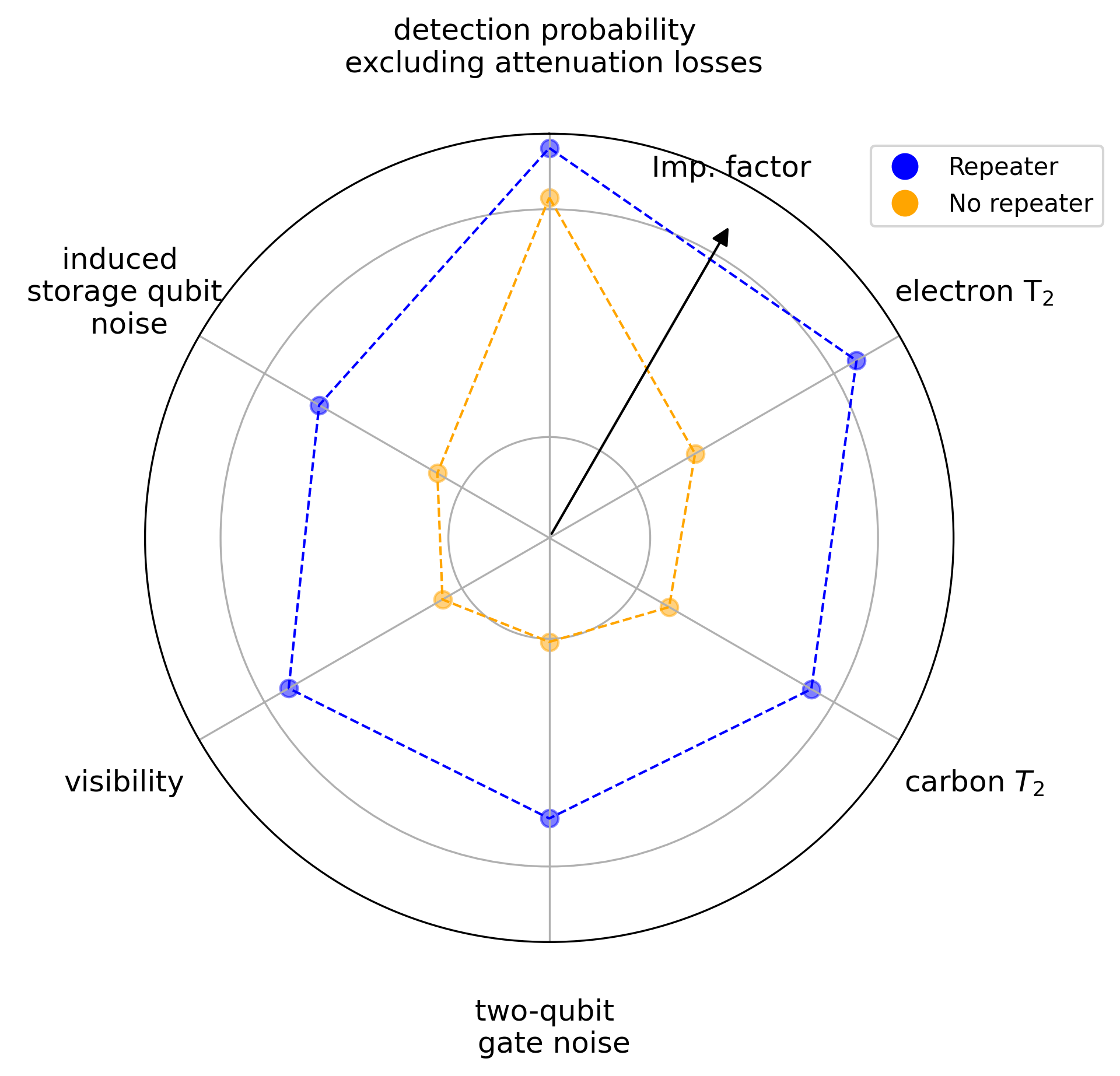}
\caption{Directions along which hardware must be improved to achieve entanglement generation rate $R = 0.1$ Hz and teleportation fidelity $F_\text{T} = 0.8717$, enabling VBQC between Delft and Eindhoven.
The blue (orange) line corresponds to the direction of hardware improvement for the case in which a repeater is (is not) used.
The repeater scenario employs a double-click entanglement generation protocol, whereas in the direct transmission case single-click entanglement generation is employed.
Note the use of a logarithmic scale.
}
\label{fig:circular_plot_double_click_vs_repeaterless}
\end{figure}

The direct transmission setup requires less improvement in all parameters.
In fact, the only parameter that requires significant improvement is photon detection probability excluding attenuation losses, although still less than what is required for the repeater setup.
The reason for this is that the elementary link state generated with the single-click protocol and state-of-the-art parameters already has high enough fidelity, so the only constraint is that these states are generated fast enough.
The required value for the photon detection probability excluding attenuation losses is 0.39, less than the 0.73 required for the repeater with double-click entanglement generation case, but still above the limit imposed by the zero-phonon line.
\\

These results indicate that performing VBQC over this particular setup might best be done without a repeater, but nevertheless do not detract from the main goal of the paper, which was to investigate hardware requirements if a repeater were to be used.

%% file: sections/single_double_click_appendix.tex
\subsection{Hardware requirements for repeaters with single and double-click entanglement generation}
\label{sec:appendix_single_double}
We investigated also how the hardware requirements for color centers running single and double-click entanglement generation protocols differ.
We considered a rate target of $R = 0.1$ Hz and an average teleportation fidelity target of $F_\text{T} = 0.8717$.
These two sets of hardware requirements are presented in Figure~\ref{fig:single_vs_double_surf_siv}.
\\

\begin{figure}[!htpb]
\centering
\includegraphics[width=0.5\columnwidth]{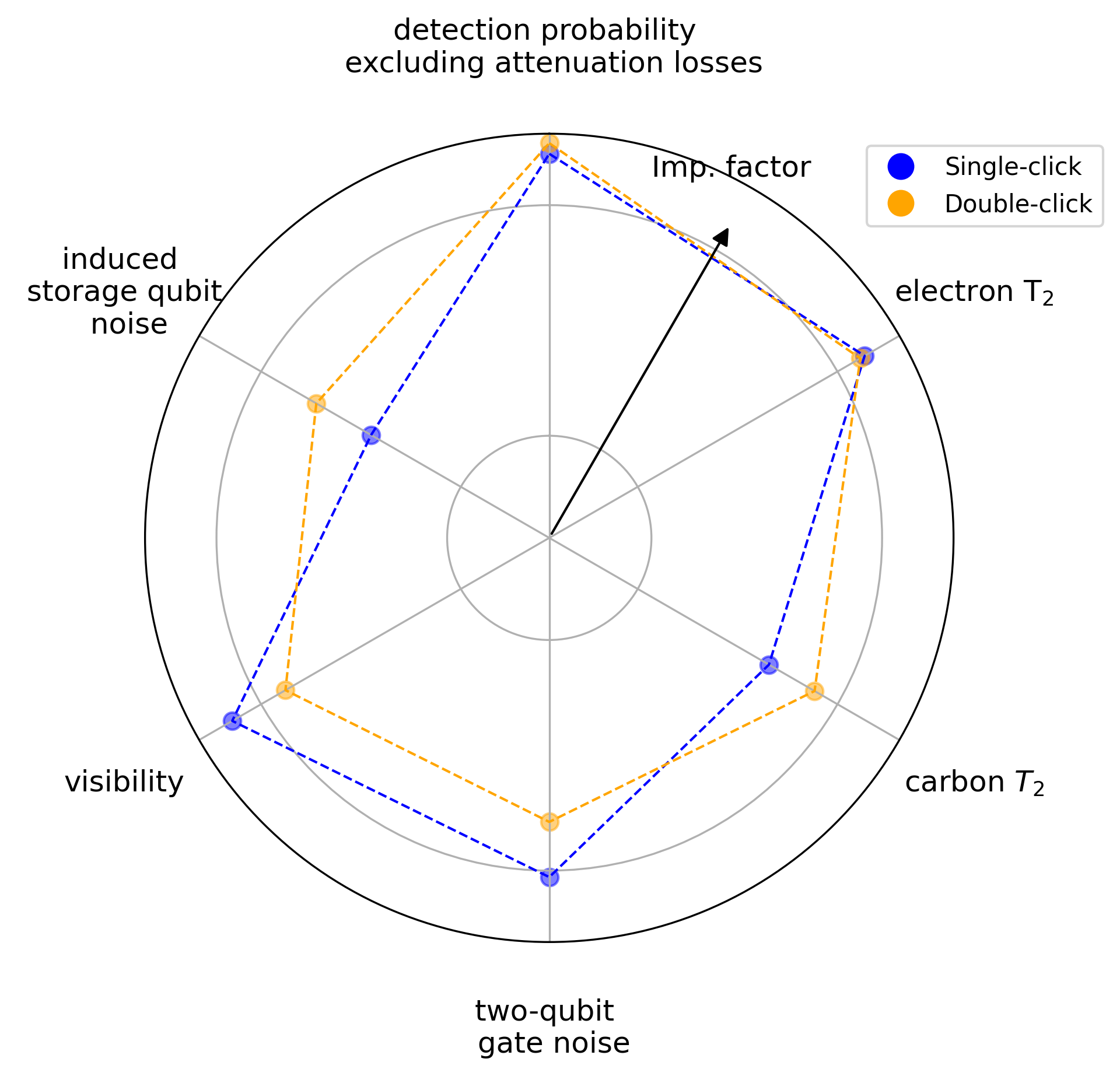}
\caption{Hardware requirements for executing 2-qubit VBQC using a color-center repeater performing double-click (orange) and single-click entanglement generation (blue).
These are the requirements for achieving an entanglement generation rate of $R = 0.1$ Hz and an average teleportation fidelity of $F_\text{T} = 0.8717$.
}
\label{fig:single_vs_double_surf_siv}
\end{figure}

The hardware requirements are more stringent for a color-center repeater performing single-click entanglement generation.
This is due to the fairly demanding fidelity target, which does not leave much room for noise in a protocol that inherently generates imperfect entangled states.
We must however stress that this conclusion is specific to this particular setup and these performance targets, and does not imply that double-click should in general be chosen over single-click.
In fact, one need only look at the second set of performance targets we considered in section~\ref{sec:results} to understand this point.
These targets are impossible to achieve using a color-center repeater performing double-click entanglement generation, but are feasible if single-click is employed.

%% file: sections/setup_costs_appendix.tex
\subsection{Hardware improvement costs}
In Table~\ref{tab:setups_costs} we present the cost of hardware improvement associated with the minimal hardware requirements found for every setup we investigated.
\label{sec:appendix_hardware_improvement_costs}
\\

\begin{table}[!ht]
\begin{tabular}{|c|c|c|c|c|}
\hline
Platform                      & Target                      & Setup                                      & Protocol                      & Cost  \\ \hline
\multirow{9}{*}{Color center} & \multirow{7}{*}{R = 0.1 Hz} & Standard (see Section~\ref{sec:discussion})                                   & Double-click, no-move         & 26.2  \\ \cline{3-5} 
                              &                             & \multirow{4}{*}{Fiber network (see Section~\ref{sec:results})}             & Single-click, no-move         & 82.6  \\ \cline{4-5} 
                              &                             &                                            & Single-click, move            & 165.5 \\ \cline{4-5} 
                              &                             &                                            & Double-click, no-move         & 59.8  \\ \cline{4-5} 
                              &                             &                                            & Double-click, move            & 100.1 \\ \cline{3-5} 
                              &                             & Fiber network (repeaterless, see Appendix~\ref{appendix:sec_repeaterless})               & Single-click                  & 20.5  \\ \cline{3-5} 
                              &                             & Alternative fiber network (see Appendix~\ref{appendix:sec_alt_architecture})                 & Double-click, no-move         & 116.1 \\ \cline{2-5} 
                              & \multirow{2}{*}{R = 0.5 Hz} & \multirow{2}{*}{Fiber network (see Section~\ref{sec:results})}             & Single-click, no-move         & 153.3 \\ \cline{4-5} 
                              &                             &                                            & Single-click, move            & 227.3 \\ \hline
Trapped ions                  & R = 0.1 Hz                  & Fiber network (see Section~\ref{sec:results})                             & Double-click                  & 171.1 \\ \hline
\multirow{4}{*}{Abstract}     & \multirow{4}{*}{R = 0.1 Hz} & Fiber network, color center baseline (see Section~\ref{sec:discussion})      & \multirow{4}{*}{Double-click} & 40.7  \\ \cline{3-3} \cline{5-5} 
                              &                             & Fiber network, trapped ion baseline (see Section~\ref{sec:discussion})       &                               & 50.1  \\ \cline{3-3} \cline{5-5} 
                              &                             & Fiber network, converted from color center (see Section~\ref{sec:discussion}) &                               & 37.2  \\ \cline{3-3} \cline{5-5} 
                              &                             & Fiber network, converted from trapped ion (see Section~\ref{sec:discussion})  &                               & 121.0 \\ \hline
\end{tabular}
\caption{Improvement cost, as defined in Appendix~\ref{sec:appendix_optimization}, of minimal hardware requirements for all setups we investigated.}
\label{tab:setups_costs}
\end{table}

%% file: main.bbl
\begin{thebibliography}{100}
\expandafter\ifx\csname url\endcsname\relax
  \def\url#1{\texttt{#1}}\fi
\expandafter\ifx\csname urlprefix\endcsname\relax\def\urlprefix{URL }\fi
\providecommand{\bibinfo}[2]{#2}
\providecommand{\eprint}[2][]{\url{#2}}

\bibitem{ekert1991quantum}
\bibinfo{author}{Ekert, A.~K.}
\newblock \bibinfo{title}{Quantum cryptography based on {B}ell’s theorem}.
\newblock \emph{\bibinfo{journal}{Phys. Rev. Lett.}}
  \textbf{\bibinfo{volume}{67}}, \bibinfo{pages}{661} (\bibinfo{year}{1991}).

\bibitem{bennett2020quantum}
\bibinfo{author}{Bennett, C.~H.} \& \bibinfo{author}{Brassard, G.}
\newblock \bibinfo{title}{Quantum cryptography: {{Public}} key distribution and
  coin tossing}.
\newblock \emph{\bibinfo{journal}{Theoret. Comput. Sci.}}
  \textbf{\bibinfo{volume}{560}}, \bibinfo{pages}{7--11}
  (\bibinfo{year}{2014}).

\bibitem{hillery1999quantum}
\bibinfo{author}{Hillery, M.}, \bibinfo{author}{Bu{\v{z}}ek, V.} \&
  \bibinfo{author}{Berthiaume, A.}
\newblock \bibinfo{title}{Quantum secret sharing}.
\newblock \emph{\bibinfo{journal}{Phys. Rev. A}} \textbf{\bibinfo{volume}{59}},
  \bibinfo{pages}{1829} (\bibinfo{year}{1999}).

\bibitem{komar2014quantum}
\bibinfo{author}{Komar, P.} \emph{et~al.}
\newblock \bibinfo{title}{A quantum network of clocks}.
\newblock \emph{\bibinfo{journal}{Nat. Phys.}} \textbf{\bibinfo{volume}{10}},
  \bibinfo{pages}{582--587} (\bibinfo{year}{2014}).

\bibitem{wehner2018quantum}
\bibinfo{author}{Wehner, S.}, \bibinfo{author}{Elkouss, D.} \&
  \bibinfo{author}{Hanson, R.}
\newblock \bibinfo{title}{Quantum internet: {{A}} vision for the road ahead}.
\newblock \emph{\bibinfo{journal}{Science}} \textbf{\bibinfo{volume}{362}},
  \bibinfo{pages}{eaam9288} (\bibinfo{year}{2018}).

\bibitem{briegel1998quantum}
\bibinfo{author}{Briegel, H.-J.}, \bibinfo{author}{D{\"u}r, W.},
  \bibinfo{author}{Cirac, J.~I.} \& \bibinfo{author}{Zoller, P.}
\newblock \bibinfo{title}{Quantum repeaters: the role of imperfect local
  operations in quantum communication}.
\newblock \emph{\bibinfo{journal}{Phys. Rev. Lett.}}
  \textbf{\bibinfo{volume}{81}}, \bibinfo{pages}{5932} (\bibinfo{year}{1998}).

\bibitem{dur1999quantum}
\bibinfo{author}{D{\"u}r, W.}, \bibinfo{author}{Briegel, H.-J.},
  \bibinfo{author}{Cirac, J.~I.} \& \bibinfo{author}{Zoller, P.}
\newblock \bibinfo{title}{Quantum repeaters based on entanglement
  purification}.
\newblock \emph{\bibinfo{journal}{Phys. Rev. A}} \textbf{\bibinfo{volume}{59}},
  \bibinfo{pages}{169} (\bibinfo{year}{1999}).

\bibitem{munro2015inside}
\bibinfo{author}{Munro, W.~J.}, \bibinfo{author}{Azuma, K.},
  \bibinfo{author}{Tamaki, K.} \& \bibinfo{author}{Nemoto, K.}
\newblock \bibinfo{title}{Inside quantum repeaters}.
\newblock \emph{\bibinfo{journal}{IEEE J. Sel. Topics Quantum Electron.}}
  \textbf{\bibinfo{volume}{21}}, \bibinfo{pages}{78--90}
  (\bibinfo{year}{2015}).

\bibitem{muralidharan2016optimal}
\bibinfo{author}{Muralidharan, S.} \emph{et~al.}
\newblock \bibinfo{title}{Optimal architectures for long distance quantum
  communication}.
\newblock \emph{\bibinfo{journal}{Sci. Rep.}} \textbf{\bibinfo{volume}{6}},
  \bibinfo{pages}{20463} (\bibinfo{year}{2016}).

\bibitem{sangouard2011quantum}
\bibinfo{author}{Sangouard, N.}, \bibinfo{author}{Simon, C.},
  \bibinfo{author}{De~Riedmatten, H.} \& \bibinfo{author}{Gisin, N.}
\newblock \bibinfo{title}{Quantum repeaters based on atomic ensembles and
  linear optics}.
\newblock \emph{\bibinfo{journal}{Rev. Mod. Phys.}}
  \textbf{\bibinfo{volume}{83}}, \bibinfo{pages}{33} (\bibinfo{year}{2011}).

\bibitem{duanLongdistanceQuantumCommunication2001}
\bibinfo{author}{Duan, L.-M.}, \bibinfo{author}{Lukin, M.},
  \bibinfo{author}{Cirac, I.} \& \bibinfo{author}{Zoller, P.}
\newblock \bibinfo{title}{Long-distance quantum communication with atomic
  ensembles and linear optics}.
\newblock \emph{\bibinfo{journal}{Nature}} \textbf{\bibinfo{volume}{414}},
  \bibinfo{pages}{413--418} (\bibinfo{year}{2001}).

\bibitem{rozpedek2019near}
\bibinfo{author}{Rozp{\k{e}}dek, F.} \emph{et~al.}
\newblock \bibinfo{title}{Near-term quantum-repeater experiments with
  nitrogen-vacancy centers: Overcoming the limitations of direct transmission}.
\newblock \emph{\bibinfo{journal}{Phys. Rev. A}} \textbf{\bibinfo{volume}{99}},
  \bibinfo{pages}{052330} (\bibinfo{year}{2019}).

\bibitem{duan2010colloquium}
\bibinfo{author}{Duan, L.-M.} \& \bibinfo{author}{Monroe, C.}
\newblock \bibinfo{title}{Colloquium: Quantum networks with trapped ions}.
\newblock \emph{\bibinfo{journal}{Rev. Mod. Phys.}}
  \textbf{\bibinfo{volume}{82}}, \bibinfo{pages}{1209} (\bibinfo{year}{2010}).

\bibitem{uphoff2016integrated}
\bibinfo{author}{Uphoff, M.}, \bibinfo{author}{Brekenfeld, M.},
  \bibinfo{author}{Rempe, G.} \& \bibinfo{author}{Ritter, S.}
\newblock \bibinfo{title}{An integrated quantum repeater at telecom wavelength
  with single atoms in optical fiber cavities}.
\newblock \emph{\bibinfo{journal}{Appl. Phys. B}}
  \textbf{\bibinfo{volume}{122}}, \bibinfo{pages}{1--15}
  (\bibinfo{year}{2016}).

\bibitem{monroe2013}
\bibinfo{author}{Monroe, C.} \emph{et~al.}
\newblock \bibinfo{title}{Quantum {{Networks}} with {{Atoms}} and {{Photons}}}.
\newblock \emph{\bibinfo{journal}{J. Phys. Conf. Ser.}}
  \textbf{\bibinfo{volume}{467}}, \bibinfo{pages}{012008}
  (\bibinfo{year}{2013}).

\bibitem{reiserer2015a}
\bibinfo{author}{Reiserer, A.} \& \bibinfo{author}{Rempe, G.}
\newblock \bibinfo{title}{Cavity-based quantum networks with single atoms and
  optical photons}.
\newblock \emph{\bibinfo{journal}{Rev. Mod. Phys.}}
  \textbf{\bibinfo{volume}{87}}, \bibinfo{pages}{1379--1418}
  (\bibinfo{year}{2015}).

\bibitem{langenfeld2021a}
\bibinfo{author}{Langenfeld, S.}, \bibinfo{author}{Thomas, P.},
  \bibinfo{author}{Morin, O.} \& \bibinfo{author}{Rempe, G.}
\newblock \bibinfo{title}{Quantum {{Repeater Node Demonstrating Unconditionally
  Secure Key Distribution}}}.
\newblock \emph{\bibinfo{journal}{Phys. Rev. Lett.}}
  \textbf{\bibinfo{volume}{126}}, \bibinfo{pages}{230506}
  (\bibinfo{year}{2021}).

\bibitem{ruf2021quantum}
\bibinfo{author}{Ruf, M.}, \bibinfo{author}{Wan, N.~H.}, \bibinfo{author}{Choi,
  H.}, \bibinfo{author}{Englund, D.} \& \bibinfo{author}{Hanson, R.}
\newblock \bibinfo{title}{Quantum networks based on color centers in diamond}.
\newblock \emph{\bibinfo{journal}{J. Appl. Phys.}}
  \textbf{\bibinfo{volume}{130}}, \bibinfo{pages}{070901}
  (\bibinfo{year}{2021}).

\bibitem{bhaskar2020experimental}
\bibinfo{author}{Bhaskar, M.~K.} \emph{et~al.}
\newblock \bibinfo{title}{Experimental demonstration of memory-enhanced quantum
  communication}.
\newblock \emph{\bibinfo{journal}{Nature}} \textbf{\bibinfo{volume}{580}},
  \bibinfo{pages}{60--64} (\bibinfo{year}{2020}).

\bibitem{pompili2021realization}
\bibinfo{author}{Pompili, M.} \emph{et~al.}
\newblock \bibinfo{title}{Realization of a multinode quantum network of remote
  solid-state qubits}.
\newblock \emph{\bibinfo{journal}{Science}} \textbf{\bibinfo{volume}{372}},
  \bibinfo{pages}{259--264} (\bibinfo{year}{2021}).

\bibitem{amirloo2010quantum}
\bibinfo{author}{Amirloo, J.}, \bibinfo{author}{Razavi, M.} \&
  \bibinfo{author}{Majedi, A.~H.}
\newblock \bibinfo{title}{Quantum key distribution over probabilistic quantum
  repeaters}.
\newblock \emph{\bibinfo{journal}{Phys. Rev. A}} \textbf{\bibinfo{volume}{82}},
  \bibinfo{pages}{032304} (\bibinfo{year}{2010}).

\bibitem{asadi2018quantum}
\bibinfo{author}{Asadi, F.~K.} \emph{et~al.}
\newblock \bibinfo{title}{Quantum repeaters with individual rare-earth ions at
  telecommunication wavelengths}.
\newblock \emph{\bibinfo{journal}{Quantum}} \textbf{\bibinfo{volume}{2}},
  \bibinfo{pages}{93} (\bibinfo{year}{2018}).

\bibitem{bernardes2011rate}
\bibinfo{author}{Bernardes, N.~K.}, \bibinfo{author}{Praxmeyer, L.} \&
  \bibinfo{author}{van Loock, P.}
\newblock \bibinfo{title}{Rate analysis for a hybrid quantum repeater}.
\newblock \emph{\bibinfo{journal}{Phys. Rev. A}} \textbf{\bibinfo{volume}{83}},
  \bibinfo{pages}{012323} (\bibinfo{year}{2011}).

\bibitem{borregaard2015heralded}
\bibinfo{author}{Borregaard, J.}, \bibinfo{author}{Komar, P.},
  \bibinfo{author}{Kessler, E.}, \bibinfo{author}{S{\o}rensen, A.~S.} \&
  \bibinfo{author}{Lukin, M.~D.}
\newblock \bibinfo{title}{Heralded quantum gates with integrated error
  detection in optical cavities}.
\newblock \emph{\bibinfo{journal}{Phys. Rev. Lett.}}
  \textbf{\bibinfo{volume}{114}}, \bibinfo{pages}{110502}
  (\bibinfo{year}{2015}).

\bibitem{bruschi2014repeat}
\bibinfo{author}{Bruschi, D.~E.}, \bibinfo{author}{Barlow, T.~M.},
  \bibinfo{author}{Razavi, M.} \& \bibinfo{author}{Beige, A.}
\newblock \bibinfo{title}{Repeat-until-success quantum repeaters}.
\newblock \emph{\bibinfo{journal}{Phys. Rev. A}} \textbf{\bibinfo{volume}{90}},
  \bibinfo{pages}{032306} (\bibinfo{year}{2014}).

\bibitem{chen2007fault}
\bibinfo{author}{Chen, Z.-B.}, \bibinfo{author}{Zhao, B.},
  \bibinfo{author}{Chen, Y.-A.}, \bibinfo{author}{Schmiedmayer, J.} \&
  \bibinfo{author}{Pan, J.-W.}
\newblock \bibinfo{title}{Fault-tolerant quantum repeater with atomic ensembles
  and linear optics}.
\newblock \emph{\bibinfo{journal}{Phys. Rev. A}} \textbf{\bibinfo{volume}{76}},
  \bibinfo{pages}{022329} (\bibinfo{year}{2007}).

\bibitem{collins2007multiplexed}
\bibinfo{author}{Collins, O.}, \bibinfo{author}{Jenkins, S.},
  \bibinfo{author}{Kuzmich, A.} \& \bibinfo{author}{Kennedy, T.}
\newblock \bibinfo{title}{Multiplexed memory-insensitive quantum repeaters}.
\newblock \emph{\bibinfo{journal}{Phys. Rev. Lett.}}
  \textbf{\bibinfo{volume}{98}}, \bibinfo{pages}{060502}
  (\bibinfo{year}{2007}).

\bibitem{guha2015rate}
\bibinfo{author}{Guha, S.} \emph{et~al.}
\newblock \bibinfo{title}{Rate-loss analysis of an efficient quantum repeater
  architecture}.
\newblock \emph{\bibinfo{journal}{Phys. Rev. A}} \textbf{\bibinfo{volume}{92}},
  \bibinfo{pages}{022357} (\bibinfo{year}{2015}).

\bibitem{hartmann2007role}
\bibinfo{author}{Hartmann, L.}, \bibinfo{author}{Kraus, B.},
  \bibinfo{author}{Briegel, H.-J.} \& \bibinfo{author}{D{\"u}r, W.}
\newblock \bibinfo{title}{Role of memory errors in quantum repeaters}.
\newblock \emph{\bibinfo{journal}{Phys. Rev. A}} \textbf{\bibinfo{volume}{75}},
  \bibinfo{pages}{032310} (\bibinfo{year}{2007}).

\bibitem{jiang2009quantum}
\bibinfo{author}{Jiang, L.} \emph{et~al.}
\newblock \bibinfo{title}{Quantum repeater with encoding}.
\newblock \emph{\bibinfo{journal}{Phys. Rev. A}} \textbf{\bibinfo{volume}{79}},
  \bibinfo{pages}{032325} (\bibinfo{year}{2009}).

\bibitem{nemoto2016photonic}
\bibinfo{author}{Nemoto, K.} \emph{et~al.}
\newblock \bibinfo{title}{Photonic quantum networks formed from nv- centers}.
\newblock \emph{\bibinfo{journal}{Sci. Rep.}} \textbf{\bibinfo{volume}{6}},
  \bibinfo{pages}{1--12} (\bibinfo{year}{2016}).

\bibitem{razavi2009quantum}
\bibinfo{author}{Razavi, M.}, \bibinfo{author}{Piani, M.} \&
  \bibinfo{author}{L{\"u}tkenhaus, N.}
\newblock \bibinfo{title}{Quantum repeaters with imperfect memories: Cost and
  scalability}.
\newblock \emph{\bibinfo{journal}{Phys. Rev. A}} \textbf{\bibinfo{volume}{80}},
  \bibinfo{pages}{032301} (\bibinfo{year}{2009}).

\bibitem{razavi2006long}
\bibinfo{author}{Razavi, M.} \& \bibinfo{author}{Shapiro, J.~H.}
\newblock \bibinfo{title}{Long-distance quantum communication with neutral
  atoms}.
\newblock \emph{\bibinfo{journal}{Phys. Rev. A}} \textbf{\bibinfo{volume}{73}},
  \bibinfo{pages}{042303} (\bibinfo{year}{2006}).

\bibitem{simon2007quantum}
\bibinfo{author}{Simon, C.} \emph{et~al.}
\newblock \bibinfo{title}{Quantum repeaters with photon pair sources and
  multimode memories}.
\newblock \emph{\bibinfo{journal}{Phys. Rev. Lett.}}
  \textbf{\bibinfo{volume}{98}}, \bibinfo{pages}{190503}
  (\bibinfo{year}{2007}).

\bibitem{vinay2017practical}
\bibinfo{author}{Vinay, S.~E.} \& \bibinfo{author}{Kok, P.}
\newblock \bibinfo{title}{Practical repeaters for ultralong-distance quantum
  communication}.
\newblock \emph{\bibinfo{journal}{Phys. Rev. A}} \textbf{\bibinfo{volume}{95}},
  \bibinfo{pages}{052336} (\bibinfo{year}{2017}).

\bibitem{wu2020near}
\bibinfo{author}{Wu, Y.}, \bibinfo{author}{Liu, J.} \& \bibinfo{author}{Simon,
  C.}
\newblock \bibinfo{title}{Near-term performance of quantum repeaters with
  imperfect ensemble-based quantum memories}.
\newblock \emph{\bibinfo{journal}{Phys. Rev. A}}
  \textbf{\bibinfo{volume}{101}}, \bibinfo{pages}{042301}
  (\bibinfo{year}{2020}).

\bibitem{sangouard2007long}
\bibinfo{author}{Sangouard, N.} \emph{et~al.}
\newblock \bibinfo{title}{Long-distance entanglement distribution with
  single-photon sources}.
\newblock \emph{\bibinfo{journal}{Phys. Rev. A}} \textbf{\bibinfo{volume}{76}},
  \bibinfo{pages}{050301} (\bibinfo{year}{2007}).

\bibitem{sangouard2008robust}
\bibinfo{author}{Sangouard, N.} \emph{et~al.}
\newblock \bibinfo{title}{Robust and efficient quantum repeaters with atomic
  ensembles and linear optics}.
\newblock \emph{\bibinfo{journal}{Phys. Rev. A}} \textbf{\bibinfo{volume}{77}},
  \bibinfo{pages}{062301} (\bibinfo{year}{2008}).

\bibitem{borregaard2020one}
\bibinfo{author}{Borregaard, J.} \emph{et~al.}
\newblock \bibinfo{title}{One-way quantum repeater based on near-deterministic
  photon-emitter interfaces}.
\newblock \emph{\bibinfo{journal}{Phys. Rev. X}} \textbf{\bibinfo{volume}{10}},
  \bibinfo{pages}{021071} (\bibinfo{year}{2020}).

\bibitem{luong2016overcoming}
\bibinfo{author}{Luong, D.}, \bibinfo{author}{Jiang, L.}, \bibinfo{author}{Kim,
  J.} \& \bibinfo{author}{L{\"u}tkenhaus, N.}
\newblock \bibinfo{title}{Overcoming lossy channel bounds using a single
  quantum repeater node}.
\newblock \emph{\bibinfo{journal}{Appl. Phys. B}}
  \textbf{\bibinfo{volume}{122}}, \bibinfo{pages}{1--10}
  (\bibinfo{year}{2016}).

\bibitem{rozpedek2018parameter}
\bibinfo{author}{Rozp{\k{e}}dek, F.} \emph{et~al.}
\newblock \bibinfo{title}{Parameter regimes for a single sequential quantum
  repeater}.
\newblock \emph{\bibinfo{journal}{Quantum Sci. Technol.}}
  \textbf{\bibinfo{volume}{3}}, \bibinfo{pages}{034002} (\bibinfo{year}{2018}).

\bibitem{vanloock2020}
\bibinfo{author}{{van Loock}, P.} \emph{et~al.}
\newblock \bibinfo{title}{Extending {{Quantum Links}}: {{Modules}} for
  {{Fiber-}} and {{Memory-Based Quantum Repeaters}}}.
\newblock \emph{\bibinfo{journal}{Adv. Quantum Technol.}}
  \textbf{\bibinfo{volume}{3}}, \bibinfo{pages}{1900141}
  (\bibinfo{year}{2020}).

\bibitem{kamin2022}
\bibinfo{author}{Kamin, L.}, \bibinfo{author}{Shchukin, E.},
  \bibinfo{author}{Schmidt, F.} \& \bibinfo{author}{{van Loock}, P.}
\newblock \bibinfo{title}{Exact rate analysis for quantum repeaters with
  imperfect memories and entanglement swapping as soon as possible}.
\newblock \emph{\bibinfo{journal}{Phys. Rev. Research}}
  \textbf{\bibinfo{volume}{5}}, \bibinfo{pages}{023086} (\bibinfo{year}{2023}).

\bibitem{abruzzo2013quantum}
\bibinfo{author}{Abruzzo, S.} \emph{et~al.}
\newblock \bibinfo{title}{Quantum repeaters and quantum key distribution:
  Analysis of secret-key rates}.
\newblock \emph{\bibinfo{journal}{Phys. Rev. A}} \textbf{\bibinfo{volume}{87}},
  \bibinfo{pages}{052315} (\bibinfo{year}{2013}).

\bibitem{brask2008memory}
\bibinfo{author}{Brask, J.~B.} \& \bibinfo{author}{S{\o}rensen, A.~S.}
\newblock \bibinfo{title}{Memory imperfections in atomic-ensemble-based quantum
  repeaters}.
\newblock \emph{\bibinfo{journal}{Phys. Rev. A}} \textbf{\bibinfo{volume}{78}},
  \bibinfo{pages}{012350} (\bibinfo{year}{2008}).

\bibitem{muralidharan2014ultrafast}
\bibinfo{author}{Muralidharan, S.}, \bibinfo{author}{Kim, J.},
  \bibinfo{author}{L{\"u}tkenhaus, N.}, \bibinfo{author}{Lukin, M.~D.} \&
  \bibinfo{author}{Jiang, L.}
\newblock \bibinfo{title}{Ultrafast and fault-tolerant quantum communication
  across long distances}.
\newblock \emph{\bibinfo{journal}{Phys. Rev. Lett.}}
  \textbf{\bibinfo{volume}{112}}, \bibinfo{pages}{250501}
  (\bibinfo{year}{2014}).

\bibitem{pant2017rate}
\bibinfo{author}{Pant, M.}, \bibinfo{author}{Krovi, H.},
  \bibinfo{author}{Englund, D.} \& \bibinfo{author}{Guha, S.}
\newblock \bibinfo{title}{Rate-distance tradeoff and resource costs for
  all-optical quantum repeaters}.
\newblock \emph{\bibinfo{journal}{Phys. Rev. A}} \textbf{\bibinfo{volume}{95}},
  \bibinfo{pages}{012304} (\bibinfo{year}{2017}).

\bibitem{ladd2006hybrid}
\bibinfo{author}{Ladd, T.~D.}, \bibinfo{author}{van Loock, P.},
  \bibinfo{author}{Nemoto, K.}, \bibinfo{author}{Munro, W.~J.} \&
  \bibinfo{author}{Yamamoto, Y.}
\newblock \bibinfo{title}{Hybrid quantum repeater based on dispersive cqed
  interactions between matter qubits and bright coherent light}.
\newblock \emph{\bibinfo{journal}{New J. Phys.}} \textbf{\bibinfo{volume}{8}},
  \bibinfo{pages}{184} (\bibinfo{year}{2006}).

\bibitem{van2006hybrid}
\bibinfo{author}{Van~Loock, P.} \emph{et~al.}
\newblock \bibinfo{title}{Hybrid quantum repeater using bright coherent light}.
\newblock \emph{\bibinfo{journal}{Phys. Rev. Lett.}}
  \textbf{\bibinfo{volume}{96}}, \bibinfo{pages}{240501}
  (\bibinfo{year}{2006}).

\bibitem{zwerger2017quantum}
\bibinfo{author}{Zwerger, M.} \emph{et~al.}
\newblock \bibinfo{title}{Quantum repeaters based on trapped ions with
  decoherence-free subspace encoding}.
\newblock \emph{\bibinfo{journal}{Quantum Sci. Technol.}}
  \textbf{\bibinfo{volume}{2}}, \bibinfo{pages}{044001} (\bibinfo{year}{2017}).

\bibitem{jiang2007fast}
\bibinfo{author}{Jiang, L.}, \bibinfo{author}{Taylor, J.} \&
  \bibinfo{author}{Lukin, M.}
\newblock \bibinfo{title}{Fast and robust approach to long-distance quantum
  communication with atomic ensembles}.
\newblock \emph{\bibinfo{journal}{Phys. Rev. A}} \textbf{\bibinfo{volume}{76}},
  \bibinfo{pages}{012301} (\bibinfo{year}{2007}).

\bibitem{wu2021b}
\bibinfo{author}{Wu, X.} \emph{et~al.}
\newblock \bibinfo{title}{{{SeQUeNCe}}: A customizable discrete-event simulator
  of quantum networks}.
\newblock \emph{\bibinfo{journal}{Quantum Sci. Technol.}}
  \textbf{\bibinfo{volume}{6}}, \bibinfo{pages}{045027} (\bibinfo{year}{2021}).

\bibitem{kalb2018dephasing}
\bibinfo{author}{Kalb, N.}, \bibinfo{author}{Humphreys, P.~C.},
  \bibinfo{author}{Slim, J.} \& \bibinfo{author}{Hanson, R.}
\newblock \bibinfo{title}{Dephasing mechanisms of diamond-based nuclear-spin
  memories for quantum networks}.
\newblock \emph{\bibinfo{journal}{Phys. Rev. A}} \textbf{\bibinfo{volume}{97}},
  \bibinfo{pages}{062330} (\bibinfo{year}{2018}).

\bibitem{da2023requirements}
\bibinfo{author}{da~Silva, F.~F.}, \bibinfo{author}{Avis, G.},
  \bibinfo{author}{Slater, J.~A.} \& \bibinfo{author}{Wehner, S.}
\newblock \bibinfo{title}{Requirements for upgrading trusted nodes to a
  repeater chain over 900 km of optical fiber}.
\newblock \bibinfo{note}{Preprint at \url{http://arXiv.org/abs/2303.03234}
  (2023)}.

\bibitem{leichtle2021verifying}
\bibinfo{author}{Leichtle, D.}, \bibinfo{author}{Music, L.},
  \bibinfo{author}{Kashefi, E.} \& \bibinfo{author}{Ollivier, H.}
\newblock \bibinfo{title}{Verifying bqp computations on noisy devices with
  minimal overhead}.
\newblock \emph{\bibinfo{journal}{PRX Quantum}} \textbf{\bibinfo{volume}{2}},
  \bibinfo{pages}{040302} (\bibinfo{year}{2021}).

\bibitem{fitzsimons2017unconditionally}
\bibinfo{author}{Fitzsimons, J.~F.} \& \bibinfo{author}{Kashefi, E.}
\newblock \bibinfo{title}{Unconditionally verifiable blind quantum
  computation}.
\newblock \emph{\bibinfo{journal}{Phys. Rev. A}} \textbf{\bibinfo{volume}{96}},
  \bibinfo{pages}{012303} (\bibinfo{year}{2017}).

\bibitem{morimae2012blind}
\bibinfo{author}{Morimae, T.} \& \bibinfo{author}{Fujii, K.}
\newblock \bibinfo{title}{Blind topological measurement-based quantum
  computation}.
\newblock \emph{\bibinfo{journal}{Nat. Commun.}} \textbf{\bibinfo{volume}{3}},
  \bibinfo{pages}{1036} (\bibinfo{year}{2012}).

\bibitem{huang2017experimental}
\bibinfo{author}{Huang, H.-L.} \emph{et~al.}
\newblock \bibinfo{title}{Experimental blind quantum computing for a classical
  client}.
\newblock \emph{\bibinfo{journal}{Phys. Rev. Lett.}}
  \textbf{\bibinfo{volume}{119}}, \bibinfo{pages}{050503}
  (\bibinfo{year}{2017}).

\bibitem{gheorghiu2015robustness}
\bibinfo{author}{Gheorghiu, A.}, \bibinfo{author}{Kashefi, E.} \&
  \bibinfo{author}{Wallden, P.}
\newblock \bibinfo{title}{Robustness and device independence of verifiable
  blind quantum computing}.
\newblock \emph{\bibinfo{journal}{New J. Phys.}} \textbf{\bibinfo{volume}{17}},
  \bibinfo{pages}{083040} (\bibinfo{year}{2015}).

\bibitem{dunjko2012blind}
\bibinfo{author}{Dunjko, V.}, \bibinfo{author}{Kashefi, E.} \&
  \bibinfo{author}{Leverrier, A.}
\newblock \bibinfo{title}{Blind quantum computing with weak coherent pulses}.
\newblock \emph{\bibinfo{journal}{Phys. Rev. Lett.}}
  \textbf{\bibinfo{volume}{108}}, \bibinfo{pages}{200502}
  (\bibinfo{year}{2012}).

\bibitem{broadbent2009universal}
\bibinfo{author}{Broadbent, A.}, \bibinfo{author}{Fitzsimons, J.} \&
  \bibinfo{author}{Kashefi, E.}
\newblock \bibinfo{title}{Universal blind quantum computation}.
\newblock In \emph{\bibinfo{booktitle}{2009 50th Annual IEEE Symposium on
  Foundations of Computer Science}}, \bibinfo{pages}{517--526}
  (\bibinfo{organization}{IEEE}, \bibinfo{year}{2009}).

\bibitem{barz2012demonstration}
\bibinfo{author}{Barz, S.} \emph{et~al.}
\newblock \bibinfo{title}{Demonstration of blind quantum computing}.
\newblock \emph{\bibinfo{journal}{Science}} \textbf{\bibinfo{volume}{335}},
  \bibinfo{pages}{303--308} (\bibinfo{year}{2012}).

\bibitem{bennett2001remote}
\bibinfo{author}{Bennett, C.~H.} \emph{et~al.}
\newblock \bibinfo{title}{Remote state preparation}.
\newblock \emph{\bibinfo{journal}{Phys. Rev. Lett.}}
  \textbf{\bibinfo{volume}{87}}, \bibinfo{pages}{077902}
  (\bibinfo{year}{2001}).

\bibitem{bernien2013heralded}
\bibinfo{author}{Bernien, H.} \emph{et~al.}
\newblock \bibinfo{title}{Heralded entanglement between solid-state qubits
  separated by three metres}.
\newblock \emph{\bibinfo{journal}{Nature}} \textbf{\bibinfo{volume}{497}},
  \bibinfo{pages}{86--90} (\bibinfo{year}{2013}).

\bibitem{hensen2015loophole}
\bibinfo{author}{Hensen, B.} \emph{et~al.}
\newblock \bibinfo{title}{Loophole-free bell inequality violation using
  electron spins separated by 1.3 kilometres}.
\newblock \emph{\bibinfo{journal}{Nature}} \textbf{\bibinfo{volume}{526}},
  \bibinfo{pages}{682--686} (\bibinfo{year}{2015}).

\bibitem{kalb2017entanglement}
\bibinfo{author}{Kalb, N.} \emph{et~al.}
\newblock \bibinfo{title}{Entanglement distillation between solid-state quantum
  network nodes}.
\newblock \emph{\bibinfo{journal}{Science}} \textbf{\bibinfo{volume}{356}},
  \bibinfo{pages}{928--932} (\bibinfo{year}{2017}).

\bibitem{humphreys2018deterministic}
\bibinfo{author}{Humphreys, P.~C.} \emph{et~al.}
\newblock \bibinfo{title}{Deterministic delivery of remote entanglement on a
  quantum network}.
\newblock \emph{\bibinfo{journal}{Nature}} \textbf{\bibinfo{volume}{558}},
  \bibinfo{pages}{268--273} (\bibinfo{year}{2018}).

\bibitem{hermans2022qubit}
\bibinfo{author}{Hermans, S.} \emph{et~al.}
\newblock \bibinfo{title}{Qubit teleportation between non-neighbouring nodes in
  a quantum network}.
\newblock \emph{\bibinfo{journal}{Nature}} \textbf{\bibinfo{volume}{605}},
  \bibinfo{pages}{663--668} (\bibinfo{year}{2022}).

\bibitem{abobeih2018one}
\bibinfo{author}{Abobeih, M.~H.} \emph{et~al.}
\newblock \bibinfo{title}{One-second coherence for a single electron spin
  coupled to a multi-qubit nuclear-spin environment}.
\newblock \emph{\bibinfo{journal}{Nat. Commun.}} \textbf{\bibinfo{volume}{9}},
  \bibinfo{pages}{1--8} (\bibinfo{year}{2018}).

\bibitem{bradley2019ten}
\bibinfo{author}{Bradley, C.} \emph{et~al.}
\newblock \bibinfo{title}{A ten-qubit solid-state spin register with quantum
  memory up to one minute}.
\newblock \emph{\bibinfo{journal}{Phys. Rev. X}} \textbf{\bibinfo{volume}{9}},
  \bibinfo{pages}{031045} (\bibinfo{year}{2019}).

\bibitem{tirepeater}
\bibinfo{author}{Krutyanskiy, V.} \emph{et~al.}
\newblock \bibinfo{title}{Telecom-{{Wavelength Quantum Repeater Node Based}} on
  a {{Trapped-Ion Processor}}}.
\newblock \emph{\bibinfo{journal}{Phys. Rev. Lett.}}
  \textbf{\bibinfo{volume}{130}}, \bibinfo{pages}{213601}
  (\bibinfo{year}{2023}).

\bibitem{krutyanskiy2023}
\bibinfo{author}{Krutyanskiy, V.} \emph{et~al.}
\newblock \bibinfo{title}{Entanglement of trapped-ion qubits separated by 230
  meters}.
\newblock \emph{\bibinfo{journal}{Phys. Rev. Lett.}}
  \textbf{\bibinfo{volume}{130}}, \bibinfo{pages}{050803}
  (\bibinfo{year}{2023}).

\bibitem{krutyanskiy2019light}
\bibinfo{author}{Krutyanskiy, V.} \emph{et~al.}
\newblock \bibinfo{title}{Light-matter entanglement over 50 km of optical
  fibre}.
\newblock \emph{\bibinfo{journal}{npj Quantum Inf.}}
  \textbf{\bibinfo{volume}{5}}, \bibinfo{pages}{1--5} (\bibinfo{year}{2019}).

\bibitem{schupp_interface_2021}
\bibinfo{author}{Schupp, J.} \emph{et~al.}
\newblock \bibinfo{title}{Interface between {Trapped}-{Ion} {Qubits} and
  {Traveling} {Photons} with {Close}-to-{Optimal} {Efficiency}}.
\newblock \emph{\bibinfo{journal}{PRX Quantum}} \textbf{\bibinfo{volume}{2}},
  \bibinfo{pages}{020331} (\bibinfo{year}{2021}).

\bibitem{krutyanskiy2017}
\bibinfo{author}{Krutyanskiy, V.}, \bibinfo{author}{Meraner, M.},
  \bibinfo{author}{Schupp, J.} \& \bibinfo{author}{Lanyon, B.~P.}
\newblock \bibinfo{title}{Polarisation-preserving photon frequency conversion
  from a trapped-ion-compatible wavelength to the telecom {{C-band}}}.
\newblock \emph{\bibinfo{journal}{Appl. Phys. B}}
  \textbf{\bibinfo{volume}{123}}, \bibinfo{pages}{228} (\bibinfo{year}{2017}).

\bibitem{myerson2008}
\bibinfo{author}{Myerson, A.~H.} \emph{et~al.}
\newblock \bibinfo{title}{High-{{Fidelity Readout}} of {{Trapped-Ion Qubits}}}.
\newblock \emph{\bibinfo{journal}{Phys. Rev. Lett.}}
  \textbf{\bibinfo{volume}{100}}, \bibinfo{pages}{200502}
  (\bibinfo{year}{2008}).

\bibitem{roos2006}
\bibinfo{author}{Roos, C.~F.}, \bibinfo{author}{Chwalla, M.},
  \bibinfo{author}{Kim, K.}, \bibinfo{author}{Riebe, M.} \&
  \bibinfo{author}{Blatt, R.}
\newblock \bibinfo{title}{`{{Designer}} atoms' for quantum metrology}.
\newblock \emph{\bibinfo{journal}{Nature}} \textbf{\bibinfo{volume}{443}},
  \bibinfo{pages}{316--319} (\bibinfo{year}{2006}).

\bibitem{tiprivate}
\bibinfo{author}{Baier, S.}, \bibinfo{author}{Galli, M.},
  \bibinfo{author}{Krutyanskii, V.}, \bibinfo{author}{Lanyon, B.} \&
  \bibinfo{author}{Northup, T.}
\newblock \bibinfo{howpublished}{private communications}
  (\bibinfo{year}{2022}).

\bibitem{coopmans2021netsquid}
\bibinfo{author}{Coopmans, T.} \emph{et~al.}
\newblock \bibinfo{title}{Netsquid, a network simulator for quantum information
  using discrete events}.
\newblock \emph{\bibinfo{journal}{Communications Physics}}
  \textbf{\bibinfo{volume}{4}}, \bibinfo{pages}{1--15} (\bibinfo{year}{2021}).

\bibitem{ruf2021resonant}
\bibinfo{author}{Ruf, M.}, \bibinfo{author}{Weaver, M.~J.},
  \bibinfo{author}{van Dam, S.~B.} \& \bibinfo{author}{Hanson, R.}
\newblock \bibinfo{title}{Resonant excitation and purcell enhancement of
  coherent nitrogen-vacancy centers coupled to a fabry-perot microcavity}.
\newblock \emph{\bibinfo{journal}{Phys. Rev. Applied}}
  \textbf{\bibinfo{volume}{15}}, \bibinfo{pages}{024049}
  (\bibinfo{year}{2021}).

\bibitem{schindlerQuantumInformationProcessor2013}
\bibinfo{author}{Schindler, P.} \emph{et~al.}
\newblock \bibinfo{title}{A quantum information processor with trapped ions}.
\newblock \emph{\bibinfo{journal}{New J. Phys.}} \textbf{\bibinfo{volume}{15}},
  \bibinfo{pages}{123012} (\bibinfo{year}{2013}).

\bibitem{molmer1999multiparticle}
\bibinfo{author}{M{\o}lmer, K.} \& \bibinfo{author}{S{\o}rensen, A.}
\newblock \bibinfo{title}{Multiparticle entanglement of hot trapped ions}.
\newblock \emph{\bibinfo{journal}{Phys. Rev. Lett.}}
  \textbf{\bibinfo{volume}{82}}, \bibinfo{pages}{1835} (\bibinfo{year}{1999}).

\bibitem{cabrillo1999creation}
\bibinfo{author}{Cabrillo, C.}, \bibinfo{author}{Cirac, J.~I.},
  \bibinfo{author}{Garcia-Fernandez, P.} \& \bibinfo{author}{Zoller, P.}
\newblock \bibinfo{title}{Creation of entangled states of distant atoms by
  interference}.
\newblock \emph{\bibinfo{journal}{Phys. Rev. A}} \textbf{\bibinfo{volume}{59}},
  \bibinfo{pages}{1025} (\bibinfo{year}{1999}).

\bibitem{barrett2005efficient}
\bibinfo{author}{Barrett, S.~D.} \& \bibinfo{author}{Kok, P.}
\newblock \bibinfo{title}{Efficient high-fidelity quantum computation using
  matter qubits and linear optics}.
\newblock \emph{\bibinfo{journal}{Phys. Rev. A}} \textbf{\bibinfo{volume}{71}},
  \bibinfo{pages}{060310} (\bibinfo{year}{2005}).

\bibitem{vardoyan2022quantum}
\bibinfo{author}{Vardoyan, G.}, \bibinfo{author}{Skrzypczyk, M.} \&
  \bibinfo{author}{Wehner, S.}
\newblock \bibinfo{title}{On the quantum performance evaluation of two
  distributed quantum architectures}.
\newblock \emph{\bibinfo{journal}{Performance Evaluation}}
  \textbf{\bibinfo{volume}{153}}, \bibinfo{pages}{102242}
  (\bibinfo{year}{2022}).

\bibitem{horodeckiGeneralTeleportationChannel1999}
\bibinfo{author}{Horodecki, M.}, \bibinfo{author}{Horodecki, P.} \&
  \bibinfo{author}{Horodecki, R.}
\newblock \bibinfo{title}{General teleportation channel, singlet fraction, and
  quasidistillation}.
\newblock \emph{\bibinfo{journal}{Phys. Rev. A}} \textbf{\bibinfo{volume}{60}},
  \bibinfo{pages}{1888--1898} (\bibinfo{year}{1999}).

\bibitem{jiang2007optimal}
\bibinfo{author}{Jiang, L.}, \bibinfo{author}{Taylor, J.~M.},
  \bibinfo{author}{Khaneja, N.} \& \bibinfo{author}{Lukin, M.~D.}
\newblock \bibinfo{title}{Optimal approach to quantum communication using
  dynamic programming}.
\newblock \emph{\bibinfo{journal}{Proceedings of the National Academy of
  Sciences}} \textbf{\bibinfo{volume}{104}}, \bibinfo{pages}{17291--17296}
  (\bibinfo{year}{2007}).

\bibitem{coopmans2022improved}
\bibinfo{author}{Coopmans, T.}, \bibinfo{author}{Brand, S.} \&
  \bibinfo{author}{Elkouss, D.}
\newblock \bibinfo{title}{Improved analytical bounds on delivery times of
  long-distance entanglement}.
\newblock \emph{\bibinfo{journal}{Phys. Rev. A}}
  \textbf{\bibinfo{volume}{105}}, \bibinfo{pages}{012608}
  (\bibinfo{year}{2022}).

\bibitem{dur2007entanglement}
\bibinfo{author}{D{\"u}r, W.} \& \bibinfo{author}{Briegel, H.~J.}
\newblock \bibinfo{title}{Entanglement purification and quantum error
  correction}.
\newblock \emph{\bibinfo{journal}{Rep. Prog. Phys.}}
  \textbf{\bibinfo{volume}{70}}, \bibinfo{pages}{1381} (\bibinfo{year}{2007}).

\bibitem{hongMeasurementSubpicosecondTime1987}
\bibinfo{author}{Hong, C.~K.}, \bibinfo{author}{Ou, Z.~Y.} \&
  \bibinfo{author}{Mandel, L.}
\newblock \bibinfo{title}{Measurement of subpicosecond time intervals between
  two photons by interference}.
\newblock \emph{\bibinfo{journal}{Phys. Rev. Lett.}}
  \textbf{\bibinfo{volume}{59}}, \bibinfo{pages}{2044--2046}
  (\bibinfo{year}{1987}).

\bibitem{bouchard2020}
\bibinfo{author}{Bouchard, F.} \emph{et~al.}
\newblock \bibinfo{title}{Two-photon interference: The
  {{Hong}}–{{Ou}}–{{Mandel}} effect}.
\newblock \emph{\bibinfo{journal}{Rep. Prog. Phys.}}
  \textbf{\bibinfo{volume}{84}}, \bibinfo{pages}{012402}
  (\bibinfo{year}{2020}).

\bibitem{pfaff2014unconditional}
\bibinfo{author}{Pfaff, W.} \emph{et~al.}
\newblock \bibinfo{title}{Unconditional quantum teleportation between distant
  solid-state quantum bits}.
\newblock \emph{\bibinfo{journal}{Science}} \textbf{\bibinfo{volume}{345}},
  \bibinfo{pages}{532--535} (\bibinfo{year}{2014}).

\bibitem{stute_tunable_2012}
\bibinfo{author}{Stute, A.} \emph{et~al.}
\newblock \bibinfo{title}{Tunable ion–photon entanglement in an optical
  cavity}.
\newblock \emph{\bibinfo{journal}{Nature}} \textbf{\bibinfo{volume}{485}},
  \bibinfo{pages}{482--485} (\bibinfo{year}{2012}).

\bibitem{dahlberg2019link}
\bibinfo{author}{Dahlberg, A.} \emph{et~al.}
\newblock \bibinfo{title}{A link layer protocol for quantum networks}.
\newblock In \emph{\bibinfo{booktitle}{Proceedings of the {{ACM Special
  Interest Group}} on {{Data Communication}}}}, {{SIGCOMM}} '19,
  \bibinfo{pages}{159--173} (\bibinfo{publisher}{{Association for Computing
  Machinery}}, \bibinfo{address}{{New York, NY, USA}}, \bibinfo{year}{2019}).

\bibitem{da2021optimizing}
\bibinfo{author}{da~Silva, F.~F.}, \bibinfo{author}{{Torres-Knoop}, A.},
  \bibinfo{author}{Coopmans, T.}, \bibinfo{author}{Maier, D.} \&
  \bibinfo{author}{Wehner, S.}
\newblock \bibinfo{title}{Optimizing entanglement generation and distribution
  using genetic algorithms}.
\newblock \emph{\bibinfo{journal}{Quantum Sci. Technol.}}
  \textbf{\bibinfo{volume}{6}}, \bibinfo{pages}{035007} (\bibinfo{year}{2021}).

\bibitem{rabbie2022designing}
\bibinfo{author}{Rabbie, J.}, \bibinfo{author}{Chakraborty, K.},
  \bibinfo{author}{Avis, G.} \& \bibinfo{author}{Wehner, S.}
\newblock \bibinfo{title}{Designing quantum networks using preexisting
  infrastructure}.
\newblock \emph{\bibinfo{journal}{npj Quantum Inf.}}
  \textbf{\bibinfo{volume}{8}}, \bibinfo{pages}{1--12} (\bibinfo{year}{2022}).

\bibitem{cramer2016repeated}
\bibinfo{author}{Cramer, J.} \emph{et~al.}
\newblock \bibinfo{title}{Repeated quantum error correction on a continuously
  encoded qubit by real-time feedback}.
\newblock \emph{\bibinfo{journal}{Nat. Commun.}} \textbf{\bibinfo{volume}{7}},
  \bibinfo{pages}{1--7} (\bibinfo{year}{2016}).

\bibitem{taminiau2014universal}
\bibinfo{author}{Taminiau, T.~H.}, \bibinfo{author}{Cramer, J.},
  \bibinfo{author}{van~der Sar, T.}, \bibinfo{author}{Dobrovitski, V.~V.} \&
  \bibinfo{author}{Hanson, R.}
\newblock \bibinfo{title}{Universal control and error correction in multi-qubit
  spin registers in diamond}.
\newblock \emph{\bibinfo{journal}{Nat. Nanotechnol.}}
  \textbf{\bibinfo{volume}{9}}, \bibinfo{pages}{171--176}
  (\bibinfo{year}{2014}).

\bibitem{reiserer2016robust}
\bibinfo{author}{Reiserer, A.} \emph{et~al.}
\newblock \bibinfo{title}{Robust quantum-network memory using
  decoherence-protected subspaces of nuclear spins}.
\newblock \emph{\bibinfo{journal}{Phys. Rev. X}} \textbf{\bibinfo{volume}{6}},
  \bibinfo{pages}{021040} (\bibinfo{year}{2016}).

\bibitem{barrosDeterministicSinglephotonSource2009}
\bibinfo{author}{Barros, H.~G.} \emph{et~al.}
\newblock \bibinfo{title}{Deterministic single-photon source from a single
  ion}.
\newblock \emph{\bibinfo{journal}{New J. Phys.}} \textbf{\bibinfo{volume}{11}},
  \bibinfo{pages}{103004} (\bibinfo{year}{2009}).

\bibitem{casaboneHeraldedEntanglementTwo2013}
\bibinfo{author}{Casabone, B.} \emph{et~al.}
\newblock \bibinfo{title}{Heralded {{Entanglement}} of {{Two Ions}} in an
  {{Optical Cavity}}}.
\newblock \emph{\bibinfo{journal}{Phys. Rev. Lett.}}
  \textbf{\bibinfo{volume}{111}}, \bibinfo{pages}{100505}
  (\bibinfo{year}{2013}).

\bibitem{kellerContinuousGenerationSingle2004}
\bibinfo{author}{Keller, M.}, \bibinfo{author}{Lange, B.},
  \bibinfo{author}{Hayasaka, K.}, \bibinfo{author}{Lange, W.} \&
  \bibinfo{author}{Walther, H.}
\newblock \bibinfo{title}{Continuous generation of single photons with
  controlled waveform in an ion-trap cavity system}.
\newblock \emph{\bibinfo{journal}{Nature}} \textbf{\bibinfo{volume}{431}},
  \bibinfo{pages}{1075--1078} (\bibinfo{year}{2004}).

\bibitem{meraner_indistinguishable_2020}
\bibinfo{author}{Meraner, M.} \emph{et~al.}
\newblock \bibinfo{title}{Indistinguishable photons from a trapped-ion quantum
  network node}.
\newblock \emph{\bibinfo{journal}{Phys. Rev. A}}
  \textbf{\bibinfo{volume}{102}}, \bibinfo{pages}{052614}
  (\bibinfo{year}{2020}).

\bibitem{stuteIonPhotonQuantum2012}
\bibinfo{author}{Stute, A.} \emph{et~al.}
\newblock \bibinfo{title}{Toward an ion\textendash photon quantum interface in
  an optical cavity}.
\newblock \emph{\bibinfo{journal}{Appl. Phys. B}}
  \textbf{\bibinfo{volume}{107}}, \bibinfo{pages}{1145--1157}
  (\bibinfo{year}{2012}).

\bibitem{walker2020}
\bibinfo{author}{Walker, T.}, \bibinfo{author}{Kashanian, S.~V.},
  \bibinfo{author}{Ward, T.} \& \bibinfo{author}{Keller, M.}
\newblock \bibinfo{title}{Improving the indistinguishability of single photons
  from an ion-cavity system}.
\newblock \emph{\bibinfo{journal}{Phys. Rev. A}}
  \textbf{\bibinfo{volume}{102}}, \bibinfo{pages}{032616}
  (\bibinfo{year}{2020}).

\bibitem{walker2018}
\bibinfo{author}{Walker, T.} \emph{et~al.}
\newblock \bibinfo{title}{Long-{{Distance Single Photon Transmission}} from a
  {{Trapped Ion}} via {{Quantum Frequency Conversion}}}.
\newblock \emph{\bibinfo{journal}{Phys. Rev. Lett.}}
  \textbf{\bibinfo{volume}{120}}, \bibinfo{pages}{203601}
  (\bibinfo{year}{2018}).

\bibitem{stephensonHighRateHighFidelityEntanglement2020}
\bibinfo{author}{Stephenson, L.~J.} \emph{et~al.}
\newblock \bibinfo{title}{High-{{Rate}}, {{High-Fidelity Entanglement}} of
  {{Qubits Across}} an {{Elementary Quantum Network}}}.
\newblock \emph{\bibinfo{journal}{Phys. Rev. Lett.}}
  \textbf{\bibinfo{volume}{124}}, \bibinfo{pages}{110501}
  (\bibinfo{year}{2020}).

\bibitem{vanleent2021}
\bibinfo{author}{{van Leent}, T.} \emph{et~al.}
\newblock \bibinfo{title}{Entangling single atoms over 33 km telecom fibre}.
\newblock \emph{\bibinfo{journal}{Nature}} \textbf{\bibinfo{volume}{607}},
  \bibinfo{pages}{69--73} (\bibinfo{year}{2022}).

\bibitem{crocker2019a}
\bibinfo{author}{Crocker, C.} \emph{et~al.}
\newblock \bibinfo{title}{High purity single photons entangled with an atomic
  qubit}.
\newblock \emph{\bibinfo{journal}{Opt. Express}} \textbf{\bibinfo{volume}{27}},
  \bibinfo{pages}{28143--28149} (\bibinfo{year}{2019}).

\bibitem{inlek2017}
\bibinfo{author}{Inlek, I.~V.}, \bibinfo{author}{Crocker, C.},
  \bibinfo{author}{Lichtman, M.}, \bibinfo{author}{Sosnova, K.} \&
  \bibinfo{author}{Monroe, C.}
\newblock \bibinfo{title}{Multispecies {{Trapped-Ion Node}} for {{Quantum
  Networking}}}.
\newblock \emph{\bibinfo{journal}{Phys. Rev. Lett.}}
  \textbf{\bibinfo{volume}{118}}, \bibinfo{pages}{250502}
  (\bibinfo{year}{2017}).

\bibitem{nadlinger2021}
\bibinfo{author}{Nadlinger, D.~P.} \emph{et~al.}
\newblock \bibinfo{title}{Device-{{Independent Quantum Key Distribution}}}.
\newblock \emph{\bibinfo{journal}{Nature}} \textbf{\bibinfo{volume}{607}},
  \bibinfo{pages}{682--686} (\bibinfo{year}{2022}).

\bibitem{netsquid-magic}
\bibinfo{title}{{{NetSquid-Magic}}}.
\newblock
  \bibinfo{howpublished}{\url{https://gitlab.com/softwarequtech/netsquid-snippets/netsquid-magic}}.

\bibitem{watrous2018}
\bibinfo{author}{Watrous, J.}
\newblock \emph{\bibinfo{title}{The {{Theory}} of {{Quantum Information}}}}
  (\bibinfo{publisher}{{Cambridge University Press}}, \bibinfo{year}{2018}),
  \bibinfo{edition}{1} edn.

\bibitem{bennet2001remote}
\bibinfo{author}{Bennett, C.~H.} \emph{et~al.}
\newblock \bibinfo{title}{Remote state preparation}.
\newblock \emph{\bibinfo{journal}{Phys. Rev. Lett.}}
  \textbf{\bibinfo{volume}{87}}, \bibinfo{pages}{077902}
  (\bibinfo{year}{2001}).

\bibitem{jensen}
\bibinfo{author}{Jensen, J. L. W.~V.}
\newblock \bibinfo{title}{{Sur les fonctions convexes et les in\'egalit\'es
  entre les valeurs moyennes}}.
\newblock \emph{\bibinfo{journal}{Acta Math.}} \textbf{\bibinfo{volume}{30}},
  \bibinfo{pages}{175--193} (\bibinfo{year}{1906}).

\bibitem{bowdreyFidelitySingleQubit2002}
\bibinfo{author}{Bowdrey, M.~D.}, \bibinfo{author}{Oi, D. K.~L.},
  \bibinfo{author}{Short, A.}, \bibinfo{author}{Banaszek, K.} \&
  \bibinfo{author}{Jones, J.}
\newblock \bibinfo{title}{Fidelity of single qubit maps}.
\newblock \emph{\bibinfo{journal}{Phys. Lett. A}}
  \textbf{\bibinfo{volume}{294}}, \bibinfo{pages}{258--260}
  (\bibinfo{year}{2002}).

\bibitem{hansson2008we}
\bibinfo{author}{Hansson, S.~O.}
\newblock \bibinfo{title}{Do we need second-order probabilities?}
\newblock \emph{\bibinfo{journal}{Dialectica}} \textbf{\bibinfo{volume}{62}},
  \bibinfo{pages}{525--533} (\bibinfo{year}{2008}).

\bibitem{buzek1999}
\bibinfo{author}{Bu{\v z}ek, V.}, \bibinfo{author}{Hillery, M.} \&
  \bibinfo{author}{Werner, R.~F.}
\newblock \bibinfo{title}{Optimal manipulations with qubits: {{Universal-NOT}}
  gate}.
\newblock \emph{\bibinfo{journal}{Phys. Rev. A}} \textbf{\bibinfo{volume}{60}},
  \bibinfo{pages}{R2626--R2629} (\bibinfo{year}{1999}).

\bibitem{klappenecker2005}
\bibinfo{author}{Klappenecker, A.} \& \bibinfo{author}{Rotteler, M.}
\newblock \bibinfo{title}{Mutually unbiased bases are complex projective
  2-designs}.
\newblock In \emph{\bibinfo{booktitle}{Proceedings. {{International Symposium}}
  on {{Information Theory}}, 2005. {{ISIT}} 2005.}},
  \bibinfo{pages}{1740--1744} (\bibinfo{year}{2005}).

\bibitem{barrettEfficientHighfidelityQuantum2005}
\bibinfo{author}{Barrett, S.~D.} \& \bibinfo{author}{Kok, P.}
\newblock \bibinfo{title}{Efficient high-fidelity quantum computation using
  matter qubits and linear optics}.
\newblock \emph{\bibinfo{journal}{Phys. Rev. A}} \textbf{\bibinfo{volume}{71}},
  \bibinfo{pages}{060310} (\bibinfo{year}{2005}).

\bibitem{sympy}
\bibinfo{author}{Meurer, A.} \emph{et~al.}
\newblock \bibinfo{title}{Sympy: symbolic computing in python}.
\newblock \emph{\bibinfo{journal}{PeerJ Comput. Sci.}}
  \textbf{\bibinfo{volume}{3}}, \bibinfo{pages}{e103} (\bibinfo{year}{2017}).

\bibitem{delft_eindhoven_code}
\bibinfo{author}{Avis, G.} \emph{et~al.}
\newblock \bibinfo{title}{Simulation code for {{Requirements}} for a
  processing-node quantum repeater on a real-world fiber grid · {{GitLab}}}.
\newblock
  \bibinfo{howpublished}{\url{https://gitlab.com/softwarequtech/simulation-code-for-requirements-for-a-processing-node-quantum-repeater-on-a-real-world-fiber-grid}}.

\bibitem{kambs2018}
\bibinfo{author}{Kambs, B.} \& \bibinfo{author}{Becher, C.}
\newblock \bibinfo{title}{Limitations on the indistinguishability of photons
  from remote solid state sources}.
\newblock \emph{\bibinfo{journal}{New J. Phys.}} \textbf{\bibinfo{volume}{20}},
  \bibinfo{pages}{115003} (\bibinfo{year}{2018}).

\bibitem{fioretto2020}
\bibinfo{author}{Fioretto, D.~A.}
\newblock \emph{\bibinfo{title}{Towards a Flexible Source for Indistinguishable
  Photons Based on Trapped Ions and Cavities}}.
\newblock Ph.D. thesis, \bibinfo{school}{Innsbruck} (\bibinfo{year}{2020}).
\newblock \urlprefix\url{http://diglib.uibk.ac.at/ulbtirolhs/5459504}.

\bibitem{smart-stopos}
\bibinfo{title}{{{Smart-Stopos}}}.
\newblock
  \bibinfo{howpublished}{\url{https://gitlab.com/aritoka/smart-stopos}}.

\bibitem{jain2001termination}
\bibinfo{author}{Jain, B.~J.}, \bibinfo{author}{Pohlheim, H.} \&
  \bibinfo{author}{Wegener, J.}
\newblock \bibinfo{title}{On termination criteria of evolutionary algorithms}.
\newblock In \emph{\bibinfo{booktitle}{Proceedings of the 3rd Annual Conference
  on Genetic and Evolutionary Computation}}, \bibinfo{pages}{768--768}
  (\bibinfo{year}{2001}).

\bibitem{labay2021genetic}
\bibinfo{author}{Labay~Mora, A.}
\newblock \bibinfo{title}{Genetic algorithm-based optimisation of entanglement
  distribution to minimise hardware cost} (\bibinfo{year}{2021}).
\newblock
  \urlprefix\url{http://resolver.tudelft.nl/uuid:5dd40a56-8c8d-4766-a2fe-0a8c45e1ee3f}.

\bibitem{netsquid-netconf}
\bibinfo{title}{{{NetSquid-NetConf}}}.
\newblock
  \bibinfo{howpublished}{\url{https://gitlab.com/softwarequtech/netsquid-snippets/netsquid-netconf}}.

\bibitem{netsquid-nv}
\bibinfo{title}{{{NetSquid-NV}}}.
\newblock
  \bibinfo{howpublished}{\url{https://gitlab.com/softwarequtech/netsquid-snippets/netsquid-nv}}.

\bibitem{netsquid-physlayer}
\bibinfo{title}{{{NetSquid-PhysLayer}}}.
\newblock
  \bibinfo{howpublished}{\url{https://gitlab.com/softwarequtech/netsquid-snippets/netsquid-physlayer}}.

\bibitem{netsquid-trappedions}
\bibinfo{title}{{{NetSquid-TrappedIons}}}.
\newblock
  \bibinfo{howpublished}{\url{https://gitlab.com/softwarequtech/netsquid-snippets/netsquid-trappedions}}.

\bibitem{netsquid-simulationtools}
\bibinfo{title}{{{NetSquid-SimulationTools}}}.
\newblock
  \bibinfo{howpublished}{\url{https://gitlab.com/softwarequtech/netsquid-snippets/netsquid-simulationtools}}.

\bibitem{pompiliExperimentalDemonstrationEntanglement2021}
\bibinfo{author}{Pompili, M.} \emph{et~al.}
\newblock \bibinfo{title}{Experimental demonstration of entanglement delivery
  using a quantum network stack}.
\newblock \emph{\bibinfo{journal}{npj Quantum Inf.}}
  \textbf{\bibinfo{volume}{8}}, \bibinfo{pages}{1--10} (\bibinfo{year}{2022}).

\bibitem{briegelQuantumRepeatersRole1998}
\bibinfo{author}{Briegel, H.-J.}, \bibinfo{author}{D{\"u}r, W.},
  \bibinfo{author}{Cirac, J.~I.} \& \bibinfo{author}{Zoller, P.}
\newblock \bibinfo{title}{Quantum {{Repeaters}}: {{The Role}} of {{Imperfect
  Local Operations}} in {{Quantum Communication}}}.
\newblock \emph{\bibinfo{journal}{Phys. Rev. Lett.}}
  \textbf{\bibinfo{volume}{81}}, \bibinfo{pages}{5932--5935}
  (\bibinfo{year}{1998}).

\bibitem{netsquid-abstractmodel}
\bibinfo{title}{{{NetSquid-AbstractModel}}}.
\newblock
  \bibinfo{howpublished}{\url{https://gitlab.com/softwarequtech/netsquid-snippets/netsquid-abstractmodel}}.

\end{thebibliography}
